\documentclass{dalthesis} 
\usepackage{bookmark}

\input{thesis.sty}

\begin{document}

\title{Algebraic and Logical methods in Quantum Computation}

\author{Neil J. Ross}

\phd \degree{Doctor of Philosophy}
\degreeinitial{Ph.D.}

\faculty{Science} \dept{Mathematics and Statistics}

\supervisor{Peter Selinger}

\reader{Dorette Pronk} \reader{Richard J. Wood}

\nolistoftables

\defenceday{15}\defencemonth{August}\defenceyear{2015}
\convocation{October}{2015}

\dedicate{\`A mes ma\^itres.}

\frontmatter

% =====================================================================
\begin{abstract}

This thesis contains contributions to the theory of quantum
computation.

We first define a new method to efficiently approximate special
unitary operators. Specifically, given a special unitary $U$ and a
precision $\epsilon>0$, we show how to efficiently find a sequence of
Clifford+$V$ or Clifford+$T$ operators whose product approximates $U$
up to $\epsilon$ in the operator norm. In the general case, the length
of the approximating sequence is asymptotically optimal. If the
unitary to approximate is diagonal then our method is optimal: it
yields the shortest sequence approximating $U$ up to $\epsilon$.

Next, we introduce a mathematical formalization of a fragment of the
Quipper quantum programming language. We define a typed lambda
calculus called Proto-Quipper which formalizes a restricted but
expressive fragment of Quipper. The type system of Proto-Quipper is
based on intuitionistic linear logic and prohibits the duplication of
quantum data, in accordance with the no-cloning property of quantum
computation. We prove that Proto-Quipper is type-safe in the sense
that it enjoys the subject reduction and progress properties.

\end{abstract}

% =====================================================================
\prefacesection{List of Abbreviations and Symbols Used}

\begin{longtable}{ll}
  $\N$ & The natural numbers. \\
  $\Z$ & The integers. \\
  $\Qu$ & The rational numbers. \\
  $\R$ & The real numbers. \\
  $\Comp$ & The complex numbers. \\
  $R^\times$ & The group of units of the ring $R$. \\
  $\norm{.}$ & The norm of a scalar, vector, or matrix. \\
  $\x$ & The tensor product. \\
  $\Matrices_{m,n}(R)$ & The set of $m$ by $n$ matrices over the ring
  $R$. \\
  $\uset(n)$ & The unitary group of order $n$. \\
  $\suset(n)$ & The special unitary group of order $n$. \\
  $U\inv$ & The inverse of the matrix $U$. \\
  $U\da$ & The conjugate transpose of the matrix $U$. \\
  $\det(U)$ & The determinant of the matrix $U$. \\
  $\tr(U)$ & The trace of the matrix $U$. \\
  $K(\alpha)$ & The extension of the field $K$ by $\alpha$. \\
  $R[\alpha]$ & The extension of the ring $R$ by $\alpha$. \\
  $\O_K$ & The ring of integers of the field $K$. \\
  $(.)^\bullet$ & The bullet automorphism of $\Zomega$. \\
  $\equiv$ & Modular congruence. \\
  $\divides$ & Divisibility. \\
  $\Norm_R(.)$ & The norm for the ring of integers $R$. \\
  $a\s{y/x}$ & The renaming of $x$ by $y$ in $a$. \\
  $\FV(a)$ & The free variables of $a$. \\
  $a[b/x]$ & The substitution of $b$ for $x$ in $a$. \\
  $\to,\too$ & The one-step and multi-step $\beta$-reduction. \\
  $=_\alpha$ & The $\alpha$-equivalence. \\
  $\equiv_\beta$ & The $\beta$-equivalence. \\
  $<:$ & The subtyping relation. \\
  $[Q,L,a]$ & A closure of the quantum lambda calculus. \\
  $\up(A)$ & The uprightness of the set $A$. \\
  $\BBox(A)$ & The bounding box of the set $A$. \\
  $\area(A)$ & The area of the set $A$. \\
  $\Grid(A)$ & The grid for the set $A$. \\
  $O(.)$ & The big-$O$ notation. \\
  $\sk(D)$ & The skew of the ellipse $D$. \\
  $\bias\state$ & The bias of the state $\state$. \\
  $\sinl(.)$ & The hyperbolic sine in base $\lambda$. \\
  $\cosl(.)$ & The hyperbolic cosine in base $\lambda$. \\
  $\Repsilon$ & The $\epsilon$-region. \\
  $\Disk$ & The closed unit disk. \\
  $\Realpart(a)$ & The real part of the complex number $a$. \\
  $\pset_f(X)$ & The set of finite subsets of the set $X$. \\
  $\FQ(a)$ & The free quantum variables of a term $a$. \\
  $\finbij(X)$ & The set of finite bijections on the set $X$. \\
  $\spec_X(T)$ & An $X$-specimen for $T$. \\
  $[C,a]$ & A closure of Proto-Quipper. \\
  $\floor{.}$ & The floor function. \\
  $\cupdot$ & The disjoint union. \\
\end{longtable}

% =====================================================================
\begin{acknowledgements}

I would like to express my profound gratitude to the people who, in
one way or another, have contributed to the writing of this thesis.

I want to thank my supervisor Peter Selinger, for his insight, his
guidance, and, most importantly, his contagious passion for
mathematics. Special thanks are due to the members of my supervisory
committee, Dorette Pronk and Richard Wood, for accepting to read a
thesis which, alas, contains so few arrows. I am very grateful to
Prakash Panangaden for accepting to be my external examiner.

The staff of the Department of Mathematics and Statistics of Dalhousie
University have played a large part in making my stay in Halifax
enjoyable. I am very appreciative of their hard work.

I am greatly indebted to the many teachers I had the chance to learn
from during the past ten years. In particular, I want to acknowledge
the profound influence of Susana Berestovoy, Julien Dutant, Vincent
Homer, Jean-Baptiste Joinet, and Damiano Mazza.

I want to thank the researchers I had the pleasure to collaborate
with, notably D. Scott Alexander, Henri Chataing, Alexander S. Green,
Peter LeFanu Lumsdaine, Jonathan M. Smith, and Beno\^it Valiron.

My fellow students have always provided inspiration, support, and
laughter. I especially want to thank Abdullah Al-Shaghay, Ali
Alilooee, Chlo\'e Berruyer, Samir Blakaj, M\'even Cadet, Antonio
Chavez, Hoda Chuangpishit, Hugo F\'er\'ee, Florent Franchette, Brett
Giles, Giulio Guerrieri, Fran\c cois Guignot, Zhenyu Victor Guo,
D. Leigh Herman, Ben Hersey, Joey Mingrone, Lucas Mol, Alberto Naibo,
Mattia Petrolo, Francisco Rios, Kira Scheibelhut, Matthew Stephen,
Aur\'elien Tonneau, Antonio Vargas, Kim Whoriskey, Bian Xiaoning,
Amelia Yzaguirre, and Kevin Zatloukal.

Finally, I want to thank my friends and loved ones, far and near. My
brother and my sister. My parents, for their unwavering support. And
Kira, for her kindness.

\end{acknowledgements}

% =====================================================================
\mainmatter

% =====================================================================
\chapter{Introduction}
\label{chap:intro}

\emph{Quantum computation}, introduced in the early 1980s by Feynman
\cite{Feynman82}, is a paradigm for computation based on the laws of
quantum physics. The interest in quantum computation lies in the fact
that \emph{quantum computers} can solve certain problems more
efficiently than their \emph{classical} counterparts. Most famously,
Shor showed in 1994 that integers can be factored in polynomial time
on a quantum computer \cite{Shor}. This is in striking contrast with
the exponential running time of the best known classical
methods. Since then, many algorithms leveraging the power of quantum
computers have been introduced (e.g., \cite{Grover96}, \cite{TF},
\cite{CN}). This promised increase in efficiency has provided great
incentive to solve the theoretical and practical challenges associated
with building quantum computers.

The fundamental unit of quantum computation is the \emph{quantum bit}
or \emph{qubit}. The \emph{state} of a qubit is described by a unit
vector in the two-dimensional Hilbert space $\Comp^2$. A system of $n$
qubits is similarly described by a unit vector in the Hilbert space
\[
\Comp^{2^n}=\underbrace{\Comp^2 \x \ldots \x \Comp^2}_{n}.
\]
A computation is performed by acting on the state of a system of
qubits. This can be done in two ways. One can either apply a
\emph{unitary transformation} to the state or one can \emph{measure}
some of the qubits making up the system. A quantum algorithm describes
a sequence of unitary operations to be performed on the state of a
system of qubits, usually followed by a single final measurement, or
in some cases, measurements throughout the computation. It is
customary to represent a sequence of unitary operations in the form of
a \emph{quantum circuit}, which is built up from a basic set of
unitaries, called \emph{gates}, using composition and tensor
product. An important peculiarity of quantum computation is that the
state of a quantum system cannot in general be \emph{duplicated}. This
is the so-called \emph{no-cloning} property. This contrasts with the
situation in classical computation, where the state of a bit can be
freely copied.

Quantum algorithms are run on a physical device. Just as in classical
computing, there is a chain of successive translations, starting from
a mathematical description of an algorithm in a research paper. The
abstract algorithm is first implemented in a programming language,
which is then turned into a quantum circuit. This first circuit is
then rewritten using a finite set of basic unitaries available on the
hardware. Finally, this second circuit is rewritten according to some
error correcting scheme, which redundantly encodes the circuit to make
it robust to some of the errors that are bound to occur on the
physical machine. At this point, the abstract description of the
algorithm can be mapped to a physical system for the computation to be
effectively performed.

All of the above-mentioned steps in the execution of a quantum
algorithm raise interesting mathematical questions. For many years the
problems at the top of this translation chain, those closer to the
abstract description of the algorithm, were either overlooked or
considered adequately solved. This was justified by the fact that the
challenges of building a reliable physical quantum computer were so
far from being met that the higher-level problems seemed somewhat
irrelevant. However, as the prospect of usable quantum computers draws
nearer, the need to develop tools to define proper solutions to these
problems has become more pressing.

In this thesis, we contribute to two of the mathematical problems that
arise in the higher level of this execution phase. We first develop
methods to decompose unitaries into certain basic sets of
unitaries. Secondly, we introduce a lambda calculus which serves as a
foundation for the \emph{Quipper} quantum programming language. We now
briefly outline each of these contributions.

% ====================================================================
\section{Approximate synthesis}
\label{sect:intro-synthesis}

The \emph{unitary group of order 2}, denoted $\uset(2)$, is the group
of $2\by 2$ complex unitary matrices. The \emph{special unitary group
  of order 2}, denoted $\suset(2)$, is the subgroup of $\uset(2)$
consisting of unitary matrices of determinant 1.

Let $S\seq\uset(2)$ be a set of unitaries and $\p{S}$ be the set of
words over $S$. In the context of quantum computing, the elements of
$S$ are called \emph{single-qubit gates} and the elements of $\p{S}$
are called \emph{single-qubit circuits over $S$}. Matrix
multiplication defines a map $\mu: \p{S} \to \uset(2)$. However, we
often abuse notation and simply write $W$ for $\mu(W)$.

We think of $S\seq\uset(2)$ as the set of unitary operations that can
be performed natively on a quantum computer. Because a quantum
computer is a physical device, $S$ is finite and the set $\p{S}$ of
circuits expressible on our quantum computer is countable. In
contrast, both $\uset(2)$ and $\suset(2)$ are uncountable. Hence,
regardless of which $S$ is chosen, there will be $U\in\suset(2)$ such
that $U\notin \p{S}$. This is a fortiori true of $\uset(2)$ also. This
might be problematic, as it is often desirable to have all unitaries
at our disposal when writing quantum algorithms. This tension is
alleviated by the fact that unitaries can be \emph{approximated}.

\begin{definition}
  \label{def:distance}
  The \emph{distance} between two operators $U,W \in \uset(2)$ is
  defined as
  \[
  \norm{U-W} = \mbox{sup}\s{||Uv-Wv|| \such v\in\Comp^2 \mbox{ and }
    ||v||=1}.
  \]
\end{definition}

The notion of distance introduced in Definition~\ref{def:distance},
based on the operator norm, is adopted because the physically
observable difference between two unitaries $U$ and $W$ is a function
of $\norm{U-W}$. As a result, if $\epsilon$ is small enough, the
action of the unitaries $U$ and $W$ are observably almost
indistinguishable.  The finiteness of the gate set $S$ can therefore
be remedied if $\p{S}$ is dense in $\suset(2)$.

\begin{definition}
  \label{def:universal-U}
  A set $S\seq\uset(2)$ is \emph{universal} if for any $U\in\suset(2)$
  and any $\epsilon>0$, there exists $W\in\p{S}$ such that $\norm{U-
    W} < \epsilon$.
\end{definition}

Note that a set $S$ of gates is universal if it is dense in
$\suset(2)$, rather than in $\uset(2)$. This is due to the fact that
since a global phase has no observable effect in quantum mechanics, we
can without loss of generality focus on special unitary matrices.

If a gate set $S$ is universal, then any special unitary can be
approximated up to an arbitrarily small precision by a circuit over
$S$. However, the fact that $S$ is universal does not provide an
efficient method which, given a special unitary $U$ and a precision
$\epsilon$, allows us to construct the approximating circuit $W$. An
algorithm that performs such a task is a solution to the
\emph{approximate synthesis problem for $S$}.

\begin{problem}[Approximate synthesis problem for $S$]
  \label{prob:app-synth}
  Given a unitary $U\in\suset(2)$ and a precision $\epsilon\geq 0$,
  construct a circuit $W$ over $S$ such that $\norm{W-U}\leq\epsilon$.
\end{problem}

The approximate synthesis problem is important for quantum computing
because it significantly impacts the resources required to run a
quantum algorithm. Indeed, a quantum circuit, to be executed by a
quantum computer, must first be compiled into some universal gate
set. The complexity of the final physical circuit therefore crucially
depends on the chosen synthesis method. In view of the considerable
resources required for most quantum algorithms on interesting problem
sizes, which can require upwards of 30 trillion gates \cite{pldi}, a
universal gate set can be realistically considered for practical
quantum computing only if it comes equipped with a good synthesis
algorithm.

In this thesis, we are interested in Problem~\ref{prob:app-synth} for
two universal extensions of the \emph{Clifford group}. The Clifford
group is generated by the following gates
\[
\omega = e^{i\pi/4}, \quad H = \frac{1}{\sqrt 2}
\begin{bmatrix}
  1 & 1 \\
  1 & -1
\end{bmatrix}, \quad \mbox{ and } S =
\begin{bmatrix}
  1 & 0 \\
  0 & -i
\end{bmatrix}.
\]
Note that $\omega$ is a complex number, rather than a matrix. By a
slight abuse of notation, we write $\omega$ to denote the unitary
$\omega I$, where $I$ is the identity $2\by 2$ matrix. The Clifford
group is of great interest in quantum computation because Clifford
circuits can be fault-tolerantly implemented at very low cost in most
error-correcting schemes. For this reason, the Clifford group is often
seen as a prime candidate for practical quantum computing. However,
the Clifford group is finite and therefore not universal for quantum
computing. Moreover, Gottesman and Knill showed that Clifford circuits
can be efficiently simulated on a classical computer
\cite{Gottesman98}. It is therefore necessary to consider universal
extensions of the Clifford group.

The first extension we will consider, the \emph{Clifford+$V$} gate
set, arises by adding the following \emph{$V$-gates} to the set of
Clifford generators
\[
V_X = \frac{1}{\sqrt 5}
\begin{bmatrix}
  1 & 2i \\
  2i & 1
\end{bmatrix}, \quad V_Y = \frac{1}{\sqrt 5}
\begin{bmatrix}
  1 & 2 \\
  -2 & 1
\end{bmatrix}, \mbox{ and } V_Z = \frac{1}{\sqrt 5}
\begin{bmatrix}
  1 + 2i & 0 \\
  0 & 1-2i
\end{bmatrix}.
\]
The $V$-gates were introduced in \cite{LPS1986} and \cite{LPS1987} and
later considered in the context of approximate synthesis in
\cite{HRC2002}, \cite{BGS2013}, \cite{vsynth}, and \cite{BBG2014}.

The second extension we will consider, the \emph{Clifford+$T$} gate
set, arises by adding the following $T$-gate to the set of Clifford
generators
\[
T =
\begin{bmatrix}
  1 & 0 \\
  0 & \omega
\end{bmatrix}
\]
where $\omega=e^{i\pi/4}$ as above. The $T$ gate is the most common
extension of the Clifford group considered in the literature. One
reason for this is that the Clifford+$T$ gate set enjoys nice
error-correction properties \cite{Nielsen-Chuang}. This gate set has
often been considered in the approximate synthesis literature (see,
e.g., \cite{KMM-approx}, \cite{KMM-practical},
\cite{Selinger-newsynth}, and \cite{gridsynth}).

We evaluate an approximate synthesis algorithm with respect to its
\emph{time complexity} and its \emph{circuit complexity}. The time
complexity of an algorithm is the number of arithmetic operations it
requires to produce an approximating circuit. The circuit complexity
is the length of the produced circuit, which we identify with the
number of non-Clifford gates that appear in the circuit. This is
motivated by the high cost of error correction for non-Clifford gates.

Until recently, there were two main approaches to the approximate
synthesis problem: the ones based on exhaustive search, like Fowler's
algorithm of \cite{Fowler04}, and the ones based on geometric methods,
like the Solovay-Kitaev algorithm (\cite{Kitaev97b},
\cite{Dawson-Nielsen}, \cite{KSV2002}). The methods based on
exhaustive search achieve optimal circuit sizes but, due to their
exponential runtimes, are impractical for small $\epsilon$. In
contrast, the well-known Solovay-Kitaev algorithm has polynomial
runtime and achieves circuit sizes of $O(\log^c(1/\epsilon))$, where
$c>3$. However, the information-theoretic lower bound is
$O(\log(1/\epsilon))$, so the Solovay-Kitaev algorithm leaves ample
room for improvement.

In the past few years, number theoretic methods, and in particular
Diophantine equations, have been used to define new synthesis
algorithms. This has rejuvenated the field of quantum circuit
synthesis and significant progress has been made
(\cite{Kliuchnikov-etal}, \cite{BGS2013}, \cite{vsynth},
\cite{BBG2014}, \cite{KMM-approx}, \cite{KMM-practical},
\cite{Selinger-newsynth}, \cite{gridsynth}).

The algorithms we introduce in this thesis belong to this new kind of
number-theoretic algorithms. We give algorithms for Clifford+$V$ and
Clifford+$T$ circuits. Both algorithms are efficient (they run in
probabilistic polynomial time) and achieve near-optimal circuit
length. In certain specific cases, the algorithms are optimal. That
is, the produced circuits are the shortest approximations
possible. This solves a long-standing open problem, as no such
efficient optimal synthesis method was previously known for any gate
set.

% ====================================================================
\section{The mathematical foundations of Quipper}
\label{sect:intro-quipper}

\emph{Quipper} is a programming language for quantum computation (see
\cite{CACM}, \cite{rc}, \cite{pldi}, and \cite{QAPL}). The Quipper
language was developed in 2011--2013 in the context of a research
contract for the U.S. Intelligence Advanced Research Project Activity
(see \cite{BAA}). As part of this project, seven non-trivial
algorithms from the quantum computing literature were implemented in
Quipper (\cite{BWT}, \cite{BF}, \cite{CN}, \cite{GSE}, \cite{LS},
\cite{SV}, and \cite{TF}). I participated in the development of the
Quipper language and in the implementation of the Triangle Finding
Algorithm \cite{TF} and the Unique Shortest Vector Algorithm
\cite{SV}.

An important aspect of the Quipper language is that it acts as a
\emph{circuit description language}. This means that Quipper provides
a syntax in which to express quantum circuits. Quipper moreover
provides the ability to treat circuits as data and to manipulate them
as a whole. For example, Quipper has operators for reversing and
iterating circuits, decomposing them into gate sets, etc. This
circuit-as-data paradigm is remarkably useful for the programmer, as
it is very close to the way in which quantum algorithms are described
in the literature.

Currently, Quipper is implemented as an \emph{embedded} language
\cite{Claessen-2001}. This means that Quipper can be seen as a
collection of functions and data types within some pre-existing
\emph{host} language. Quipper's host language is \emph{Haskell}
\cite{Haskell}, a strongly-typed functional programming language. An
advantage of this embedded language approach is that it allows for
the implementation of a large-scale system without having to first
design and implement a compiler, a parser, etc. The embedded language
approach also has drawbacks, however. In particular, Quipper inherits
the type system of its host and while Haskell's type system provides
many type-safety properties, it is not in general strong enough to
ensure the full type-safety of quantum programs. In the current
Quipper implementation, it is therefore the programmer's
responsibility to ensure that quantum components are plugged together
in physically meaningful ways. This means that certain types of
programming errors will not be prevented by the compiler. In the worst
case, this may lead to ill-formed output or run-time errors.

In this thesis, we introduce a quantum programming language which we
call \emph{Proto-Quipper}. It is defined as a typed lambda calculus
and can be seen as a mathematical formalization of a fragment of
Quipper. Proto-Quipper is meant to provide a foundation for the
development of a stand-alone (i.e., non-embedded) version of
Quipper. Moreover, Proto-Quipper is designed to ``enforce the
physics'', in the sense that it detects, at compile-time,
programming errors that could lead to ill-formed or undefined
circuits. In particular, the no-cloning property of quantum
computation is enforced.

In designing the Proto-Quipper language, our approach was to start
with a limited, but still expressive, fragment of the Quipper language
and make it type safe. This fragment will serve then as a robust basis
for future language extensions. The idea is to eventually close the
gap between Proto-Quipper and Quipper by extending Proto-Quipper with
one feature at a time while retaining type safety.

Our main inspiration for the design of Proto-Quipper is the quantum
lambda calculus (see \cite{valiron04}, \cite{SeVa06}, or
\cite{SeVa09}). The quantum lambda calculus represents an ideal
starting point for the design of Proto-Quipper because it is equipped
with a type system tailored for quantum computation. However, the
quantum lambda calculus only manipulates qubits and all quantum
operations are immediately carried out on a quantum device, not stored
for symbolic manipulation. We therefore extend the quantum lambda
calculus with the minimal set of features that makes it
Quipper-like. The current version of Proto-Quipper is designed to
\begin{itemize}
\item incorporate Quipper's ability to generate and act on quantum
  circuits, and to
\item provide a linear type system to guarantee that the produced
  circuits are physically meaningful (in particular, properties like
  no-cloning are respected).
\end{itemize}
To achieve these goals, we extend the types of the quantum lambda
calculus with a type $Circ(T,U)$ of circuits, and add constant terms
to capture some of Quipper's circuit-level operations, like reversing.
We give a formal operational semantics of Proto-Quipper in terms of a
reduction relation on pairs $[C,t]$ of a term $t$ of the language and
a so-called circuit state $C$. The state $C$ represents the circuit
currently being built. The reduction is defined as a rewrite procedure
on such pairs, with the state being affected when terms involve
quantum constants.

% ====================================================================
\section{Outline}
\label{sect:intro-outline}

The thesis can be divided into three parts, whose contents are
outlined below.
\begin{itemize}
\item The first part of the thesis, which corresponds to
  chapters~\ref{chap:quantum} -- \ref{chap:lambda}, contains
  background material. Chapter~\ref{chap:quantum} is an exposition of
  the basic notions of quantum computation. Chapter~\ref{chap:nb-th}
  introduces concepts and methods from algebraic number
  theory. Chapter~\ref{chap:lambda} presents the lambda calculus as
  well as the quantum lambda calculus.
\item The second part of the thesis, which corresponds to
  chapters~\ref{chap:grid-pb} -- \ref{chap:synth-T}, contains
  algebraic contributions to quantum computation. In
  Chapter~\ref{chap:grid-pb}, we introduce and solve grid problems. In
  Chapter~\ref{chap:synth-V}, we define an algorithm to solve the
  problem of approximate synthesis of special unitaries over the
  Clifford+$V$ gate set. In Chapter~\ref{chap:synth-T} we show how the
  methods of the previous chapter can be adapted to the Clifford+$T$
  gate set.
\item The third and last part of the thesis, which corresponds to
  chapters~\ref{chap:pq} and \ref{chap:pq-safe}, contains logical
  contributions to quantum computation. In Chapter~\ref{chap:pq} we
  introduce the syntax, type system, and operational semantics of
  Proto-Quipper. In Chapter~\ref{chap:pq-safe} we prove that
  Proto-Quipper enjoys the subject reduction and progress properties.
\end{itemize}
In the interest of brevity, the introductory chapters
\ref{chap:quantum} -- \ref{chap:lambda} contain only the material that
is necessary to the subsequent chapters. In particular, most proofs
are omitted. In each of these chapters, we provide references to the
relevant literature.

% ====================================================================
\section{Contributions}
\label{sect:intro-contribs}

My original contributions are contained in chapters~\ref{chap:grid-pb}
-- \ref{chap:pq-safe} of the thesis. They have appeared in several
published papers. The algorithm for the Clifford+$V$ approximate
synthesis of unitaries and its analysis (Section~\ref{sect:grid-pb-Zi}
and Chapter~\ref{chap:synth-V}) appeared in the single-authored paper
\cite{vsynth}. The algorithm for the Clifford+$T$ approximate
synthesis of unitaries and its analysis
(Section~\ref{sect:grid-pb-Zomega} and Chapter~\ref{chap:synth-T})
appeared in the paper \cite{gridsynth}, co-authored with my supervisor
Peter Selinger, following earlier work by Selinger
\cite{Selinger-newsynth}. The results of
Section~\ref{ssect:grid-pb-Zomega-ellipses}, providing a method to
make two ellipses simultaneously upright, are my original
contribution. The other results of \cite{gridsynth} are the product of
an equal collaboration. Finally, the definition of the Proto-Quipper
language, as well as the proof of its type safety (Chapter
\ref{chap:pq} and \ref{chap:pq-safe}) appeared in the report
\cite{PQ}, co-authored with Henri Chataing and Peter Selinger. The
results in these chapters are my original work; Peter Selinger's role
was supervisory, and Henri Chataing was a summer intern whom I helped
supervise.

% =====================================================================
% Background

% ---------------------------------------------------------------------
\chapter{Quantum computation}
\label{chap:quantum}

In this chapter, we provide a brief introduction to quantum
computation. Further details can be found in the
literature. References for this material include \cite{Nielsen-Chuang}
and \cite{KLM2007}. For a concise introduction, we also refer the
reader to \cite{Selinger-TQPL} which we loosely follow here.

% ====================================================================
\section{Linear Algebra}
\label{sect:lin-alg}

We write $\N$\label{numbN} for the semiring of non-negative integers
(including 0) and $\Z$\label{numbZ} for the ring of integers. The
fields of rational numbers, real numbers, and complex numbers are
denoted by $\Qu$\label{numbQ}, $\R$\label{numbR}, and
$\Comp$\label{numbC} respectively. Recall that $\Comp=\s{a+bi \such
  a,b\in\R}$ so that the complex numbers can be identified with the
two-dimensional real plane $\R^2$. Recall moreover that if $\alpha =
a+bi$ is a complex number, then its \emph{conjugate}\label{dag} is
$\alpha\da = a-bi$.

% --------------------------------------------------------------------
\subsection{Finite dimensional Hilbert spaces}
\label{ssect:Hilbert-spaces}

We will be interested in complex vector spaces of the form $\Comp^n$
for some $n\in\N$. The integer $n$ is called the \emph{dimension} of
$\Comp^n$. The vector space $\Comp^n$ has a canonical basis, whose
elements we denote by $e_i$, for $1 \leq i \leq n$. Every element
$\alpha\in\Comp^n$ can be uniquely written as a linear combination
\begin{equation}
  \label{eq:lin-comb}
  \alpha = a_1e_1 + \ldots + a_ne_n
\end{equation}
with $a_1, \ldots a_n\in\Comp$. We will often represent the vector
$\alpha$ of (\ref{eq:lin-comb}) as a \emph{column vector}
\begin{equation}
  \label{eq:col-vect}
  \alpha = \begin{bmatrix}
    a_1 \\
    \vdots \\
    a_n
  \end{bmatrix}.
\end{equation}
We identify $\Comp^1$ with $\Comp$ and frequently refer to the
elements of $\Comp$ as \emph{scalars}. The vector space $\Comp^n$ is
equipped with the usual operations of addition and scalar
multiplication. Given a vector $\alpha$ as in (\ref{eq:col-vect}), its
\emph{dual} is the \emph{row vector}
\[
\alpha\da =
\begin{bmatrix}
  a_1\da & \ldots & a_n\da
\end{bmatrix}.
\]
Note that under the identification of $\Comp^1$ and $\Comp$, no
ambiguity arises from using $(-)\da$ to denote the dual of a vector
and the conjugate of a scalar. The \emph{norm}\label{defnorm} of a
vector $\alpha$ is defined as $||\alpha|| = \sqrt{\alpha\da\alpha}$,
where $\alpha\da\alpha$ is obtained by matrix multiplication. A vector
whose norm is 1 is called a \emph{unit vector}. Equipped with the norm
$||-||$, the vector space $\Comp^n$ has the structure of an
$n$-dimensional \emph{Hilbert space}.

% --------------------------------------------------------------------
\subsection{Operators and matrices}
\label{ssec:matrices}

A linear operator $\Comp^n \to \Comp^m$ can be represented by an $m\by
n$ complex matrix. We write $\Matrices_{m,n}(\Comp)$\label{defmats}
for the set of all $m\by n$ complex matrices and we say that $m$ and
$n$ are the \emph{dimensions} of $U\in\Matrices_{m,n}(\Comp)$. If
$m=n$, then we simply say that $U$ has dimension $n$. We identify
$\Matrices_{n,1}(\Comp)$ and $\Comp^n$.

We will use some well-known operations on matrices. For any $n$, the
identity matrix of dimension $n$ is denoted by $I_n$, or simply by $I$
if the dimension is clear from context. If $U\in
\Matrices_{n,n}(\Comp)$, an \emph{inverse}\label{definv} of $U$,
written $U\inv$, is a matrix such that $UU\inv = U\inv U = I$. If it
exists, the inverse of a matrix is unique. We note that for any $n$,
the set of invertible matrices of dimension $n$ forms a group under
matrix multiplication. We denote this group by
$\Matrices_{n,n}(\Comp)^\times$\label{defgroupunits}. If $U\in
\Matrices_{m,n}(\Comp)$, the \emph{conjugate transpose} of $U$,
written $U\da$, is defined by $U\da_{i,j} = (U_{j,i})\da$. The
identification of $\Matrices_{n,1}(\Comp)$ and $\Comp^n$ ensures that
no ambiguity arises from our use of $(-)\da$ to denote the conjugate
transpose of a matrix, since the conjugate transpose of a column
vector is its dual.

We will also use the well-known notions of \emph{trace} and
\emph{determinant}\label{deftracedet} of a square matrix. Both are
functions that assign a complex number to any complex square
matrix. For $U\in \Matrices_{n,n}(\Comp)$, the trace of $U$ is defined
as
\[
\tr(U) = \sum_{i=1}^n U_{i,i}.
\]
The formula to express the determinant of an arbitrary $n\by n$ matrix
is somewhat cumbersome. Since we will only be considering determinants
of $2\by 2$ matrices, we give an explicit definition of the
determinant in this case only. For $U\in \Matrices_{2,2}(\Comp)$, the
determinant of $U$ is defined as
\[
\det \left(
  \begin{bmatrix}
    \alpha & \beta \\
    \gamma & \delta
  \end{bmatrix} \right) = \alpha\delta-\beta\gamma.
\]
We note that in any dimension the determinant is multiplicative and
$\det(I) = 1$.

Finally, if $U\in \Matrices_{n,m}(\Comp)$, $\alpha\in\Comp^n$,
$\lambda\in\Comp$, and $U\alpha = \lambda \alpha$, then $\alpha$ is an
\emph{eigenvector} of $U$ with \emph{eigenvalue} $\lambda$.

% --------------------------------------------------------------------
\subsection{Unitary, Hermitian, and positive matrices}
\label{ssec:unitaries}

A complex matrix $U$ is \emph{unitary} if $U\inv = U\da$. Unitary
matrices preserve the norm of vectors in $\Comp^n$. That is, if $U$ is
unitary then $||U\alpha||=||\alpha||$ for any vector
$\alpha\in\Comp^n$. The composition of two unitary matrices is again
unitary. Moreover, the identity matrix is unitary. Hence, the set
$\uset(n)$ of all unitary matrices of dimension $n$ forms a group,
called the \emph{unitary group of order $n$}\label{defuset}. Since the
determinant is a multiplicative function, $\uset(n)$ has a subgroup
which consists of those unitary matrices whose determinant is 1. This
group is called the \emph{special unitary group of order
  $n$}\label{defsuset} and is denoted by $\suset(n)$. We thus have the
inclusions
\[
\suset(n) \seq \uset(n) \seq \Matrices_{n,n}(\Comp)^\times.
\]
Examples of useful unitary matrices are provided in
Figure~\ref{fig:Pauli-Hadamard}. The matrix $H$ is known as the
\emph{Hadamard}\label{defHad} matrix and the matrices $X$, $Y$, and
$Z$ are known as the \emph{Pauli}\label{defPauli} matrices.

% ....................................................................
\begin{figure}
  \[
  H = \frac{1}{\sqrt 2}
  \begin{bmatrix}
    1 & 1  \\
    1 & -1
  \end{bmatrix} \quad X =
  \begin{bmatrix}
    0 & 1  \\
    1 & 0
  \end{bmatrix} \quad Y =
  \begin{bmatrix}
    0 & -i  \\
    i & 0
  \end{bmatrix} \quad Z =
  \begin{bmatrix}
    1 & 0  \\
    0 & -1
  \end{bmatrix}
  \]
  \label{fig:Pauli-Hadamard}
  \caption[The Hadamard and Pauli matrices.]{The Hadamard matrix $H$
    and the Pauli matrices $X$, $Y$, and $Z$.}
  \rule{\textwidth}{0.1mm}
\end{figure}
% ....................................................................

A complex matrix $U$ is \emph{Hermitian} if $U=U\da$. If $U$ is
hermitian, then $u\da U u$ is always real. Note that all the matrices
in Figure~\ref{fig:Pauli-Hadamard} are Hermitian.

A matrix $U$ is \emph{positive semidefinite} (resp. \emph{positive
  definite}) if $U$ is hermitian and $\alpha\da U\alpha \geq 0$
(resp. $\alpha\da U\alpha > 0$) for all $\alpha\in\Comp^n$.

% --------------------------------------------------------------------
\subsection{Tensor products}
\label{ssect:tensors}

The \emph{tensor product}\label{deftensor} of vectors spaces and
matrices is defined as usual and denoted by $\x$. We note that
$\Comp^n \x \Comp^m = \Comp^{nm}$ and that the tensor product acts on
square matrices as
\[
\x : \Matrices_{n,n}(\Comp) \times \Matrices_{m,m}(\Comp) \to
\Matrices_{nm,nm}(\Comp).
\]
Given two vectors $\alpha \in \Comp^n$, $\beta\in \Comp^m$, their
tensor product $\gamma=\alpha \x \beta\in\Comp^{nm}$ is defined by
$\gamma_{(i,j)} = \alpha_i\beta_j$, with the pairs $(i,j)$ ordered
lexicographically. One obtains a basis for the tensor product
$\mathcal{H}\x \mathcal{H}'$ of two vector spaces $\mathcal{H}$ and
$\mathcal{H}'$ by tensoring the elements of the bases for
$\mathcal{H}$ and $\mathcal{H}'$. For example, if
$\s{\alpha_1,\alpha_2}$ and $\s{\beta_1, \beta_2}$ are two bases for
$\Comp^2$, then
\[
\s{ \alpha_1\x \beta_1, \alpha_1\x \beta_2, \alpha_2\x \beta_1,
  \alpha_2\x \beta_2}
\]
forms a basis for $\Comp^2 \x \Comp^2 = \Comp^4$. We note, however,
that not all elements of $\Comp^4$ are of the form $\alpha \x \beta$
with $\alpha, \beta\in \Comp^2$.

% ====================================================================
\section{Quantum Computation}
\label{sect:quantum-computation}

% --------------------------------------------------------------------
\subsection{A single quantum bit}
\label{ssect:single-qubit}

Recall that the fundamental unit of classical computation is the
\emph{bit}. By analogy, the fundamental unit of quantum computation is
called the \emph{quantum bit}, or \emph{qubit}.

\begin{definition}
  \label{def:qubit}
  The \emph{state of a qubit} is a unit vector in $\Comp^2$ considered
  \emph{up to a phase}, that is, up to multiplication by a unit vector
  of $\Comp$.
\end{definition}

As is customary, we use the so-called \emph{ket}\label{defket}
notation to denote the state of a qubit, which is written
$\ket{\phi}$. Moreover, we write $\ket{0}$ and $\ket{1}$ for the
elements $e_1$ and $e_2$ of the standard basis for $\Comp^2$. In the
context of quantum computation, the basis $\s{\ket{0}, \ket{1}}$ is
referred to as the \emph{computational basis}\label{defcompbasis}.
Since the state of a (classical) bit is an element of the set
$\s{0,1}$, we sometimes refer to the states $\ket{0} = 1\ket{0} +
0\ket{1}$ and $\ket{1} = 0\ket{0} + 1\ket{1}$ as the \emph{classical
  states}.

By Definition~\ref{def:qubit}, the state of a qubit can be one of the
classical states but also any linear combination $\alpha\ket{0} +
\beta\ket{1}$ such that $||\alpha||^2 + ||\beta||^2 = 1$. The
coefficients $\alpha$ and $\beta$ in such a linear combination are
called the \emph{amplitudes} of the state. When both amplitudes of the
state of a qubit are non-zero, then the qubit is said to be in a
\emph{superposition} of $\ket{0}$ and $\ket{1}$.

The fact that the vector in Definition~\ref{def:qubit} is considered
only up to a phase gives rise to a nice geometric representation of
the state of a qubit. Let $\alpha\ket{0} + \beta\ket{1}$ be a unit
vector. We can rewrite this linear combination as
\[
\alpha\ket{0} + \beta\ket{1} =
e^{i\gamma}(\cos(\frac{\theta}{2})\ket{0} +
e^{i\phi}\sin(\frac{\theta}{2})\ket{1})
\]
for some real numbers $\theta\in [0,\pi]$ and $\gamma,\phi\in
[0,2\pi]$. Since the state of qubit is defined up to a phase, the same
state is described by
\[
\cos(\frac{\theta}{2})\ket{0} + e^{i\phi}\sin(\frac{\theta}{2})\ket{1}
\]
with $\theta, \phi \in \R$. The pair $(\theta, \phi)$ defines a point
on the 2-sphere known as the \emph{Bloch sphere} representation of the
state, as pictured in Figure~\ref{fig:Bloch}. The Cartesian
coordinates of the point $(\theta,\phi)$ on the Bloch sphere are given
by $(\sin\theta\cos\phi, \sin\theta\sin\phi,\cos\theta)$.

% ....................................................................
\begin{figure}
  \[
  \mp{0.8}{\begin{tikzpicture}
  % ------------------------------------------------------------------
  % The circle
  \draw (0,0) circle (2);

  % ------------------------------------------------------------------
  % The shading
  \shade[ball color=blue!10!white,opacity=0.20] (0,0) circle (2cm);

  % ------------------------------------------------------------------
  % The z-axis
  \draw[name path = zaxis, ->] (0,-2.3) -- (0,2.3) node[above] {$z$};

  % ------------------------------------------------------------------
  % The y-axis
  \draw[name path = yaxis, ->] (-2.3,0) -- (2.3,0) node[right] {$y$};

  % ------------------------------------------------------------------
  % The x-axis
  \draw[name path = xaxis, ->] ({sqrt(((2.3)*(2.3))*(4/5))},
  {(sqrt(((2.3)*(2.3))*(4/5)))/2}) -- (-{sqrt(((2.3)*(2.3))*(4/5))},
  -{(sqrt(((2.3)*(2.3))*(4/5)))/2}) node[below left] {$x$};

  % ------------------------------------------------------------------
  % Latitude (front)
  \draw[name path = latitude, dashed] (-2,0) arc (180:360:2 and
  0.75);

  % ------------------------------------------------------------------
  % Latitude (back)
  \draw[dashed] (-2,0) arc (180:0:2 and 0.75);

  % ------------------------------------------------------------------
  % Longitude
  \draw[name path = longitude, dashed] (0,-2) arc (-90:90:0.75 and 2);

  % ------------------------------------------------------------------
  % The intersection of the longitude and latitude
  \path[name intersections={of= latitude and longitude}];

  % A line segment from the origin to intersection-1
  \draw[-] (0,0) -- (intersection-1);

  % ------------------------------------------------------------------
  % The red region (hack)
  \begin{scope}
    \clip (0,-2) -- (0,2) -- (2,2) -- (2,-2) -- cycle;

    \clip (0,0) -- (30:0.75 and 2) -- (2,2) -- (0,2) -- cycle;

    \fill [fill=red!50, opacity=0.5] (0,0) ellipse (0.75 and 2);
  \end{scope}

  % ------------------------------------------------------------------
  % The green region (hack)
  \begin{scope}
    \clip (-{sqrt(((2.3)*(2.3))*(4/5))},
    -{(sqrt(((2.3)*(2.3))*(4/5)))/2}) -- ({sqrt(((2.3)*(2.3))*(4/5))},
    {(sqrt(((2.3)*(2.3))*(4/5)))/2}) -- ({sqrt(((2.3)*(2.3))*(4/5))},
    -{(sqrt(((2.3)*(2.3))*(4/5)))/2}) -- cycle ;

    \clip (0,0) -- (intersection-1) -- (0,-2) --
    (-{sqrt(((2.3)*(2.3))*(4/5))}, -{(sqrt(((2.3)*(2.3))*(4/5)))/2})
    -- cycle ;

    \fill [fill=green!50, opacity=0.5] (0,0) ellipse (2 and 0.75);
  \end{scope}

  % ------------------------------------------------------------------
  % The state
  \draw[name path = state, -] (0,0) -- (30:0.75 and 2);

  \fill (30:0.75 and 2) circle (.1);
\end{tikzpicture}}
  \]
  \label{fig:Bloch}
  \caption[The Bloch sphere representation of the state of a
  qubit.]{The Bloch sphere representation of the state of a qubit. The
    state $(\theta,\phi)$ is represented by the black dot, with
    $\theta$ corresponding to the angle pictured in red and $\phi$ to
    the angle pictured in green.}
  \rule{\textwidth}{0.1mm}
\end{figure}
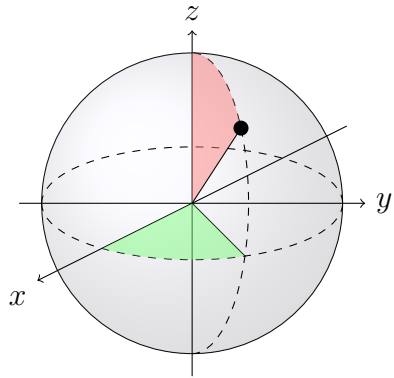
% ....................................................................

% --------------------------------------------------------------------
\subsection{Multiple quantum bits}
\label{ssect:multi-qubits}

In classical computation, the state of a system of $n$ bits is
represented by an element of the set $\s{0,1}^n$. The set of states of
a complex system therefore arises as a \emph{Cartesian product}. In
contrast, the set of states of a system composed of multiple qubits is
obtained using the \emph{tensor product}.

\begin{definition}
  \label{def:qubits}
  The \emph{state of a system of $n$ qubits} is a unit vector in
  $\Comp^{2^n}$ considered up to a phase.
\end{definition}

The basis $\s{\ket{0},\ket{1}}$ for $\Comp^2$ can be used to construct
a basis for $\Comp^{2^n}$, which we also call the computational
basis. As is customary, we denote the basis element
$\ket{x_1}\x\ldots\x\ket{x_n}$ by $\ket{x_1\ldots x_n}$, for any
$x_1,\ldots,x_n\in\s{0,1}$. For example, the state of a system of two
qubits is described, up to a phase, by a linear combination
\[
\alpha_0 \ket{00} + \alpha_1 \ket{01}+ \alpha_2 \ket{10} + \alpha_3
\ket{11}
\]
for some complex numbers $\alpha_i$ such that $\sum ||\alpha_i||^2 =
1$. As mentioned in Section~\ref{ssect:tensors}, not all elements of
$\Comp^4$ arise as the tensor of two elements of $\Comp^2$. If the
state of a two-qubit system can be written as
\[
(\alpha\ket{0} + \beta\ket{1})\x(\alpha'\ket{0} + \beta'\ket{1})
\]
then the two qubits are said to be \emph{separable}. Otherwise, they
are said to be \emph{entangled}.

% --------------------------------------------------------------------
\subsection{Evolution of a quantum system}
\label{ssect:evolution}

A computation is performed by acting on the state of a system of
qubits. This can be done in two ways: via \emph{unitary
  transformation} or via \emph{measurement}.

% --------------------------------------------------------------------
\subsubsection{Unitary evolution}
\label{sssect:unitary-evolution}

The state of a system of $n$ qubits can evolve under the action of a
unitary operator. If $\ket{\phi}$ is such a state (viewed as a column
vector in $\Comp^{2^n}$) and $U$ is a unitary matrix, the evolution of
$\ket{\phi}$ under $U$ is given by $\ket{\phi} \mapsto
U\ket{\phi}$. For this reason, we sometimes refer to a unitary matrix
of dimension $2^n$ as an \emph{$n$-qubit unitary}.

Recall from Section~\ref{ssect:single-qubit} that the state of a qubit
can be interpreted as a point on the Bloch sphere. This interpretation
extends to single-qubit unitary matrices, which can be seen as
\emph{rotations}\label{defrot} of the Bloch sphere. Let $v=(x,y,z)$ be
a unit vector in $\R^3$ and $\theta\in\R$ and define the matrix
\[
R_v(\theta) = \cos(\frac{\theta}{2})I - \sin(\frac{\theta}{2})(xX + yY
+ zZ),
\]
where $X$, $Y$, and $Z$ are the Pauli matrices. The matrix
$R_v(\theta)$ defines a rotation of the Bloch sphere by $\theta$
radians about the $v$-axis. For example, the Pauli matrices $X$, $Y$,
and $Z$ correspond to rotations about the $x$-, $y$-, and $z$-axes by
$\pi$ radians. The following theorem states that, up to a phase, every
unitary can be seen as a rotation of the Bloch sphere.

\begin{theorem}
  \label{thm:unitaries-as-rots}
  If $U\in\uset(2)$, then there exist real numbers $\alpha$ and
  $\theta$, and a unit vector $w$ in $\R^3$ such that $U =
  e^{i\alpha}R_w(\theta)$.
\end{theorem}

% --------------------------------------------------------------------
\subsubsection{Measurement}
\label{sssect:measurement}

The second way in which one can act on the state of a system of qubits
is by measurement. Unlike unitary evolutions, measurements are
probabilistic processes.

We first describe the effect of a measurement on a single
qubit. Assume a qubit is in the state $\alpha\ket{0} +
\beta\ket{1}$. If the qubit is measured, then the result of the
measurement is either 0 or 1 and the state of the qubit
post-measurement is the corresponding classical state. Moreover, the
measurement result 0 occurs with probability $||\alpha||^2$ and the
measurement result 1 occurs with probability $||\beta||^2$. Since the
state of a qubit was described by a unit vector, these probabilities
sum to 1.

We now discuss the case of a complex system. For simplicity, we only
consider a two-qubit system. Assume the state of our system is given
by the following vector
\[
\alpha_0 \ket{00} + \alpha_1 \ket{01}+ \alpha_2 \ket{10} + \alpha_3
\ket{11}.
\]
If the first qubit is measured, then with probability $||\alpha_0||^2
+ ||\alpha_1||^2$ the measurement result is 0 and the
post-measurement state is
\[
\frac{\alpha_0}{\sqrt{||\alpha_0||^2 + ||\alpha_1||^2}} \ket{00} +
\frac{\alpha_1}{\sqrt{||\alpha_0||^2 + ||\alpha_1||^2}} \ket{01},
\]
while with probability $||\alpha_2||^2 + ||\alpha_3||^2$ the
measurement result is 1 and the post-measurement state is
\[
\frac{\alpha_2}{\sqrt{||\alpha_2||^2 + ||\alpha_3||^2}} \ket{10} +
\frac{\alpha_3}{\sqrt{||\alpha_2||^2 + ||\alpha_3||^2}} \ket{11}.
\]
Note that the linear combinations have been renormalized to ensure
that the resulting states are described by unit vectors.

% --------------------------------------------------------------------
\subsubsection{No-cloning}
\label{sssect:no-cloning}

The well-known \emph{no-cloning theorem} is a property of quantum
computation which will be of importance in chapters~\ref{chap:lambda},
\ref{chap:pq}, and \ref{chap:pq-safe}. The theorem states that there
is no physical device whose action is described by the following
mapping
\[
\ket{\psi}\x\ket{0} \mapsto \ket{\psi}\x \ket{\psi}.
\]
In other words, it is impossible to ``clone'' quantum states.

% --------------------------------------------------------------------
\subsection{The QRAM model and quantum circuits}
\label{ssect:circuit-model}

In Section~\ref{sect:quantum-computation}, we described the
mathematical formalism of quantum computation, but did not spend any
time explaining what a quantum computer might look like, nor how one
would describe quantum programs and protocols. Various models of
quantum computation have been devised in the literature (see, e.g.,
\cite{Deutsch85}, \cite{Raussendorf2002}). Here, we discuss two
complementary approaches: the \emph{QRAM model} introduced by Knill in
\cite{Knill96} and the \emph{circuit model} described by Deutsch in
\cite{Deutsch89}. We also refer the reader to \cite{CACM} where this
model was described in more detail.

% --------------------------------------------------------------------
\subsubsection{The QRAM model of quantum computation}
\label{sssect:qram}

In the QRAM model of a quantum computation, a quantum computer is
thought of as consisting of two devices, a \emph{classical device} and
a \emph{quantum device}, sharing computational tasks. The classical
device performs operations such as compilation and correctness
checking. The quantum device only performs the specifically quantum
operations. In particular, it is assumed to hold an array of qubits
and to be able to
\begin{itemize}
\item initialize qubits to a specified state,
\item perform unitary operations on qubits, and
\item measure qubits.
\end{itemize}
The execution of a program in this model proceeds as follows. The
source code of the program resides on the classical device where it is
compiled into an executable. When the program is executed on the
classical device, it can communicate with the quantum device if
required. Through this communication, it can instruct the quantum
device to perform one of the above described quantum
operations. Measurement results, if any, are returned to the classical
device for post-processing or further computation. This system is
schematically represented in Figure~\ref{fig:QRAM}.

% ....................................................................
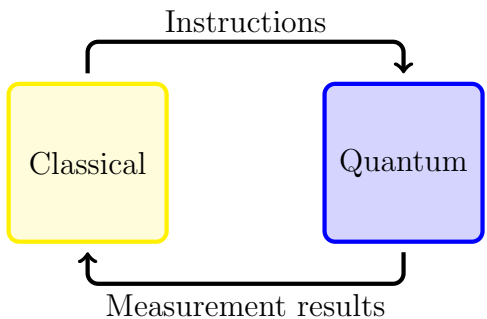
\begin{figure}
  \[
  \mp{0.8}{\begin{tikzpicture}[scale=0.7]
\draw [yellow, ultra thick, rounded corners, fill=yellow!17] 
  (0,0) rectangle (3,3);
\draw [blue, ultra thick, rounded corners, fill=blue!17] 
  (6,0) rectangle (9,3);
\draw [->, ultra thick, rounded corners] 
  (7.5,-0.2) to (7.5, -0.8) to (1.5, -0.8) to (1.5,-0.2);
\draw [->, ultra thick, rounded corners] 
  (1.5,3.2) to (1.5, 3.8) to (7.5, 3.8) to (7.5,3.2);

\node at (1.5,1.5) {Classical};
\node at (7.5,1.5) {Quantum};
\node at (4.5,-1.2) {Measurement results};
\node at (4.5,4.2) {Instructions};
\end{tikzpicture}}
  \]
  \label{fig:QRAM}
  \caption[The QRAM model of quantum computation.]{The QRAM model of
    quantum computation.}
  \rule{\textwidth}{0.1mm}
\end{figure}
% ....................................................................

% --------------------------------------------------------------------
\subsubsection{Quantum circuits}
\label{sssect:qcircs}

While the QRAM model of quantum computation describes the overall
architecture of a quantum computer, quantum circuits are a language to
express (parts of) quantum algorithms. In particular, one can think of
a quantum circuit as the description of a sequence of operations that,
in the QRAM model, the classical device would send to its quantum
counterpart.

In the quantum circuit model, a quantum computation is thought of as a
sequence of unitary gates applied to an array of qubits, followed by a
measurement of some or all of the qubits. The sequence of unitary
operations are arranged in the form of a circuit, akin to the
classical boolean circuits.

The identity matrix on $n$ qubits, i.e., the identity matrix of
dimension $2^n$, is represented by $n$ distinct \emph{wires}. For
example, the identity on 3 qubits is represented by
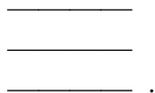
\[
\mbox{\Qcircuit @C=1em @R=1.3em {
    &\qw &\qw &\qw &\qw \\
    &\qw &\qw &\qw &\qw \\
    &\qw &\qw &\qw &\rstick{.}\qw }}
\]
A non-identity unitary $U$ acting on $n$ qubits is depicted as a box,
labelled $U$, with $n$ \emph{input wires} and $n$ \emph{output
  wires}. For example, the Hadamard matrix $H$ is represented as
\[
\mbox{ \Qcircuit @R=.2em @C=.5em @!R { &\qw &\gate{H} &\qw
    &\rstick{.}\qw }}
\]
Because of the similarities between quantum circuits and classical
boolean circuits, we sometimes refer to a unitary $U$ as a
\emph{quantum gate}.

The graphical representation of the composition $UV$, of two unitary
matrices $U$ and $V$ acting on $n$ qubits, is obtained by horizontally
concatenating the gates for $U$ and $V$. That is, by connecting the
output wires of the gate for $V$ to the input wires of the gate for
$U$, as illustrated below in the case of matrices on 4 qubits.
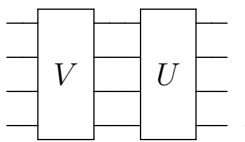
\[
\mbox{ \Qcircuit @R=.2em @C=.5em @!R {
    &\qw &\multigate{3}{V} &\qw &\qw &\multigate{3}{U} &\qw &\qw \\
    &\qw &\ghost{V} &\qw &\qw &\ghost{U} &\qw &\qw\\
    &\qw &\ghost{V} &\qw &\qw &\ghost{U} &\qw &\qw \\
    &\qw &\ghost{V} &\qw &\qw &\ghost{U} &\qw &\rstick{.}\qw }}
\]

Finally, the graphical representation of the tensor product $U\x V$ of
two unitary matrices $U$ and $V$ is obtained by vertically
concatenating the gates for $U$ and $V$, as illustrated below in the
case of unitary matrices $U$ and $V$ on 3 and 2 qubits respectively.
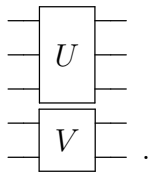
\[
\mbox{ \Qcircuit @R=.2em @C=.5em @!R {
    & \qw & \multigate{2}{U} & \qw & \qw \\
    & \qw & \ghost{U} & \qw & \qw &  \\
    & \qw & \ghost{U} & \qw & \qw & \\
    & \qw & \multigate{1}{V} & \qw & \qw \\
    & \qw & \ghost{V} & \qw & \rstick{.}\qw }}
\]

Now let $S$ be a set of unitary matrices and write $S\da$ for the set
$\s{U\da \such U\in S}$. A \emph{circuit over $S$} is constructed
using the gates of $S\cup S\da$, as well as arbitrary identity gates,
using the graphical representations for composition and tensor
product. For example, a circuit over $S=\s{U, V, W}$, where $U$ acts
on 2 qubits and $V$ and $W$ both act on 3 qubits, is represented below
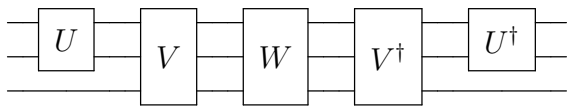
\[
\mbox{ \Qcircuit @R=.2em @C=.5em @!R { &\qw &\multigate{1}{U} &\qw
    &\qw &\multigate{2}{V} &\qw &\qw &\multigate{2}{W} &\qw &\qw
    &\multigate{2}{V\da} &\qw &\qw
    &\multigate{1}{U\da} &\qw &\qw  \\
    &\qw &\ghost{U} &\qw &\qw &\ghost{V} &\qw &\qw &\ghost{W} &\qw
    &\qw &\ghost{V\da} &\qw &\qw &\ghost{U\da} &\qw &\qw \\
    &\qw &\qw &\qw &\qw &\ghost{V} &\qw &\qw &\ghost{W} &\qw &\qw
    &\ghost{V\da} &\qw &\qw &\qw &\qw &\rstick{.}\qw }}
\]
We write $\p{S}$ for the set of all circuits over $S$.

One can recover the matrix represented by a circuit by interpreting
the operations of horizontal and vertical concatenation as composition
and tensor product respectively. For example, let $S$ consist of the
Hadamard gate $H$ and the Pauli gate $X$, and let $C$ be the following
circuit
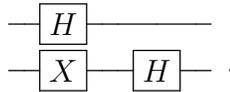
\[
\mbox{ \Qcircuit @R=.2em @C=.5em @!R {
    &\qw &\gate{H} &\qw &\qw &\qw &\qw &\qw\\
    &\qw &\gate{X} & \qw &\qw &\gate{H} &\qw &\rstick{.}\qw }}
\]
Then $C$ represents the unitary matrix $U$ given by
\[
U = (I_2 \x H) \circ (H \x X) = \frac{1}{2}\begin{bmatrix}
  1 & -1 & 1 & -1 \\
  1 & 1 & 1 & 1 \\
  1 & -1 & -1 & 1 \\
  1 & 1 & -1 & -1
\end{bmatrix}.
\]
Note that $C$ could have alternatively been interpreted as the unitary
$(I_2\circ H) \x (H\circ X)$. However, both interpretations coincide
since $(I_2\circ H) \x (H\circ X) = (I_2 \x H) \circ (H \x X)$. This
is due to the so-called \emph{bifunctoriality} of $\x$, which
guarantees in particular that
\[
(U\circ U')\x (V\circ V') = (U \x V) \circ (U' \x V')
\]
for any $U,U',V,V'$.  By a slight abuse of notation, we often write
$C=U$ if the circuit $C$ represents the matrix $U$.

The language of quantum circuits can be extended with measurement
gates, which are depicted as
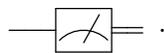
\[
\mbox{ \Qcircuit @R=.2em @C=.5em @!R { &\qw & \qw &\meter & \cw &
    \rstick{.}\cw }}
\]

% ---------------------------------------------------------------------
\chapter{Algebraic number theory}
\label{chap:nb-th}

In this chapter, we introduce basic concepts of algebraic number
theory. In particular, we describe ring extensions of $\Z$ and
computational methods to solve certain Diophantine equations, known as
\emph{relative norm equations}, over these rings. References for this
material include \cite{Wla90} and \cite{Cohen2000}.

% ====================================================================
\section{Rings of integers}
\label{sec:int-rings}

If $K$ is a field, $S$ a subfield of $K$, and $\alpha$ an element of
$K$, then the \emph{field extension}\label{deffieldext} $S(\alpha)$ is
the smallest subfield of $K$ which contains $S$ and $\alpha$.

An element $\alpha\in\Comp$ is an \emph{algebraic number} if it is the
root of some polynomial over $\Qu$. A field extension of $\Qu$ of the
form $\Qu(\alpha)$ for some algebraic number $\alpha$ is called an
\emph{algebraic number field}.

An element $\beta\in\Comp$ is an \emph{algebraic integer} if it is the
root of some monic polynomial over $\Z$, i.e., of some polynomial over
$\Z$ whose leading coefficient is 1. The set of algebraic integers of
a number field $\Qu(\alpha)$ forms a ring, called \emph{the ring of
  integers}\label{defringints} of $\Qu(\alpha)$ and denoted by
$\O_{\Qu(\alpha)}$. For any algebraic number field $\Qu(\alpha)$, the
ring $\O_{\Qu(\alpha)}$ is an integral domain whose field of fractions
is $\Qu(\alpha)$.

% --------------------------------------------------------------------
\subsection{Extensions of \texorpdfstring{$\Z$}{Z}}
\label{ssec:Z-exts}

If $R$ is a ring, $R'$ a subring of $R$, and $\alpha$ an element of
$R$, then the \emph{ring extension}\label{defringext} $R'[\alpha]$ is
the smallest subring of $R$ which contains $R'$ and $\alpha$.

\begin{definition}[Extensions of $\Z$]
  \label{def:ext-Z}
  We are interested in the following four ring extensions of $\Z$.
  \begin{itemize}
  \item The ring $\Z$ of \emph{integers}.
  \item The ring $\Zi$ of \emph{Gaussian integers}.
  \item The ring $\Zrt$ of \emph{quadratic integers with radicand 2}.
  \item The ring $\Zomega$ of \emph{cyclotomic integers of degree 8},
    where $\omega= e^{i\pi/4}=(1+i)/\sqrt{2}$.
  \end{itemize}
\end{definition}

We note that since $i=\omega^2$ and $\sqrt{2} = \omega - \omega^3$, we
have the inclusions $\Z \seq \Zi \seq \Zomega$ and $\Z \seq \Zrt \seq
\Zomega$.  Moreover, one can show that $\Zrt$ and $\Zomega$ are dense
in $\R$ and $\Comp$ respectively. The rings introduced in
Definition~\ref{def:ext-Z} are rings of algebraic integers. Indeed we
have
\begin{center}
  $\Z = \O_{\Qu}$, $\Zi = \O_{\Qu(i)}$, $\Zrt = \O_{\Qu(\sqrt 2)}$,
  and $\Zomega = \O_{\Qu(\omega)}$.
\end{center}
Explicit expressions for the elements of $\Zi$, $\Zrt$, and $\Zomega$
are given below.
\begin{itemize}
\item $\Zi = \s{a_0 + a_1i \such a_j \in \Z}$.
\item $\Zrt = \s{a_0 + a_1\sqrt{2} \such a_j \in \Z}$.
\item $\Zomega = \s{a_0 + a_1\omega + a_2\omega^2 + a_3\omega^3 \such
    a_j \in \Z}$.
\end{itemize}
Note that there is a bijection between $\Zi$ and $\Z^2$. As the
following proposition shows, there is also a bijection between
$\Zomega$ and and two disjoint copies of $\Zrt\times\Zrt$.

\begin{proposition}
  \label{prop:zomega}
  A complex number $\alpha$ is in $\Zomega$ if and only if it can be
  written of the form $\alpha=a_0+a_1 i$ or of the form
  $\alpha=a_0+a_1 i+\omega$, where $a_0, a_1\in\Zrt$.
\end{proposition}

\begin{proof}
  The right-to-left implication is trivial. For the left-to-right
  implication, let $\alpha=a+b\omega+c\omega^2+d\omega^3$, where
  $a,b,c,d\in\Z$. Noting that $\omega=\frac{1+i}{\sqrt 2}$, we have
  \[
  \alpha = (a + \frac{b-d}{2} \sqrt2) + (c + \frac{b+d}{2}{\sqrt2})i.
  \]
  If $b-d$ (and therefore $b+d$) is even, then $\alpha$ is of the
  first form, with $a_0=a + \frac{b-d}{2} \sqrt2$ and $a_1=c +
  \frac{b+d}{2}{\sqrt2}$. If $b-d$ (and therefore $b+d$) is odd, then
  $\alpha$ is of the second form, with $a_0=a + \frac{b-d-1}{2}
  \sqrt2$ and $a_1 = c + \frac{b+d-1}{2}{\sqrt2}$.
\end{proof}

We close this subsection with the definition of two algebraic integers
which will be useful in chapters \ref{chap:grid-pb},
\ref{chap:synth-V}, and \ref{chap:synth-T}.

\begin{definition}
  \label{def:lambda-delta}
  The algebraic integers $\lambda\in\Zrt$ and $\delta\in\Zomega$ are
  defined as follows
  \begin{itemize}
  \item $\lambda=1+\sqrt{2}$ and
  \item $\delta = 1+\omega$.
  \end{itemize}
\end{definition}

% --------------------------------------------------------------------
\subsection{Automorphisms}
\label{ssec:autos}

Recall that an automorphism of a ring $R$ is an isomorphism $R \to
R$. The ring $\Zomega$ has four automorphisms. One of these
automorphisms is \emph{complex conjugation}, which we denote $(-)\da$
as in Chapter~\ref{chap:quantum}. Explicitly, $(-)\da$ acts on an
arbitrary element of $\Zomega$ as follows
\[
(a_0 + a_1\omega + a_2\omega^2 +a_3\omega^3)^\dagger = a_0 -a_3\omega
- a_2\omega^2 -a_1\omega^3.
\]
A second automorphism of $\Zomega$ is
\emph{$\sqrt{2}$-conjugation}\label{defbul}, denoted $(-)\bul$, which
acts on an arbitrary element of $\Zomega$ as follows
\[
(a_0 + a_1\omega +a_2\omega^2 +a_3\omega^3)\bul = a_0 -a_1\omega
+a_2\omega^2 -a_3\omega^3.
\]
The remaining two automorphisms of $\Zomega$ are the identity as well
as ${(-)\bul}\da={(-)\da}\bul$.

The rings $\Zi$ and $\Zrt$ both have two automorphisms, while $\Z$ has
exactly one. All of these are obtained by restricting the
automorphisms of $\Zomega$. Because $(-)\da$ acts trivially on $\Zrt$,
the only non-identity automorphism of $\Zrt$ is
$(-)\bul$. Explicitly, the action of $(-)\bul$ on an element of
$\Zrt$ is given by $(a+b\sqrt{2})\bul = a-b\sqrt{2}$. Similarly, the
only non-identity automorphism of $\Zi$ is $(-)\da$, whose action is
explicitly given by $(a+bi)\da = a-bi$. The ring $\Z$ has no
non-trivial automorphism.

We note that for $t\in\Zomega$, we have $t\in\Zrt$ iff $t=t\da$,
$t\in\Zi$ iff $t=t\bul$, and $t\in\Z$ iff $t=t\da$ and $t=t\bul$.

% --------------------------------------------------------------------
\subsection{Norms}
\label{ssec:norms}

Let $R$ be one of the rings $\Z$, $\Zi$, $\Zrt$, or $\Zomega$. We
define the \emph{norm}\label{nbnorm} $\Norm_R(\alpha)$ of an element
$\alpha\in R$ to be
\[
\Norm_R(\alpha) = \prod_\sigma \sigma(\alpha),
\]
where the product is taken over all automorphisms $\sigma : R \to
R$. We provide explicit formulas for each norm in the definition
below.

\begin{definition}[Norms]~
  \label{def:norms}
  \begin{itemize}
  \item If $\alpha\in\Z$, then $\Norm_\Z(\alpha) = \alpha$.
  \item If $\alpha = a_0+a_1i\in\Zi$, then $\Norm_{\Zi}(\alpha) =
    \alpha\da \alpha = a_0^2 + a_1^2$\label{defnormZi}.
  \item If $\alpha=a_0+a_1\sqrt{2}\in\Zrt$, then $\Norm_{\Zrt}(\alpha)
    = \alpha\bul\alpha = a_0^2 - 2a_1^2$\label{defnormZrt}.
  \item If $\alpha=a_0+a_1\omega+a_2\omega^2+a_3\omega^3\in\Zomega$,
    then
    \[
    \Norm_{\Zomega}(\alpha) = {\alpha\da}\bul\alpha\da\alpha\bul\alpha
    = (a_0^2 + a_1^2 + a_2^2 + a_3^2)^2 -
    2(a_3a_2+a_2a_1+a_1a_0-a_3a_0)^2\label{defnormZomega}.
    \]
  \end{itemize}
\end{definition}

All the norms introduced in Definition~\ref{def:norms} are
multiplicative and integer valued. This means that
$\Norm_R(\alpha\beta)=\Norm_R(\alpha)\Norm_R(\beta)$ and
$\Norm_R(\alpha)\in\Z$. The norms $\Norm_{\Zomega}$ and $\Norm_{\Zi}$
are moreover valued in the non-negative integers. Finally, we have
$\Norm_R(\alpha)=0$ iff $\alpha=0$ and $\Norm_R(\alpha)$ is a unit if
and only if $\alpha$ is a unit, that is, an invertible element.

\begin{remark}
  \label{rem:bound-bullet}
  If $\alpha$ and $\beta$ are two distinct elements of $\Zrt$, then
  the following inequality holds:
  \begin{equation}
    \label{eqn:discrete}
    |\alpha-\beta|\cdot |\alpha\bul-\beta\bul| \geq 1,
  \end{equation}
  This follows from the fact that $|\alpha-\beta|\cdot
  |\alpha\bul-\beta\bul| = |\Norm_{\Zrt}(\alpha - \beta)|$. The same
  inequality holds for $\alpha,\beta\in\Zomega$.
\end{remark}

% ====================================================================
\section{Diophantine equations}
\label{sect:diophantine}

% --------------------------------------------------------------------
\subsection{Euclidean domains}
\label{ssect:euclidian}

The rings $\Z$, $\Zi$, $\Zrt$, and $\Zomega$ are integral domains,
whose fields of fractions are $\Qu$, $\Qu(i)$, $\Qu(\sqrt 2)$, and
$\Qu(\omega)$. An important property of these rings is that they are
\emph{Euclidean domains}.

\begin{definition}
  \label{def:euclidian-domain}
  A \emph{Euclidean domain} is an integral domain $R$ equipped with a
  function $f:R\setminus \s{0} \to \N$ such that for every $a\in R$
  and $b\in R\setminus \s{0}$, there exist $q,r\in R$ such that $a =
  bq + r$ and $r=0$ or $f(r)<f(b)$.
\end{definition}

\begin{proposition}
  \label{prop:ext-Euc-dom}
  Let $R$ be one of $\Z$, $\Zi$, $\Zrt$, or $\Zomega$. Then the
  function $|\Norm_R(-)|$ makes $R$ into a Euclidean domain.
\end{proposition}

The notion of divisibility, as well as many essential properties of
the divisibility of integers can be defined in an arbitrary Euclidean
domain. In particular, we write $x\divides y$ if $x$ is a divisor of
$y$, and $x\sim y$ if $x\divides y$ and $y\divides x$. An element $x$
is {\em prime} if $x$ is not a unit and $x=ab$ implies that either $a$
or $b$ is a unit. The notion of greatest common divisor, as well as
Euclid's algorithm, can be defined in any Euclidean domain. Finally,
every Euclidean domain is also a \emph{unique factorization domain},
which means that every non-zero non-unit element of the ring can be
factored into primes in an essentially unique way.

% --------------------------------------------------------------------
\subsection{Relative norm equations}
\label{ssect:norm-eqs}

In chapters~\ref{chap:synth-V} and \ref{chap:synth-T}, we will be
interested in solving certain equations known as \emph{relative norm
  equations}. Specifically, we will be concerned with the following
two problems.

\begin{problem}[Relative norm equation over $\Zi$]
  \label{pb:diophantine-Zi}
  Given $\beta\in\Z$, find $\alpha\in\Zi$ such that $\alpha\da\alpha =
  \beta$.
\end{problem}

\begin{problem}[Relative norm equation over $\Zomega$]
  \label{pb:diophantine-Zomega}
  Given $\beta\in\Zrt$, find $\alpha\in\Zomega$ such that
  $\alpha\da\alpha = \beta$.
\end{problem}

Solving Problem~\ref{pb:diophantine-Zi} amounts to finding
$\alpha\in\Zi$ such that $\Norm_{\Zi}(\alpha)=\beta$. The equation to
be solved is therefore a norm equation. In
Problem~\ref{pb:diophantine-Zomega}, however, $\alpha\da\alpha \neq
\Norm_{\Zomega}(\alpha)$. In this case, we do not consider all the
automorphic images of $\alpha$. Instead, we consider the automorphic
images of $\alpha$ under the automorphisms of $\Zomega$ that fix
$\Zrt$. For this reason the equation in
Problem~\ref{pb:diophantine-Zomega} is a \emph{relative} norm
equation. For uniformity, we refer to both equations as relative norm
equations.

A \emph{Diophantine equation} is a polynomial equation in integer
variables. By writing $\alpha= a+bi$ with $a,b\in\Z$,
Problem~\ref{pb:diophantine-Zi} becomes equivalent to the Diophantine
equation
\[
a^2 + b^2 = \beta.
\]
Similarly, by writing $\alpha = a+ b\omega +c\omega^2 +d\omega^3$ and
$\beta = a' + b'\sqrt{2}$ with $a,a',b,b',c,d\in\Z$,
Problem~\ref{pb:diophantine-Zomega} becomes equivalent to the system
of Diophantine equations
\begin{align*}
  a^2 + b^2 +c^2 +d^2 &= a'\\
  ab-ad+cb+cd &= b'.
\end{align*}
In light of these equivalences, we sometimes abuse terminology and
refer to the equations of problems \ref{pb:diophantine-Zi} and
\ref{pb:diophantine-Zomega} as Diophantine equations.

\begin{remark}
  \label{rem:arithm-op}
  Problems \ref{pb:diophantine-Zi} and \ref{pb:diophantine-Zomega} are
  computational problems. This means that a solution to either of
  these problems is an algorithm which decides whether the given
  equation has a solution and produces a solution if one exists. The
  algorithms solving problems \ref{pb:diophantine-Zi} and
  \ref{pb:diophantine-Zomega} that we consider here are
  \emph{probabilistic} in the sense that they make certain choices at
  random. We evaluate the time-complexity of an algorithm by
  estimating the number of \emph{arithmetic operations} it requires to
  solve one the above problems. By arithmetic operations, we mean
  addition, subtraction, multiplication, division, exponentiation, and
  logarithm. When we say that an algorithm runs in probabilistic
  polynomial time, we mean that the algorithm is probabilistic,
  requires a polynomial number of expected arithmetic operations to
  solve the given problem, and produces a correct solution with
  probability greater than $1/2$.
\end{remark}

Since $\alpha\da\alpha\geq 0$, a necessary condition for
Problem~\ref{pb:diophantine-Zi} to have a solution is $\beta\geq
0$. Similarly, necessary conditions for
Problem~\ref{pb:diophantine-Zomega} to have a solution are $\beta\geq
0$ and $\beta\bul\geq 0$. This follows from the fact that
$\alpha\da\alpha=\beta$ implies
$(\alpha\bul)\da(\alpha\bul)=\beta\bul$. There are also sufficient
conditions for the above relative norm equations to have solutions.

\begin{proposition}
  \label{prop:prime-1mod4-Zi}
  Let $\beta\in\Z$ be such that $\beta\geq 0$. If $\beta$ is prime and
  $\beta\equiv 1\mmod{4}$, then the equation $\alpha\da \alpha=\beta$
  has a solution.
\end{proposition}

\begin{proposition}
  \label{prop:prime-1mod8-Zomega}
  Let $\beta \in \Zrt$ be such that $\beta\geq 0$ and $\beta\bul \geq
  0$ and let $n=\beta\bul\beta\in\Z$. If $n$ is prime and $n\equiv
  1\mmod{8}$, then the equation $\alpha\da \alpha=\beta$ has a
  solution.
\end{proposition}

We close this chapter by stating that there are solutions to problems
\ref{pb:diophantine-Zi} and \ref{pb:diophantine-Zomega}.

\begin{proposition}
  \label{prop:diophantineZi}
  Let $\beta\in\Z$. Given the prime factorization of $\beta$, there
  exists an algorithm that determines, in probabilistic polynomial
  time, whether the equation $\alpha\da \alpha = \beta$ has a solution
  $\alpha\in\Zi$ or not, and finds a solution if there is one.
\end{proposition}

\begin{proposition}
  \label{prop:diophantineZomega}
  Let $\beta\in\Zrt$, and let $n=\beta\bul\beta$. Given the prime
  factorization of $n$, there exists an algorithm that determines, in
  probabilistic polynomial time, whether the equation $\alpha\da
  \alpha = \beta$ has a solution $\alpha\in\Zomega$ or not, and finds
  a solution if there is one.
\end{proposition}

% ---------------------------------------------------------------------
\chapter{The lambda calculus}
\label{chap:lambda}

In this chapter, we introduce the untyped lambda calculus as well as
various typed lambda calculi, including the quantum lambda
calculus. The standard reference for the untyped lambda calculus is
\cite{Baren}. For typed lambda calculi, see
\cite{proofs-and-types}. For the quantum lambda calculus, see
\cite{valiron04}, \cite{SeVa06}, or \cite{SeVa09}.

% ====================================================================
\section{The untyped lambda calculus}
\label{sect:lambda}

% --------------------------------------------------------------------
\subsection{Concrete terms}
\label{ssect:c-terms}

We start by defining the \emph{syntax} of the untyped lambda calculus.

\begin{definition}
  \label{def:l-conc-terms}
  The \emph{concrete terms} of the untyped lambda calculus are defined
  by
  \[
  a, b \quad \bnf \quad x \bor (\lambda x. a) \bor (a b)
  \]
  where $x$ comes from a countable set $\vset$ of \emph{variables}.
\end{definition}

In Definition~\ref{def:l-conc-terms}, concrete terms are given in the
so-called \emph{Backus-Naur Form}. This notation should be interpreted
as defining the collection $L$\label{defconcterms} of all concrete
terms as the smallest set of words on the alphabet $\vset \cup
\s{(,),\lambda,.}$ such that
\begin{itemize}
\item $\vset \seq L$,
\item if $a\in L$ and $x\in \vset$, then $(\lambda x.a) \in L$, and
\item if $a,b \in L$, then $(ab) \in L$.
\end{itemize}

A concrete term of the form $(\lambda x.a)$ is called a \emph{lambda
  abstraction} and we say that $a$ is the \emph{body} of the
abstraction. A concrete term of the form $(ab)$ is called an
\emph{application}. The operations of lambda abstraction and
application are called \emph{term forming operations}.

To increase the readability of concrete terms, we adopt the following
notational conventions
\begin{itemize}
\item outermost parentheses are omitted,
\item applications associate to the left,
\item the body of a lambda abstraction extends as far to the right as
  possible, and
\item multiple lambda abstractions are contracted.
\end{itemize}
This implies, for example, that the term $(\lambda x. (\lambda
y. ((xy)x)))$ will be written $\lambda xy. xyx$.

The intended interpretation of concrete terms is as follows. The
lambda abstraction $\lambda x.a$ represents the function defined by
the rule $x \mapsto a$. The application $ab$ represents the
application of the function $a$ to the argument $b$.

As a first example, consider the identity function. Since it acts as
$x \mapsto x$, its representation as a concrete term should be
$\lambda x.x$. The application of the identity function to some input
$a$ is written as $(\lambda x.x)a$. As a second example, consider the
function that acts as $x \mapsto (y \mapsto x)$. This function inputs
$x$ and outputs the constant function to $x$. Its representation as a
concrete term is $\lambda xy.x$.

% --------------------------------------------------------------------
\subsection{Reduction}
\label{ssect:red}

We now define the \emph{operational semantics} of the terms of the
untyped lambda calculus. That is, we attribute meaning to the terms by
specifying their behavior.

For the concrete term $\lambda x.x$ to be an acceptable representative
for the identity function, it should ``behave'' accordingly, i.e., the
concrete term $(\lambda x.x)a$ should reduce to $a$. One way to
achieve this is to define the reduction relation by
\begin{equation}
  \label{eq:naive-subst}
  (\lambda x.b)a \to b[a/x],
\end{equation}
where $b[a/x]$ stands for the substitution of $a$ for every occurrence
of $x$ in $b$. Under this definition of the reduction relation we have
$(\lambda x.x)a \to a$, as intended.

Now consider the concrete term $\lambda xy.x$. Under the
interpretation of terms given above, it represents the function $x
\mapsto (y \mapsto x)$. If we apply this concrete term to $x'y'$, then
(\ref{eq:naive-subst}) yields the expected result
\[
(\lambda xy.x)(x'y') \to (\lambda y.x)[(x'y')/x]=\lambda y.x'y'.
\]
However, if we apply $\lambda xy.x$ to the variable $y$ and reduce
according to (\ref{eq:naive-subst}) again, we get
\[
(\lambda xy.x)y \to (\lambda y.x)[y/x]=\lambda y.y.
\]
This is unsatisfactory because the concrete term $\lambda y.y$
represents the identity function, not the constant function to
$y$. The way around this problem is to start by renaming $\lambda
xy.x$ to, say, $\lambda xz.x$. Under this renaming, the reduction of
(\ref{eq:naive-subst}) would produce the concrete term $\lambda z.y$,
representing a constant function to $y$. We therefore need to define a
substitution method that appropriately renames variables.

\begin{definition}
  \label{def:renaming}
  Let $x$ and $y$ be two variables and $a$ be a concrete term. The
  \emph{renaming of $x$ by $y$ in $a$}, written $a\s{y/x}$, is defined
  as
  \begin{itemize}
  \item $x\s{y/x} = y$,
  \item $z\s{y/x} = z$ if $x\neq z$,
  \item $ab\s{y/x} = (a\s{y/x})(b\s{y/x})$,
  \item $\lambda x.a\s{y/x} = \lambda y. (a\s{y/x})$, and
  \item $\lambda z.a\s{y/x} = \lambda z. (a\s{y/x})$ if $x \neq z$.
  \end{itemize}
\end{definition}

\begin{definition}
  \label{def:free-bound-var}
  The set of \emph{free variables} of a concrete term $a$, written
  $\FV(a)$, is defined as
  \begin{itemize}
  \item $\FV(x)=\s{x}$,
  \item $\FV(ab)=\FV(a)\cup \FV(b)$, and
  \item $\FV(\lambda x.a)=\FV(a)\setminus\s{x}$.
  \end{itemize}
  Any variable that appears in $a$ but does not belong to $\FV(a)$ is
  said to be \emph{bound}. For this reason, we call $\lambda$ a
  \emph{binder}.
\end{definition}

\begin{definition}
  \label{def:closed-term}
  If $a$ is a concrete term such that $\FV(a)=\emptyset$, then $a$ is
  said to be \emph{closed}.
\end{definition}

In the terminology of Definition~\ref{def:free-bound-var}, the problem
with the relation defined by (\ref{eq:naive-subst}) is that the
variable $y$ is free on the left of $\to$ but bound on the right. The
variable $y$ is said to have been \emph{captured} in the course of the
reduction. Hence, we introduce a \emph{capture-avoiding} notion of
substitution.

\begin{definition}
  \label{def:substitution}
  Let $x$ be a variable, and $a$ and $b$ be two concrete terms. The
  \emph{substitution of $b$ for $x$ in $a$}, written $a[b/x]$, is
  defined as
  \begin{itemize}
  \item if $a = x$, then $a[b/x] = b$,
  \item if $a = y$ and $y \neq x$, then $a[b/x] = a$,
  \item if $a = cc'$, then $a[b/x] = c[b/x]c'[b/x]$,
  \item if $a = \lambda x.c$, then $a[b/x] = a$, and
  \item if $a = \lambda y.c$, then $a[b/x] = \lambda
    z. (c\s{z/y})[b/x]$ where $z \notin \FV(b) \cup \FV(c)$.
  \end{itemize}
\end{definition}

\begin{remark}
  \label{rem:fresh}
  The last clause of Definition~\ref{def:substitution} is slightly
  ambiguous. Indeed, the variable $z$ is not specified, but only
  required to belong to $\vset \setminus (\FV(b)\cup \FV(c))$. We say
  of such a $z$ that it is \emph{fresh with respect to $b$ and
    $c$}. To be more rigorous, we should require the set $\vset$ to be
  well-ordered and choose $z$ to be the least element of $\vset
  \setminus (\FV(b)\cup \FV(c))$.  Since this ambiguity will be lifted
  below when we move from concrete terms to \emph{abstract} ones, we
  leave the definition unchanged.
\end{remark}

To formally define the reduction relation, we start by identifying the
strings, within a concrete term, that will give rise to a
reduction. As one might expect, a reduction will take place whenever a
function is applied to an argument.

\begin{definition}
  \label{def:l-redex}
  A \emph{redex} is a concrete term of the form $(\lambda x.a)b$. By
  extension, a \emph{redex of a concrete term} $c$ is a redex that
  appears in $c$.
\end{definition}

\begin{definition}
  \label{def:l-beta}
  The \emph{one-step $\beta$-reduction}, written $\to$, is defined on
  concrete terms by the rules
  \[
  \infer[]{(\lambda x.a)b \to a[b/x]}{} \quad \infer[]{ab \to a'b}{ a
    \to a' } \quad \infer[]{ab \to ab'}{b \to b'} \quad
  \infer[.]{\lambda x.a \to \lambda x.a'}{a \to a'}
  \]
  The \emph{$\beta$-reduction}, written $\too$, is the reflexive and
  transitive closure of $\to$.
\end{definition}

\begin{remark}
  \label{rem:ref-trans}
  In Definition~\ref{def:l-beta}, $\too$ is defined as the
  \emph{reflexive and transitive closure} of $\to$. This defines
  $\too$ as the smallest relation containing $\to$ that is reflexive
  and transitive.
\end{remark}

With the reduction of concrete terms now defined, we can confirm that
the concrete terms $\lambda x.x$ and $\lambda xy.x$ behave as expected
since for any concrete term $a$ we have
\[
(\lambda x.x)a \to a \quad \mbox{ and } \quad (\lambda xy.x)a \to
\lambda z.a
\]
where $z\notin \FV(a)\cup \s{x}$.

% --------------------------------------------------------------------
\subsection{Abstract terms}
\label{ssect:alpha-eq}

As noted in Remark~\ref{rem:fresh}, the reduction $(\lambda xy.x)a \to
\lambda z.a$ depends on the choice of an ordering of the set $\vset$,
since $z$ is the least element of $\vset \setminus (\FV(a)\cup
\s{x})$. Thus, two different orderings of $\vset$ will yield two
different concrete terms $\lambda z.a$ and $\lambda z'.a$. This
difference, however, is inessential since both terms define the same
function $L\to L$ because $z,z'\notin\FV(a)$. The same inessential
difference occurs between the concrete terms $\lambda z.z$ and
$\lambda z'.z'$. To remove this distinction, we define an equivalence
relation $=_\alpha$ on $L$ that equates the concrete terms differing
only in the name of their bound variables. The idea behind this
equivalence is that in the concrete term $\lambda x.a$, the
occurrences of the variable $x$ in $a$ are place holders, rather than
variables possessing an intrinsic identity. They can therefore be
renamed without affecting the overall meaning of the concrete term.

\begin{definition}
  \label{def:alpha-eq}
  The \emph{$\alpha$-equivalence}, written $=_\alpha$, is defined on
  concrete terms as the smallest equivalence relation satisfying the
  rules
  \[
  \infer[]{\lambda x.a =_\alpha \lambda y. (a\s{y/x})}{y \notin a}
  \quad \infer[]{ab =_\alpha a'b}{a =_\alpha a'} \quad \infer[]{ab
    =_\alpha ab'}{b =_\alpha b'} \quad \infer[.]{\lambda x.a =_\alpha
    \lambda x.a'}{a =_\alpha a'}
  \]
  where $y\notin a$ means that $y$ does not appear in $a$.
\end{definition}

\begin{definition}
  \label{def:l-abs-terms}
  The \emph{abstract terms} of the untyped lambda calculus are the
  elements of $L/{=_{\alpha}}$. We write $\Lambda$ for the set of all
  abstract terms.
\end{definition}

If two concrete terms $a$ and $b$ are $\alpha$-equivalent, then they
have the same structure (i.e., $a$ is an application if and only if
$b$ is an application, and so on). For this reason, we keep writing
$x$, $ab$, and $\lambda x.a$ for abstract terms even though we are
dealing with equivalence classes of concrete terms.

When defining a function or a relation on abstract terms using the
underlying concrete terms, we should make sure that this function or
relation is well-defined. For example, if we define a function $F$ on
abstract terms in this way, then we should verify that $a =_\alpha b$
implies $F(a)=F(b)$. One can check that the notions of free variable,
capture-avoiding substitution and $\beta$-reduction are well-defined
on abstract terms. In these cases, and in what follows, we generally
overlook this obligation and omit the proofs of well-definedness.

Because from now on we will always be manipulating abstract terms, we
refer to these as \emph{terms} for brevity.

% --------------------------------------------------------------------
\subsection{Properties of the untyped lambda calculus}
\label{ssect:props-untyped}

An important property of the untyped lambda calculus is that it forms
a complete model of computation in the sense of the Church-Turing
thesis.

\begin{proposition}
  \label{prop:l-T-complete}
  The untyped lambda calculus is Turing-complete, i.e., every Turing
  machine can be simulated by a lambda term.
\end{proposition}

A term $a$ may have any number of redexes. If $a$ has no redexes, then
the computation of $a$ is finished.

\begin{definition}
  \label{def:l-normal-form}
  A term that does not contain any redexes is \emph{reduced} or
  \emph{in normal form}. If $b$ is reduced and $a\too b$ we say that
  $b$ is a normal form for $a$.
\end{definition}

For example, a variable $x$ is reduced. Similarly, the terms $\lambda
x.x$ and $\lambda xy.x$ are both reduced. The following term, on the
other hand, is not reduced, since it contains two redexes
\[
(\lambda x.x)((\lambda y.z)x').
\]
When reducing such a term, nothing in the definition of the
$\beta$-reduction tells us which redex to reduce first. We say that
the $\beta$-reduction is \emph{non-deterministic}.

\begin{proposition}
  \label{prop:l-confluence}
  The untyped lambda calculus is confluent, i.e., if $a$, $b$, and
  $b'$ are terms such that $a\too b$ and $a\too b'$, then there exists
  a term $c$ such that $b\too c$ and $b'\too c$.
\end{proposition}

Proposition~\ref{prop:l-confluence} was first established by Church
and Rosser in \cite{Church-Rosser} and is therefore known as the
\emph{Church-Rosser property}. It is also referred to as the
\emph{confluence} property of the $\beta$-reduction. Confluence
guarantees that, despite the non-determinism of the reduction, normal
forms are unique.

\begin{corollary}
  \label{cor:l-uniqueness-normal-form}
  Let $a$ be a term. If $a$ has a normal form, then it is unique.
\end{corollary}

The uniqueness of normal forms implies that the lambda calculus is
\emph{consistent} in the sense that not all terms are equated by the
$\beta$-reduction.

\begin{corollary}
  \label{cor:l-consistency}
  The untyped lambda calculus is consistent, i.e., there exists two
  terms $a$ and $b$ such that $a \not\equiv_\beta b$, where
  $\equiv_\beta$ denotes the reflexive, symmetric, and transitive
  closure of $\to$.
\end{corollary}

% ====================================================================
\section{The simply typed lambda calculus}
\label{sect:st}

The term $xx$ is a well-formed term of the untyped lambda calculus. If
we think of this term in light of the interpretation discussed in
Section~\ref{sect:lambda} we are led to interpret the variable $x$ as
both function and argument in $xx$. This unusual construction is
admitted in the untyped lambda calculus because there are no notions
of domain and codomain for terms. \emph{Types} can be seen as a method
to endow the lambda calculus with these notions.

% --------------------------------------------------------------------
\subsection{Terms}
\label{ssect:st-terms}

The language of the simply typed lambda calculus is an extension of
the language of the untyped lambda calculus.

\begin{definition}
  \label{def:st-terms}
  The \emph{terms} of the simply typed lambda calculus are defined by
  \[
  a,b \quad \bnf \quad x \bor \lambda x.a \bor ab \bor * \bor \p{a,b}
  \bor \letin{*}{a}{b} \bor \letin{\p{x,y}}{a}{b}
  \]
  where $x$ and $y$ are variables of the untyped lambda calculus.
\end{definition}

Definition~\ref{def:st-terms} extends the language of the untyped
lambda calculus by adding the \emph{constant} $*$ as well as two new
term forming operations. The intended meaning of these new terms is as
follows.
\begin{itemize}
\item $\p{a,b}$ is the pair of $a$ and $b$.
\item $*$ is the empty pair, i.e., the 0-ary version of $\p{a,b}$.
\item $\letin{\p{x,y}}{a}{b}$ is a term that will reduce $a$ and in
  case $a\too \p{b_1, b_2}$ will assign $b_1$ to $x$ and $b_2$ to $y$.
\item $\letin{*}{a}{b}$ is the nullary version of
  $\letin{\p{x,y}}{a}{b}$.
\end{itemize}

We extend the notion of free-variables of a term to account for the
new term forming operations.

\begin{definition}
  \label{def:st-free-var}
  The set of \emph{free variables} of a term $a$ of the simply typed
  lambda calculus, written $\FV(a)$, is defined as
  \begin{itemize}
  \item $\FV(x)=\s{x}$,
  \item $\FV(ab)=\FV(a)\cup \FV(b)$,
  \item $\FV(\lambda x.a)=\FV(a)\setminus\s{x}$,
  \item $\FV(*)=\emptyset$,
  \item $\FV(\p{a,b})=\FV(a)\cup \FV(b)$,
  \item $\FV(\letin{*}{a}{b})=\FV(a)\cup \FV(b)$, and
  \item $\FV(\letin{\p{x,y}}{a}{b})=\FV(a)\cup (\FV(b)\setminus
    \s{x,y})$.
  \end{itemize}
\end{definition}

The notions of $\alpha$-equivalence and capture-avoiding substitution
can be extended to the setting of the simply typed lambda
calculus. Note that in $\letin{\p{x,y}}{a}{b}$ the variables $x$ and
$y$ are bound in $b$ (but not in $a$) so that $\mathtt{let}$ is a
binder. As in the untyped case, we say of a term $a$ such that
$\FV(a)=\emptyset$ that it is closed.

% --------------------------------------------------------------------
\subsection{Operational semantics}
\label{ssect:st-red}

To account for the new term forming operations of our extended
language, we introduce additional reduction rules.

\begin{definition}
  \label{def:st-beta}
  The \emph{one-step $\beta$-reduction}, written $\to$, is defined on
  the terms of the simply typed lambda calculus by the rules of
  Definition~\ref{def:l-beta} as well as those given in
  Figure~\ref{fig:st-beta-red}. The \emph{$\beta$-reduction}, written
  $\too$, is the reflexive and transitive closure of $\to$.
\end{definition}

% ....................................................................
\begin{figure}
  \[
  \infer[]{\letin{*}{*}{a} \to a}{} \quad
  \infer[]{\letin{\p{x,y}}{\p{b,c}}{a} \to a[b/x,c/y]}{}
  \]
  \[
  \infer[]{\p{a,b} \to \p{a,b'}}{ b \to b'} \quad \infer[]{\p{a,b} \to
    \p{a',b}}{a \to a'}
  \]
  \[
  \infer[]{\letin{*}{a}{b} \to \letin{*}{a'}{b}}{a \to a'}
  \]
  \[
  \infer[]{\letin{\p{x,y}}{a}{b} \to \letin{\p{x,y}}{a'}{b}}{ a \to
    a'}
  \]
  \[
  \infer[]{\letin{*}{a}{b} \to \letin{*}{a}{b'}}{b \to b'}
  \]
  \[
  \infer[]{\letin{\p{x,y}}{a}{b} \to \letin{\p{x,y}}{a}{b'}}{ b \to
    b'}
  \]
  \label{fig:st-beta-red}
  \caption[Reduction rules for the simply typed lambda
  calculus.]{Additional reduction rules for the simply typed lambda
    calculus.}
  \rule{\textwidth}{0.1mm}
\end{figure}
% ....................................................................

% --------------------------------------------------------------------
\subsection{Types}
\label{ssect:st-types}

\begin{definition}
  \label{def:st-types}
  The \emph{types} of the simply typed lambda calculus are defined by
  \[
  A,B \quad \bnf \quad X \bor (A \times B) \bor 1 \bor (A \to B)
  \]
  where $X$ comes from a set $\tset$ of \emph{basic types}.
\end{definition}

It can be useful to think of types as sets of terms. Under this
interpretation, we have
\begin{itemize}
\item $(A \times B)$ is the set of pairs,
\item $1$ is the set containing the unique empty tuple, and
\item $(A \to B)$ is the set of functions from $A$ to $B$.
\end{itemize}

We now explain how to use types to restrict the formation of terms.

\begin{definition}
  \label{def:st-typing-context}
  A \emph{typing context} is a finite set $\s{x_1:A_1,\ldots,x_n:A_n}$
  of pairs of a variable and a type, such that no variable occurs more
  than once. The expressions of the form $x:A$ in a typing context are
  called \emph{type declarations}.
\end{definition}

\begin{definition}
  \label{def:st-typing-judgment}
  A \emph{typing judgment} is an expression of the form
  \[
  \Gamma \entails a:A
  \]
  where $\Gamma$ is a typing context, $a$ is a term, and $A$ is a
  type.
\end{definition}

\begin{definition}
  \label{def:st-valid}
  A typing judgment is \emph{valid} if it can be inferred from the
  rules given in Figure~\ref{fig:st-trules}.
\end{definition}

% ....................................................................
\begin{figure}
  \[
  \infer[\rul{ax}]{\Gamma, x:A \entails x:A}{}
  \]
  \[
  \infer[\rul{\lambda}]{\Gamma \entails \lambda x.b:A\to B}{
    \Gamma,x:A \entails b:B} \quad \infer[\rul{app}]{\Gamma \entails
    ca:B}{ \Gamma \entails c:A\to B & \Gamma\entails a:A }
  \]
  \[
  \infer[\rul{*_i}]{\Gamma \entails *: 1}{ } \quad
  \infer[\rul{*_e}]{\Gamma \entails \letin{*}{b}{a}:A}{ \Gamma
    \entails b:1 & \Gamma \entails a:A }
  \]
  \[
  \infer[\rul{\times_i}]{\Gamma \entails \p{a,b} : A \times B}{\Gamma
    \entails a:A & \Gamma \entails b:B} \quad
  \infer[\rul{\times_e}]{\Gamma \entails \letin{\p{x,y}}{b}{a}:A}{
    \Gamma \entails b: (B_1\times B_2) & \Gamma, x:B_1, y:B_2 \entails
    a:A }
  \]
  \label{fig:st-trules}
  \caption{Typing rules for the simply typed lambda calculus.}
  \rule{\textwidth}{0.1mm}
\end{figure}
% ....................................................................

If $a$ is a term, one shows that $\Gamma \entails a:A$ is valid by
exhibiting a \emph{typing derivation}. If such a derivation exists, we
say that $a$ is \emph{well-typed of type $A$}, or sometimes simply
\emph{well-typed}. For example, below are two typing derivations,
establishing that both $\lambda x.x$ and $\lambda xy.x$ are
well-typed.
\[
\infer[]{\entails \lambda x.x : X \to X}{ \infer[]{x:X\entails x:X}{}}
\qquad \infer[]{\entails \lambda xy.x : X \to (Y\to X)}{
  \infer[]{x:X\entails \lambda y.x:Y\to X}{ \infer[]{x:X, y:Y
      \entails x:X}{}}}
\]

The term $xx$, however, is not well-typed. Indeed, suppose a typing
derivation $\pi$ of $\Gamma \entails xx : B$ exists. Then the last
rule of $\pi$ must be the $\rul{app}$ rule, since this rule is the
only one allowing the construction of an application. Moreover, the
only rule permitting the introduction of a variable is the $\rul{ax}$
rule. The typing derivation $\pi$ must therefore be the following.
\[
\infer[.]{\Gamma \entails xx:B}{ \infer[]{\Gamma, x:A\to B \entails
    x:A\to B}{} & \infer[]{\Gamma, x:A\entails x:A}{}}
\]
But there are no types $A$ and $B$ such that $A=A\to B$. Hence there
is no typing derivation of $\Gamma \entails xx : B$.

% --------------------------------------------------------------------
\subsection{Properties of the type system}
\label{ssect:st-strong-norm}

The type system of the simply typed lambda calculus restricts the
construction of terms in order to syntactically rule out
``ill-behaved'' terms. To verify that the type system achieves this
intended goal, we need to prove that all well-typed terms ``behave
well''. This intuitive idea is captured by establishing the \emph{type
  safety} of the language. Following \cite{WrFe94}, we consider that a
language is type safe if it enjoys the \emph{subject reduction} and
\emph{progress} properties. The latter property relies on a notion of
\emph{value}. These values are a particular set of distinguished
normal forms. The definition of value varies from language to
language. For now, we can take the set of values to consist of all
normal forms. Later, when we consider specific languages, this
definition will be adjusted and explicitly stated.
\begin{description}
\item[Subject reduction:] This property guarantees that the type of a
  term is stable under reduction.  As a corollary, it also shows that
  if a term is well-typed, then it never reduces to an ill-typed term.
\item[Progress:] This property shows that a well-typed closed term is
  either a value or admits further reductions.
\end{description}

The simply typed lambda calculus is type safe as it enjoys both of the
above properties.

\begin{proposition}[Subject reduction]
  \label{prop:st-subj-red}
  If $\Gamma \entails a:A$ and $a\to a'$, then $\Gamma \entails a':A$.
\end{proposition}

\begin{proposition}[Progress]
  \label{prop:st-progress}
  If $\entails a:A$, then either $a$ is a value or there exists $a'$
  such that $a\to a'$.
\end{proposition}

The simply typed lambda calculus also enjoys a property known as
\emph{strong normalization}.

\begin{definition}
  \label{def:l-normalization}
  A term $a$ is \emph{weakly normalizing} if there exists a finite
  sequence of reductions $a \to \ldots \to b$ where $b$ is in normal
  form, and \emph{strongly normalizing} if every sequence of
  reductions starting from $a$ is finite.
\end{definition}

Note that any strongly normalizing term is also weakly
normalizing. Variables are examples of strongly normalizing terms. As
an example of a term that is neither strongly nor weakly normalizing,
consider $\Omega\Omega$, where $\Omega$\label{defomegalambda} is the
term $\lambda x.xx$. $\Omega\Omega$ is neither weakly nor strongly
normalizing since we have
\[
\Omega\Omega = (\lambda x.xx)\lambda x.xx \to xx[\lambda x.xx / x] =
(\lambda x.xx)\lambda x.xx = \Omega\Omega.
\]
As an example of a term that is weakly but not strongly normalizing,
consider $(\lambda z.y)(\Omega\Omega)$.

Since $\Omega$ contains $xx$, we know that $\Omega\Omega$ is not
well-typed. In contrast, the well-typed terms we have encountered so
far, $\lambda x.x$ and $\lambda xy.x$, are both strongly
normalizing. In fact, the simply typed lambda calculus has the
property that \emph{all} well-typed terms are strongly normalizable.

\begin{proposition}[Strong normalization]
  \label{prop:st-strong-norm}
  If $\entails a:A$ is a valid typing judgment, then $a$ is strongly
  normalizing.
\end{proposition}

% ====================================================================
\section{Linearity}
\label{sect:linlog}

We now sketch a version of Girard's \emph{intuitionistic linear logic}
(\cite{Gir87}). As we shall see in the next section the use of linear
logic in the context of quantum computation is motivated by the
no-cloning property of quantum information.

% --------------------------------------------------------------------
\subsection{Contraction, weakening, and strict linearity}
\label{ssect:strict-lin}

Informally, a variable is used \emph{linearly} if it is used exactly
once. In the simply typed lambda calculus, variables can be used
\emph{non-linearly}. As a first example, consider the following typing
derivation.
\[
\infer[]{x:A \entails \p{x,x}:A\times A}{\infer{x:A \entails x:A }{} &
  \infer{x:A \entails x:A}{}}
\]
In the above derivation the variable $x$ is used non-linearly, in the
sense that only a single occurrence of $x$ in the context is required
to construct the pair $\p{x,x}$ in which $x$ occurs twice. This is
possible because the two occurrences of the declaration $x:A$ in the
leaves of the typing derivation were implicitly \emph{contracted} by
the application of the $(\times_i)$ rule. As another example, note
that the typing judgement $x:A, y:B \entails y:B$ is valid by the
$(ax)$ rule. In this case, the variable $x$ is handled non-linearly
because it appears in the context but is not used at all. This second
kind of non-linearity is due to the implicit \emph{weakening} of the
context in the $(ax)$ rule.

It is possible to modify the typing rules of the simply typed lambda
calculus to force variables to be used in a strictly linear
fashion. To obtain such a system, we can replace the $\rul{ax}$ and
$\rul{*_i}$ rules with
\[
\infer{x:A \entails x:A}{} \quad \mbox{ and } \quad \infer[.]{\entails
  *:1}{}
\]
In the above rules, the typing contexts are minimal, which guarantees
that no implicit weakening can occur. To forbid implicit contractions
we need to ensure that contexts are not merged but juxtaposed in
binary rules. This can be achieved in the case of the $\rul{\times_i}$
rule as follows.
\[
\infer[]{\Gamma_1,\Gamma_2 \entails \p{a,b}:A\times
  B}{\deduce{\Gamma_1\entails a:A}{} & \deduce{\Gamma_2\entails
    b:B}{}}
\]
The above rule carries the side condition that the contexts $\Gamma_1$
and $\Gamma_2$ are distinct, so that the notation $\Gamma_1,\Gamma_2$
denotes the disjoint union of the two contexts.

If contexts are removed from all nullary rules and juxtaposed rather
than merged in all binary rules, then we obtain a strictly linear
system. In this system, a variable occurs in the context of a valid
typing judgement if and only if it appears exactly once in the term
being typed.

% --------------------------------------------------------------------
\subsection{Reintroducing non-linearity}
\label{ssect:exponentials}

The restrictions imposed by the strictly linear type system sketched
above are very strong. Our interest in Girard's linear logic is the
fact that it allows us to reintroduce a controlled form of
non-linearity. The idea is to use a modality called \emph{bang} and
denoted $!$ to identify the variables that can be used
non-linearly. To this end, we extend the grammar of types as follows.
\[
A,B \quad \bnf \quad X \bor 1 \bor (A \times B) \bor (A \to B) \bor
!A \label{defbang}
\]
The new type $!A$ consists of all the elements of the type $A$ that
can be used non-linearly. One can think of the elements of type $!A$
as those elements of $A$ that have the property of being
\emph{reusable} or \emph{duplicable}.

We can modify the typing rules to account for this new modality. For
example, the $\rul{\times_i}$ rule becomes
\[
\infer[.]{!\Delta,\Gamma_1,\Gamma_2 \entails \p{a,b}:A\times
  B}{\deduce{!\Delta,\Gamma_1\entails a:A}{} &
  \deduce{!\Delta,\Gamma_2\entails b:B}{}}
\]
where the context $!\Delta$ denotes a set of declarations of the form
$x_1:{!A_1},\ldots x_n:{!A_n}$. In this new rule, the contracted part
of the context consists exclusively of declarations of the form
$x:{!A}$. This ensures that the only variables that are used
non-linearly are the ones of a non-linear type.

It should be possible to use a duplicable variable only once. In other
words, if a variable $x$ is declared of type $!A$, it should also have
type $A$. One way to achieve this is to equip the type system with a
\emph{subtyping relation}\label{defsubtyping}, denoted $<:$,
satisfying $!A<:A$ for every type $A$.

% ====================================================================
\section{The quantum lambda calculus}
\label{sect:qlc}

Various lambda calculi for quantum computation have appeared in the
literature (e.g., \cite{Tonder}, \cite{Arrighi-Dowek}). Here, we focus
on the \emph{quantum lambda calculus} (see \cite{valiron04},
\cite{SeVa06}, or \cite{SeVa09}) as it is the main inspiration for the
Proto-Quipper language defined and studied in chapters~\ref{chap:pq}
and \ref{chap:pq-safe}.

The quantum lambda calculus is based on the QRAM model of quantum
computation described in Section~\ref{sssect:qram}. To embody the QRAM
model, the reduction relation of the quantum lambda calculus is
defined on \emph{closures}. These closures are triples $[Q,L,a]$ where
$Q$ is a unit vector in $\bigotimes_{i=1}^{n}\Comp^2$ for some integer
$n$, $L$ is a list of $n$ distinct term variables, and $a$ is a term
of the quantum lambda calculus.

The vector $Q$ represents the state of a system of $n$ qubits held in
some hypothetical quantum device. In a well-formed closure, the free
variables of $a$ are required to form a subset of $L$ and the list $L$
is interpreted as a link between the variables of $a$ and the qubits
of $Q$. This way, the qubits whose state is described by $Q$ become
accessible to the operations of the quantum lambda calculus.

% --------------------------------------------------------------------
\subsection{Terms}
\label{sect:qlc-terms}

\begin{definition}
  \label{def:ql-terms}
  The \emph{terms} of the quantum lambda calculus are defined by
  \begin{center}
    \begin{tabular}{rl}
      $a,b,c \quad \bnf$ & $x \bor u \bor \lambda x.a \bor ab \bor 
      \p{a,b} \bor * \bor$ \\[0.05in]
      & $\letin{\p{x,y}}{a}{b} \bor \letrecin{x~y}{b}{c} 
      \bor$ \\[0.05in]
      & $\injl(a) \bor \injr(a) \bor \matchwith{a}{b}{c}$
    \end{tabular}
  \end{center}
  where $u$ comes from a set $U$ of \emph{quantum constants} and $x,y$
  come from a countable set $\vset$ of \emph{variables}.
\end{definition}

The meaning of most terms is intended to be the standard one, as
described in the previous sections. The term $\letrecin{x~y}{b}{c}$ is
a recursion operator. The terms $\injl(a)$ and $\injr(a)$ denote the
left and right inclusion in a disjoint union respectively. The term
$\matchwith{a}{b}{c}$ denotes a case distinction depending on $a$. The
terms $\injl(*)$ and $\injr(*)$ form a two element set on which one
can perform a case distinction. The classical \emph{bits} are defined
as the elements of this set, with $0=\injr(*)$ and
$1=\injl(*)$\label{defbitsqlc}.

The set $U$ contains syntactical representatives of certain operations
that can be executed by the quantum device. In particular, $U$ is
assumed to contain the constants $\unew$ and $\umeas$, whose intended
interpretation is as follows.
\begin{itemize}\label{defnewmeas}
\item The term $\unew$ represents an initialization function. It
  inputs a bit (i.e., one of $0$ or $1$ as defined above) and produces
  a qubit in the corresponding classical state (i.e., $\ket{1}$ or
  $\ket{0}$ respectively).
\item The term $\umeas$ represents a measurement function. It inputs a
  qubit and measures it in the computational basis, returning the
  corresponding bit.
\end{itemize}

Assume that $U$ contains a constant $H$ representing the Hadamard gate
and consider the term $\coin$\label{defcoin} defined as
\[
\coin = \lambda * . \umeas (H (\unew 0)).
\]
This term represents a ``fair coin''. When applied to any argument, it
will prepare a qubit in the state $\ket{0}$ and apply a Hadamard gate
to it. This results in the superposition
\[
\frac{\ket{0}+\ket{1}}{\sqrt 2}.
\]
The qubit is then measured, which results in $0$ or $1$ with equal
probability.

% --------------------------------------------------------------------
\subsection{Operational semantics}
\label{sect:qlc-red}

In the quantum lambda calculus, one must choose a \emph{reduction
  strategy}. To see why this is the case, assume that
$\bplus$\label{defbplus} is a term from the quantum lambda calculus
representing addition modulo 2 and consider the following term
\[
a=(\lambda x. \bplus \p{x,x}) (\coin *).
\]
The term $a$ has two redexes. If the outer redex is reduced first, we
obtain the term $\bplus \p{(\coin *),(\coin *)}$, which will reduce to
$0$ or $1$ with equal probability. However, if we evaluate $(\coin *)$
first, then $a$ reduces to
\[
(\lambda x. \bplus \p{x,x})0 \quad \mbox{ or } \quad (\lambda
x. \bplus \p{x,x})1
\]
with equal probability. Either way, the final result of computation
will be $0$. This example shows that confluence fails in the quantum
lambda calculus, forcing us to choose an order of evaluation. In the
quantum lambda calculus, evaluation of terms follows a
\emph{call-by-value} reduction strategy. In particular, this means
that when evaluating an application we reduce the argument before
applying the function. To determine when a term is reduced, we define
a notion of \emph{value}.

\begin{definition}
  \label{def:ql-values}
  The \emph{values} of the quantum lambda calculus are defined by
  \[
  v,w \quad \bnf \quad x \bor u \bor * \bor \p{v,w} \bor \lambda x.a
  \bor \injr(v) \bor \injl(v).
  \]
\end{definition}

\begin{definition}
  \label{def:ql-closure}
  A \emph{closure} is a triple $[Q,L,a]$ where
  \begin{itemize}
  \item $Q$ is a normalized vector of $\Comp^{2^n}$ for some $n\geq 0$
    called a \emph{quantum array},
  \item $L$ is a list of $n$ distinct term variables, and
  \item $a$ is a term whose free variables appear in $L$.
  \end{itemize}
  The closure $[Q,L,a]$ is a \emph{value} if $a$ is a value.
\end{definition}

\begin{definition}
  \label{def:ql-beta}
  The \emph{one-step $\beta$-reduction}, written $\to_p$, is defined
  on closures by the rules given in Figure~\ref{fig:ql-red-rules}. The
  notation $[Q,L,a]\to_p [Q',L',a']$ means that the reduction takes
  place with probability $p$.
\end{definition}

% ....................................................................
\begin{figure}
  \[
  \infer[]{[Q,L,av] \to_p [Q',L',a'v]}{[Q,L,a] \to_p [Q',L',a']} \quad
  \infer[]{[Q,L,ab] \to_p [Q',L',ab']}{[Q,L,b] \to_p [Q',L',b']}
  \]
  \[
  \infer[]{[Q,L,\p{a,b}] \to_p [Q', L',\p{a,b'}]}{[Q,L,b] \to_p
    [Q',L',b']} \quad \infer[]{[Q,L,\p{a,v}] \to_p
    [Q',L',\p{a',v}]}{[Q,L,a] \to_p [Q',L',a']}
  \]
  \[
  \infer[]{[Q,L,\injl(a)] \to_p [Q', L',\injl(a')]}{[Q,L,a] \to_p
    [Q',L',a']} \quad \infer[]{[Q,L,\injr(a)] \to_p [Q',
    L',\injr(a')]}{[Q,L,a] \to_p [Q',L',a']}
  \]
  \[
  \infer[]{[Q,L,\letin{\p{x,y}}{a}{b}] \to_p
    [Q',L',\letin{\p{x,y}}{a'}{b}]}{[Q,L,a] \to_p [Q',L',a']}
  \]
  \[
  \infer[]{[Q,L,\matchwith{a}{b}{c}] \to_p
    [Q',L',\matchwith{a'}{b}{c}]}{[Q,L,a] \to_p [Q',L',a']}
  \]
  ~
  \[
  \deduce[]{[Q,L,(\lambda x.a)v] \to_1 [Q,L, a[v/x]]}{}
  \]
  \[
  \deduce[]{[Q,L,\letin{*}{*}{a}] \to_1 [Q,L,a]}{}
  \]
  \[
  \deduce[]{[Q,L,\letin{\p{x,y}}{\p{v,w}}{a}] \to_1
    [Q,L,a[v/x,w/y]]}{}
  \]
  \[
  \deduce[]{[Q,L,\matchwith{\injl(v)}{b}{c}] \to_1 [Q,L,b[v/x]]}{}
  \]
  \[
  \deduce[]{[Q,L,\matchwith{\injr(v)}{b}{c}] \to_1 [Q,L,c[v/y]]}{}
  \]
  \[
  \deduce[]{[Q,L,\letrecin{x~y}{b}{c}] \to_1 [Q,L,c[(\lambda
    y.\letrecin{x~y}{b}{b})/x]]}{}
  \]
  ~
  \[
  \deduce[]{[Q, L, u\p{x_{j_1},\ldots,x_{j_n}}] \to_1 [Q', L,
    \p{x_{j_1},\ldots,x_{j_n}}]}{}
  \]
  \[
  \deduce[]{[\alpha\ket{Q_0}+\beta\ket{Q_1}, L, \umeas(x_i)]
    \to_{|\alpha|^2} [\ket{Q_0}, L, 0]}{}
  \]
  \[
  \deduce[]{[\alpha\ket{Q_0}+\beta\ket{Q_1},L, \umeas(x_i)]
    \to_{|\beta|^2} [\ket{Q_1}, L, 1]}{}
  \]
  \[
  \deduce[]{[Q, \ket{x_1\ldots x_n},\unew(0)] \to_1 [Q\X\ket{0},
    \ket{x_1\ldots x_{n+1}},x_{n+1}]}{}
  \]
  \[
  \deduce[]{[Q, \ket{x_1\ldots x_n},\unew(1)] \to_1 [Q\X\ket{1},
    \ket{x_1\ldots x_{n+1}},x_{n+1}]}{}
  \]
  \label{fig:ql-red-rules}
  \caption{Reduction rules for the quantum lambda calculus.}
  % \rule{\textwidth}{0.1mm}
\end{figure}
% ....................................................................

The rules are separated in three groups. The first group contains the
\emph{congruence rules}. In particular, the rules for the reduction of
an application can be seen to define a call-by-value reduction
strategy. The second group of rules contains the \emph{classical
  rules}. These rules define the reduction of redexes that do not
involve any of the constants from the set $U$. The last group of rules
contains the \emph{quantum rules}. These rules define the interaction
between the classical device and the quantum device.  In the first
quantum rule, we have $Q'=u(Q)$. This rule corresponds to the
application of the unitary $u$ to the relevant qubits. Note that the
only probabilistic reduction step is the one corresponding to
measurement.

The chosen reduction strategy guarantees that, at every step of a
reduction, only one rule applies. Hence, unlike the untyped lambda
calculus, the quantum lambda calculus is \emph{deterministic}.

\begin{proposition}
  \label{ql-determinicity}
  If $[Q,L,a]$ is a closure, then at most one reduction rule applies
  to it.
\end{proposition}

% --------------------------------------------------------------------
\subsection{Types}
\label{sect:qlc-types}

\begin{definition}
  \label{def:ql-types}
  The \emph{types} of the quantum lambda calculus are defined by
  \[
  A,B \quad \bnf \quad \qubit \bor 1 \bor \bang A \bor A\X B \bor A
  \oplus B \bor A\loli B .
  \]
\end{definition}

The type system of the quantum lambda calculus is based on
intuitionistic linear logic as sketched in
Section~\ref{sect:linlog}. The notation is adopted from linear logic,
with $A\x B$ for the type of pairs, $A\+ B$ for the type of sums, and
$A\loli B$ for the type of functions. The type $\qubit$ represents the
set of all 1-qubit states. As in Section~\ref{sect:linlog}, the type
${!}A$ can be understood as the subset of $A$ consisting of values
that have the additional property of being \emph{duplicable} or
\emph{reusable}. We will sometimes write ${!}^nA$, with $n\in\N$, to
mean
\[
\underbrace{{!}\ldots {!}}_{n} A.
\]
Similarly, we sometimes write $A^{\x n}$ to mean
\[
\underbrace{A\x \ldots \x A}_{n}.
\]
We also write $\bit$ for the type $1\+ 1$.

The fact that a term is reusable should not prevent us from using it
exactly once. Intuitively, this should imply that if ${!}A$ is a valid
type for a given term, then $A$ should also be a valid type for it.
To capture this idea, we use a subtyping relation on types.

\begin{definition}
  \label{def:ql-subtyping}
  The \emph{subtyping relation} $<:$ is the smallest relation on types
  satisfying the rules given in Figure~\ref{fig:ql-subtyping}.
\end{definition}

% ....................................................................
\begin{figure}
  \[
  \infer[]{\qubit <: \qubit}{} \quad \infer[]{1 <: 1}{}
  \]
  \[
  \infer[]{(A_1\X A_2) <: (B_1 \X B_2)}{A_1<:B_1 & A_2<:B_2} \quad
  \infer[]{(A_1\+ A_2) <: (B_1 \+ B_2)}{A_1<:B_1 & A_2<:B_2} \quad
  \infer[]{(A_1\loli B_1) <: (A_2\loli B_2)}{A_2<:A_1 & B_1<:B_2}
  \]
  \[
  \infer[]{ \bang^nA <: \bang^mB}{ A<:B & (n=0 \imp m=0) }
  \]
  \label{fig:ql-subtyping}
  \caption{Subtyping rules for the quantum lambda calculus.}
  \rule{\textwidth}{0.1mm}
\end{figure}
% ....................................................................

\begin{proposition}
  The subtyping relation is reflexive and transitive.
\end{proposition}

To define the type system of the quantum lambda calculus, we first
introduce axioms for the elements of the set $U$.

\begin{definition}
  \label{def:constants-type}
  We introduce a type for the constants of $U$. For $\unew$ and
  $\umeas$ we set
  \[
  A_{\unew}= \bit \loli \qubit, \quad A_{\umeas} = \qubit \loli
  {!}\bit,
  \]
  and for the remaining elements $v\in U$ we set
  \[
  A_v= \qubit^{\x n}\loli \qubit^{\x n}.
  \]
\end{definition}

\begin{definition}
  \label{def:ql-valid}
  A typing judgment of the quantum lambda calculus is \emph{valid} if
  it can be inferred from the rules given in
  Figure~\ref{fig:ql-trules}.
\end{definition}

% ....................................................................
\begin{figure}[!ht]
  \[
  \infer[\rul{ax_1}]{\Gamma, x:A \entails x:B}{A<:B} \quad
  \infer[\rul{ax_2}]{\Gamma \entails u:B}{\bang A_u <:B}
  \]
  \[
  \infer[\rul{\+_{i_1}}]{\Gamma \entails \injl(a):\bang^n(A \+
    B)}{\Gamma \entails a:\bang^nA} \quad
  \infer[\rul{\+_{i_2}}]{\Gamma \entails \injr(b):\bang^n(A \+
    B)}{\Gamma \entails b:\bang^nB}
  \]
  \[
  \infer[\rul{\+_e}]{\Gamma_1, \Gamma_2, \bang \Delta \entails
    \matchwith{a}{b}{c}:C}{\Gamma_1, \bang \Delta \entails a :
    \bang^n(A \+ B) & \deduce[]{\Gamma_2, \bang \Delta, x:\bang^n A
      \entails b:C}{\Gamma_2, \bang \Delta, y:\bang^n A \entails c:C}}
  \]
  \[
  \infer[\rul{\lambda_1}]{\Gamma \entails \lambda x.b:A\loli B}{
    \Gamma,x:A \entails b:B } \quad \infer[\rul{\lambda_2}]{\Gamma,
    \bang \Delta \entails \lambda x.b:\bang^{n+1}(A\loli B)}{\Gamma,
    \bang \Delta, x:A\entails b:B & \FV(b)\cap |\Gamma| = \emptyset}
  \]
  \[
  \infer[\rul{app}]{\Gamma_1,\Gamma_2, \bang \Delta \entails
    ca:B}{\Gamma_1, \bang \Delta\entails c:A\loli B & \Gamma_2, \bang
    \Delta \entails a:A }
  \]
  \[
  \infer[\rul{*_i}]{\Gamma \entails *:\bang^n 1}{}
  \]
  \[
  \infer[\rul{\X_i}]{\Gamma_1,\Gamma_2, \bang \Delta \entails
    \p{a,b}:\bang^n(A\X B)}{\Gamma_1, \bang \Delta;Q_1\entails
    a:\bang^nA & \Gamma_2, \bang \Delta ;Q_2\entails b:\bang^nB }
  \]
  \[
  \infer[\rul{\X_e}]{\Gamma_1,\Gamma_2, \bang \Delta \entails
    \letin{\p{x,y}}{b}{a}:A}{\Gamma_1, \bang \Delta \entails
    b:\bang^n(B_1\X B_2) & \Gamma_2, \bang \Delta, x:\bang^nB_1,
    y:\bang^nB_2 \entails a:A }
  \]
  \[
  \infer[\rul{rec}]{\Gamma, \bang \Delta \entails
    \letrecin{x~y}{b}{c}:C}{\bang \Delta, x: \bang (A \loli B), y:A
    \entails b:B & \Gamma, \bang \Delta, x: \bang (A\loli B) \entails
    c:C}
  \]
  \label{fig:ql-trules}
  \caption{Typing rules for the quantum lambda calculus.}
  \rule{\textwidth}{0.1mm}
\end{figure}
% ....................................................................

There are two rules for the construction of a lambda abstraction. The
rule $\rul{\lambda_1}$ is similar to the $\rul{\lambda}$ rule of the
simply typed lambda calculus. However, the produced function is not
duplicable. In contrast, the rule $\rul{\lambda_2}$ produces a
duplicable function. The main difference is that in the
$\rul{\lambda_2}$ rule, the free variables that appear in $b$ must all
be of a duplicable type. This prevents $b$ from having any embedded
quantum data, which could not be cloned. Note that the type system
prevents us from assigning the type $\qubit \loli \qubit \X \qubit$ to
the term $\lambda x. \p{x,x}$.

\begin{definition}
  \label{def:ql-typed closure}
  A \emph{typed closure} is an expression of the form
  \[
  [Q,L,a]:A,
  \]
  where $[Q,L,a]$ is a closure and $A$ is a type. It is \emph{valid}
  if
  \[
  x_1:\qubit, \ldots, x_n:\qubit \entails a:A
  \]
  is a valid typing judgement, with $L=\ket{x_1,\ldots,x_n}$.
\end{definition}

The quantum lambda calculus is a type safe language, in the sense that
it enjoys the subject reduction and progress properties.

\begin{proposition}
  \label{prop:ql-subject-red}
  If $[Q,L,a]:A$ is a valid typed closure and 
  \[
  [Q,L,a] \to_p [Q',L',a'],
  \]
  then $[Q',L',a']:A$ is a valid typed closure.
\end{proposition}

\begin{proposition}
  \label{prop:ql-progress}
  Let $[Q,L,a]$ be a valid typed closure of type $A$. Then either
  $[Q,L,a]$ is a value, or there is a valid typed closure $[Q',L',a']$
  such that $[Q,L,a] \to_p [Q',L',a']$. Moreover, the total
  probability of all possible one-step reductions from $[Q,L,a]$ is 1.
\end{proposition}

% =====================================================================
% Algebraic Methods in Quantum Computation

% ---------------------------------------------------------------------
\chapter{Grid problems}
\label{chap:grid-pb}

In this chapter, we present an efficient method to solve a type of
lattice point enumeration problem which we call a \emph{grid
  problem}. As a first approximation, a grid problem can be thought of
as follows: given a discrete subset $L\seq\R^2$ (such as a lattice),
which we call the \emph{grid}, and a bounded convex subset $A\seq\R^2$
with non-empty interior, enumerate all the points $u \in A \cap
L$. Specifically, we will be interested in grid problems for which $L$
is a subset of $\Zi$ or of $\Zomega$, as defined in
Chapter~\ref{chap:nb-th}. We refer to the first kind of problem as a
\emph{grid problem over $\Zi$} and to the second kind of problem as a
\emph{grid problem over $\Zomega$}. As we shall see in chapters
\ref{chap:synth-V} and \ref{chap:synth-T}, these problems have
applications in quantum computation. However, they are treated here
independently of any quantum considerations. The results contained in
this chapter first appeared in \cite{vsynth} and \cite{gridsynth}.

We note that the method presented here is not the only method for
solving grid problems. Alternatively, grid problems over $\Zi$ or
$\Zomega$ can be reduced to so-called \emph{integer programming
  problems} in some fixed dimension which can be efficiently solved
using the techniques pioneered by Lenstra in
\cite{Lenstra1983}. Nevertheless, we believe our method is novel and
interesting.

Even though grid problems over $\Zi$ are significantly simpler than
grid problems over $\Zomega$, the overall method remains the same. For
this reason, the case of $\Zi$, which is treated first, is used as an
introduction to the case of $\Zomega$, which is treated second.

Let $R$ be one of $\Zi$ or $\Zomega$. As in Chapter~\ref{chap:nb-th},
we quantify the complexity of our methods by estimating the number of
arithmetic operations required to produce an element $u \in A \cap L$,
for $L\seq R$. Our algorithms will input bounded convex subsets
$A\seq\R^2$ (as well as closed intervals $[x_0,x_1]\seq \R)$. If we
were to give a rigorous complexity-theoretic account, we should
indicate what it means for a subset $A$ of $\R^2$ to be ``given'' as
the input to an algorithm. The details of this do not matter much. For
our purposes, it will suffice to assume that a convex set is given
along with the following information.
\begin{itemize}
\item A convex polygon enclosing $A$, say with rational vertices, and
  such that the area of the polygon exceeds that of $A$ by at most a
  fixed constant factor;
\item a method to decide, for any given point of $R$, whether it is in
  $A$ or not; and
\item a method to compute the intersection of $A$ with any straight
  line in $R$. More precisely, given any straight line parameterized
  as $L(t)=p+tq$, with $p,q\in R$, we can effectively determine the
  interval $\s{t\mid L(t)\in A}$ in the sense of the above.
\end{itemize}

% ====================================================================
\section{Grid problems over \texorpdfstring{$\Zi$}{Z[i]}}
\label{sect:grid-pb-Zi}

Recall from Chapter~\ref{chap:nb-th} that the elements of $\Zi$ are of
the form $a+bi$, with $a$ and $b$ in $\Z$. We can therefore identify
$\Zi$ with the set $\Z^2\subseteq \R^2$. When viewed in this way, we
refer to $\Zi$ as the \emph{grid} and to elements $u\in\Zi$ as
\emph{grid points}.

\begin{problem}[Grid problem over $\Zi$]
  \label{pb:grid-Zi}
  Given a bounded convex subset $A$ of $\R^2$ with non-empty interior,
  enumerate all the points $u \in A \cap \Zi$.
\end{problem}

A point $u \in A\cap \Zi$ is called a \emph{solution} to the grid
problem over $\Zi$ for $A$. Figure~\ref{fig:grid-pb-Zi}~(a)
illustrates a grid problem for which $A$ is a disk centered at the
origin. The grid is shown as black dots and the set $A$ is shown in
red.

% ......................................................................
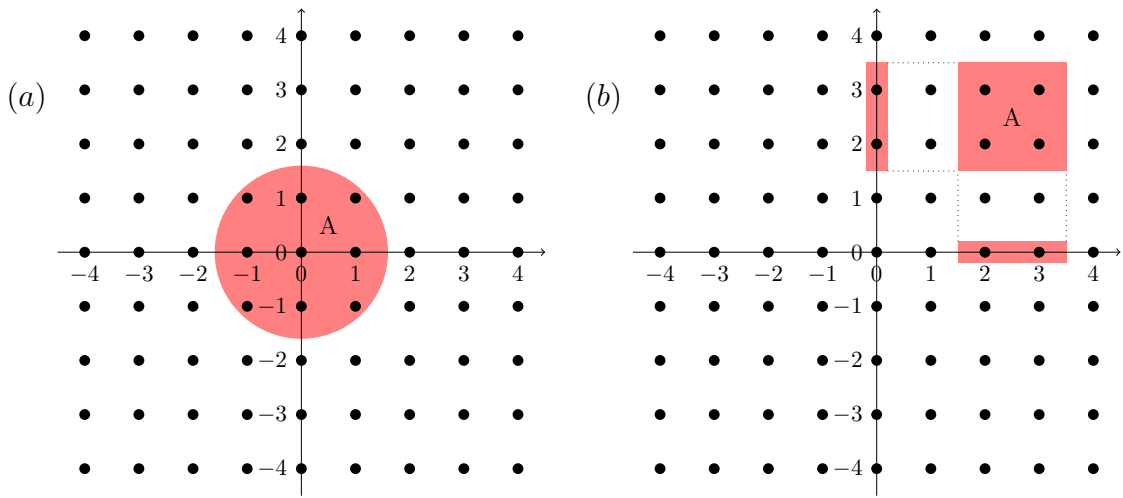
\begin{figure}
  \[
  (a)~ \mp{0.8}{\scalebox{0.8}{\begin{tikzpicture}[scale=0.9]

% The convex set A:
\fill[color=red!50] (0,0) circle (1.6);
\draw (0.5,0.5) node {A};

% The labelled x-axis and y-axis:
\draw[->] (-4.5,0) -- (4.5, 0);
\draw (-4,0) -- (-4,-0.1) node[below] {\small $-4$};
\draw (-3,0) -- (-3,-0.1) node[below] {\small $-3$};
\draw (-2,0) -- (-2,-0.1) node[below] {\small $-2$};
\draw (-1,0) -- (-1,-0.1) node[below] {\small $-1$};
\draw (0,0) -- (0,-0.1) node[below] {\small $0$};
\draw (1,0) -- (1,-0.1) node[below] {\small $1$};
\draw (2,0) -- (2,-0.1) node[below] {\small $2$};
\draw (3,0) -- (3,-0.1) node[below] {\small $3$};
\draw (4,0) -- (4,-0.1) node[below] {\small $4$};
\draw[->] (0,-4.5) -- (0,4.5);
\draw (0,-4) -- (-0.1,-4) node[left] {\small $-4$};
\draw (0,-3) -- (-0.1,-3) node[left] {\small $-3$};
\draw (0,-2) -- (-0.1,-2) node[left] {\small $-2$};
\draw (0,-1) -- (-0.1,-1) node[left] {\small $-1$};
\draw (0,0) -- (-0.1,0) node[left] {\small $0$};
\draw (0,1) -- (-0.1,1) node[left] {\small $1$};
\draw (0,2) -- (-0.1,2) node[left] {\small $2$};
\draw (0,3) -- (-0.1,3) node[left] {\small $3$};
\draw (0,4) -- (-0.1,4) node[left] {\small $4$};

% The points of the integer grid:
\foreach \x in {-4,...,4}
  \foreach \y in {-4,...,4}
    {\fill (\x,\y) circle (.1);}

\end{tikzpicture}}} \quad (b)~
  \mp{0.8}{\scalebox{0.8}{\begin{tikzpicture}[scale=0.9]

% The convex set A:
\fill[color=red!50] 
  (1.5,1.5) -- (1.5,3.5) -- (3.5,3.5) -- (3.5,1.5) -- cycle;
\draw (2.5,2.5) node {A};

% The projection of $B$ on the x- and y- axes:
\fill[thick,color=red!50] 
  (1.5,0.2) -- (3.5,0.2) -- (3.5,-0.2) -- (1.5,-0.2);
\fill[thick,color=red!50] 
  (0.2,1.5) -- (0.2,3.5) -- (-0.2,3.5) -- (-0.2,1.5);
\draw[dotted] (1.5,0.2) -- (1.5,1.5);
\draw[dotted] (0.2,1.5) -- (1.5,1.5);
\draw[dotted] (3.5,0.2) -- (3.5,1.5);
\draw[dotted] (0.2,3.5) -- (1.5,3.5);

% The labelled x-axis and y-axis:
\draw[->] (-4.5,0) -- (4.5, 0);
\draw (-4,0) -- (-4,-0.1) node[below] {\small $-4$};
\draw (-3,0) -- (-3,-0.1) node[below] {\small $-3$};
\draw (-2,0) -- (-2,-0.1) node[below] {\small $-2$};
\draw (-1,0) -- (-1,-0.1) node[below] {\small $-1$};
\draw (0,0) -- (0,-0.1) node[below] {\small $0$};
\draw (1,0) -- (1,-0.1) node[below] {\small $1$};
\draw (2,0) -- (2,-0.1) node[below] {\small $2$};
\draw (3,0) -- (3,-0.1) node[below] {\small $3$};
\draw (4,0) -- (4,-0.1) node[below] {\small $4$};
\draw[->] (0,-4.5) -- (0,4.5);
\draw (0,-4) -- (-0.1,-4) node[left] {\small $-4$};
\draw (0,-3) -- (-0.1,-3) node[left] {\small $-3$};
\draw (0,-2) -- (-0.1,-2) node[left] {\small $-2$};
\draw (0,-1) -- (-0.1,-1) node[left] {\small $-1$};
\draw (0,0) -- (-0.1,0) node[left] {\small $0$};
\draw (0,1) -- (-0.1,1) node[left] {\small $1$};
\draw (0,2) -- (-0.1,2) node[left] {\small $2$};
\draw (0,3) -- (-0.1,3) node[left] {\small $3$};
\draw (0,4) -- (-0.1,4) node[left] {\small $4$};

% The points of the integer grid:
\foreach \x in {-4,...,4}
  \foreach \y in {-4,...,4}
    {\fill (\x,\y) circle (.1);}

\end{tikzpicture}}}
  \]
  \caption[Grid problems over $\Zi$.]{(a) A grid problem over $\Zi$
    for which $A$ is a disk of radius $1.5$ centered at the
    origin. (b) A grid problem over $\Zi$ for which $A$ is an upright
    rectangle, together with the projections of $A$ along the $x$- and
    $y$-axes.}
  \label{fig:grid-pb-Zi}
  \rule{\textwidth}{0.1mm}
\end{figure}
% ......................................................................

% --------------------------------------------------------------------
\subsection{Upright rectangles}
\label{ssect:grid-pb-Zi-up-rect}

We first consider grid problems over $\Zi$ for which $A$ is an
\emph{upright rectangle}, i.e., of the form $[x_1,x_2]\times
[y_1,y_2]$. These instances are easily solved, as they can be reduced
to a problem in a lower dimension. Indeed, it suffices to
independently solve the grid problem in the $x$-axis (i.e., by
enumerating the integers in $[x_1, x_2]$) and in the $y$-axis (i.e. by
enumerating the integers in interval $[y_1,y_2]$), as illustrated in
Figure~\ref{fig:grid-pb-Zi}~(b).

\begin{proposition}
  \label{prop:up-rect-Zi}
  Let $A$ be an upright rectangle. Then there is an algorithm which
  enumerates all the solutions to the grid problem over $\Zi$ for
  $A$. Moreover, the algorithm requires only a constant number of
  arithmetic operations per solution produced.
\end{proposition}

% --------------------------------------------------------------------
\subsection{Upright sets}
\label{ssect:grid-pb-Zi-up-set}

We now generalize the method of the previous subsection to convex sets
that are \emph{close} to upright rectangles in a suitable sense.

\begin{definition}[Uprightness]
  \label{def:upright}
  Let $A$ be a bounded convex subset of $\R^2$ with non-empty
  interior. The bounding box of $A$, denoted $\BBox(A)$, is the
  smallest set of the form $[x_1,x_2]\by [y_1,y_2]$ that contains
  $A$. The \emph{uprightness of $A$}, denoted $\up(A)$, is defined to
  be the ratio of the area of A to the area of its bounding box:
  \[
  \up(A) =\frac{\area(A)}{\area(\BBox(A))}.
  \]
  We say that $A$ is $M$-upright if $\up(A) \geq M$.
\end{definition}

\begin{proposition}
  \label{prop:up-set-Zi}
  Let $A$ be an $M$-upright set. Then there exists an algorithm which
  enumerates all the solutions to the grid problem over $\Zi$ for
  $A$. Moreover, the algorithm requires $O(1/M)$\label{defbigO}
  arithmetic operations per solution produced. In particular, when
  $M>0$ is fixed, it requires only a constant number of operations per
  solution.
\end{proposition}

\begin{proof}
  By Proposition~\ref{prop:up-rect-Zi}, we can efficiently enumerate
  the solutions of the grid problem for $\BBox(A)$. For each such
  candidate solution $u$, we only need to check whether $u$ is also a
  solution for $A$.  To establish the efficiency of the algorithm, we
  need to ensure that the total number of solutions is not too small
  in relation to the total number of candidates produced. To see this,
  note that, with the exception of trivial cases, when the number of
  rows or columns is very small, $M$-uprightness and convexity ensure
  that the proportion of candidates $u$ that are solutions for $A$ is
  approximately $M : 1$. Therefore, the runtime per solution differs
  from that of Proposition~\ref{prop:up-rect-Zi} by at most a factor
  of $O(1/M)$.
\end{proof}

Figure~\ref{fig:grid-pb-Zi-Bbox} shows three different examples of
grid problems over $\Zi$. The sets $A_i$ are again shown in red, for
$i=1,2,3$, and their bounding boxes are shown in outline. The typical
case of an upright set is $A_1$. Here, a fixed proportion of grid
points from the bounding box of $A_1$ are elements of $A_1$. The
exceptional case of an upright set is $A_2$. Its bounding box spans
only two columns of the grid. Therefore, although the bounding box
contains many grid points, $A_2$ does not. However, this case is
easily dealt with by solving the problem in a lower dimension for
each of the grid columns separately. Finally, the set $A_3$ is not
upright. In this case, Lemma~\ref{prop:up-set-Zi} is not helpful, and
a priori, it could be a difficult problem to find grid points in
$A_3$.

% ......................................................................
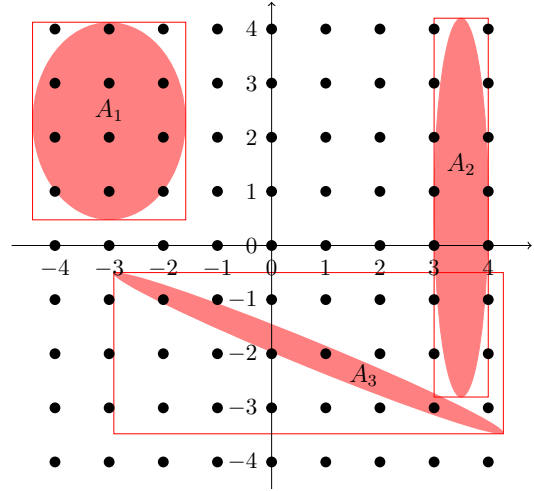
\begin{figure}
  \[
  \mp{0.8}{\scalebox{0.8}{\begin{tikzpicture}[scale=0.9]

% Ellipses:
% Debug: 0.2 0.5 0.30000000000000004 (-0.2,-0.0) (-0.0,0.19999999999999996) (-1.414213562373095,-0.0) (-0.0,1.8257418583505536)
\fill[color=red!50] 
(-4.414213562373095,2.3) .. controls 
  (-4.414213562373095,3.3083293854947433) and 
  (-3.78104858350254,4.125741858350553) .. 
(-3.0,4.125741858350553) .. controls 
  (-2.21895141649746,4.125741858350553) and 
  (-1.585786437626905,3.3083293854947433) .. 
(-1.585786437626905,2.3) .. controls 
  (-1.585786437626905,1.2916706145052563) and 
  (-2.21895141649746,0.47425814164944624) .. 
(-3.0,0.47425814164944624) .. controls 
  (-3.78104858350254,0.47425814164944624) and 
  (-4.414213562373095,1.2916706145052563) .. 
(-4.414213562373095,2.3) -- cycle;
% Debug: 6.92 7.0 8.000000000000007e-2 (-6.92,-0.0) (-0.0,6.92) (-0.3779644730092272,-0.0) (-0.0,3.535533905932736)

\draw[fill=red!50, color=red!50] (3.5,0.7) ellipse (0.5 and 3.5);

% Debug: 16.255674739806633 16.321949925903315 6.627518609668215e-2 (-5.7,-13.921949925903315) (13.921949925903316,-5.7) (-9.378563018651263e-2,-0.22906646442559916) (3.594777835731891,-1.4717933746872234)
\fill[color=red!50] 
(0.5890570822881064,-2.2190159580867657) .. controls 
  (2.574398059992575,-3.031864993828518) and 
  (4.225824174901229,-3.588252783348306) .. 
(4.27762054820651,-3.4617428683483897) .. controls 
  (4.3294169215117915,-3.3352329533484735) and 
  (2.7619693203656004,-2.57373206497732) .. 
(0.7766283426611317,-1.7608830292355675) .. controls 
  (-1.2087126350433368,-0.9480339934938151) and 
  (-2.8601387499519904,-0.39164620397402683) .. 
  (-2.9119351232572717,-0.5181561189739432) .. controls 
  (-2.963731496562553,-0.6446660339738596) and 
  (-1.3962838954163623,-1.4061669223450133) .. 
(0.5890570822881064,-2.2190159580867657) -- cycle;

% Bounding boxes:
\draw[color=red] 
(-4.414213562373095,0.474258141649446) -- 
(-1.5857864376269049,0.474258141649446) -- 
(-1.5857864376269049,4.125741858350554) -- 
(-4.414213562373095,4.125741858350554) -- cycle;

\draw[color=red] 
(3,4.2) --
(4,4.2) --
(4,-2.8) --
(3,-2.8) -- cycle;
%(3.3898024799571416,-2.7855339059327373) -- 
%(4.145731425975596,-2.7855339059327373) -- 
%(4.145731425975596,4.285533905932738) -- 
%(3.3898024799571416,4.285533905932738) -- cycle;

\draw[color=red] 
(-2.913158320992765,-3.4794618913596915) -- 
(4.278843745942003,-3.4794618913596915) -- 
(4.278843745942003,-0.5004370959626416) -- 
(-2.913158320992765,-0.5004370959626416) -- cycle;

% Labels:
\draw (-3,2.5) node {$A_1$};
\draw (3.5,1.5) node {$A_2$};
\draw (1.7,-2.4) node {$A_3$};

% The labelled x-axis and y-axis:
\draw[->] (-4.8,0) -- (4.8, 0);
\draw (-4,0) -- (-4,-0.1) node[below] {\small $-4$};
\draw (-3,0) -- (-3,-0.1) node[below] {\small $-3$};
\draw (-2,0) -- (-2,-0.1) node[below] {\small $-2$};
\draw (-1,0) -- (-1,-0.1) node[below] {\small $-1$};
\draw (0,0) -- (0,-0.1) node[below] {\small $0$};
\draw (1,0) -- (1,-0.1) node[below] {\small $1$};
\draw (2,0) -- (2,-0.1) node[below] {\small $2$};
\draw (3,0) -- (3,-0.1) node[below] {\small $3$};
\draw (4,0) -- (4,-0.1) node[below] {\small $4$};
\draw[->] (0,-4.5) -- (0,4.5);
\draw (0,-4) -- (-0.1,-4) node[left] {\small $-4$};
\draw (0,-3) -- (-0.1,-3) node[left] {\small $-3$};
\draw (0,-2) -- (-0.1,-2) node[left] {\small $-2$};
\draw (0,-1) -- (-0.1,-1) node[left] {\small $-1$};
\draw (0,0) -- (-0.1,0) node[left] {\small $0$};
\draw (0,1) -- (-0.1,1) node[left] {\small $1$};
\draw (0,2) -- (-0.1,2) node[left] {\small $2$};
\draw (0,3) -- (-0.1,3) node[left] {\small $3$};
\draw (0,4) -- (-0.1,4) node[left] {\small $4$};

% The points of the integer grid:
\foreach \x in {-4,...,4}
  \foreach \y in {-4,...,4}
    {\fill (\x,\y) circle (.1);}
    
\end{tikzpicture}}}
  \]
  \caption{Grid problems over $\Zi$ for upright and non-upright sets.}
  \label{fig:grid-pb-Zi-Bbox}
  \rule{\textwidth}{0.1mm}
\end{figure}
% ......................................................................

% --------------------------------------------------------------------
\subsection{Grid operators}
\label{ssect:grid-op-Zi}

The method of the previous subsection can be further generalized by
using certain linear transformations to turn non-upright sets into
upright sets. The linear transformations that are useful for this
purpose are \emph{special grid operators}.

\begin{definition}[Grid operator]
  \label{def:grid-operator}
  A \emph{grid operator} is an integer matrix, or equivalently, a
  linear operator, that maps $\Z^2$ to itself. A grid operator $\Gop$
  is called \emph{special} if it has determinant $\pm 1$, in which
  case $\Gop\inv$ is also a grid operator.
\end{definition}

\begin{remark}
  \label{rem:grid-of-ell}
  If $A$ is a subset of $\R^2$ and $\Gop$ is a grid operator, then
  $\Gop(A)$, the direct image of $A$, is defined as usual by
  $\Gop(A)=\s{\Gop(v)\such v\in A}$.
\end{remark}

\begin{remark}
  \label{rem:sol-grid-op}
  The interest in special grid operators lies in the fact that $u$ is
  a solution to the grid problem over $\Zi$ for $A$ if and only if
  $\Gop(u)$ is a solution to the grid problem over $\Zi$ for
  $\Gop(A)$.
\end{remark}

% --------------------------------------------------------------------
\subsection{Ellipses}
\label{ssect:grid-pb-Zi-ellipses}

Combining the results of the previous two subsections, we know that
the grid problem over $\Zi$ for $A$ can be solved efficiently provided
that we can find a operator $\Gop$ such that $\Gop(A)$ is sufficiently
upright. In this subsection, we show that if $A$ is an ellipse, then
this can always be done.

\begin{definition}[Ellipse]
  \label{def:ellipse}
  Let $D$ be a positive definite real $2\by 2$-matrix with non-zero
  determinant, and let $p\in\R^2$ be a point. The \emph{ellipse
    defined by $D$ and centered at $p$} is the set
  \[
  E = \s{u \in \R^2 \such (u-p)\da D(u -p)\leq 1}.
  \]
\end{definition}

\begin{remark}
  \label{rem:grid-op-ellipse}
  If $\Gop$ is a grid operator and $E$ is an ellipse centered at the
  origin and defined by $D$, then $\Gop(E)$ is an ellipse centered at
  the origin and defined by $(\Gop\inv)\da D\Gop\inv$.
\end{remark}

The notion of uprightness introduced above was defined for an
arbitrary bounded convex subset of $\R^2$. If the set in question is
an ellipse, we can expand the definition of uprightness into an
explicit expression.

\begin{proposition}
  \label{prop:uprightness-ellipse}
  Let $E$ be the ellipse defined by $D$ and centered at $p$, with
  \[
  D = \left[
    \begin{array}{cc}
      a & b  \\
      b & d
    \end{array} 
  \right].
  \]
  Then $\up(E) = \frac{\pi}{4} \sqrt{\frac{\det(D)}{ad}}$.
\end{proposition}

\begin{proof}
  We can compute the area of $E$ and the area of its bounding box
  using $D$. Indeed, we have $\area(E)= \pi/\sqrt{\det(D)}$ and $\area
  (\BBox(E))=4\sqrt{ad} /\det(D)$. Substituting these in
  Definition~\ref{def:upright}, yields the desired expression for
  uprightness
  \[
  \up(E) = \frac{\area(E)} {\area(\BBox(E))} = \frac{\pi}{4} \sqrt{
    \frac{\det(D)}{ad} }.  {\hspace{1.5cm}
    \begin{tikzpicture}[baseline=0, xscale=1.6, yscale=1]
      \draw[fill=yellow!20] (-1,-1) -- (1,-1) -- (1,1) -- (-1,1) --
      cycle; \draw[fill=red!50, rotate=45, yscale=0.4] (0,0) circle
      (1.3131); \draw (0,0) node {\small $E$}; \draw (-1,.75)
      node[right] {\small $\BBox(E)$};
    \end{tikzpicture}}
  \]
\end{proof}

\begin{remark}
  \label{rem:uprightness-invariant}
  The uprightness of an ellipse is invariant under translation and
  scalar multiplication.
\end{remark}

\begin{definition}[Skew]
  \label{def:skew}
  The \emph{skew} of a matrix is the product of its anti-diagonal
  entries. The \emph{skew} of an ellipse defined by $D$ is the skew of
  $D$.
\end{definition}

By Proposition~\ref{prop:uprightness-ellipse}, the skew of an ellipse
is small if and only if its uprightness is large. Our strategy for
increasing the uprightness will therefore be to reduce the skew.

\begin{proposition}
  \label{prop:to-upright-Zi}
  Let $E$ be an ellipse. There exists a grid operator $\Gop$ such that
  $\Gop(E)$ is $1/2$-upright. Moreover, if $E$ is $M$-upright, then
  $\Gop$ can be efficiently computed in $O(\log(1/M))$ arithmetic
  operations.
\end{proposition}

\begin{proof}
  Let $A$ and $B$ be the following special grid operators
  \[
  A = \begin{bmatrix}
    1 & 1 \\
    0 & 1
  \end{bmatrix}
  , \quad B = \begin{bmatrix}
    1 & 0 \\
    1 & 1
  \end{bmatrix},
  \]
  and consider an ellipse $E$ defined by $D$ and centered at
  $p$. Since uprightness is invariant under translation and scaling,
  we may without loss of generality assume that $E$ is centered at the
  origin and that $D$ has determinant 1. Suppose moreover that the
  entries of $D$ are as follows
  \[
  \begin{bmatrix}
    a & b \\
    b & d
  \end{bmatrix}.
  \]
  Note that $D$ can be written in this form because it is
  symmetric. We first show that there exists a grid operator $\Gop$
  such that $\sk(G(E))\leq 1$. Indeed, assume that $\sk(E) = b^2 \geq
  1$. In case $a\leq d$, choose $n$ such that $|na+b|\leq a/2$. Then
  we have:
  \[
  {A^n}\da D A^n = \begin{bmatrix}
    \cdots & na+b \\
    na+b & \cdots
  \end{bmatrix}.
  \]
  Therefore, using Remark~\ref{rem:grid-op-ellipse} with
  $\Gop_1=(A^n)\inv$, we have:
  \begin{align*}
    \sk(\Gop_1(E)) = (na+b)^2 \leq \frac{a^2}{4} \leq \frac{ad}{4} & =
    \frac{1+b^2}{4}  \\
    & =  \frac{1+\sk (E)}{4} \\
    & \leq \frac{2~\sk (E)}{4} = \frac{1}{2}\sk(E).
  \end{align*}
  Similarly, in case $d<a$, then choose $n$ such that $|nd+b|\leq
  d/2$.  A similar calculation shows that in this case, with
  $\Gop_1=(B^n)\inv$, we get $\sk(\Gop_1(E)) \leq
  \frac{1}{2}\sk(E)$. In both cases, the skew of $E$ is reduced by a
  factor of 2 or more. Applying this process repeatedly yields a
  sequence of operators $\Gop_1,\ldots,\Gop_m$ and letting
  $\Gop=\Gop_m\cdot\ldots\cdot\Gop_1$ we find that $\sk(\Gop(E)) \leq
  1$.
  
  Now let $D'$ be the matrix defining $\Gop(E)$, with entries as
  follows:
  \[
  D' =\begin{bmatrix}
    \alpha & \beta \\
    \beta & \delta
  \end{bmatrix}.
  \]
  Then $\sk(\Gop(E)) \leq 1$ implies that $\beta^2\leq 1$. Moreover,
  since $A$ and $B$ are special grid operators we have
  $\det(D')=\alpha\delta -\beta^2=1$. Using the expression from
  Proposition~\ref{prop:uprightness-ellipse} for the uprightness of
  $\Gop(E)$ we get the desired result:
  \[
  \up(\Gop(E)) = \frac{\pi}{4} \sqrt{\frac{\det(D')}{\alpha\delta}} =
  \frac{\pi}{4\sqrt{\alpha\delta}} = \frac{\pi}{4\sqrt{\beta^2+1}}
  \geq \frac{\pi}{4\sqrt{2}} \geq \frac{1}{2}.
  \]

  Finally, to bound the number of arithmetic operations, note that
  each application of $\Gop_j$ reduces the skew by at least a factor
  of 2. Therefore, the number $n$ of grid operators required satisfies
  $n\leq \log_2(\sk(E))$. Now note that since $D$ has determinant 1,
  we have
  \[
  M \leq \up(E) = \frac{\pi}{4}\frac{1}{\sqrt{ad}} =
  \frac{\pi}{4\sqrt{b^2+1}}.
  \]
  Therefore $\sk(E)=b^2\leq (\pi^2/16M^2)-1$, so that the computation
  of $\Gop$ requires $O(\log(1/M))$ arithmetic operations.
\end{proof}

% --------------------------------------------------------------------
\subsection{The enclosing ellipse of a bounded convex set}
\label{ssect:enclosing}

The final step in our solution of grid problems over $\Z[i]$ is to
generalize Proposition~\ref{prop:to-upright-Zi} from ellipses to
arbitrary bounded convex sets with non-empty interior. This can be
done because every such set $A$ can be inscribed in an ellipse whose
area is not much greater than that of $A$, as stated in the following
proposition, which was proved in \cite{gridsynth}.

\begin{proposition}
  \label{prop:enclosing-ellipse}
  Let $A$ be a bounded convex subset of $\R^2$ with non-empty
  interior. Then there exists an ellipse $E$ such that $A\seq E$, and
  such that
  \begin{equation}
    \label{eqn:area-EA}
    \area(E) \leq \frac{4\pi}{3\sqrt 3}\area(A). 
  \end{equation}
\end{proposition}

Note that $\frac{4\pi}{3\sqrt 3} \approx 2.4184$.  We remark that the
bound in Proposition~\ref{prop:enclosing-ellipse} is sharp; the bound
is attained in case $A$ is an equilateral triangle. In this case, the
enclosing ellipse is a circle, and the ratio of the areas is exactly
$\frac{4\pi}{3\sqrt 3}$.
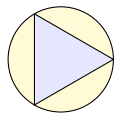
\[
\m{\begin{tikzpicture}[scale=0.7] \draw[fill=yellow!20] (0,0) circle
    (1); \draw[fill=blue!10] (1,0) -- (-0.5,.8660254037) --
    (-0.5,-.8660254037) -- cycle;
  \end{tikzpicture}}
\]

% --------------------------------------------------------------------
\subsection{General solution to grid problems over
  \texorpdfstring{$\Zi$}{Z[i]}}

We can now describe our algorithm to solve Problem~\ref{pb:grid-Zi}.

\begin{proposition}
  \label{prop:algo-grid-pb-Zi}
  There is an algorithm which, given a bounded convex subset $A$ of
  $\R^2$ with non-empty interior, enumerates all solutions of the grid
  problem over $\Zi$ for $A$. Moreover, if $A$ is $M$-upright, then
  the algorithm requires $O(\log(1/M))$ arithmetic operations overall,
  plus a constant number of arithmetic operations per solution
  produced.
\end{proposition}

\begin{proof}
  Given $A$, with an enclosing ellipse $A'$ whose area only exceeds
  that of $A$ by a fixed constant factor $N$, use
  Proposition~\ref{prop:to-upright-Zi} to find a grid operator $\Gop$
  such that $G(A')$ is $1/2$-upright. Then, use
  Proposition~\ref{prop:up-set-Zi} to enumerate the grid points of
  $G(A')$. For each grid point $u$ found, check whether it belongs to
  $G(A)$. This is the case if and only if $G\inv(u)$ is a solution to
  the grid problem over $\Zi$ for $A$.
\end{proof}

\begin{remark}
  Note that the complexity of $O(\log(1/M))$ overall operations in
  Proposition~\ref{prop:algo-grid-pb-Zi} is exponentially better than
  the complexity of $O(1/M)$ per candidate we obtained in
  Proposition~\ref{prop:up-set-Zi}. This improvement is entirely due
  to the use of grid operators in
  Proposition~\ref{prop:to-upright-Zi}.
\end{remark}

% --------------------------------------------------------------------
\subsection{Scaled grid problems over \texorpdfstring{$\Zi$}{Z[i]}}
\label{ssect:scaled-grid-pb-Zi}

If $A$ is a convex bounded subset of $\R^2$ with non-empty interior,
then so is the set $rA$, for any non-zero real number $r$. Hence, by
the results of the previous subsection, we can solve grid problems
over $\Zi$ for $rA$ where $r$ is a non-zero scalar and $A$ is a
bounded convex subset of $\R^2$ with non empty interior. We call such
a problem a \emph{scaled grid problem over $\Zi$ for $A$ and $r$}. In
Chapter~\ref{chap:synth-V}, we will be interested in solving a
sequence of such scaled grid problems for specific values of $r$. The
reasons behind our particular choice of a sequence will be detailed in
Chapter~\ref{chap:synth-V}.

Consider the set of real numbers of the form $\rt{k}\rf{\ell}$, with
$k,\ell\in\N$ and $0\leq k\leq 2$. Note that the set
$\s{\rt{k}\rf{\ell}}_{k,\ell}$ is ordered as a subset of $\R$. In
particular, if $\rt{k}\rf{\ell} \leq \rt{k'}\rf{\ell'}$, then
$\ell\leq \ell'$. When we say that an algorithm ``enumerates all
solutions of the scaled grid problem over $\Zi$ for $A$ and
$\rt{k}\rf{\ell}$ in order of increasing $\ell$'', we mean that the
algorithm first outputs all solutions for $\ell=0$, then for $\ell=1$,
etc.

\begin{proposition}
  \label{prop:algo-grid-pb-scaled-Zi}
  There is an algorithm which, given a bounded convex subset $A$ of
  $\R^2$ with non-empty interior, enumerates (the infinite sequence
  of) all solutions of the scaled grid problem over $\Zi$ for $A$ and
  $\rt{k}\rf{\ell}$ in order of increasing $\ell$. Moreover, if $A$ is
  $M$-upright, then the algorithm requires $O(\log(1/M))$ arithmetic
  operations overall, plus a constant number of arithmetic operations
  per solution produced.
\end{proposition}

\begin{proof}
  This follows from Proposition~\ref{prop:algo-grid-pb-Zi}.
\end{proof}

We finish this subsection with some lower bounds on the number of
solutions to scaled grid problems.

\begin{remark}
  \label{rem:size-for-3-sols}
  If a bounded convex subset $A\seq\R^2$ contains a circle of radius
  $1/\rf{k}$, then the grid problem over $\Zi$ for $A$ and
  $\rt{k}\rf{\ell}$ has at least three solutions.
\end{remark}

\begin{proposition}
  \label{prop:evolution-grid-Zi}
  Let $A$ be a bounded convex subset of $\R^2$ with non-empty interior
  and assume that the scaled grid problem over $\Zi$ for $A$ and
  $\rt{k}\rf{\ell}$ has at least two distinct solutions. Then for all
  $j\geq 0$, the scaled grid problem over $\Z[i]$ for $A$ and
  $\rt{k}\rf{\ell+2j}$ has at least $5^j+1$ solutions.
\end{proposition}

\begin{proof}
  Let $u\neq v$ be solutions of the scaled grid problem over $\Zi$ for
  $A$ and $\rt{k}\rf{\ell}$. That is, $u,v \in (\rt{k}\rf{\ell}A)
  \cap\Zi$. For each $n=0,1,\ldots,5^j$, let $\phi=\frac{n}{5^j}$, and
  consider $u_j = \phi u + (1-\phi) v$. Then $u_j$ is a convex
  combination of $u$ and $v$. Since $\rt{k}\rf{\ell}A$ is convex, it
  follows that $u_j\in \rt{k}\rf{\ell}A$, so that $5^ju_j$ is a
  solution of the scaled grid problem over $\Zi$ for $A$ and
  $\rt{k}\rf{\ell+2j}$, yielding $5^j+1$ distinct such solutions.
\end{proof}

% --------------------------------------------------------------------
\section{Grid problems over \texorpdfstring{$\Zomega$}{Z[omega]}}
\label{sect:grid-pb-Zomega}

% ......................................................................
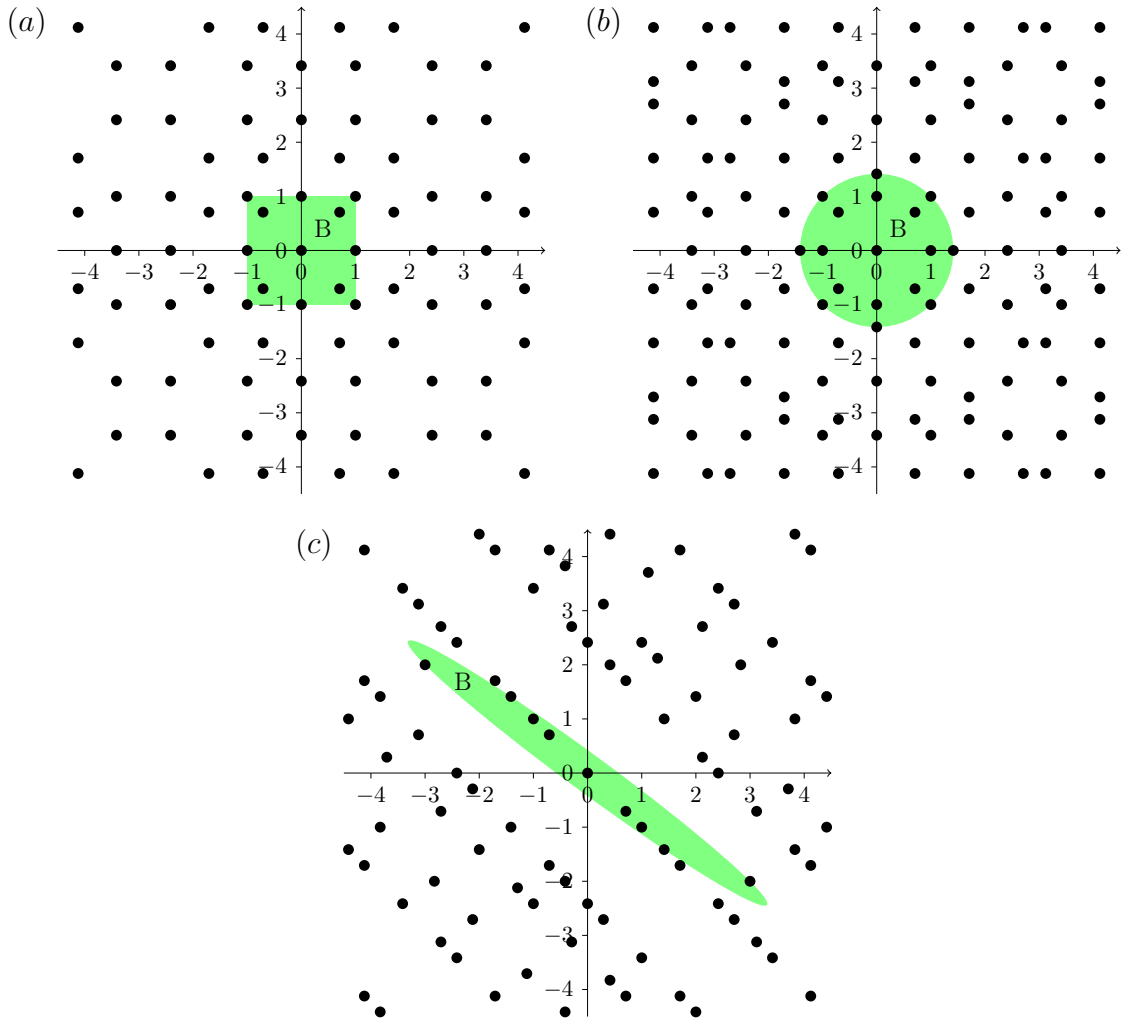
\begin{figure}[!t]
  \[ (a)~ \mp{0.95}{\scalebox{0.8}{\begin{tikzpicture}[scale=0.9]
\fill[color=green!50] (1.001,1.001) -- (-1.001,1.001) -- (-1.001,-1.001) -- (1.001,-1.001) -- cycle;
\draw[->] (-4.5,0) -- (4.5, 0);
\draw (-4,0) -- (-4,-0.1) node[below] {\small $-4$};
\draw (-3,0) -- (-3,-0.1) node[below] {\small $-3$};
\draw (-2,0) -- (-2,-0.1) node[below] {\small $-2$};
\draw (-1,0) -- (-1,-0.1) node[below] {\small $-1$};
\draw (0,0) -- (0,-0.1) node[below] {\small $0$};
\draw (1,0) -- (1,-0.1) node[below] {\small $1$};
\draw (2,0) -- (2,-0.1) node[below] {\small $2$};
\draw (3,0) -- (3,-0.1) node[below] {\small $3$};
\draw (4,0) -- (4,-0.1) node[below] {\small $4$};
\draw[->] (0,-4.5) -- (0,4.5);
\draw (0,-4) -- (-0.1,-4) node[left] {\small $-4$};
\draw (0,-3) -- (-0.1,-3) node[left] {\small $-3$};
\draw (0,-2) -- (-0.1,-2) node[left] {\small $-2$};
\draw (0,-1) -- (-0.1,-1) node[left] {\small $-1$};
\draw (0,0) -- (-0.1,0) node[left] {\small $0$};
\draw (0,1) -- (-0.1,1) node[left] {\small $1$};
\draw (0,2) -- (-0.1,2) node[left] {\small $2$};
\draw (0,3) -- (-0.1,3) node[left] {\small $3$};
\draw (0,4) -- (-0.1,4) node[left] {\small $4$};
\fill (-3.414213562373095,-3.414213562373096) circle (.1);
\fill (-3.414213562373095,-2.4142135623730954) circle (.1);
\fill (-3.4142135623730954,-1.0) circle (.1);
\fill (-3.4142135623730954,0.00000000000000011102230246251565) circle (.1);
\fill (-3.4142135623730954,1.0000000000000004) circle (.1);
\fill (-3.4142135623730954,2.414213562373096) circle (.1);
\fill (-3.4142135623730954,3.414213562373096) circle (.1);
\fill (-2.414213562373095,-3.414213562373096) circle (.1);
\fill (-2.414213562373095,-2.4142135623730954) circle (.1);
\fill (-2.4142135623730954,-1.0) circle (.1);
\fill (-2.4142135623730954,0.00000000000000011102230246251565) circle (.1);
\fill (-2.4142135623730954,1.0000000000000004) circle (.1);
\fill (-2.4142135623730954,2.414213562373096) circle (.1);
\fill (-2.4142135623730954,3.414213562373096) circle (.1);
\fill (-0.9999999999999999,-3.4142135623730954) circle (.1);
\fill (-0.9999999999999999,-2.4142135623730954) circle (.1);
\fill (-1.0,-1.0000000000000002) circle (.1);
\fill (-1.0,0.0) circle (.1);
\fill (-1.0,1.0000000000000002) circle (.1);
\fill (-1.0,2.4142135623730954) circle (.1);
\fill (-1.0,3.4142135623730954) circle (.1);
\fill (0.00000000000000011102230246251565,-3.4142135623730954) circle (.1);
\fill (0.00000000000000011102230246251565,-2.4142135623730954) circle (.1);
\fill (0.0,-1.0000000000000002) circle (.1);
\fill (0.0,0.0) circle (.1);
\fill (0.0,1.0000000000000002) circle (.1);
\fill (-0.00000000000000011102230246251565,2.4142135623730954) circle (.1);
\fill (-0.00000000000000011102230246251565,3.4142135623730954) circle (.1);
\fill (1.0,-3.4142135623730954) circle (.1);
\fill (1.0,-2.4142135623730954) circle (.1);
\fill (1.0,-1.0000000000000002) circle (.1);
\fill (1.0,0.0) circle (.1);
\fill (1.0,1.0000000000000002) circle (.1);
\fill (0.9999999999999999,2.4142135623730954) circle (.1);
\fill (0.9999999999999999,3.4142135623730954) circle (.1);
\fill (2.4142135623730954,-3.414213562373096) circle (.1);
\fill (2.4142135623730954,-2.414213562373096) circle (.1);
\fill (2.4142135623730954,-1.0000000000000004) circle (.1);
\fill (2.4142135623730954,-0.00000000000000011102230246251565) circle (.1);
\fill (2.4142135623730954,1.0) circle (.1);
\fill (2.414213562373095,2.4142135623730954) circle (.1);
\fill (2.414213562373095,3.414213562373096) circle (.1);
\fill (3.4142135623730954,-3.414213562373096) circle (.1);
\fill (3.4142135623730954,-2.414213562373096) circle (.1);
\fill (3.4142135623730954,-1.0000000000000004) circle (.1);
\fill (3.4142135623730954,-0.00000000000000011102230246251565) circle (.1);
\fill (3.4142135623730954,1.0) circle (.1);
\fill (3.414213562373095,2.4142135623730954) circle (.1);
\fill (3.414213562373095,3.414213562373096) circle (.1);
\fill (-4.121320343559643,-4.121320343559644) circle (.1);
\fill (-4.121320343559643,-1.7071067811865477) circle (.1);
\fill (-4.121320343559643,-0.7071067811865475) circle (.1);
\fill (-4.121320343559643,0.7071067811865478) circle (.1);
\fill (-4.121320343559643,1.7071067811865483) circle (.1);
\fill (-4.121320343559643,4.121320343559644) circle (.1);
\fill (-1.7071067811865475,-4.121320343559643) circle (.1);
\fill (-1.7071067811865475,-1.707106781186548) circle (.1);
\fill (-1.7071067811865475,-0.7071067811865476) circle (.1);
\fill (-1.7071067811865477,0.7071067811865477) circle (.1);
\fill (-1.7071067811865477,1.707106781186548) circle (.1);
\fill (-1.707106781186548,4.121320343559644) circle (.1);
\fill (-0.7071067811865475,-4.121320343559643) circle (.1);
\fill (-0.7071067811865476,-1.707106781186548) circle (.1);
\fill (-0.7071067811865476,-0.7071067811865476) circle (.1);
\fill (-0.7071067811865477,0.7071067811865477) circle (.1);
\fill (-0.7071067811865477,1.707106781186548) circle (.1);
\fill (-0.7071067811865478,4.121320343559644) circle (.1);
\fill (0.7071067811865478,-4.121320343559644) circle (.1);
\fill (0.7071067811865477,-1.707106781186548) circle (.1);
\fill (0.7071067811865477,-0.7071067811865477) circle (.1);
\fill (0.7071067811865476,0.7071067811865476) circle (.1);
\fill (0.7071067811865476,1.707106781186548) circle (.1);
\fill (0.7071067811865475,4.121320343559643) circle (.1);
\fill (1.707106781186548,-4.121320343559644) circle (.1);
\fill (1.7071067811865477,-1.707106781186548) circle (.1);
\fill (1.7071067811865477,-0.7071067811865477) circle (.1);
\fill (1.7071067811865475,0.7071067811865476) circle (.1);
\fill (1.7071067811865475,1.707106781186548) circle (.1);
\fill (1.7071067811865475,4.121320343559643) circle (.1);
\fill (4.121320343559643,-4.121320343559644) circle (.1);
\fill (4.121320343559643,-1.7071067811865483) circle (.1);
\fill (4.121320343559643,-0.7071067811865478) circle (.1);
\fill (4.121320343559643,0.7071067811865475) circle (.1);
\fill (4.121320343559643,1.7071067811865477) circle (.1);
\fill (4.121320343559643,4.121320343559644) circle (.1);
\draw (0.4,0.4) node {B};
\end{tikzpicture}}} \quad (b)~
  \mp{0.95}{\scalebox{0.8}{\begin{tikzpicture}[scale=0.9]
\fill[color=green!50] (0.0,0.0) circle (1.415213562373095);
\draw[->] (-4.5,0) -- (4.5, 0);
\draw (-4,0) -- (-4,-0.1) node[below] {\small $-4$};
\draw (-3,0) -- (-3,-0.1) node[below] {\small $-3$};
\draw (-2,0) -- (-2,-0.1) node[below] {\small $-2$};
\draw (-1,0) -- (-1,-0.1) node[below] {\small $-1$};
\draw (0,0) -- (0,-0.1) node[below] {\small $0$};
\draw (1,0) -- (1,-0.1) node[below] {\small $1$};
\draw (2,0) -- (2,-0.1) node[below] {\small $2$};
\draw (3,0) -- (3,-0.1) node[below] {\small $3$};
\draw (4,0) -- (4,-0.1) node[below] {\small $4$};
\draw[->] (0,-4.5) -- (0,4.5);
\draw (0,-4) -- (-0.1,-4) node[left] {\small $-4$};
\draw (0,-3) -- (-0.1,-3) node[left] {\small $-3$};
\draw (0,-2) -- (-0.1,-2) node[left] {\small $-2$};
\draw (0,-1) -- (-0.1,-1) node[left] {\small $-1$};
\draw (0,0) -- (-0.1,0) node[left] {\small $0$};
\draw (0,1) -- (-0.1,1) node[left] {\small $1$};
\draw (0,2) -- (-0.1,2) node[left] {\small $2$};
\draw (0,3) -- (-0.1,3) node[left] {\small $3$};
\draw (0,4) -- (-0.1,4) node[left] {\small $4$};
\fill (-3.414213562373095,-3.414213562373096) circle (.1);
\fill (-3.414213562373095,-2.4142135623730954) circle (.1);
\fill (-3.4142135623730954,-1.0) circle (.1);
\fill (-3.4142135623730954,0.00000000000000011102230246251565) circle (.1);
\fill (-3.4142135623730954,1.0000000000000004) circle (.1);
\fill (-3.4142135623730954,2.414213562373096) circle (.1);
\fill (-3.4142135623730954,3.414213562373096) circle (.1);
\fill (-2.414213562373095,-3.414213562373096) circle (.1);
\fill (-2.414213562373095,-2.4142135623730954) circle (.1);
\fill (-2.4142135623730954,-1.0) circle (.1);
\fill (-2.4142135623730954,0.00000000000000011102230246251565) circle (.1);
\fill (-2.4142135623730954,1.0000000000000004) circle (.1);
\fill (-2.4142135623730954,2.414213562373096) circle (.1);
\fill (-2.4142135623730954,3.414213562373096) circle (.1);
\fill (-1.4142135623730954,0.00000000000000011102230246251565) circle (.1);
\fill (-0.9999999999999999,-3.4142135623730954) circle (.1);
\fill (-0.9999999999999999,-2.4142135623730954) circle (.1);
\fill (-1.0,-1.0000000000000002) circle (.1);
\fill (-1.0,0.0) circle (.1);
\fill (-1.0,1.0000000000000002) circle (.1);
\fill (-1.0,2.4142135623730954) circle (.1);
\fill (-1.0,3.4142135623730954) circle (.1);
\fill (0.00000000000000011102230246251565,-3.4142135623730954) circle (.1);
\fill (0.00000000000000011102230246251565,-2.4142135623730954) circle (.1);
\fill (0.00000000000000011102230246251565,-1.4142135623730954) circle (.1);
\fill (0.0,-1.0000000000000002) circle (.1);
\fill (0.0,0.0) circle (.1);
\fill (0.0,1.0000000000000002) circle (.1);
\fill (-0.00000000000000011102230246251565,1.4142135623730954) circle (.1);
\fill (-0.00000000000000011102230246251565,2.4142135623730954) circle (.1);
\fill (-0.00000000000000011102230246251565,3.4142135623730954) circle (.1);
\fill (1.0,-3.4142135623730954) circle (.1);
\fill (1.0,-2.4142135623730954) circle (.1);
\fill (1.0,-1.0000000000000002) circle (.1);
\fill (1.0,0.0) circle (.1);
\fill (1.0,1.0000000000000002) circle (.1);
\fill (0.9999999999999999,2.4142135623730954) circle (.1);
\fill (0.9999999999999999,3.4142135623730954) circle (.1);
\fill (1.4142135623730954,-0.00000000000000011102230246251565) circle (.1);
\fill (2.4142135623730954,-3.414213562373096) circle (.1);
\fill (2.4142135623730954,-2.414213562373096) circle (.1);
\fill (2.4142135623730954,-1.0000000000000004) circle (.1);
\fill (2.4142135623730954,-0.00000000000000011102230246251565) circle (.1);
\fill (2.4142135623730954,1.0) circle (.1);
\fill (2.414213562373095,2.4142135623730954) circle (.1);
\fill (2.414213562373095,3.414213562373096) circle (.1);
\fill (3.4142135623730954,-3.414213562373096) circle (.1);
\fill (3.4142135623730954,-2.414213562373096) circle (.1);
\fill (3.4142135623730954,-1.0000000000000004) circle (.1);
\fill (3.4142135623730954,-0.00000000000000011102230246251565) circle (.1);
\fill (3.4142135623730954,1.0) circle (.1);
\fill (3.414213562373095,2.4142135623730954) circle (.1);
\fill (3.414213562373095,3.414213562373096) circle (.1);
\fill (-4.121320343559643,-4.121320343559644) circle (.1);
\fill (-4.121320343559643,-3.121320343559643) circle (.1);
\fill (-4.121320343559643,-2.707106781186548) circle (.1);
\fill (-4.121320343559643,-1.7071067811865477) circle (.1);
\fill (-4.121320343559643,-0.7071067811865475) circle (.1);
\fill (-4.121320343559643,0.7071067811865478) circle (.1);
\fill (-4.121320343559643,1.7071067811865483) circle (.1);
\fill (-4.121320343559643,2.7071067811865483) circle (.1);
\fill (-4.121320343559643,3.121320343559643) circle (.1);
\fill (-4.121320343559643,4.121320343559644) circle (.1);
\fill (-3.121320343559643,-4.121320343559644) circle (.1);
\fill (-3.121320343559643,-1.7071067811865477) circle (.1);
\fill (-3.121320343559643,-0.7071067811865475) circle (.1);
\fill (-3.121320343559643,0.7071067811865478) circle (.1);
\fill (-3.121320343559643,1.7071067811865483) circle (.1);
\fill (-3.121320343559643,4.121320343559644) circle (.1);
\fill (-2.7071067811865475,-4.121320343559643) circle (.1);
\fill (-2.7071067811865475,-1.707106781186548) circle (.1);
\fill (-2.707106781186548,1.707106781186548) circle (.1);
\fill (-2.707106781186548,4.121320343559644) circle (.1);
\fill (-1.7071067811865475,-4.121320343559643) circle (.1);
\fill (-1.7071067811865475,-3.121320343559643) circle (.1);
\fill (-1.7071067811865475,-2.707106781186548) circle (.1);
\fill (-1.7071067811865475,-1.707106781186548) circle (.1);
\fill (-1.7071067811865475,-0.7071067811865476) circle (.1);
\fill (-1.7071067811865477,0.7071067811865477) circle (.1);
\fill (-1.7071067811865477,1.707106781186548) circle (.1);
\fill (-1.7071067811865477,2.707106781186548) circle (.1);
\fill (-1.707106781186548,3.1213203435596433) circle (.1);
\fill (-1.707106781186548,4.121320343559644) circle (.1);
\fill (-0.7071067811865475,-4.121320343559643) circle (.1);
\fill (-0.7071067811865475,-3.121320343559643) circle (.1);
\fill (-0.7071067811865476,-1.707106781186548) circle (.1);
\fill (-0.7071067811865476,-0.7071067811865476) circle (.1);
\fill (-0.7071067811865477,0.7071067811865477) circle (.1);
\fill (-0.7071067811865477,1.707106781186548) circle (.1);
\fill (-0.7071067811865478,3.1213203435596433) circle (.1);
\fill (-0.7071067811865478,4.121320343559644) circle (.1);
\fill (0.7071067811865478,-4.121320343559644) circle (.1);
\fill (0.7071067811865478,-3.1213203435596433) circle (.1);
\fill (0.7071067811865477,-1.707106781186548) circle (.1);
\fill (0.7071067811865477,-0.7071067811865477) circle (.1);
\fill (0.7071067811865476,0.7071067811865476) circle (.1);
\fill (0.7071067811865476,1.707106781186548) circle (.1);
\fill (0.7071067811865475,3.121320343559643) circle (.1);
\fill (0.7071067811865475,4.121320343559643) circle (.1);
\fill (1.707106781186548,-4.121320343559644) circle (.1);
\fill (1.707106781186548,-3.1213203435596433) circle (.1);
\fill (1.7071067811865477,-2.707106781186548) circle (.1);
\fill (1.7071067811865477,-1.707106781186548) circle (.1);
\fill (1.7071067811865477,-0.7071067811865477) circle (.1);
\fill (1.7071067811865475,0.7071067811865476) circle (.1);
\fill (1.7071067811865475,1.707106781186548) circle (.1);
\fill (1.7071067811865475,2.707106781186548) circle (.1);
\fill (1.7071067811865475,3.121320343559643) circle (.1);
\fill (1.7071067811865475,4.121320343559643) circle (.1);
\fill (2.707106781186548,-4.121320343559644) circle (.1);
\fill (2.707106781186548,-1.707106781186548) circle (.1);
\fill (2.7071067811865475,1.707106781186548) circle (.1);
\fill (2.7071067811865475,4.121320343559643) circle (.1);
\fill (3.121320343559643,-4.121320343559644) circle (.1);
\fill (3.121320343559643,-1.7071067811865483) circle (.1);
\fill (3.121320343559643,-0.7071067811865478) circle (.1);
\fill (3.121320343559643,0.7071067811865475) circle (.1);
\fill (3.121320343559643,1.7071067811865477) circle (.1);
\fill (3.121320343559643,4.121320343559644) circle (.1);
\fill (4.121320343559643,-4.121320343559644) circle (.1);
\fill (4.121320343559643,-3.121320343559643) circle (.1);
\fill (4.121320343559643,-2.7071067811865483) circle (.1);
\fill (4.121320343559643,-1.7071067811865483) circle (.1);
\fill (4.121320343559643,-0.7071067811865478) circle (.1);
\fill (4.121320343559643,0.7071067811865475) circle (.1);
\fill (4.121320343559643,1.7071067811865477) circle (.1);
\fill (4.121320343559643,2.707106781186548) circle (.1);
\fill (4.121320343559643,3.121320343559643) circle (.1);
\fill (4.121320343559643,4.121320343559644) circle (.1);
\draw (0.4,0.4) node {B};
\end{tikzpicture}}}
  \]
  \[ (c)~ \mp{0.95}{\scalebox{0.8}{\begin{tikzpicture}[scale=0.9]
% Debug: 8.381527307120106 8.440663653560053 5.9136346439947474e-2 (-4.0,-5.440763653560053) (5.440763653560053,-4.0) (-0.20388260857871685,-0.2773192715870234) (3.3131500365557374,-2.435797801573568)
\fill[color=green!50] (-0.20388260857871685,-0.2773192715870234)
.. controls (1.6259196305123527,-1.622573251067478) and (3.200548781081991,-2.588957006105265) .. (3.3131500365557374,-2.435797801573568)
.. controls (3.4257512920294837,-2.282638597041871) and (2.0336848476697864,-1.0679347078934311) .. (0.20388260857871685,0.2773192715870234)
.. controls (-1.6259196305123527,1.622573251067478) and (-3.200548781081991,2.588957006105265) .. (-3.3131500365557374,2.435797801573568)
.. controls (-3.4257512920294837,2.282638597041871) and (-2.0336848476697864,1.0679347078934311) .. (-0.20388260857871685,-0.2773192715870234)
-- cycle;
\draw[->] (-4.5,0) -- (4.5, 0);
\draw (-4,0) -- (-4,-0.1) node[below] {\small $-4$};
\draw (-3,0) -- (-3,-0.1) node[below] {\small $-3$};
\draw (-2,0) -- (-2,-0.1) node[below] {\small $-2$};
\draw (-1,0) -- (-1,-0.1) node[below] {\small $-1$};
\draw (0,0) -- (0,-0.1) node[below] {\small $0$};
\draw (1,0) -- (1,-0.1) node[below] {\small $1$};
\draw (2,0) -- (2,-0.1) node[below] {\small $2$};
\draw (3,0) -- (3,-0.1) node[below] {\small $3$};
\draw (4,0) -- (4,-0.1) node[below] {\small $4$};
\draw[->] (0,-4.5) -- (0,4.5);
\draw (0,-4) -- (-0.1,-4) node[left] {\small $-4$};
\draw (0,-3) -- (-0.1,-3) node[left] {\small $-3$};
\draw (0,-2) -- (-0.1,-2) node[left] {\small $-2$};
\draw (0,-1) -- (-0.1,-1) node[left] {\small $-1$};
\draw (0,0) -- (-0.1,0) node[left] {\small $0$};
\draw (0,1) -- (-0.1,1) node[left] {\small $1$};
\draw (0,2) -- (-0.1,2) node[left] {\small $2$};
\draw (0,3) -- (-0.1,3) node[left] {\small $3$};
\draw (0,4) -- (-0.1,4) node[left] {\small $4$};
\fill (-4.414213562373095,-1.4142135623730951) circle (.1);
\fill (-4.414213562373096,1.0000000000000004) circle (.1);
\fill (-3.8284271247461907,-4.414213562373096) circle (.1);
\fill (-3.8284271247461907,-1.0) circle (.1);
\fill (-3.8284271247461903,1.4142135623730954) circle (.1);
\fill (-3.414213562373095,-2.4142135623730954) circle (.1);
\fill (-3.4142135623730954,3.414213562373096) circle (.1);
\fill (-3.0,2.0000000000000004) circle (.1);
\fill (-2.8284271247461907,-2.0) circle (.1);
\fill (-2.414213562373095,-3.414213562373096) circle (.1);
\fill (-2.4142135623730954,0.00000000000000011102230246251565) circle (.1);
\fill (-2.4142135623730954,2.414213562373096) circle (.1);
\fill (-2.0,-1.4142135623730954) circle (.1);
\fill (-2.0,4.414213562373097) circle (.1);
\fill (-1.4142135623730954,-1.0) circle (.1);
\fill (-1.4142135623730954,1.4142135623730954) circle (.1);
\fill (-0.9999999999999999,-2.4142135623730954) circle (.1);
\fill (-1.0,1.0000000000000002) circle (.1);
\fill (-1.0,3.4142135623730954) circle (.1);
\fill (-0.41421356237309515,-4.414213562373096) circle (.1);
\fill (-0.41421356237309537,-2.0000000000000004) circle (.1);
\fill (-0.41421356237309537,3.8284271247461903) circle (.1);
\fill (0.00000000000000011102230246251565,-2.4142135623730954) circle (.1);
\fill (0.0,0.0) circle (.1);
\fill (-0.00000000000000011102230246251565,2.4142135623730954) circle (.1);
\fill (0.41421356237309537,-3.8284271247461903) circle (.1);
\fill (0.41421356237309537,2.0000000000000004) circle (.1);
\fill (0.41421356237309515,4.414213562373096) circle (.1);
\fill (1.0,-3.4142135623730954) circle (.1);
\fill (1.0,-1.0000000000000002) circle (.1);
\fill (0.9999999999999999,2.4142135623730954) circle (.1);
\fill (1.4142135623730954,-1.4142135623730954) circle (.1);
\fill (1.4142135623730954,1.0) circle (.1);
\fill (2.0,-4.414213562373097) circle (.1);
\fill (2.0,1.4142135623730954) circle (.1);
\fill (2.4142135623730954,-2.414213562373096) circle (.1);
\fill (2.4142135623730954,-0.00000000000000011102230246251565) circle (.1);
\fill (2.414213562373095,3.414213562373096) circle (.1);
\fill (2.8284271247461907,2.0) circle (.1);
\fill (3.0,-2.0000000000000004) circle (.1);
\fill (3.4142135623730954,-3.414213562373096) circle (.1);
\fill (3.414213562373095,2.4142135623730954) circle (.1);
\fill (3.8284271247461903,-1.4142135623730954) circle (.1);
\fill (3.8284271247461907,1.0) circle (.1);
\fill (3.8284271247461907,4.414213562373096) circle (.1);
\fill (4.414213562373096,-1.0000000000000004) circle (.1);
\fill (4.414213562373095,1.4142135623730951) circle (.1);
\fill (-4.121320343559643,-4.121320343559644) circle (.1);
\fill (-4.121320343559643,-1.7071067811865477) circle (.1);
\fill (-4.121320343559643,1.7071067811865483) circle (.1);
\fill (-4.121320343559643,4.121320343559644) circle (.1);
\fill (-3.7071067811865475,0.29289321881345265) circle (.1);
\fill (-3.121320343559643,0.7071067811865478) circle (.1);
\fill (-3.121320343559643,3.121320343559643) circle (.1);
\fill (-2.7071067811865475,-3.121320343559643) circle (.1);
\fill (-2.7071067811865475,-0.7071067811865476) circle (.1);
\fill (-2.707106781186548,2.707106781186548) circle (.1);
\fill (-2.121320343559643,-2.707106781186548) circle (.1);
\fill (-2.121320343559643,-0.2928932188134524) circle (.1);
\fill (-1.7071067811865475,-4.121320343559643) circle (.1);
\fill (-1.7071067811865477,1.707106781186548) circle (.1);
\fill (-1.707106781186548,4.121320343559644) circle (.1);
\fill (-1.292893218813452,-2.121320343559643) circle (.1);
\fill (-1.1213203435596428,-3.707106781186548) circle (.1);
\fill (-0.7071067811865476,-1.707106781186548) circle (.1);
\fill (-0.7071067811865477,0.7071067811865477) circle (.1);
\fill (-0.7071067811865478,4.121320343559644) circle (.1);
\fill (-0.2928932188134522,-3.1213203435596433) circle (.1);
\fill (-0.2928932188134524,2.707106781186548) circle (.1);
\fill (0.2928932188134524,-2.707106781186548) circle (.1);
\fill (0.2928932188134522,3.1213203435596433) circle (.1);
\fill (0.7071067811865478,-4.121320343559644) circle (.1);
\fill (0.7071067811865477,-0.7071067811865477) circle (.1);
\fill (0.7071067811865476,1.707106781186548) circle (.1);
\fill (1.1213203435596428,3.707106781186548) circle (.1);
\fill (1.292893218813452,2.121320343559643) circle (.1);
\fill (1.707106781186548,-4.121320343559644) circle (.1);
\fill (1.7071067811865477,-1.707106781186548) circle (.1);
\fill (1.7071067811865475,4.121320343559643) circle (.1);
\fill (2.121320343559643,0.2928932188134524) circle (.1);
\fill (2.121320343559643,2.707106781186548) circle (.1);
\fill (2.707106781186548,-2.707106781186548) circle (.1);
\fill (2.7071067811865475,0.7071067811865476) circle (.1);
\fill (2.7071067811865475,3.121320343559643) circle (.1);
\fill (3.121320343559643,-3.121320343559643) circle (.1);
\fill (3.121320343559643,-0.7071067811865478) circle (.1);
\fill (3.7071067811865475,-0.29289321881345265) circle (.1);
\fill (4.121320343559643,-4.121320343559644) circle (.1);
\fill (4.121320343559643,-1.7071067811865483) circle (.1);
\fill (4.121320343559643,1.7071067811865477) circle (.1);
\fill (4.121320343559643,4.121320343559644) circle (.1);
\draw (-2.3,1.7) node {B};
\end{tikzpicture}}}
  \]
  \caption[Complex grids.]{The complex grid for three different convex
    sets $B$. In each case, the set $B$ is shown in green and grid
    points are shown as black dots. (a) $B=[-1,1]^2$. (b)
    $B=\s{(x,y)\mid x^2+y^2\leq 2}$. (c) $B=\s{(x,y)\mid 6x^2 + 16xy +
      11y^2 \leq 2}$.}
  \label{fig-grid-2d}
  \rule{\textwidth}{0.1mm}
\end{figure}
% ......................................................................

We now turn to grid problems where the lattice $L$ is a subset of
$\Zomega$. Viewing $\Comp$ as the real plane, we can consider the
elements of $\Zomega$ as points in $\R^2$. However, we know from
Chapter~\ref{chap:nb-th} that $\Zomega$ is dense in $\Comp$ and
therefore does not form a lattice in the plane. To circumvent this
issue, we consider subsets of $\Zomega$ that arise as the image, under
the automorphism $(-)\bul$, of the intersection of $\Zomega$ and a
bounded convex set $B\in\R^2$ with non-empty interior. We use this
notion of grid to formulate grid problems over $\Zomega$.

\begin{definition}
  \label{def:grid-Zomega}
  Let $B$ be a subset of $\R^2$. The {\em (complex) grid} for $B$ is
  the set
  \begin{equation}
    \Grid(B) = \s{ u\in\Zomega \mid u\bul \in B}.
  \end{equation}
\end{definition}

\begin{remark}
  \label{rem:discrete-inf-2d}
  We will only be interested in the case where $B$ is a bounded convex
  set with non-empty interior. In this case, the grid is discrete and
  infinite. It is infinite by the density of $\Zomega$: there are
  infinitely many points $u\in B\cap \Zomega$, and for each of them,
  $u\bul$ is a grid point. To see that it is discrete, recall from
  Remark~\ref{rem:bound-bullet} of Chapter~\ref{chap:nb-th} that for
  $u,v\in\Zomega$ we have
  \[
  |u-v|\cdot |u\bul -v\bul| \geq 1.
  \]
  Since the distance between points in $B$ is bounded above, the
  distance between their bullets is bounded below.
\end{remark}

Figure~\ref{fig-grid-2d} illustrates the complex grids for several
different convex sets $B$. Note that the grid has a $90$-degree
symmetry in (a), a $45$-degree symmetry in (b), and a $180$-degree
symmetry in (c).

\begin{problem}[Grid problem over $\Zomega$]
  \label{pb:grid-Zomega}
  Given two bounded convex subsets $A$ and $B$ of $\R^2$ with
  non-empty interior, enumerate all the points $u \in A \cap
  \Grid(B)$.
\end{problem}

Alternatively, grid problems over $\Z[\omega]$ can be understood as
looking for points in $\Z[\omega]$ such that $u\in A$ and $u\bul \in
B$. We also refer to the conditions $u\in A$ and $u\bul\in B$ as
\emph{grid constraints}. As before, an element $u\in A\cap \Grid(B)$
is called a \emph{solution} to the grid problem over $\Zomega$ for $A$
and $B$.

We solve grid problems over $\Zomega$ by reasoning as in the previous
section. That is, we first deal with the case of upright rectangles
and then generalize our methods to arbitrary convex sets using
ellipses.

% ----------------------------------------------------------------------
\subsection{One-dimensional grid problems}
\label{ssect:1-d-grid-pb-Zomega}

A grid problem over $\Zi$ for an upright rectangle $A$ is solved by
considering the $x$ and $y$ coordinates of the problem
independently. To extend this method to grid problems over $\Zomega$,
we define a one-dimensional analogue of Problem~\ref{pb:grid-Zomega}.

\begin{definition}
  \label{def:grid-Zrt}
  Let $B$ be a subset of $\R$. The {\em (real) grid} for $B$ is the
  set
  \begin{equation}
    \Grid(B) = \s{ u\in\Zrt \mid u\bul\in B}.
  \end{equation}
\end{definition}

\begin{remark}
  \label{rem:discrete-inf-1d}
  In the following, we will only be interested in the case where $B$
  is a closed interval $[y_0,y_1]$ with $y_0<y_1$. In this case also
  the grid is discrete and infinite.
\end{remark}

Figure~\ref{fig:grid-1d} illustrates the grids for the intervals
$[-1,1]$ and $[-3,3]$, respectively. For example, the first few
non-negative points in $\Grid([-1,1])$ are $0$, $1$, $1+\sqrt2$,
$2+\sqrt2$, $2+2\sqrt2$, $3+2\sqrt 2$, and $4+3\sqrt 2$. As one would
expect, the grid for $[-3,3]$ is about three times denser than that
for $[-1,1]$. We also note that $B\seq B'$ implies
$\Grid(B)\seq\Grid(B')$.

% ......................................................................
\begin{figure}[t]
  \[
  (a)~ \mp{0.8}{\scalebox{0.8}{\begin{tikzpicture}[scale=0.9]
\fill[thick,color=green!50] (-1.0,0.2) -- (1.0,0.2) -- (1.0,-0.2) -- (-1.0,-0.2);
\draw[->] (-8.8,0) -- (8.8, 0);
\draw (-8,0) -- (-8,-0.3) node[below] {$-8$};
\draw (-7,0) -- (-7,-0.3) node[below] {$-7$};
\draw (-6,0) -- (-6,-0.3) node[below] {$-6$};
\draw (-5,0) -- (-5,-0.3) node[below] {$-5$};
\draw (-4,0) -- (-4,-0.3) node[below] {$-4$};
\draw (-3,0) -- (-3,-0.3) node[below] {$-3$};
\draw (-2,0) -- (-2,-0.3) node[below] {$-2$};
\draw (-1,0) -- (-1,-0.3) node[below] {$-1$};
\draw (0,0) -- (0,-0.3) node[below] {$0$};
\draw (1,0) -- (1,-0.3) node[below] {$1$};
\draw (2,0) -- (2,-0.3) node[below] {$2$};
\draw (3,0) -- (3,-0.3) node[below] {$3$};
\draw (4,0) -- (4,-0.3) node[below] {$4$};
\draw (5,0) -- (5,-0.3) node[below] {$5$};
\draw (6,0) -- (6,-0.3) node[below] {$6$};
\draw (7,0) -- (7,-0.3) node[below] {$7$};
\draw (8,0) -- (8,-0.3) node[below] {$8$};
\fill (-8.242640687119286,0) circle (.1);
\fill (-5.82842712474619,0) circle (.1);
\fill (-4.82842712474619,0) circle (.1);
\fill (-3.414213562373095,0) circle (.1);
\fill (-2.414213562373095,0) circle (.1);
\fill (-1.0,0) circle (.1);
\fill (0.0,0) circle (.1);
\fill (1.0,0) circle (.1);
\fill (2.414213562373095,0) circle (.1);
\fill (3.414213562373095,0) circle (.1);
\fill (4.82842712474619,0) circle (.1);
\fill (5.82842712474619,0) circle (.1);
\fill (8.242640687119286,0) circle (.1);
\draw (-0.5,0) node[above] {$B$};\end{tikzpicture}}}
  \]
  \[
  (b)~ \mp{0.8}{\scalebox{0.8}{\begin{tikzpicture}[scale=0.9]
\fill[thick,color=green!50] (-3.0,0.2) -- (3.0,0.2) -- (3.0,-0.2) -- (-3.0,-0.2);
\draw[->] (-8.8,0) -- (8.8, 0);
\draw (-8,0) -- (-8,-0.3) node[below] {$-8$};
\draw (-7,0) -- (-7,-0.3) node[below] {$-7$};
\draw (-6,0) -- (-6,-0.3) node[below] {$-6$};
\draw (-5,0) -- (-5,-0.3) node[below] {$-5$};
\draw (-4,0) -- (-4,-0.3) node[below] {$-4$};
\draw (-3,0) -- (-3,-0.3) node[below] {$-3$};
\draw (-2,0) -- (-2,-0.3) node[below] {$-2$};
\draw (-1,0) -- (-1,-0.3) node[below] {$-1$};
\draw (0,0) -- (0,-0.3) node[below] {$0$};
\draw (1,0) -- (1,-0.3) node[below] {$1$};
\draw (2,0) -- (2,-0.3) node[below] {$2$};
\draw (3,0) -- (3,-0.3) node[below] {$3$};
\draw (4,0) -- (4,-0.3) node[below] {$4$};
\draw (5,0) -- (5,-0.3) node[below] {$5$};
\draw (6,0) -- (6,-0.3) node[below] {$6$};
\draw (7,0) -- (7,-0.3) node[below] {$7$};
\draw (8,0) -- (8,-0.3) node[below] {$8$};
\fill (-8.65685424949238,0) circle (.1);
\fill (-8.242640687119286,0) circle (.1);
\fill (-7.82842712474619,0) circle (.1);
\fill (-7.242640687119286,0) circle (.1);
\fill (-6.82842712474619,0) circle (.1);
\fill (-6.242640687119286,0) circle (.1);
\fill (-5.82842712474619,0) circle (.1);
\fill (-5.414213562373095,0) circle (.1);
\fill (-4.82842712474619,0) circle (.1);
\fill (-4.414213562373095,0) circle (.1);
\fill (-3.8284271247461903,0) circle (.1);
\fill (-3.414213562373095,0) circle (.1);
\fill (-3.0,0) circle (.1);
\fill (-2.8284271247461903,0) circle (.1);
\fill (-2.414213562373095,0) circle (.1);
\fill (-2.0,0) circle (.1);
\fill (-1.4142135623730951,0) circle (.1);
\fill (-1.0,0) circle (.1);
\fill (-0.41421356237309515,0) circle (.1);
\fill (0.0,0) circle (.1);
\fill (0.41421356237309515,0) circle (.1);
\fill (1.0,0) circle (.1);
\fill (1.4142135623730951,0) circle (.1);
\fill (2.0,0) circle (.1);
\fill (2.414213562373095,0) circle (.1);
\fill (2.8284271247461903,0) circle (.1);
\fill (3.0,0) circle (.1);
\fill (3.414213562373095,0) circle (.1);
\fill (3.8284271247461903,0) circle (.1);
\fill (4.414213562373095,0) circle (.1);
\fill (4.82842712474619,0) circle (.1);
\fill (5.414213562373095,0) circle (.1);
\fill (5.82842712474619,0) circle (.1);
\fill (6.242640687119286,0) circle (.1);
\fill (6.82842712474619,0) circle (.1);
\fill (7.242640687119286,0) circle (.1);
\fill (7.82842712474619,0) circle (.1);
\fill (8.242640687119286,0) circle (.1);
\fill (8.65685424949238,0) circle (.1);
\draw (-0.5,0) node[above] {$B$};\end{tikzpicture}}}
  \]
  \caption[Real grids.]{The real grid for two different intervals
    $B$. In both cases, the interval $B$ is shown in green, and grid
    points are shown as black dots.}
  \label{fig:grid-1d}
  \rule{\textwidth}{0.1mm}
\end{figure}
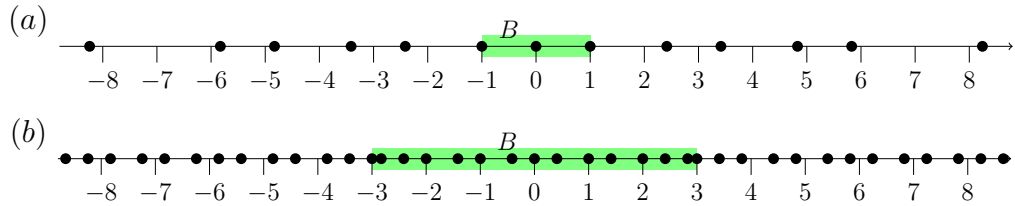
% ......................................................................

\begin{problem}[One-dimensional grid problem]
  \label{pb:grid-1-d-Zomega}
  Given two sets of real numbers $A$ and $B$ of $\R$, enumerate all
  the points $u \in A \cap \Grid(B)$.
\end{problem}

As in the two-dimensional case, Problem~\ref{pb:grid-1-d-Zomega} can
be equivalently expressed by the grid constraints $u\in A$ and
$u\bul\in B$ for $u\in\Zrt$. In the case where $A$ and $B$ are finite
intervals, the grid problem is guaranteed to have a finite number of
solutions. We recall the following facts from
{\cite{Selinger-newsynth}}.

\begin{lemma}
  \label{lem:grid-bounds}
  Let $A=[x_0,x_1]$ and $B=[y_0,y_1]$ be closed real intervals, such
  that $x_1-x_0=\delta$ and $y_1-y_0=\Delta$. If $\delta\Delta < 1$,
  then the one-dimensional grid problem for $A$ and $B$ has at most
  one solution. If $\delta\Delta \geq (1+\sqrt2)^2$, then the
  one-dimensional grid problem for $A$ and $B$ has at least one
  solution.
\end{lemma}

\begin{proof}
  Lemmas~16 and 17 of {\cite{Selinger-newsynth}}.
\end{proof}

\begin{proposition}
  \label{prop:1-d-Zomega}
  Let $A=[x_0,x_1]$ and $B=[y_0,y_1]$ be closed real intervals.  There
  is an algorithm which enumerates all solutions to the
  one-dimensional grid problem for $A$ and $B$. Moreover, the
  algorithm only requires a constant number of arithmetic operations
  per solution produced.
\end{proposition}

\begin{proof}
  It was already noted in {\cite[Lemma~17]{Selinger-newsynth}} that
  there is an efficient algorithm for computing one solution. To see
  that we can efficiently enumerate all solutions, let
  $\delta=x_1-x_0$ and $\Delta=y_1-y_0$ as before. Recall from
  Definition~\ref{def:lambda-delta} of Chapter~\ref{chap:nb-th} that
  $\lambda=\sqrt{2} +1$ and that $\lambda\inv=-\lambda\bul$. The grid
  problem for the sets $A$ and $B$ is equivalent to the grid problem
  for $\lambda\inv A$ and $-\lambda B$, because $u\in A$ and $u\bul\in
  B$ hold if and only if $\lambda\inv u\in\lambda\inv A$ and
  $(\lambda\inv u)\bul\in-\lambda B$. Using such rescaling, we may
  without loss of generality assume that $\lambda\inv\leq \delta<1$.

  Now consider any solution $u = a+b\sqrt 2\in\Z[\sqrt2]$. From
  $u\in[x_0,x_1]$, we know that $x_0-b\sqrt 2 \leq a \leq x_1-b\sqrt
  2$. But since $x_1-x_0<1$, it follows that for any $b\in\Z$, there
  is at most one $a\in\Z$ yielding a solution. Moreover, we note that
  $b=(u-u\bul)/\rt{3}$, so that any solution satisfies
  $(x_0-y_1)/\rt{3} \leq b \leq (x_1-y_0)/\rt{3}$. The algorithm then
  proceeds by enumerating all the integers $b$ in the interval
  $[(x_0-y_1)/\rt{3}, (x_1-y_0)/\rt{3}]$. For each such $b$, find the
  unique integer $a$ (if any) in the interval $[x_0-b\sqrt 2,
  x_1-b\sqrt 2]$. Finally, check if $a+b\sqrt2$ is a solution. The
  runtime is governed by the number of $b\in\Z$ that need to be
  checked, of which there are at most $O(y_1-y_0) =
  O(\delta\Delta)$. As a consequence of Lemma~\ref{lem:grid-bounds},
  the total number of solutions is at least $\Omega(\delta\Delta)$,
  and so the algorithm is efficient.
\end{proof}

% ----------------------------------------------------------------------
\subsection{Upright rectangles and upright sets}
\label{ssect:grid-pb-Zomega-up-rect}

Recall, from Proposition~\ref{prop:zomega} of
Chapter~\ref{chap:nb-th}, that $\Zomega$ can be seen as two disjoint
copies of $\Zrt\times\Zrt$. We use this fact to reduce two dimensional
grid problems over $\Zomega$ for upright rectangles $A$ and $B$ to
independent one-dimensional grid problems.

\begin{proposition}
  \label{prop:up-rect-Zomega}
  Let $A$ and $B$ be upright rectangles. Then there is an algorithm
  which enumerates all the solutions to the grid problem over
  $\Zomega$ for $A$ and $B$. Moreover, the algorithm requires only a
  constant number of arithmetic operations per solution produced.
\end{proposition}

\begin{proof}
  By assumption, $A=A_x\times A_y$ and $B=B_x\times B_y$, where $A_x$,
  $A_y$, $B_x$, and $B_y$ are closed intervals. By
  Proposition~\ref{prop:zomega}, any potential solution is of the form
  $u=a+b i$ or $u=a+b i+\omega$, where $a,b\in\Z[\sqrt 2]$. When
  $u=a+b i$, then $u\bul = a\bul + b\bul i$. Therefore, the
  two-dimensional grid constraints $u\in A$ and $u\bul\in B$ are
  equivalent to the one-dimensional constraints $a\in A_x$, $a\bul\in
  B_x$ and $b\in A_y$, $b\bul\in B_y$.  On the other hand, when $u=a+b
  i+\omega$, let $v = u-\omega = a+b i$. Then $v\bul = u\bul +
  \omega$, and the constraints $u\in A$ and $u\bul\in B$ are
  equivalent to $v\in A-\omega$ and $v\bul \in B+\omega$, which
  reduces to the one-dimensional constraints $a\in
  A_x-\frac{1}{\sqrt2}$, $a\bul\in B_x+\frac{1}{\sqrt2}$ and $b\in
  A_y-\frac{1}{\sqrt2}$, $b\bul\in B_y+\frac{1}{\sqrt2}$.  In both
  cases, the solutions to the one-dimensional constraints can be
  efficiently enumerated by Proposition~\ref{prop:1-d-Zomega}.
\end{proof}

Now that we are able to solve grid problem over $\Zomega$ for upright
rectangles $A$ and $B$, we can reason as in
Subsection~\ref{ssect:grid-pb-Zi-up-set} to establish the following
proposition.

\begin{proposition}
  \label{prop:up-set-Zomega}
  Let $A,B$ be a pair of $M$-upright sets. Then there exists an
  algorithm which enumerates all the solutions to the grid problem
  over $\Zomega$ for $A$ and $B$. Moreover, the algorithm requires
  $O(1/M^2)$ arithmetic operations per solution produced. In
  particular, when $M>0$ is fixed, it requires only a constant number
  of operations per solution.
\end{proposition}

% ----------------------------------------------------------------------
\subsection{Grid operators}
\label{ssect:grid-op-Zomega}

We now adapt the notion of grid operator from
Subsection~\ref{ssect:grid-op-Zi} to the setting of grid problems over
$\Zomega$.

\begin{definition}
  \label{def:grid-op-Zomega}
  We regard $\Z[\omega]$ as a subset of $\R^2$. A real linear operator
  $\Gop:\R^2 \to \R^2$ is called a {\em grid operator} if
  $\Gop(\Z[\omega]) \seq \Z[\omega]$. Moreover, a grid operator $\Gop$
  is called {\em special} if it has determinant $\pm 1$.
\end{definition}

Grid operators are characterized by the following lemma.

\begin{lemma}
  \label{lem-gridoperators}
  Let $\Gop:\R^2\to\R^2$ be a linear operator, which we can identify
  with a real $2\times 2$-matrix. Then $\Gop$ is a grid operator if
  and only if it is of the form
  \begin{equation}\label{eqn:gridoperator}
    \Gop =
    \left[
      \begin{array}{cc}
        a+\frac{a'}{\sqrt{2}} & b+\frac{b'}{\sqrt{2}}\\
        c+\frac{c'}{\sqrt{2}} & d+\frac{d'}{\sqrt{2}}
      \end{array}
    \right],
  \end{equation}
  where $a,b,c,d,a',b',c',d'$ are integers satisfying $a+b+c+d \equiv
  0\mmod{2}$ and $a'\equiv b'\equiv c' \equiv d'\mmod{2}$.
\end{lemma}

\begin{proof}
  By Proposition~\ref{prop:zomega} from Chapter~\ref{chap:nb-th}, we
  know that a vector $u\in\R^2$ is in $\Z[\omega]$ if and only if it
  can be written of the form
  \begin{equation}\label{eqn:zomega-point}
    u = \left[\begin{matrix} 
        x_1+\frac{x_2}{\sqrt{2}} \\ 
        y_1+\frac{y_2}{\sqrt{2}}
      \end{matrix}\right],
  \end{equation}
  where $x_1,x_2,y_1,y_2$ are integers and $x_2\equiv y_2\mmod{2}$. It
  can then be shown by computation that every operator of the form
  {\eqref{eqn:gridoperator}} is a grid operator. For the converse,
  consider an arbitrary grid operator $\Gop$. We prove the claim by
  applying $\Gop$ to the three points
  $\begin{bsmallmatrix}1\\0\end{bsmallmatrix}$,
  $\begin{bsmallmatrix}0\\1\end{bsmallmatrix}$, and
  $\frac{1}{\sqrt{2}}\begin{bsmallmatrix}1\\1\end{bsmallmatrix} \in
  \Z[\omega]$.  From
  $\Gop\begin{bsmallmatrix}1\\0\end{bsmallmatrix}\in\Z[\omega]$ and
  $\Gop\begin{bsmallmatrix}0\\1\end{bsmallmatrix}\in\Z[\omega]$, it
  follows that the columns of $\Gop$ are of the form
  {\eqref{eqn:zomega-point}}, so that $\Gop$ is of the form
  {\eqref{eqn:gridoperator}}, with integers $a,b,c,d,a',b',c',d'$
  satisfying $a'\equiv c'\mmod{2}$ and $b'\equiv
  d'\mmod{2}$. Moreover, we have
  \[ \Gop\begin{bsmallmatrix}1/\sqrt 2\\1/\sqrt 2\end{bsmallmatrix} =
  \begin{bmatrix}
    \frac{a'+b'}{2}+\frac{a+b}{\sqrt2}\\
    \frac{c'+d'}{2}+\frac{c+d}{\sqrt2}
  \end{bmatrix}\in\Z[\omega],
  \]
  which implies $a+b\equiv c+d\mmod{2}$ and $a'+b'\equiv c'+d'\equiv
  0\mmod{2}$. Together, these conditions imply $a+b+c+d \equiv
  0\mmod{2}$ and $a'\equiv b'\equiv c' \equiv d'\mmod{2}$, as claimed.
\end{proof}

\begin{remark}
  \label{gridoperatorcomposition}
  The composition of two (special) grid operators is again a (special)
  grid operator. If $\Gop$ is a special grid operator, then $\Gop$ is
  invertible and $\Gop\inv$ is a special grid operator. If $\Gop$ is a
  (special) grid operator, then $\Gop\bul$ is a (special) grid
  operator, defined by applying $(-)\bul$ separately to each matrix
  entry, and satisfying $\Gop\bul u\bul = (\Gop u)\bul$.
\end{remark}

The interest of special grid operators lies in the following fact.

\begin{proposition}
  \label{prop:operator-on-constraint}
  Let $\Gop$ be a special grid operator, and let $A$ and $B$ be
  subsets of $\R^2$. Define
  \[ \begin{array}{r@{~}c@{~}l}
    \Gop(A) &=& \s{\Gop u \mid u\in A}, \\
    \Gop\bul(B) &=& \s{\Gop\bul u \mid u\in B}.
  \end{array}
  \]
  Then $u\in\Z[\omega]$ is a solution to the grid problem over
  $\Zomega$ for $A$ and $B$ if and only if $\Gop u$ is a solution to
  the grid problem over $\Zomega$ for $\Gop(A)$ and $\Gop\bul(B)$. In
  particular, the grid problem over $\Zomega$ for $A$ and $B$ is
  computationally equivalent to that for $\Gop(A)$ and $\Gop\bul(B)$.
\end{proposition}

\begin{proof}
  Let $u\in\Z[\omega]$. Then $u$ is a solution to the grid problem for
  $A$ and $B$ if and only if $u\in A$ and $u\bul\in B$, if and only if
  $\Gop u\in \Gop(A)$ and $\Gop\bul u\bul = (\Gop u)\bul \in
  \Gop\bul(B)$, if and only if $\Gop u$ is a solution to the grid
  problem for $\Gop(A)$ and $\Gop\bul(B)$.
\end{proof}

% ......................................................................
\begin{figure}
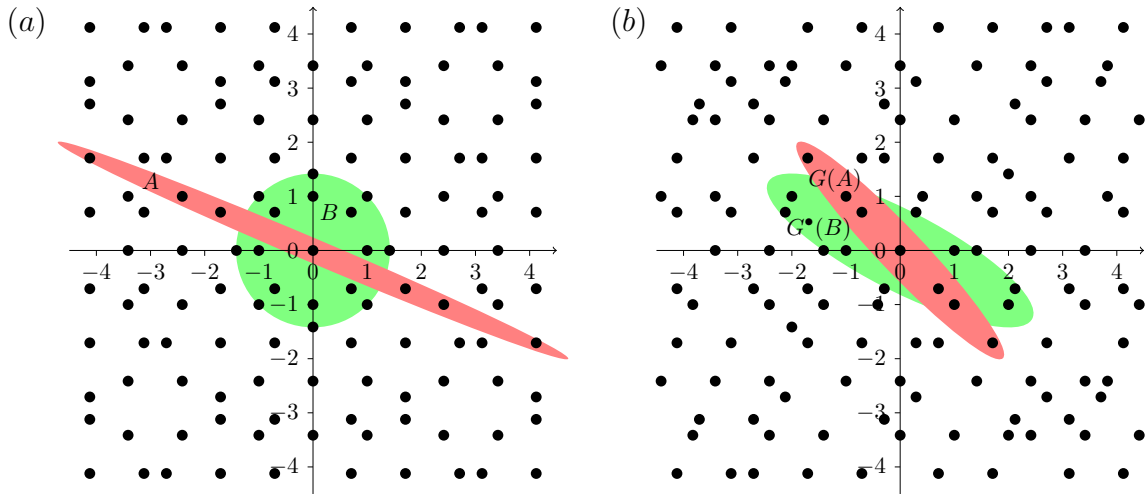

  \[ (a)~ \mp{0.95}{\scalebox{0.8}{\input{grid-4-a.tex}}} \quad (b)~
  \mp{0.95}{\scalebox{0.8}{\input{grid-4-b.tex}}}
  \]
  \caption[The action of a grid operator on a grid problem over
  $\Zomega$.]{(a) The grid problem over $\Zomega$ for two sets $A$ and
    $B$. (b) The grid problem over $\Zomega$ for $\Gop(A)$ and
    $\Gop\bul(B)$. Note that the solutions of (a), which are the grid
    points in the set $A$, are in one-to-one correspondence with the
    solutions of (b), which are the grid points in the set $\Gop(A)$.}
  \label{fig:grid-4}
  \rule{\textwidth}{0.1mm}
\end{figure}
% ......................................................................

Figure~\ref{fig:grid-4}(a) illustrates the grid problem for a pair of
sets $A$ and $B$. As before, the set $B$ is shown in green, and
$\Grid(B)$ is shown as black dots. The set $A$ is shown in red, and
the solutions to the grid problem are the seven grid points that lie
in $A$. Figure~\ref{fig:grid-4}(b) shows the grid problem for the sets
$\Gop(A)$ and $\Gop\bul(B)$, where $\Gop$ is the special grid operator
\[
\Gop = \left[\begin{matrix} 1 & \sqrt2 \\ 0 & 1\end{matrix}\right].
\]
Note that, as predicted by
Proposition~\ref{prop:operator-on-constraint}, the solutions of the
transformed grid problem are in one-to-one correspondence with those
of the original problem; namely, in each case, there are seven
solutions.

% ----------------------------------------------------------------------
\subsection{Ellipses}
\label{ssect:grid-pb-Zomega-ellipses}

Proceeding as in Subsection~\ref{ssect:grid-pb-Zi-ellipses} we now
prove that if $A$ and $B$ are two ellipses, then the grid problem for
$A$ and $B$ over $\Zomega$ can be solved efficiently. For this, we
show that one can find a grid operator $\Gop$ such that $\Gop(A)$ and
$\Gop\bul(B)$ are \emph{both} sufficiently upright. Indeed, we know
from Proposition~\ref{prop:up-set-Zomega} that if both $A$ and $B$ are
upright sets, then the grid problem over $\Zomega$ for $A$ and $B$ can
be solved efficiently. The fact that two ellipses have to be made
simultaneously upright make the problem of finding an appropriate grid
operator significantly more complicated.

We start by reformulating the problem in more convenient terms. Recall
from Definition~\ref{def:lambda-delta} of Chapter~\ref{chap:nb-th}
that $\lambda=\sqrt{2} +1$. The matrix $D$ corresponding to an ellipse
$E$ therefore has determinant 1 if and only if it can be written in
the form
\begin{equation}
  D =
  \left[ 
    \begin{array} {cc} 
      e{\lambda^{-z}} & b \\
      b & e{\lambda^{z}}
    \end{array} 
  \right]
\end{equation}
for some $b,e,z\in\R$ with $e>0$ and $e^2=b^2+1$. As established in
Proposition~\ref{prop:uprightness-ellipse}, the definition of
uprightness simplifies in this case to
\begin{equation}
  \label{eqn:upe}
  \up(E) = \frac{\pi}{4e^2} = \frac{\pi}{4\sqrt{b^2+1}}.
\end{equation}
Equivalently, if $\up(E)=M$, then
\begin{equation}
  \label{eqn:b2}
  b^2 = \frac{\pi^2}{16M^2} - 1.
\end{equation}

Since we now have to deal with pairs of ellipses, it is convenient to
introduce the following terminology for discussing pairs of matrices.

\begin{definition}
  \label{def:state}
  A \emph{state} is a pair of real symmetric positive definite
  matrices of determinant 1. Given a state $\state$ with
  \begin{equation}
    \label{eqn:state}
    D =
    \left[ 
      \begin{array}{cc} 
        e{\lambda^{-z}} & b \\
        b & e{\lambda^{z}}
      \end{array} 
    \right]
    \quad
    \Delta =
    \left[ 
      \begin{array}{cc} 
        \varepsilon{\lambda^{-\zeta}} & \beta \\
        \beta & \varepsilon{\lambda^{\zeta}}
      \end{array} 
    \right]
  \end{equation}
  we define its \emph{skew} as $\sk\state=b^2+\beta^2$ and its
  \emph{bias} as $\bias \state =\zeta-z$.
\end{definition}

Note that the skew of a state is small if and only if both $b^2$ and
$\beta^2$ are small, which happens, by {\eqref{eqn:upe}}, if and only
if the ellipses corresponding to $D$ and $\Delta$ both have large
uprightness. So our strategy for increasing the uprightness will be to
reduce the skew, as in Subsection~\ref{ssect:grid-pb-Zi-ellipses}. In
what follows, we use $\state$ to denote an arbitrary state and always
assume that the entries of $D$ and $\Delta$ are given as in
\eqref{eqn:state}. For future reference, we record here another useful
property of states.

\begin{remark}
  \label{rem-be}
  If $\state$ is a state with $b\geq 0$, then $-be\leq -b^2$. Indeed:
  \[
  e^2=b^2+1 ~\imp~ e^2\geq b^2 ~\imp~ e\geq b ~\imp~ -be\leq -b^2.
  \]
  Similarly, if $b\leq 0$, then $be\leq -b^2$. Analogous inequalities
  also hold for $\beta$ and $\varepsilon$.
\end{remark}

The action of a grid operator on an ellipse can be adapted to states
in a natural way, provided that the operator is special.

\begin{definition}
  \label{def:action}
  The action of special grid operators on states is defined as
  follows. Here, $\Gop\da$ denotes the transpose of $\Gop$, and
  $\Gop\bul$ is defined by applying $(-)\bul$ separately to each
  matrix entry, as in Remark~\ref{gridoperatorcomposition}.
  \[
  \state\cdot \Gop = (\Gop^\dagger D\Gop,\Gop^{\bullet\dagger} \Delta
  \Gop^\bullet).
  \]
\end{definition}

\begin{lemma}
  \label{lem:ellipse-action}
  Let $(D,\Delta)$ be a state, and let $A$ and $B$ be the ellipses
  centered at the origin that are defined by $D$ and $\Delta$,
  respectively.  Then the ellipses $\Gop(A)$ and $\Gop\bul(B)$ are
  defined by the matrices $D'$ and $\Delta'$, where
  \[
  (D',\Delta') = (D,\Delta) \cdot \Gop\inv
  \]
\end{lemma}

\begin{proof}
  We have
  \[
  \begin{array}{rcl}
    \Gop(A) &=& \s{\Gop(u)\in\R^2 \mid u\da D u \leq 1} \\
    &=& \s{v\in\R^2 \mid (\Gop\inv v)\da D(\Gop\inv v) \leq 1} \\
    &=& \s{v\in\R^2 \mid v\da (\Gop\inv)\da D \Gop\inv v \leq 1},
  \end{array}
  \]
  so the ellipse $\Gop(A)$ is defined by the positive operator
  $D'=(\Gop\inv)\da D \Gop\inv$. The proof for $\Gop\bul(B)$ is
  similar.
\end{proof}

The main ingredient in our proof that states can be made upright is
the following \emph{Step Lemma}.

\begin{lemma}[Step Lemma]
  \label{lem:step}
  For any state $\state$, if $\sk\state \geq \formula{\P}$, then there
  exists a special grid operator $\Gop$ such that $\sk (\state \cdot
  \Gop)\leq \formula{\Q} ~\sk\state$.  Moreover, $\Gop$ can be
  computed using a constant number of arithmetic operations.
\end{lemma}

Before proving the Step Lemma, we show how it can be used to derive
the following proposition.

\begin{proposition}
  \label{prop:ellipse-Zomega}
  Let $A$ and $B$ be ellipses. Then there exists a grid operator
  $\Gop$ such that $\Gop(A)$ and $\Gop\bul (B)$ are
  $1/\formula{\oneoverM}$-upright. Moreover, if $A$ and $B$ are
  $M$-upright, then $\Gop$ can be efficiently computed in
  $O(\log(1/M))$ arithmetic operations.
\end{proposition}

\begin{proof}
  Let $D$ and $\Delta$ be the matrices defining $A$ and $B$
  respectively, in the sense of Definition~\ref{def:ellipse}. Since
  uprightness is invariant under translations and scaling, we may
  without loss of generality assume that both ellipses are centered at
  the origin, and that $\det D = \det \Delta = 1$.

  The pair $\state$ is a state. By applying Lemma~\ref{lem:step}
  repeatedly, we get grid operators $\Gop_1,\ldots,\Gop_n$ such that:
  \begin{equation}
    \label{eqn-skew-reduction}
    \sk(\state\cdot \Gop_1\ldots \Gop_n) \leq \formula{\P}.
  \end{equation}
  Now let $(D',\Delta')=\state\cdot \Gop_1\ldots \Gop_n$ and set
  $\Gop=(\Gop_1\cdots \Gop_n)^{-1}$. By
  Lemma~\ref{lem:ellipse-action}, the ellipses $\Gop(A)$ and
  $\Gop^{\bullet}(B)$ are defined by the matrices $D'$ and $\Delta'$,
  respectively.  Let $b$ and $\beta$ be the anti-diagonal entries of
  the matrices $D'$ and $\Delta'$, respectively.  We have:
  \[
  b^2+\beta^2 = \sk (D',\Delta')= \sk(\state\cdot \Gop\inv) = \sk
  (\state\cdot \Gop_1\ldots \Gop_n)\leq \formula{\P},
  \]
  hence $b^2\leq \formula{\P}$ and $\beta^2\leq\formula{\P}$.  Using
  {\eqref{eqn:upe}}, we get
  \[
  \up(\Gop(A)) = \frac{\pi}{4\sqrt{b^2+1}} \geq
  \frac{\pi}{4\sqrt{\formula{\P+1}}} \geq 1/\formula{\oneoverM}
  \]
  and similarly $\up(\Gop\bul(B)) \geq 1/\formula{\oneoverM}$, as
  desired.
  % Check!!!
  \assert{\P <= (pi^2 / (16 * \M^2))-1}%

  To bound the number of operations, note that each application of
  $\Gop_j$ reduces the skew by at least
  $\formula{\roundone{100-100*\Q}}$ percent.  Therefore, the number
  $n$ in {\eqref{eqn-skew-reduction}} satisfies $n\leq
  \log_{\formula{\Q}}({\formula{\P}}/{\sk\state}) =
  O(\log(\sk\state))$. Using {\eqref{eqn:b2}}, we have
  \[
  \log (\sk\state) = \log (b^2+\beta^2) \leq \log
  ((\frac{\pi^2}{16M^2} - 1) + (\frac{\pi^2}{16M^2} - 1)) =
  O(\log(1/M)).
  \]
  It follows that the computation of $\Gop$ requires $O(\log(1/M))$
  applications of the Step Lemma, each of which requires a constant
  number of arithmetic operations, proving the final claim of the
  proposition.
\end{proof}

The remainder of this subsection is devoted to proving the Step
Lemma. To each state, we associate the pair $(z, \zeta)$. The proof of
the Step Lemma is essentially a case distinction on the location of
the pair $(z, \zeta)$ in the plane. We find coverings of the plane
with the property that if the point $(z, \zeta)$ belongs to some
region $\mathcal{O}$ of our covering, then it is easy to compute a
special grid operator $\Gop$ such that $\sk(\state\cdot \Gop)\leq
\formula{\Q}~\sk\state$. The relevant grid operators are given in
Figure~\ref{list_operators}.

% ......................................................................
\begin{figure}
  \[
  R= \frac{1}{\sqrt{2}}\left[ \begin{array} {cc}
      1 & -1  \\
      1 & 1
    \end{array} \right]
  ~~
  A = \left[ \begin{array} {cc} 
      1 & -2  \\
      0 & 1
    \end{array} \right]
  ~~ 
  B = \left[ \begin{array} {cc} 
      1 & \sqrt{2}  \\
      0 & 1
    \end{array} \right]
  \]
  \[
  K = \frac{1}{\sqrt{2}}\left[ \begin{array} {cc}
      -\lambda^{-1} & -1  \\
      \lambda & 1
    \end{array} \right]
  ~~
  X= \left[ \begin{array} {cc} 
      0 & 1  \\
      1 & 0
    \end{array} \right]
  ~~ 
  Z= \left[ \begin{array} {cc} 
      1 & 0  \\
      0 & -1
    \end{array} \right]
  \]
  \caption[The grid operators $R$, $A$, $B$, $X$, $K$, and $Z$.]{The
    grid operators $R$, $A$, $B$, $X$, $K$, and $Z$.}
  \label{list_operators}
  \rule{\textwidth}{0.1mm}
\end{figure}
% ......................................................................

Each one of the next five subsections is dedicated to a particular
region of the plane. 

% --------------------------------------------------------------------
\subsubsection{The Shift Lemma}
\label{sssect-shift}

In this section, we consider states $\state$ such that
$|\bias\state|>1$. Any such state can be ``shifted" to a state $(D',
\Delta')$ of equal skew but with $|\bias(D', \Delta')|\leq 1$.

\begin{definition}
  \label{def:shift-ops}
  The \emph{shift operators} $\sigma$ and $\tau$ are defined by:
  \[
  \sigma = \sqrt{\lambda^{-1}}\left[\begin{array}{cc}
      \lambda & 0 \\
      0 & 1
    \end{array} \right],
  \tau =
  \sqrt{\lambda^{-1}}
  \left[\begin{array}{cc}
      1 & 0 \\
      0 & -\lambda
    \end{array} \right]
  \]
\end{definition}

Even though $\sigma$ and $\tau$ are not grid operators, we can use
them to define an operation on states called a \emph{shift by $k$}.
By abuse of notation, we write this operation as an action.

\begin{definition}
  \label{def:k-shift}
  Given a state $\state$ and $k\in\Z$, the \emph{k-shift of $\state$}
  is defined as:
  \[
  \state \cdot \shift^k = (\sigma^k D \sigma^k, \tau^k \Delta \tau^k).
  \]
\end{definition}

The notation $\state \cdot \shift^k$ is justified by the following
lemma.

\begin{lemma}
  \label{propshift}
  The shift of a state is a state and moreover:
  \[
  \sk(\state\cdot\shift^k)= \sk\state ~\mbox{ and }~
  \bias(\state\cdot\shift^k)=\bias\state +2k
  \]
\end{lemma}

\begin{proof}
  Compute $\state\cdot\shift^k$:
  \[
  \begin{array}{rcl}
    \state\cdot\shift^k & = & 
    (\sigma^k D \sigma^k, \tau^k \Delta \tau^k) \\
    & = & (\sigma^k\left[ \begin{array} {cc} 
        e{\lambda^{-z}} & b \\
        b & e{\lambda^{z}}
      \end{array} \right] \sigma^k,
    \tau^k\left[ \begin{array} {cc} 
        \varepsilon{\lambda^{-\zeta}} & \beta \\
        \beta & \varepsilon{\lambda^{\zeta}}
      \end{array} \right]\tau^k) \\
    & = & (\left[ \begin{array} {cc} 
        e\lambda^{-z+k} & b \\
        b & e\lambda^{z-k}
      \end{array} \right],
    \left[ \begin{array} {cc} 
        \varepsilon\lambda^{-\zeta-k} & (-1)^k\beta \\
        (-1)^k\beta & \varepsilon\lambda^{\zeta+k}
      \end{array} \right])
  \end{array}
  \]
  The resulting matrices are clearly symmetric and positive
  definite. Moreover, since $\sigma^k$ and $\tau^k$ have determinant
  $\pm 1$, both $\sigma^k D\sigma^k$ and $\tau^k \Delta\tau^k$ have
  determinant 1. Finally:
  \begin{itemize}
  \item $\sk(\state\cdot\shift^k)= b^2+((-1)^k\beta)^2=b^2+\beta^2=
    \sk\state$ and
  \item $\bias(\state\cdot\shift^k)=(\zeta+k)-(z-k)=\bias\state +2k$,
  \end{itemize}
  which completes the proof.
\end{proof}

For every special grid operator $\Gop$, there is a special grid
operator $\Gop'$ whose action on a state corresponds to shifting the
state by $k$, applying $\Gop$ and then shifting the state by $-k$.

\begin{lemma}
  \label{conjugationbysigma1}
  If $\Gop$ is a special grid operator and $k\in\Z$, then
  $\Gop'=\sigma^k \Gop \sigma^k$ is a special grid operator and
  moreover $\Gop'^\bullet = (-\tau)^k \Gop^\bullet \tau^k$.
\end{lemma}

\begin{proof}
  It suffices to show this for $k=1$. Suppose
  $\Gop=\left[ \begin{array} {cc}
      w & x \\
      y & z
    \end{array} \right]$ is a special grid operator and note that:
  \[
  \Gop'= \sigma \Gop \sigma= \left[ \begin{array} {cc}
      \lambda w & x \\
      y & \lambda^{-1}z
    \end{array} \right]
  =
  \left[ \begin{array} {cc} 
      \lambda^{-1} & 0 \\
      0 & \lambda^{-1}
    \end{array} \right]
  \left[ \begin{array} {cc} 
      \lambda & 0 \\
      0 & 1
    \end{array} \right]
  \Gop
  \left[ \begin{array} {cc} 
      \lambda  & 0 \\
      0 & 1
    \end{array} \right].
  \]
  Since all the factors in the above product are grid operators, the
  result is also a grid operator. Moreover
  $\det(\sigma\Gop\sigma)=\det(\Gop)=1$ so that $\sigma\Gop\sigma$ is
  special. Finally:
  \[
  \Gop'^\bullet=(\sigma \Gop \sigma)^\bullet= \left[ \begin{array}
      {cc}
      \lambda^\bullet w^\bullet & x^\bullet \\
      y^\bullet & (\lambda^{-1})^\bullet z^\bullet
    \end{array} \right]
  =
  \left[ \begin{array} {cc} 
      -\lambda^{-1} w^\bullet & x^\bullet \\
      y^\bullet & -\lambda z^\bullet
    \end{array} \right]
  =
  -\tau \Gop^\bullet \tau.
  \]
\end{proof}

\begin{lemma}
  \label{conjugationbysigma2}
  If $\Gop$ is a grid operator, then:
  \[
  ((\state\cdot \shift^k) \cdot \Gop )\cdot \shift^k = \state \cdot
  (\sigma^k \Gop\sigma^k).
  \]
\end{lemma}

\begin{proof}
  Write $\Gop'=\sigma^k \Gop\sigma^k$. Simple computation then yields
  the result:
  \[
  \begin{array}{rcl}
    ((\state\cdot \shift^k) \cdot \Gop )\cdot \shift^k 
    & = & ((\sigma^k D \sigma^k,\, \tau^k \Delta \tau^k)
    \cdot \Gop )\cdot \shift^k \\
    & = & (\Gop^\dagger\sigma^k D \sigma^k \Gop,\, 
    \Gop^{\bullet\dagger} \tau^k \Delta \tau^k \Gop^\bullet )
    \cdot \shift^k \\
    & = & (\sigma^k\Gop^\dagger\sigma^k D \sigma^k 
    \Gop\sigma^k,\, \tau^k \Gop^{\bullet\dagger} 
    \tau^k \Delta \tau^k \Gop^\bullet \tau^k ) \\
    & = & (\sigma^k\Gop^\dagger\sigma^k D \sigma^k 
    \Gop\sigma^k,\, ((-\tau)^k\Gop^{\bullet\dagger} 
    \tau^k )\Delta ((-\tau)^k \Gop^\bullet \tau^k) ) \\
    & = & (\Gop'^\dagger D \Gop',\, 
    \Gop'^{\bullet\dagger} \Delta {\Gop'}^\bullet ) \\
    & = & \state \cdot \Gop' \\
    & = & \state \cdot (\sigma^k \Gop\sigma^k).
  \end{array}
  \]
\end{proof}

Shifts allow us to consider only states $\state$ with
$\bias\state\in[-1,1]$ in the proof of the Step Lemma.

\begin{lemma}
  \label{shift_lemma}
  If the Step Lemma holds for all states $\state$ with $\bias\state\in
  [-1,1]$, then it holds for all states.
\end{lemma}

\begin{proof}
  Let $\state$ be some state with $\sk\state\geq\formula{\P}$.  Let
  $x=\bias\state$ and set $k=\floor{\frac{1-x}{2}}$. Then by
  Lemma~\ref{propshift}, we have $\sk(\state\cdot\shift^k)= \sk\state$
  and $\bias(\state\cdot\shift^k)\in[-1,1]$. Then by assumption, there
  exists a special grid operator $\Gop$ such that
  $\sk((\state\cdot\shift^k)\cdot \Gop)\leq \formula{\Q}~\sk(\state
  \cdot\shift^k)$. Now by Lemma~\ref{conjugationbysigma1} we know that
  $\Gop'=\sigma^k ~\Gop~ \sigma^k$ is a special grid
  operator. Moreover, by Lemma~\ref{conjugationbysigma2} and
  \ref{propshift}, we have:
  \[
  \begin{array}{rcl}
    \sk(\state\cdot \Gop') & = & 
    \sk(((\state\cdot\shift^k)\cdot \Gop)\cdot \shift^k) \\
    & = & \sk((\state \cdot\shift^k)\cdot \Gop) \\
    & \leq & \formula{\Q}~\sk(\state \cdot\shift^k) \\
    & = & \formula{\Q}~\sk\state ,
  \end{array}
  \]
  which completes the proof.
\end{proof}

% --------------------------------------------------------------------
\subsubsection[The R Lemma]{The $R$ Lemma}
\label{sssect-R}

\begin{definition}
  \label{def:sin-l}
  The \emph{hyperbolic sine in base $\lambda$} is defined as:
  \[
  \sinl(x) = \frac{\lambda^x-\lambda^{-x}}{2}.
  \]
\end{definition}

\begin{lemma}
  \label{lemmaR}
  Recall the operator $R$ from Figure~\ref{list_operators}.  If
  $\state$ is a state such that $\sk\state\geq\formula{\P}$, and such
  that $\formula{-\r}\leq z\leq\formula{\r}$ and $\formula{-\r}\leq
  \zeta\leq\formula{\r}$, then:
  \[
  \sk(\state\cdot R) \leq \formula{\Q}~\sk\state.
  \]
\end{lemma}

\begin{proof}
  Compute the action of $R$ on $\state$:
  \[
  \begin{array}{lcl}
    R^\dagger DR 
    & = & \displaystyle\frac{1}{2}
    \left[ \begin{array} {cc} 
        1 & 1  \\
        -1 & 1
      \end{array} \right]
    \left[ \begin{array} {cc} 
        e{\lambda^{-z}} & b \\
        b & e{\lambda^{z}}
      \end{array} \right]
    \left[ \begin{array} {cc} 
        1 & -1  \\
        1 & 1
      \end{array} \right] \\\\[-1.5ex]
    & = &
    \left[ \begin{array} {cc} 
        \ldots & \frac{e({\lambda^{z}}-{\lambda^{-z}})}{2}\\
        \frac{e({\lambda^{z}}-{\lambda^{-z}})}{2} & \ldots
      \end{array} \right] 
    ~ = ~ 
    \left[ \begin{array} {cc} 
        \ldots & e\sinl(z)\\
        e\sinl(z) & \ldots
      \end{array} \right], 
    \\\\[-1.5ex]
    R^{\bullet\dagger} \Delta R^\bullet 
    & = & \displaystyle\frac{1}{2}
    \left[ \begin{array} {cc} 
        -1 & -1  \\
        1 & -1
      \end{array} \right]
    \left[ \begin{array} {cc} 
        \varepsilon{\lambda^{-\zeta}} & \beta \\
        \beta & \varepsilon{\lambda^{\zeta}}
      \end{array} \right]
    \left[ \begin{array} {cc} 
        -1 & 1  \\
        -1 & -1
      \end{array} \right] \\\\[-1.5ex]
    & = & \left[ \begin{array} {cc} 
        \ldots & 
        \frac{\varepsilon({\lambda^{\zeta}}-{\lambda^{-\zeta}})}{2}\\
        \frac{\varepsilon({\lambda^{\zeta}}-{\lambda^{-\zeta}})}{2} & 
        \ldots
      \end{array} \right] 
    ~ = ~ \left[ \begin{array} {cc} 
        \ldots & \varepsilon\sinl(\zeta)\\
        \varepsilon\sinl(\zeta) & \ldots
      \end{array} \right].
  \end{array}
  \]
  Therefore $\sk(\state\cdot R)= e^2\sinl^2(z)+
  \varepsilon^2\sinl^2(\zeta)$. But recall that $e^2=b^2+1$ and
  $\varepsilon^2=\beta^2+1$, so that in fact:
  \[
  \sk(\state\cdot R) = (b^2+1)\sinl^2(z)+ (\beta^2+1)\sinl^2(\zeta).
  \]
  We assumed $\formula{-\r}\leq z,\zeta\leq\formula{\r}$ and this
  implies that $\sinl^2(\zeta), \sinl^2(z) \leq
  \sinl^2(\formula{\r})$. Writing $y=\sinl^2(\formula{\r})$ for
  brevity, and using the assumption that $\sk\state\geq\formula{\P}$,
  we get:
  \[
  \begin{array}{rcl}
    \sk(\state\cdot R) & = &
    (b^2+1)\sinl^2(z)+
    (\beta^2+1)\sinl^2(\zeta)\\
    & \leq & 
    (b^2+1)y+(\beta^2+1)y\\
    & = & (b^2+\beta^2+2)y \\
    & \leq & \sk\state (1+\frac{2}{\formula{\P}})y.
  \end{array}
  \]
  This completes the proof, since $(1+\frac{2} {\formula{\P}})y =
  (1+\frac{2} {\formula{\P}})\sinl^2(\formula{\r}) \approx
  \formula{(1+(2/\P))*(\sinhl \r)^2} \leq \formula{\Q}$.
  %% Check!!!
  \assert{(1+(2/\P))*(\sinhl \r)^2 <= \Q}%
\end{proof}

% --------------------------------------------------------------------
\subsubsection[The K Lemma]{The $K$ Lemma}
\label{sssect-K}

\begin{definition}
  \label{def:cos-l}
  The \emph{hyperbolic cosine in base $\lambda$} is defined as:
  \[
  \cosl (x) = \frac{\lambda^x+\lambda^{-x}}{2}.
  \]
\end{definition}

\begin{lemma}
  \label{lemmaK}
  Recall the operator $K$ from Figure~\ref{list_operators}.  If
  $\state$ is a state such that $\bias\state\in [-1,1]$,
  $\sk\state\geq\formula{\P}$, and such that $b,\beta\geq 0$,
  $z\leq\formula{-\a}$, and $\formula{\r}\leq\zeta$, then:
  \[
  \sk(\state\cdot K)\leq \formula{\Q}~\sk\state.
  \]
\end{lemma}

\begin{proof}
  Compute the action of $K$ on $\state$:
  \begin{eqnarray}
    \lefteqn{K^\dagger DK}\nonumber\\
    & = & \frac{1}{2}
    \left[ \begin{array}{cc} 
        -\lambda^{-1} & \lambda  \\
        -1 & 1
      \end{array} \right]
    \left[ \begin{array}{cc} 
        e{\lambda^{-z}} & b \\
        b & e{\lambda^{z}}
      \end{array} \right]
    \left[ \begin{array}{cc} 
        -\lambda^{-1} & -1  \\
        \lambda & 1
      \end{array} \right] \nonumber \\[9pt]
    & = & \frac{1}{2} 
    \left[ \begin{array}{cc} 
        \ldots & e(\lambda^{z+1}+\lambda^{-z-1})
        -2\sqrt{2}b \\
        e(\lambda^{z+1}+\lambda^{-z-1})-2\sqrt{2}b & \ldots
      \end{array} \right]  \nonumber \\[9pt]
    & = & 
    \left[ \begin{array}{cc} 
        \ldots & e\cosl(z+1) -\sqrt{2}b \\
        e\cosl(z+1)  -\sqrt{2}b & \ldots
      \end{array} \right],\nonumber
  \end{eqnarray}
  \begin{eqnarray}
    \lefteqn{K^{\bullet\dagger} \Delta K^\bullet} 
    \nonumber\\
    & = & \frac{1}{2}
    \left[ \begin{array}{cc} 
        \lambda & -\lambda^{-1}  \\
        -1 & 1
      \end{array} \right]
    \left[ \begin{array}{cc} 
        \varepsilon{\lambda^{-\zeta}} & \beta \\
        \beta & \varepsilon{\lambda^{\zeta}}
      \end{array} \right]
    \left[ \begin{array}{cc} 
        \lambda & -1  \\
        -\lambda^{-1} & 1
      \end{array} \right] \nonumber \\[9pt]
    & = & \frac{1}{2} 
    \left[ \begin{array}{cc} 
        \ldots & -\varepsilon(
        \lambda^{\zeta-1}+\lambda^{-\zeta+1})+2\sqrt{2}\beta \\
        -\varepsilon(\lambda^{\zeta-1}+\lambda^{-\zeta+1})
        +2\sqrt{2}\beta & \ldots
      \end{array} \right]  \nonumber \\[9pt]
    & = &
    \left[ \begin{array}{cc} 
        \ldots & \sqrt{2}\beta - \varepsilon \cosl(\zeta-1) \\
        \sqrt{2}\beta - \varepsilon\cosl(\zeta-1)
        & \ldots
      \end{array} \right].\nonumber
  \end{eqnarray}

  \noindent
  Therefore:
  \begin{equation}\label{k0}
    \sk(\state \cdot K)=(\sqrt{2}b -e\cosl(z+1))^2 
    +(\sqrt{2}\beta-\varepsilon\cosl(\zeta-1))^2.
  \end{equation}

  \noindent
  But recall that $e^2=b^2+1$, and from Remark~\ref{rem-be} that
  $b\geq 0$ implies $-be\leq -b^2$, so:
  \begin{eqnarray}
    \lefteqn{(\sqrt{2}b -e\cosl(z+1))^2} 
    \nonumber\\
    & = & 2b^2-2\sqrt{2}\,be\cosl(z+1)+ 
    e^2\cosl^2(z+1) \nonumber\\
    & \leq & 2b^2-2\sqrt{2}\,b^2\cosl(z+1)+ 
    (b^2+1)\cosl^2(z+1) \nonumber\\
    & = & b^2(2-2\sqrt{2}\cosl(z+1)+ 
    \cosl^2(z+1))+\cosl^2(z+1) 
    \nonumber\\
    & = & b^2(\sqrt{2}-\cosl(z+1))^2+
    \cosl^2(z+1).\label{k1}
  \end{eqnarray}
  Reasoning analogously, we also have
  \begin{equation}
    (\sqrt{2}\beta-\varepsilon\cosl(\zeta-1))^2
    \leq 
    \beta^2(\sqrt{2}-\cosl(\zeta-1))^2+\cosl^2(\zeta-1).
  \end{equation}

  \noindent
  By assumption, $\bias\state\in [-1,1]$, thus $\zeta\leq z+1$. This,
  together with the assumptions $\formula{\r}\leq\zeta$ and
  $z\leq\formula{-\a}$, implies that both $z+1$ and $\zeta-1$ are in
  the interval $[\formula{-\rr},\formula{-\aa}]$. On this interval,
  the function $\cosl^2(x)$ assumes its maximum at $x=\pgfmathparse{
    ifthenelse((\coshl {-\rr}) <= (\coshl {-\aa}), -\aa,
    -\rr)}\pgfmathresult$, and the function $f(x) = (\sqrt 2 -
  \cosl(x))^2$ assumes its maximum at $x=\assert{-\rr <= 0 && 0 <=
    -\aa}%
  \pgfmathparse{ ifthenelse((\kone{-\rr}) <= (\kone{0}) &&
    (\kone{-\aa}) <= (\kone{0}), 0, ifthenelse ((\kone {-\rr}) <=
    (\kone {-\aa}), -\aa, -\rr)}\pgfmathresult$.  Therefore,
  \begin{equation}
    b^2(\sqrt{2}-\cosl(z+1))^2  + \cosl^2(z+1)
    \leq 
    \pgfmathparse{
      ifthenelse((\kone{-\rr}) <= (\kone{0}) 
      && (\kone{-\aa}) <= (\kone{0}), 0,
      ifthenelse ((\kone {-\rr}) <= (\kone {-\aa}), -\aa, -\rr)
    }
    b^2(\sqrt{2} -\cosl(\pgfmathresult))^2 
    + 
    \pgfmathparse{
      ifthenelse((\coshl {-\rr}) <= (\coshl {-\aa}), -\aa, -\rr)
    }
    \cosl^2(\pgfmathresult)
  \end{equation}
  and
  \begin{equation}
    \beta^2(\sqrt{2}-\cosl(\zeta-1))^2+\cosl^2(\zeta-1)
    \leq  
    \beta^2\pgfmathparse{
      ifthenelse((\kone{-\rr}) <= (\kone{0}) 
      && (\kone{-\aa}) <= (\kone{0}), 0,
      ifthenelse ((\kone {-\rr}) <= (\kone {-\aa}), -\aa, -\rr)
    }
    (\sqrt{2} -\cosl(\pgfmathresult))^2
    +
    \pgfmathparse{
      ifthenelse((\coshl {-\rr}) <= (\coshl {-\aa}), -\aa, -\rr)
    }
    \cosl^2(\pgfmathresult).\label{k3}
  \end{equation}

  \noindent
  Combining (\ref{k0})--(\ref{k3}), together with the assumption that
  $\sk\state\geq\formula{\P}$, yields:
  \begin{eqnarray}
    \sk(\state \cdot K) & = & 
    (\sqrt{2}b -e\cosl(z+1))^2 
    +(\sqrt{2}\beta-\varepsilon\cosl(\zeta-1))^2 \nonumber\\
    & \leq & 
    (b^2+\beta^2)\pgfmathparse{
      ifthenelse((\kone{-\rr}) <= (\kone{0}) 
      && (\kone{-\aa}) <= (\kone{0}), 0,
      ifthenelse ((\kone {-\rr}) <= (\kone {-\aa}), -\aa, -\rr)
    }
    (\sqrt{2} -\cosl(\pgfmathresult))^2
    +
    2\pgfmathparse{
      ifthenelse((\coshl {-\rr}) <= (\coshl {-\aa}), -\aa, -\rr)
    }
    \cosl^2(\pgfmathresult) \nonumber\\
    & = & 
    \sk\state\pgfmathparse{
      ifthenelse((\kone{-\rr}) <= (\kone{0}) 
      && (\kone{-\aa}) <= (\kone{0}), 0,
      ifthenelse ((\kone {-\rr}) <= (\kone {-\aa}), -\aa, -\rr)
    }
    (\sqrt{2} -\cosl(\pgfmathresult))^2
    +
    2\pgfmathparse{
      ifthenelse((\coshl {-\rr}) <= (\coshl {-\aa}), -\aa, -\rr)
    }
    \cosl^2(\pgfmathresult) \nonumber\\
    & \leq & 
    \sk\state(\pgfmathparse{
      ifthenelse((\kone{-\rr}) <= (\kone{0}) 
      && (\kone{-\aa}) <= (\kone{0}), 0,
      ifthenelse ((\kone {-\rr}) <= (\kone {-\aa}), -\aa, -\rr)
    }
    (\sqrt{2} -\cosl(\pgfmathresult))^2
    +
    \frac{2}{\formula{\P}}\pgfmathparse{
      ifthenelse((\coshl {-\rr}) <= (\coshl {-\aa}), -\aa, -\rr)
    }
    \cosl^2(\pgfmathresult)) \nonumber
  \end{eqnarray}
  This completes the proof since $\pgfmathparse{
    ifthenelse((\kone{-\rr}) <= (\kone{0}) && (\kone{-\aa}) <=
    (\kone{0}), 0, ifthenelse ((\kone {-\rr}) <= (\kone {-\aa}), -\aa,
    -\rr) } (\sqrt{2} -\cosl(\pgfmathresult))^2 +
  \frac{2}{\formula{\P}}\pgfmathparse{ ifthenelse((\coshl {-\rr}) <=
    (\coshl {-\aa}), -\aa, -\rr) } \cosl^2(\pgfmathresult)
  \approx\formula{max((\kone {-\rr}),(\kone {-\aa}),(\kone{0})) +
    (max((\coshl {-\rr})^2, (\coshl {-\aa})^2))*(2/\P)} \leq
  \formula{\Q}$.
  %% Check!!!
  \assert{max((\kone {-\rr}),(\kone {-\aa}),(\kone{0})) + (max((\coshl
    {-\rr})^2, (\coshl {-\aa})^2))*(2/\P) <= \Q}%
\end{proof}

% --------------------------------------------------------------------
\subsubsection[The A Lemma]{The $A$ Lemma}
\label{sssect-A}

\begin{definition}
  Let $g(x) = (1-2x)^2$.
\end{definition}

\begin{lemma}
  \label{lemmaA}
  Recall the operator $A$ from Figure~\ref{list_operators}.  If
  $\state $ is a state such that $\bias\state \in[-1,1]$,
  $\sk\state\geq\formula{\P}$, and such that $b,\beta\geq 0$ and
  $\formula{-\a}\leq{z},\zeta$, then there exists $n\in \Z$ such that:
  \[
  \sk(\state \cdot A^n) \leq \formula{\Q}~\sk\state .
  \]
\end{lemma}

\begin{proof}
  Let $c=\min\s{{z},{\zeta}}$ and $n=\max\s{1,
    \floor{\frac{\lambda^{c}}{2}}}$.  Compute the action of $A^{n}$ on
  $\state$:
  \[
  \begin{array}{lcl}
    {A^{n}}^\dagger DA^{n} & = & 
    \left[\begin{array} {cc} 
        1 & 0  \\
        -2n & 1
      \end{array}\right]
    \left[ \begin{array} {cc} 
        e{\lambda^{-z}} & b \\
        b & e{\lambda^{z}}
      \end{array}\right]
    \left[ \begin{array} {cc} 
        1 & -2n  \\
        0 & 1
      \end{array}\right] \\[9pt]
    & = & \left[\begin{array} {cc} 
        \ldots & b-2ne{\lambda^{-z}}  \\
        b-2ne{\lambda^{-z}} & \ldots
      \end{array} \right], \\[18pt]
    {A^{n}}^{\bullet\dagger}
    \Delta{A^{n}}^\bullet 
    & = & {A^{n}}^\dagger \Delta A^{n} \\
    & = & \left[ \begin{array} {cc} 
        \ldots & \beta-2n\varepsilon{\lambda^{-\zeta}}  \\
        \beta-2n\varepsilon{\lambda^{-\zeta}} & \ldots
      \end{array} \right].
  \end{array}
  \]
  Therefore:
  \[
  \sk(\state \cdot A^{n}) =
  (b-2ne{\lambda^{-z}})^2+(\beta-2n\varepsilon{\lambda^{-\zeta}})^2
  \]
  But recall that $e^2=b^2+1$ and $\varepsilon^2=\beta^2+1$, and from
  Remark~\ref{rem-be} that $b,\beta\geq 0$ implies $-be\leq -b^2$ and
  $-\varepsilon\beta\leq -\beta^2$. Using these facts, we can expand
  the above formula as follows:
  \begin{eqnarray}
    \lefteqn{\sk(\state\cdot A^{n})}\nonumber \\
    & = & (b-2ne{\lambda^{-z}})^2+
    (\beta-2n\varepsilon{\lambda^{-\zeta}})^2 
    \nonumber\\
    & = & b^2
    -4nbe{\lambda^{-z}}
    +4n^2e^2{\lambda^{-2z}}
    +\beta^2
    -4n\beta\varepsilon{\lambda^{-\zeta}}
    +4n^2\varepsilon^2{\lambda^{-2\zeta}}
    \nonumber\\
    & \leq & b^2
    -4nb^2{\lambda^{-z}}
    +4n^2(b^2+1){\lambda^{-2z}}
    +\beta^2
    -4n\beta^2{\lambda^{-\zeta}}
    +4n^2(\beta^2+1){\lambda^{-2\zeta}}
    \nonumber\\
    & = & 
    b^2(
    1
    -4n{\lambda^{-z}}
    +4n^2{\lambda^{-2z}}
    )+\beta^2(
    1
    -4n{\lambda^{-\zeta}}
    +4n^2{\lambda^{-2\zeta}}
    )+4n^2({\lambda^{-2z}}+{\lambda^{-2\zeta}})
    \nonumber \\
    & = & 
    b^2(1-2n\lambda^{-z})^2
    +\beta^2(1-2n\lambda^{-\zeta})^2
    +4n^2({\lambda^{-2z}}+{\lambda^{-2\zeta}})
    \nonumber \\
    & = & b^2g(n{\lambda^{-z}})+\beta^2g(n{\lambda^{-\zeta}})+ 
    4n^2({\lambda^{-2z}}+{\lambda^{-2\zeta}}).\nonumber
  \end{eqnarray}

  \noindent
  Writing $y=\max\s{g(n{\lambda^{-z}}), g(n{\lambda^{-\zeta}})}$ for
  brevity, and using the assumption that $\sk\state\geq\formula{\P}$
  together with the fact that $c\leq z,\zeta$, we get:
  \begin{eqnarray}
    \sk(\state \cdot A^{n}) & \leq & 
    b^2y+\beta^2y+ 8n^2\lambda^{-2c} \nonumber\\
    & = & \sk\state y+ 8n^2\lambda^{-2c} \nonumber\\
    & \leq & \sk\state (y+ 
    \frac{8}{\formula{\P}}n^2\lambda^{-2c}). \nonumber
  \end{eqnarray}
  To finish the proof, it remains to show that $y+
  \frac{8}{\formula{\P}}n^2\lambda^{-2c}\leq \formula{\Q}$.  There are
  two cases:
  \begin{itemize}
  \item If $\floor{\frac{\lambda^{c}}{2}}\geq 1$, then
    $\frac{\lambda^{c}}{4}\leq n\leq \frac{\lambda^{c}}{2}$. From
    $n\leq \frac{\lambda^{c}}{2}$, we have $2n\lambda^{-c}\leq 1$, and
    so $\frac{8}{\formula{\P}}n^2\lambda^{-2c} \leq
    \frac{2}{\formula{\P}}$. Moreover, because $\bias\state\in[-1,1]$,
    we have $c \leq z,\zeta \leq c+1$. Hence $\frac{1}{4\lambda} =
    \frac{\lambda^{c}}{4} \lambda^{-c-1} \leq n\lambda^{-c-1} \leq
    n{\lambda^{-z}}, n{\lambda^{-\zeta}} \leq n\lambda^{-c} \leq
    \frac{1}{2}$. On the interval $[\frac{1}{4\lambda}, \frac{1}{2}]$,
    the function $g(x)$ assumes its maximum at
    $x=\frac{1}{4\lambda}$. This implies that $y\leq
    g(\frac{1}{4\lambda})$. This completes the present case since we
    get:
    \[
    y+ \frac{8}{\formula{\P}}n^2\lambda^{-2c}\leq
    g(\frac{1}{4\lambda}) +\frac{2}{\formula{\P}}
    \approx\formula{(1/(2*\l)-1)^2+2/\P} \leq \formula{\Q}.
    \]
    \assert{(1/(2*\l)-1)^2+2/\P <= \Q}%
  \item If $\floor{\frac{\lambda^{c}}{2}}< 1$, then $n=1$ and
    $\lambda^c< 2$. From $\formula{-\a}\leq c$, we have
    %% Check!!!
    \assert{(ln 0.5)/(ln ((sqrt 2)+1)) <= \a}%
    $\frac{8}{\formula{\P}}n^2\lambda^{-2c}\leq
    \frac{8}{\formula{\P}}\lambda^{\formula{2*\a}}$. Moreover, because
    $\bias\state\in[-1,1]$, we have ${\formula{-\a}}\leq c \leq
    z,\zeta\leq c+1$. With $\lambda^c\leq 2$, this implies that
    $\frac{1}{2\lambda}\leq \lambda^{-c-1}\leq
    \lambda^{-z},\lambda^{-\zeta}
    \leq\lambda^{\formula{\a}}$. Therefore both $\lambda^{-z}$ and
    $\lambda^{-\zeta}$ are in the interval $[\frac{1}{2\lambda},
    \lambda^{\formula{\a}}]$. On this interval, the function $g(x)$
    assumes its maximum at \assert{(1-2*(\l^(\a)))^2 <
      (1-2*(1/(2*\l)))^2}%
    $x=\frac{1}{2\lambda}$, and therefore $y\leq
    g(\frac{1}{2\lambda})$. This completes the proof since:
    \[
    y+\frac{8}{\formula{\P}}n^2\lambda^{-2c}\leq g(\frac{1}{2\lambda})
    + \frac{8}{\formula{\P}}\lambda^{\formula{2*\a}}
    \approx\formula{(max(((2*(pow(\l,\a))-1)^2),
      (((pow(\l,(-1)))-1)^2)) + ((8*(pow (\l, 2*\a))) / \P))} \leq
    \formula{\Q}.
    \]
    %% Check!!!
    \assert{(max(((2*(pow(\l,\a))-1)^2), (((pow(\l,(-1)))-1)^2)) +
      ((8*(pow (\l, 2*\a))) / \P)) <= \Q}%
  \end{itemize}
\end{proof}

% --------------------------------------------------------------------
\subsubsection[The B Lemma]{The $B$ Lemma}
\label{sssect-B}

\begin{definition}
  Let $h(x)=(1-\sqrt{2}x)^2$.
\end{definition}

\begin{lemma}
  \label{lemmaB}
  Recall the operator $B$ from Figure~\ref{list_operators}.  If
  $\state $ is a state such that $\bias\state\in [-1,1]$,
  $\sk\state\geq\formula{\P}$, and such that $b\leq 0\leq\beta$ and
  $\formula{-\b}\leq z,\zeta$, then there exists $n\in\Z$ such that:
  \[
  \sk(\state\cdot B^n) \leq \formula{\Q}~\sk\state.
  \]
\end{lemma}

\begin{proof}
  Let $c=\min\s{z, \zeta}$, $n=\max\s{1,
    \floor{\frac{\lambda^{c}}{\sqrt{2}}}}$ and compute the action of
  $B^n$ on $\state$:
  \[
  \begin{array}{lcl} {B^n}^\dagger DB^n & = & \left[ \begin{array}
        {cc}
        1 & 0  \\
        \sqrt{2}n & 1
      \end{array} \right]
    \left[ \begin{array} {cc} 
        e{\lambda^{-z}} & b \\
        b & e{\lambda^{z}}
      \end{array} \right]
    \left[ \begin{array} {cc} 
        1 & \sqrt{2}n  \\
        0 & 1
      \end{array} \right] \\[9pt]
    & = & \left[ \begin{array} {cc} 
        \ldots & b+\sqrt{2}\,ne{\lambda^{-z}}  \\
        b+\sqrt{2}\,ne{\lambda^{-z}} & \ldots
      \end{array} \right], \\[18pt]
    {B^n}^{\bullet\dagger} D{B^n}^\bullet & = & \left[ 
      \begin{array} {cc} 
        1 & 0  \\
        -\sqrt{2}n & 1
      \end{array} \right]
    \left[ \begin{array} {cc} 
        \varepsilon{\lambda^{-\zeta}} & \beta \\
        \beta & \varepsilon{\lambda^{\zeta}}
      \end{array} \right]
    \left[ \begin{array} {cc} 
        1 & -\sqrt{2}n  \\
        0 & 1
      \end{array} \right] \\[9pt]
    & = & \left[ \begin{array} {cc} 
        \ldots & \beta-\sqrt{2}\,n\varepsilon{\lambda^{-\zeta}}  \\
        \beta-\sqrt{2}\,n\varepsilon{\lambda^{-\zeta}} & \ldots
      \end{array} \right].
  \end{array}
  \]
  Therefore:
  \[
  \sk(\state\cdot B^n) = (b+\sqrt{2}\,ne{\lambda^{-z}})^2+
  (\beta-\sqrt{2}\,n\varepsilon{\lambda^{-\zeta}})^2.
  \]
  But recall that $e^2=b^2+1$, that $\varepsilon^2=\beta^2+1$, and
  from Remark~\ref{rem-be} that $b\leq 0\leq \beta$ implies $be\leq
  -b^2$ and $-\beta\varepsilon\leq -\beta^2$. Using these facts, we
  can expand the above formula as follows:
  \begin{eqnarray}
    \lefteqn{\sk(\state\cdot B^n)}\nonumber \\
    & = & (b+\sqrt{2}\,ne{\lambda^{-z}})^2+
    (\beta-\sqrt{2}\,n\varepsilon{\lambda^{-\zeta}})^2\nonumber\\
    & = & 
    b^2
    +2\sqrt{2}\,nbe{\lambda^{-z}}
    +2n^2e^2{\lambda^{-2z}}
    +\beta^2
    -2\sqrt{2}\,n\beta\varepsilon{\lambda^{-\zeta}}
    +2n^2\varepsilon^2{\lambda^{-2\zeta}}
    \nonumber\\
    & \leq & 
    b^2
    -2\sqrt{2}\,nb^2{\lambda^{-z}}
    +2n^2(b^2+1){\lambda^{-2z}}
    +\beta^2
    -2\sqrt{2}\,n\beta^2{\lambda^{-\zeta}}
    +2n^2(\beta^2+1){\lambda^{-2\zeta}}
    \nonumber\\
    & = & b^2(
    1
    -2\sqrt{2}\,n{\lambda^{-z}}
    +2n^2{\lambda^{-2z}}
    )+\beta^2(
    1
    -2\sqrt{2}\,n{\lambda^{-\zeta}}
    +2n^2{\lambda^{-2\zeta}}
    )+2n^2({\lambda^{-2z}}+{\lambda^{-2\zeta}})\nonumber \\
    & = & 
    b^2(1-\sqrt{2}\,n\lambda^{-z})^2
    +\beta^2(1-\sqrt{2}\,n{\lambda^{-\zeta}})^2
    + 2n^2({\lambda^{-2z}}+{\lambda^{-2\zeta}}).\nonumber \\
    & = & b^2h(n{\lambda^{-z}})+\beta^2h(n{\lambda^{-\zeta}})+ 
    2n^2({\lambda^{-2z}}+{\lambda^{-2\zeta}}).\nonumber
  \end{eqnarray}

  \noindent
  Writing $y=\max\s{h(n{\lambda^{-z}}), h(n{\lambda^{-\zeta}})}$ for
  brevity, and using the assumption that $\sk\state\geq\formula{\P}$,
  together with the fact that $c\leq z,\zeta$, we get:
  \begin{eqnarray}
    \sk(\state \cdot B^n) & \leq & 
    b^2y+\beta^2y+ 4n^2\lambda^{-2c} \nonumber\\
    & = & \sk\state y +4n^2\lambda^{-2c}\nonumber\\ 
    &\leq & \sk\state (y +\frac{4}{\formula{\P}}n^2\lambda^{-2c}) 
    \nonumber.
  \end{eqnarray}

  \noindent
  To finish the proof, it remains to show that
  $y+\frac{4}{\formula{\P}}n^2\lambda^{-2c}\leq \formula{\Q}$.  There
  are two cases:
  \begin{itemize}
  \item If $\floor{\frac{\lambda^{c}}{\sqrt{2}}}\geq 1$, then
    $\frac{\lambda^{c}}{2\sqrt{2}}\leq n\leq
    \frac{\lambda^{c}}{\sqrt{2}}$. From $n\leq
    \frac{\lambda^{c}}{\sqrt{2}}$, we have $2n^2\lambda^{-2c}\leq 1$,
    and so $\frac{4n^2\lambda^{-2c}}{\formula{\P}}\leq \frac{2}
    {\formula{\P}}$. Moreover, because $\bias\state\in[-1,1]$, we have
    $c\leq z,\zeta\leq c+1$. Hence $\frac{1}{2\sqrt{2}\,\lambda} =
    \frac{\lambda^{c}}{2\sqrt{2}}\lambda^{-c-1}\leq
    n\lambda^{-c-1}\leq n{\lambda^{-z}}, n{\lambda^{-\zeta}} \leq
    n\lambda^{-c}\leq \frac{1}{\sqrt{2}}$. On the interval
    $[\frac{1}{2\sqrt{2}\,\lambda}, \frac{1}{\sqrt{2}}]$, the function
    $h(x)$ assumes its maximum at
    $x=\frac{1}{2\sqrt{2}\,\lambda}$. This implies that $y\leq
    h(\frac{1}{2\sqrt{2}\,\lambda})$. This completes the present case
    since we get:
    \[
    y+ \frac{4}{\formula{\P}}n^2\lambda^{-2c}\leq
    h(\frac{1}{2\sqrt{2}\,\lambda}) +\frac{2}{\P} \approx \formula{(1
      - sqrt(2)*(1/(2*sqrt(2)*\l)))^2 + (2/\P)} \leq \formula{\Q}.
    \]
    % !!! Check
    \assert{(1 - sqrt(2)*(1/(2*sqrt(2)*\l)))^2 + (2/\P) <= (\Q)}%
  \item If $\floor{\frac{\lambda^{c}}{\sqrt{2}}}< 1$, then $n=1$ and
    $\lambda^c<\sqrt{2}$. From $\formula{-\b}\leq c$, we have
    $\frac{4}{\formula{\P}}n^2\lambda^{-2c} \leq
    \frac{4}{\formula{\P}}\lambda^{\formula{2*\b}}$. Moreover, because
    %% Check!!!
    % \assert{((ln 1/(sqrt 2)) / (ln \l)) <= \b}%
    %
    $\bias\state\in[-1,1]$, we have $\formula{-\b}\leq c\leq
    z,\zeta\leq c+1$. With $\lambda^c\leq\sqrt{2}$, this implies that
    $\frac{1}{\sqrt{2}\,\lambda} \leq \lambda^{-c-1}\leq
    {\lambda^{-z}}, {\lambda^{-\zeta}} \leq
    \lambda^{\formula{\b}}$. Therefore both $\lambda^{-z}$ and
    $\lambda^{-\zeta}$ are in the interval
    $[\frac{1}{\sqrt{2}\,\lambda},\lambda^{\formula{\b}}]$.  On this
    interval, the function $h(x)$ assumes its maximum at \assert{(1-
      sqrt(2)*(1/(sqrt(2)*\l)))^2 <= (1-sqrt(2)*((\l)^(\b)))^2}%
    $x=\lambda^{\formula{\b}}$, and therefore $y\leq
    h(\lambda^{\formula{\b}})$. This completes the proof since:
    \[
    y+\frac{4}{\formula{\P}}n^2\lambda^{-2c}\leq
    h(\lambda^{\formula{\b}}) +
    \frac{4}{\formula{\P}}\lambda^{\formula{2*\b}}
    \approx\formula{((max((1 - sqrt(2)*(1/(sqrt(2)*\l)))^2,
      (1-sqrt(2)*((\l)^(\b)))^2) + (4/\P * (\l^(2*\b))))} \leq
    \formula{\Q}.
    \]\qedhere
    % Check !!!
    \assert{((max((1 - sqrt(2)*(1/(sqrt(2)*\l)))^2,
      (1-sqrt(2)*((\l)^(\b)))^2) + (4/\P * (\l^(2*\b)))) <= (\Q)}%
  \end{itemize}
\end{proof}

% --------------------------------------------------------------------
\subsubsection{Proof of the Step Lemma}
\label{sssect-step}

The proof of the Step Lemma is now basically a case distinction, using
the cases enumerated in lemmas~\ref{shift_lemma}--\ref{lemmaB}, as
well as some additional symmetric cases. In particular, the following
remark will allow us to use the grid operators $X$ and $Z$ to reduce
the number of cases to consider.

\begin{remark}
  \label{rem-XZ}
  The grid operator $Z$ negates the anti-diagonal entries of a state
  while the operator $X$ swaps the diagonal entries of a state. This
  follows by simple computation since
  \[
  \state\cdot Z = \left( \left[
      \begin{array}{cc}
        e{\lambda^{-z}} & -b \\
        -b & e{\lambda^{z}}
      \end{array} 
    \right] , \left[
      \begin{array}{cc}
        \varepsilon{\lambda^{-\zeta}} & -\beta \\
        -\beta & \varepsilon{\lambda^{\zeta}}
      \end{array} 
    \right] \right)
  \]
  and
  \[
  \state\cdot X = \left( \left[
      \begin{array}{cc}
        e{\lambda^{z}} & b \\
        b & e{\lambda^{-z}}
      \end{array} 
    \right] , \left[
      \begin{array}{cc}
        \varepsilon{\lambda^{\zeta}} & \beta \\
        \beta & \varepsilon{\lambda^{-\zeta}}
      \end{array} 
    \right] \right).
  \]
  Moreover, $\bias (\state\cdot Z) = \bias \state$ and $\bias
  (\state\cdot X) = -\bias \state$.
\end{remark}

\begin{un-lemma}[Step Lemma]
  For any state $\state$, if $\sk\state \geq \formula{\P}$, then there
  exists a special grid operator $\Gop$ such that $\sk (\state\cdot
  \Gop)\leq \formula{\Q} ~\sk\state$.  Moreover, $\Gop$ can be
  computed using a constant number of arithmetic operations.
\end{un-lemma}

\begin{proof}
  Let $\state$ be a state such that $\sk\state \geq \formula{\P}$. By
  Lemma~\ref{shift_lemma} we can assume w.l.o.g. that
  $\bias\state\in[-1,1]$. Moreover, by Remark~\ref{rem-XZ}, we can
  also assume that $\beta \geq 0$ and $z+\zeta\geq 0$. Note that the
  application of the grid operators $X$ and/or $Z$ in
  Remark~\ref{rem-XZ} preserves the fact that
  $\bias\state\in[-1,1]$. We now treat in turn the cases $b\geq 0$ and
  $b\leq 0$.
  \begin{description}
  \item[Case 1] \label{b_beta_plus} $b \geq 0$. A covering of the
    strip defined by $z-\zeta\in[-1,1]$ and $z+\zeta\geq 0$ is
    depicted in Figure~\ref{fig-covers}(a).  The $R$ region (in green)
    and the $A$ region (in red) are defined as the intersection of
    this space with $\s{(z,\zeta)~|~\formula{-\r}\leq z,\zeta
      \leq\formula{\r}}$ and $\s{(z,\zeta)~|~z\leq\formula{-\a}\mbox{
        and }\formula{\r}\leq \zeta}$ respectively. The $K$ and
    $K^\bullet$ regions (both in blue) fill the remaining space.
    % This covering requires -r <= a.
    %% Check!!!
    \assert{-\r <= \a}%
    %      

    % ................................................................
    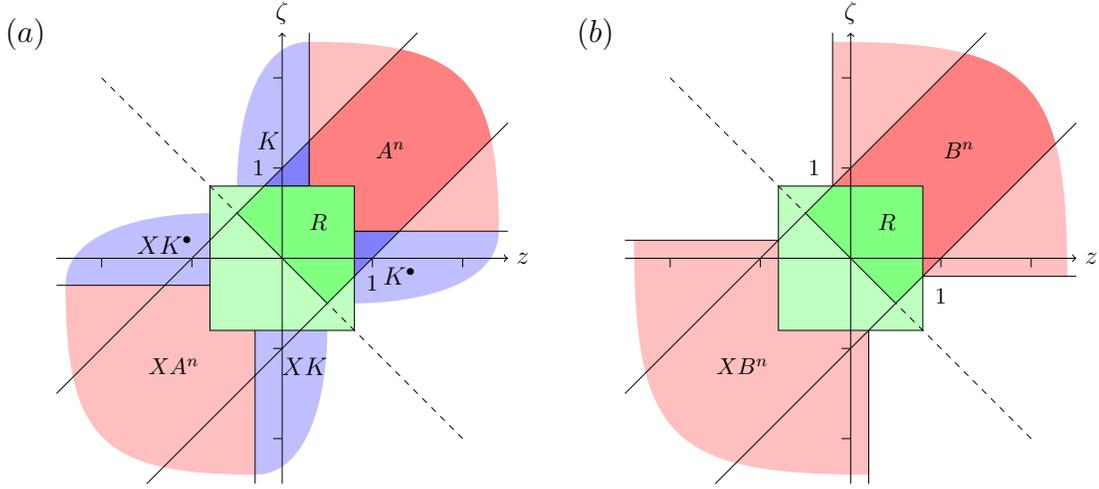
\begin{figure}
      \[ (a)~ \mp{0.9}{\scalebox{0.8}{\begin{tikzpicture}[scale=1.5]
          % Regions K: [-0.3, \infty] x [-\infty, -0.8]
          \fill[fill=blue!25] (-0.3,-0.8) -- (-0.3,-2.4) .. controls
          (0.3,-2.4) and (0.5,-1.5) .. (0.5,-0.8) -- cycle; \draw
          (0.25,-1.2) node {$XK$};
          \begin{scope}[scale=-1]
            \fill[fill=blue!25] (-0.3,-0.8) -- (-0.3,-2.4) .. controls
            (0.3,-2.4) and (0.5,-1.5) .. (0.5,-0.8) -- cycle; \draw
            (0.15,-1.3) node {$K$};
          \end{scope}
          \begin{scope}[cm={0,1,1,0,(0,0)}]
            \fill[fill=blue!25] (-0.3,-0.8) -- (-0.3,-2.4) .. controls
            (0.3,-2.4) and (0.5,-1.5) .. (0.5,-0.8) -- cycle; \draw
            (0.15,-1.3) node {$XK\bul$};
          \end{scope}
          \begin{scope}[cm={0,-1,-1,0,(0,0)}]
            \fill[fill=blue!25] (-0.3,-0.8) -- (-0.3,-2.4) .. controls
            (0.3,-2.4) and (0.5,-1.5) .. (0.5,-0.8) -- cycle; \draw
            (0.2,-1.3) node {$K\bul$};
          \end{scope}
          % 
          % Region A: [-\infty, -0.3] x [-\infty, -0.3]
          \fill[fill=red!25] (-2.4,-0.3) -- (-0.3,-0.3) -- (-0.3,-2.4)
          .. controls (-2,-2.4) and (-2.4,-2) ..  (-2.4,-0.3) --
          cycle; \draw (-2.5,-0.3) -- (-0.3,-0.3) -- (-0.3,-2.5);
          \draw (-1.2, -1.2) node {\small $XA^n$};
          % 
          % Region XA: [-\infty, -0.3] x [-\infty, -0.3]
          \begin{scope}[scale=-1]
            \fill[fill=red!25] (-2.4,-0.3) -- (-0.3,-0.3) --
            (-0.3,-2.4) .. controls (-2,-2.4) and (-2.4,-2) ..
            (-2.4,-0.3) -- cycle; \draw (-2.5,-0.3) -- (-0.3,-0.3) --
            (-0.3,-2.5); \draw (-1.2, -1.2) node {\small $A^n$};
          \end{scope}
          % Region R: [-0.8, 0.8] x [-0.8, 0.8]
          \filldraw[fill=green!25] (0.8, 0.8) -- (-0.8, 0.8) -- (-0.8,
          -0.8) -- (0.8, -0.8) -- cycle; \draw (0.4, 0.4) node {\small
            $R$};
          \begin{scope}
            \path[clip] (-1.5,-2.5) -- (2.5,1.5) -- (1.5,2.5) --
            (-2.5,-1.5) -- cycle;
            \begin{scope}[yscale=-1]
              \fill[fill=blue!50] (-0.3,-2.4) -- (-0.3,-0.8) --
              (0.3,-0.8) -- (0.3,-2.4) -- cycle; \draw ((-0.3,-2.5) --
              (-0.3,-0.8) -- (0.3,-0.8) -- (0.3,-2.5);
            \end{scope}
            \begin{scope}[cm={0,-1,-1,0,(0,0)}]
              \fill[fill=blue!50] (-0.3,-2.4) -- (-0.3,-0.8) --
              (0.3,-0.8) -- (0.3,-2.4) -- cycle; \draw ((-0.3,-2.5) --
              (-0.3,-0.8) -- (0.3,-0.8) -- (0.3,-2.5);
            \end{scope}
            % 
            % Region XA: [-\infty, -0.3] x [-\infty, -0.3]
            \begin{scope}[scale=-1]
              \fill[fill=red!50] (-2.4,-0.3) -- (-0.3,-0.3) --
              (-0.3,-2.4) .. controls (-2,-2.4) and (-2.4,-2) ..
              (-2.4,-0.3) -- cycle; \draw (-2.5,-0.3) -- (-0.3,-0.3)
              -- (-0.3,-2.5); \draw (-1.2, -1.2) node {\small $A^n$};
            \end{scope}
            % Region R: [-0.8, 0.8] x [-0.8, 0.8]
            % \filldraw[fill=green!50] (0.8, 0.8) -- (-0.8, 0.8) --
            % (-0.8, -0.8) -- (0.8, -0.8) -- cycle;
            \filldraw[fill=green!50] (0.8, 0.8) -- (-0.8, 0.8) --
            (0.8, -0.8) -- cycle; \draw (0.4, 0.4) node {\small $R$};
          \end{scope}
          % 
          % Coordinate system
          \draw[->] (-2.5,0) -- (2.5, 0) node[right] {$z$}; \draw[->]
          (0,-2.5) -- (0, 2.5) node[above] {$\zeta$}; \foreach \x in
          {-2, -1, 1, 2} { \draw (\x,0) -- (\x,-0.1); } \foreach \y in
          {-2, -1, 1, 2} { \draw (0,\y) -- (-0.1,\y); } \draw (1,-0.1)
          node[below] {\small $1$}; \draw (-0.1,1) node[left] {\small
            $1$};
          % 
          % Lines for shift lemma
          \draw (-1.5,-2.5) -- (2.5,1.5); \draw (-2.5,-1.5) --
          (1.5,2.5);
          % Line for z+zeta >= 0
          \draw [dashed] (-2, 2) -- (-0.5,0.5);
          \begin{scope}[scale=-1]
            \draw [dashed] (-2, 2) -- (-0.5,0.5);
          \end{scope}
        \end{tikzpicture}}} \quad
      (b)~ \mp{0.9}{\scalebox{0.8}{\begin{tikzpicture}[scale=1.5]
          % 
          % Region XB: [-\infty, 0.2] x [-\infty, 0.2]
          \fill[fill=red!25] (-2.4,0.2) -- (0.2,0.2) -- (0.2,-2.4)
          .. controls (-2,-2.4) and (-2.4,-2) ..  (-2.4,0.2) -- cycle;
          \draw (-2.5,0.2) -- (0.2,0.2) -- (0.2,-2.5); \draw (-1.2,
          -1.2) node {\small $XB^n$};
          % 
          % Region B: [-\infty, 0.2] x [-\infty, 0.2]
          \begin{scope}[scale=-1]
            \fill[fill=red!25] (-2.4,0.2) -- (0.2,0.2) -- (0.2,-2.4)
            .. controls (-2,-2.4) and (-2.4,-2) ..  (-2.4,0.2) --
            cycle; \draw (-2.5,0.2) -- (0.2,0.2) -- (0.2,-2.5); \draw
            (-1.2, -1.2) node {\small $B^n$};
          \end{scope}
          % Region R: [-0.8, 0.8] x [-0.8, 0.8]
          \filldraw[fill=green!25] (0.8, 0.8) -- (-0.8, 0.8) -- (-0.8,
          -0.8) -- (0.8, -0.8) -- cycle; \draw (0.4, 0.4) node {\small
            $R$};
          \begin{scope}
            \path[clip] (-1.5,-2.5) -- (2.5,1.5) -- (1.5,2.5) --
            (-2.5,-1.5) -- cycle;
            % 
            % Region XB: [-\infty, 0.2] x [-\infty, 0.2]
            % 
            % Region B: [-\infty, 0.2] x [-\infty, 0.2]
            \begin{scope}[scale=-1]
              \fill[fill=red!50] (-2.4,0.2) -- (-0.8,0.2) -- (0.2,
              -0.8) -- (0.2,-2.4) .. controls (-2,-2.4) and (-2.4,-2)
              ..  (-2.4,0.2) -- cycle; \draw (-2.5,0.2) -- (-0.8,0.2)
              -- (0.2, -0.8) -- (0.2,-2.5); \draw (-1.2, -1.2) node
              {\small $B^n$};
            \end{scope}
            % Region R: [-0.8, 0.8] x [-0.8, 0.8]
            \filldraw[fill=green!50] (0.8, 0.8) -- (-0.8, 0.8) --
            (0.8, -0.8) -- cycle; \draw (0.4, 0.4) node {\small $R$};
          \end{scope}
          % 
          % Coordinate system
          \draw[->] (-2.5,0) -- (2.5, 0) node[right] {$z$}; \draw[->]
          (0,-2.5) -- (0, 2.5) node[above] {$\zeta$}; \foreach \x in
          {-2, -1, 1, 2} { \draw (\x,0) -- (\x,-0.1); } \foreach \y in
          {-2, -1, 1, 2} { \draw (0,\y) -- (-0.1,\y); } \draw (1,-0.1)
          node[below=1.2ex] {\small $1$}; \draw (-0.1,1)
          node[left=1.2ex] {\small $1$};
          % 
          % Lines for shift lemma
          \draw (-1.5,-2.5) -- (2.5,1.5); \draw (-2.5,-1.5) --
          (1.5,2.5);
          % Line for z+zeta >= 0
          \draw [dashed] (-2, 2) -- (-0.5,0.5);
          \begin{scope}[scale=-1]
            \draw [dashed] (-2, 2) -- (-0.5,0.5);
          \end{scope}
        \end{tikzpicture}}}
      \]
      \caption[Two coverings of the plane]{(a) A covering of the
        region $z-\zeta\in [-1,1]$ and $z+\zeta\geq 0$ for the case
        $b\geq 0$. (b) A covering of the region $z-\zeta\in [-1,1]$
        and $z+\zeta\geq 0$ for the case $b\leq 0$.}
      \label{fig-covers}
      \rule{\textwidth}{0.1mm}
    \end{figure}
    % ................................................................

    We now consider in turn the possible locations of the pair
    $(z,\zeta)$ in this covering.
    \begin{enumerate}
    \item If $\formula{-\r}\leq z,\zeta \leq\formula{\r}$, then by
      Lemma~\ref{lemmaR} we have $\sk(\state\cdot R) \leq
      \formula{\Q}~\sk\state$.
    \item \label{partk} If $z\leq\formula{-\a}$ and $\formula{\r}\leq
      \zeta$, then by Lemma~\ref{lemmaK} we have $\sk(\state\cdot
      K)\leq \formula{\Q}~\sk\state$.
    \item If $\formula{-\a}\leq z,\zeta$, then there exists $n\in\Z$
      such that $\sk(\state\cdot A^n)\leq \formula{\Q}~\sk\state$ by
      Lemma~\ref{lemmaA}.
    \item \label{partsk} If $\formula{\r}\leq z$ and
      $\zeta\leq\formula{-\a}$, then note that $\state\cdot K^\bullet
      = (\Delta, D)\cdot K$, and therefore $\sk(\state\cdot K)\leq
      \formula{\Q}~\sk\state$ by Lemma~\ref{lemmaK}:
      \[
      \sk(\state\cdot K^\bullet) = \sk ((\Delta, D)\cdot K) \leq
      \formula{\Q}~ \sk (\Delta, D)= \formula{\Q}~ \sk\state.
      \]
    \end{enumerate}
  \item[Case 2] $b\leq 0$. As above, we use a covering of the strip
    defined by $z-\zeta\in[-1,1]$ and $z+\zeta\geq 0$ and consider the
    possible locations of $(z, \zeta)$ in this space. The relevant
    covering is depicted in Figure~\ref{fig-covers}(b), where the $R$
    region (in green) is defined as above and the $B$ region (in red)
    is defined as the intersection of the strip with
    $\s{(z,\zeta)~|~z,\zeta\geq \formula{-\b}}$.
    % This covering requires -b <= r-1.
    %% Check!!!
    \assert{-\b <= \r -1}%
    \begin{enumerate}
    \item If $\formula{-\r}\leq z,\zeta\leq\formula{\r}$, then by
      Lemma~\ref{lemmaR} we have $\sk(\state\cdot R)\leq
      \formula{\Q}~\sk\state$.
    \item If $z,\zeta\geq \formula{-\b}$ then there exists $n\in\Z$
      such that $\sk(\state\cdot B^n)\leq \formula{\Q}~\sk\state$ by
      Lemma~\ref{lemmaB}.
    \end{enumerate}
  \end{description}
  Finally, note that only a constant number of calculations are
  required to decide which of the above cases applies. Moreover, each
  case only requires a constant number of operations. Specifically,
  the computation of $k$ and $\sigma^k$ in Lemma~\ref{shift_lemma}, of
  $n$ and $A^n$ in Lemma~\ref{lemmaA}, and of $n$ and $B^n$ in
  Lemma~\ref{lemmaB} each require just a fixed number of operations,
  and each of the remaining cases produces a fixed grid operator.
\end{proof}

% ----------------------------------------------------------------------
\subsection{General solution to grid problems over
  \texorpdfstring{$\Zomega$}{Z[omega]}}
\label{ssec-proof-prop-algorithm-2d}

We are finally in a position to solve Problem~\ref{pb:grid-Zomega}.

\begin{proposition}
  \label{prop:algo-grid-pb-Zomega}
  There is an algorithm which, given two bounded convex subset $A$ and
  $B$ of $\R^2$ with non-empty interior, enumerates all solutions of
  the grid problem over $\Zomega$ for $A$ and $B$. Moreover, if $A$
  and $B$ are $M$-upright, then the algorithm requires $O(\log(1/M))$
  arithmetic operations overall, plus a constant number of arithmetic
  operations per solution produced.
\end{proposition}

\begin{proof}
  Analogous to the proof of Proposition~\ref{prop:algo-grid-pb-Zi},
  using Proposition~\ref{prop:ellipse-Zomega} to find the appropriate
  grid operator.
\end{proof}

% ----------------------------------------------------------------------
\subsection{Scaled grid problems over
  \texorpdfstring{$\Zomega$}{Z[omega]}}
\label{ssect:scaled-grid-pb-Zomega}

We close this chapter by considering scaled versions of
Problem~\ref{pb:grid-Zomega}, as in
Subsection~\ref{ssect:scaled-grid-pb-Zi}. More specifically, given
$k\in\N$ and two bounded convex subsets $A$ and $B$ or $\R^2$ with
non-empty interior, we are interested in solving grid problems over
$\Zomega$ for $\rt{k}A$ and $(-\sqrt{2})^kB$. We call such problems
\emph{scaled grid problems over $\Zomega$ for $A$, $B$, and
  $k$}. These scaled grid problems will also be useful in
Chapter~\ref{chap:nb-th}. Using
Proposition~\ref{prop:algo-grid-pb-Zomega}, we can establish the
following proposition.

\begin{proposition}
  \label{prop:algo-grid-pb-scaled-Zomega}
  There is an algorithm which, given two bounded convex subset $A$ and
  $B$ of $\R^2$ with non-empty interior, enumerates (the
  infinite sequence of) all solutions of the scaled grid problem over
  $\Zomega$ for $A$, $B$, and $k$ in order of increasing
  $k$. Moreover, if $A$ and $B$ are $M$-upright, then the algorithm
  requires $O(\log(1/M))$ arithmetic operations overall, plus a
  constant number of arithmetic operations per solution produced.
\end{proposition}

Finally, we give some lower bounds on the number of solutions to
scaled grid problems over $\Zomega$.

\begin{lemma}
  \label{lem-2-solutions}
  Let $A$ and $B$ be convex subsets of $\R^2$, and let $k\geq 0$.
  Assume $A$ contains a circle of radius $r$ and $B$ contains a circle
  of radius $R$, such that $rR\geq \frac{1}{2^k}(1+\sqrt2)^2$. Then
  the scaled grid problem over $\Zomega$ for $A$, $B$, and $k$ has at
  least $2$ solutions.
\end{lemma}

\begin{proof}
  By assumption, $(\rt{k})A$ contains a circle of radius $r'=\rt{k}r$
  and $(-\sqrt{2})^kB$ contains a circle of radius $R'=\rt{k}R$, with
  $rR\geq (1+\sqrt2)^2$. Let $\delta={r'}/{\sqrt2}$ and
  $\Delta=R'\sqrt2$, and inscribe two squares of size
  $\delta\times\delta$ in the first circle, and one square of size
  $\Delta\times\Delta$ in the second circle, as shown here:
  
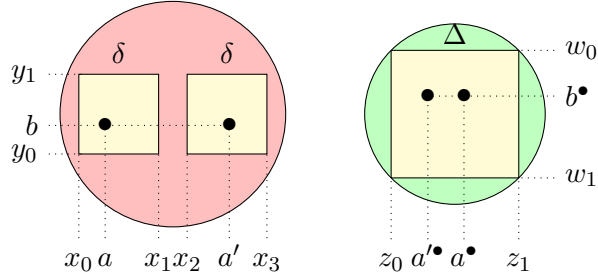
\[
  \begin{tikzpicture}[scale=1.5] \draw[fill=red!25] (0,0) circle (1);
    \draw[shift={(0.48,0)}, scale=0.5] (0,.707) node[above] {\small
      $\delta$}; \draw[shift={(0.48,0)}, scale=0.5, fill=yellow!20]
    (-.707,-.707) -- (.707,-.707) -- (.707,.707) -- (-.707,.707) --
    cycle; \draw[shift={(-0.48,0)}, scale=0.5] (0,.707) node[above]
    {\small $\delta$}; \draw[shift={(-0.48,0)}, scale=0.5,
    fill=yellow!20] (-.707,-.707) -- (.707,-.707) -- (.707,.707) --
    (-.707,.707) -- cycle; \draw[shift={(-0.48,0)}, scale=0.5, dotted]
    (-.707,-2.7) node[above, anchor=base] {\small $x_0$} +(0, 0.5) --
    (-.707,-.707); \draw[shift={(-0.48,0)}, scale=0.5, dotted]
    (.707,-2.7) node[above, anchor=base] {\small $x_1$} +(0, 0.5) --
    (.707,-.707); \draw[shift={(0.48,0)}, scale=0.5, dotted]
    (-.707,-2.7) node[above, anchor=base] {\small $x_2$} +(0, 0.5) --
    (-.707,-.707); \draw[shift={(0.48,0)}, scale=0.5, dotted]
    (.707,-2.7) node[above, anchor=base] {\small $x_3$} +(0, 0.5) --
    (.707,-.707); \draw[dotted] (-0.6, -1.35) node[above, anchor=base]
    {\small $a$} +(0, 0.25) -- (-0.6, -.1); \draw[dotted] (0.5, -1.35)
    node[above, anchor=base] {\small $a'$} +(0, 0.25) -- (0.5, -.1);
    \draw[shift={(-0.45,0)}, scale=0.5, dotted] (-1.3,-.707)
    node[left] {\small $y_0$} -- (-.707,-.707);
    \draw[shift={(-0.45,0)}, scale=0.5, dotted] (-1.3,.707) node[left]
    {\small $y_1$} -- (-.707,.707); \draw[dotted] (-1.1, -.1)
    node[left] {\small $b$} -- (0.5, -.1); \def\dot#1{\draw (#1) node
      {$\bullet$};} \dot{-0.6, -.1} \dot{0.5, -.1}

    \begin{scope}[shift={(2.5,0)}, scale=0.8]
      \draw[fill=green!25] (0,0) circle (1); \draw[fill=yellow!20,
      line join=round] (-.707,-.707) -- (.707,-.707) -- (.707,.707) --
      (-.707,.707) -- cycle; \draw (0,.707) +(0,-0.05) node[above]
      {\small $\Delta$}; \draw[dotted] (-.707,-1.6875) node[above,
      anchor=base] {\small $z_0$} +(0, 0.3) -- (-.707,-.707);
      \draw[dotted] (.707,-1.6875) node[above, anchor=base] {\small
        $z_1$} +(0, 0.3) -- (.707,-.707); \draw[dotted] (1.1,.707)
      node[right] {\small $w_0$} -- (.707,.707); \draw[dotted]
      (1.1,-.707) node[right] {\small $w_1$} -- (.707,-.707);
      \def\dot#1{\draw (#1) node {$\bullet$};} \dot{0.1, 0.2}
      \draw[dotted] (0.1, -1.6875) node[above, anchor=base] {\small
        $a\bul$} +(0, 0.3) -- (0.1, 0.2); \dot{-0.3, 0.2}
      \draw[dotted] (-0.3, -1.6875) node[above, anchor=base] {\small
        $a'^{\bullet}$} +(0, 0.3) -- (-0.3, 0.2); \draw[dotted] (1.1,
      0.2) node[right] {\small $b\bul$} -- (-0.3, 0.2);
    \end{scope}
  \end{tikzpicture}
  \]
  Since $\delta\Delta=r'R'\geq (1+\sqrt2)^2$, by
  Lemma~\ref{lem:grid-bounds}, we can find $a,a',b\in\Z[\sqrt2]$ such
  that $a\in[x_0,x_1]$, $a\bul\in[z_0,z_1]$, $a'\in[x_2,x_3]$,
  $a'^{\bullet}\in[z_0,z_1]$, $b\in[y_0,y_1]$, and
  $b\bul\in[w_0,w_1]$.  Then $u=a+ib$ and $v=a'+ib$ are two different
  solutions to the scaled grid problem over $\Zomega$ for $A$, $B$,
  and $k$ as claimed.
\end{proof}

\begin{lemma}
  \label{lem-exponential-grid}
  Let $A$ and $B$ be convex subsets of $\R^2$, and assume that the
  two-dimensional scaled grid problem for $k$ has at least two
  distinct solutions. Then for all $\ell\geq 0$, the scaled grid
  problem for $k+2\ell$ has at least $2^\ell+1$ solutions.
\end{lemma}

\begin{proof}
  Analogous to the proof of Proposition~\ref{prop:evolution-grid-Zi}.
\end{proof}

% ---------------------------------------------------------------------
\chapter{Clifford+\texorpdfstring{$V$}{V} approximate synthesis}
\label{chap:synth-V}

In this chapter, we introduce an efficient algorithm to solve the
problem of approximate synthesis of special unitaries over the
Clifford+$V$ gate set. Recall from Chapter~\ref{chap:intro} that the
Clifford group is generated by
\[
\omega = e^{i\pi/4}, \quad H= \frac{1}{\sqrt 2}
\begin{bmatrix}
  1 & 1 \\
  1 & -1
\end{bmatrix}, \quad \mbox{and} \quad S=
\begin{bmatrix}
  1 & 0 \\
  0 & i
\end{bmatrix}.
\]
Recall moreover that the Clifford+$V$ group is obtained by adding the
following $V$-gates\label{def:V-gates} to the generators of the
Clifford group
\[
V_X = \frac{1}{\sqrt 5}
\begin{bmatrix}
  1 & 2i \\
  2i & 1
\end{bmatrix}, \quad V_Y = \frac{1}{\sqrt 5}
\begin{bmatrix}
  1 & 2 \\
  -2 & 1
\end{bmatrix}, \quad \mbox{and} \quad V_Z = \frac{1}{\sqrt 5}
\begin{bmatrix}
  1 + 2i & 0 \\
  0 & 1-2i
\end{bmatrix}.
\]

The problem of approximate synthesis of special unitaries over the
Clifford+$V$ gate set is the following. Given a special unitary $U\in
\suset(2)$ and a precision $\epsilon >0$, construct a Clifford+$V$
circuit $W$ whose $V$\!-count is as small as possible and such that
$\norm{W-U}\leq \epsilon$.

We solve the problem in three steps. We first characterize the
unitaries which can be expressed \emph{exactly} as Clifford+$V$
circuits. We then use this characterization to define an algorithm
solving the problem of approximate synthesis of \emph{$z$-rotations}
over the Clifford+$V$ gate set. Finally, we show how to this method
can be used to approximately synthesize arbitrary special unitaries.

% =====================================================================
\section{Exact synthesis of Clifford+\texorpdfstring{$V$}{V}
  operators}
\label{sect:exact-synth-CliffordV}

We start by solving the problem of exact synthesis of Clifford+$V$
operators.

\begin{problem}[Exact synthesis of Clifford+$V$ operators]
  \label{pb:exact-synth-CliffordV}
  Given a unitary $U\in\uset(2)$, determine whether there exists a
  Clifford+$V$ circuit $W$ such that $U=W$ and, in case such a circuit
  exists, construct one whose $V$\!-count is minimal.
\end{problem}

The problem of exact synthesis of \emph{Pauli+$V$} operators was first
solved in \cite{BGS2013} using the arithmetic of quaternions. Here, we
extend this result to the Clifford+$V$ group.

To characterize Clifford+$V$ operators, we consider the following set
of unitaries.

\begin{definition}
  \label{def:subgroup-RV}
  The set $\vset$ consists of unitary matrices of the form
  \begin{equation}
    \label{eq:matrix}
    U =  \frac{1}{\rt{k}\rf{\ell}} \begin{bmatrix} 
      \alpha & \gamma \\
      \beta & \delta
    \end{bmatrix}
  \end{equation}
  where $k,\ell\in\N$ with $0\leq k\leq 2$,
  $\alpha,\beta,\gamma,\delta\in\Z[i]$, and such that $\det(U)$ is a
  power of $i$.
\end{definition}

We will show that a unitary $U$ is a Clifford+$V$ operator if and only
if $U\in\vset$. As a corollary, this will establish that $\vset$ is a
group, which might not be obvious, due to the seemingly arbitrary
restriction on the exponent $k$. We note that $\vset$ does not
coincide with the subgroup of $\uset(2)$ whose entries are in the ring
$\Z[1/\sqrt{2}, 1/\sqrt{5},i]$\label{def:RV} and whose determinant is
a power of $i$. Indeed, the matrix
\[
\frac{1}{5^3} \begin{bmatrix}
  2i + \sqrt{5} & - 80 + 96i \\
  80 + 96i & -2i + \sqrt{5}
\end{bmatrix}
\]
has entries in $\Z[1/\sqrt{2}, 1/\sqrt{5},i]$ and has determinant 1
but is not an element of $\vset$.

\begin{definition}
  \label{def:5-denomexp}
  Let $U\in\vset$ be as in (\ref{eq:matrix}). The integers $k$ and
  $\ell$ are called the \emph{$\sqrt{2}$-denominator exponent} and
  \emph{$\sqrt{5}$-denominator exponent} of $U$ respectively. The
  least $k$ (resp. $\ell$) such that $U$ can be written as in
  (\ref{eq:matrix}) is the \emph{least $\sqrt{2}$-denominator
    exponent} (resp. \emph{least $\sqrt{5}$-denominator exponent}) of
  $U$. These notions extend naturally to vectors and scalars of the
  form
  \begin{equation}
    \label{eq:vector}
    \frac{1}{\rt{k}\rf{\ell}} \begin{bmatrix} 
      \alpha \\
      \beta 
    \end{bmatrix}
    \quad                           
    \mbox{ and }
    \quad
    \frac{1}{\rt{k}\rf{\ell}}~ \alpha,                            
  \end{equation}
  where $k,\ell\in\N$, with $0\leq k \leq 2$, and $\alpha,
  \beta\in\Z[i]$.
\end{definition}

In what follows, we refer to the pair $(k,\ell)$ as the
\emph{denominator exponent} of a matrix, vector, or scalar. It is then
understood that the first component of the pair $(k,\ell)$ is the
$\sqrt 2$-exponent, while the second is the $\sqrt 5$-exponent. Note
that the least denominator exponent of a matrix, vector, or scalar is
the pair $(k,\ell)$, where $k$ and $\ell$ are the least $\sqrt 2$- and
$\sqrt 5$-exponents respectively.

\begin{remark}
  \label{rem:denomexp-parity}
  Since $\sqrt{5}\notin\Z[i]$, if $\ell$ and $\ell'$ are two
  $\sqrt{5}$-denominator exponents of a matrix $U\in\V$, then $\ell
  \equiv \ell' ~(\mymod 2)$. A similar property holds for
  $\sqrt{2}$-denominator exponents.
\end{remark}

We first show that if $U$ is a Clifford+$V$ operator, then $U\in \V$.

\begin{lemma}
  \label{lem:standard-form}
  If $U$ is a Clifford+$V$ operator, then $U=ABC$ where $A$ is a
  product of $V$\!-gates, $B$ is a Pauli+$S$ operator, and $C$ is one
  of $I$, $H$, $HS$, $\omega$, $H\omega$, and $HS\omega$.
\end{lemma}

\begin{proof}
  Clifford gates and $V$\!-gates can be commuted in the sense that for
  every pair $C,V$ of a Clifford gate and a $V$\!-gate, there exists a
  pair $C',V'$ such that $CV=V'C'$. This implies that a Clifford+$V$
  operator $U$ can always be written as $U=AA'$, where $A$ is a
  product of $V$\!-gates and $A'$ is a Clifford operator. Furthermore,
  the Pauli+$S$ group has index 6 as a subgroup of the Clifford group
  and its cosets are: Pauli+$S$, Pauli+$S\cdot H$, Pauli+$S\cdot HS$,
  Pauli+$S\cdot\omega$, Pauli+$S\cdot H\omega$, and Pauli+$S\cdot
  HS\omega$. It thus follows that a Clifford operator $A'$ can always
  be written as $A'=BC$ with $B$ a Pauli+$S$ operator and $C$ one of
  $I$, $H$, $HS$, $\omega$, $H\omega$, and $HS\omega$.
\end{proof}

\begin{corollary}
  \label{cor:CliffordV-in-vset}
  If $U$ is a Clifford+$V$ operator, then $U\in \V$.
\end{corollary}

To show, conversely, that every element of $\V$ can be represented by
a Clifford+$V$ circuit, we proceed as follows. First, we show that
every unit vector of the form (\ref{eq:vector}) can be reduced to
$e_1=\left[ \begin{smallmatrix} 1 \\ 0\end{smallmatrix}\right]$ by
applying a sequence of carefully chosen Clifford+$V$ gates. Then, we
show how applying this method to the first column of a unitary matrix
$U$ of the form (\ref{eq:matrix}) yields a Clifford+$V$ circuit for
$U$.

\begin{lemma}
  \label{lem-l-invariant}
  If $u$ is a unit vector of the form (\ref{eq:vector}) with least
  $\sqrt 5$-denominator exponent $\ell$ and $W$ is a Clifford circuit,
  then $Wu$ has least $\sqrt 5$-denominator exponent $\ell$.
\end{lemma}

\begin{proof}
  It suffices to show that the generators of the Clifford group
  preserve the least $\sqrt 5$-denominator exponent of $u$. The
  general result then follows by induction. To this end, write $u$ as
  in (\ref{eq:vector}), with $\alpha=a+ib$ and $\beta=c+id$:
  \[
  u=\frac{1}{\rt{k}\rf{\ell}}\begin{bmatrix}
    a+ib \\
    c+id
  \end{bmatrix}.
  \]
  Now apply $H$, $\omega$, and $S$ to $u$:
  \[
  Hu=\frac{1}{\rt{k+1}\rf{\ell}}\begin{bmatrix}
    (a+c)+i(b+d) \\
    (a-c)+i(b-d)
  \end{bmatrix}, \quad \omega u =
  \frac{1}{\rt{k+1}\rf{\ell}}\begin{bmatrix}
    (a-b)+i(a+b) \\
    (c-d)+i(c+d)
  \end{bmatrix},
  \]
  \[
  Su = \frac{1}{\rt{k}\rf{\ell}}\begin{bmatrix}
    a+ib \\
    -d+ic
  \end{bmatrix}.
  \]
  By minimality of $\ell$, one of $a,b,c,d$ is not divisible by 5. The
  least $\sqrt 5$-denominator of $Su$ is therefore $\ell$. Moreover,
  for any two integers $x$ and $y$, $x+y\equiv x-y \equiv 0 ~(\mymod
  5)$ implies $x\equiv y \equiv 0 ~(\mymod 5)$. Thus the least $\sqrt
  5$-denominator exponent of $Hu$ and $\omega u$ is also $\ell$.
\end{proof}

\begin{lemma}
  \label{lem-column-l}
  If $u$ is a unit vector of the form (\ref{eq:vector}) with least
  denominator exponent $(k,\ell)$, then there exists a Clifford
  circuit $W$ such that $Wu$ has least denominator exponent
  $(0,\ell)$.
\end{lemma}

\begin{proof}
  By Lemma~\ref{lem-l-invariant}, we need not worry about $\ell$ and
  only have to focus on reducing $k$. Write $u$ as in
  (\ref{eq:vector}), with $0\leq k \leq 2$, $\alpha=a+ib$, and
  $\beta=c+id$. Since $u$ has unit norm, we have $a^2+b^2+c^2+d^2 =
  2^k\cdot 5^\ell$. We prove the lemma by case distinction on $k$. If
  $k=0$, there is nothing to prove. The remaining cases are treated as
  follows.
  \begin{itemize}
  \item $k=1$. In this case $a^2 + b^2 +c^2 +d^2 = 2\cdot 5^\ell
    \equiv 0 ~(\mymod 2)$. Therefore only an even number amongst
    $a,b,c,d$ can be odd. Using a Pauli+$S$ operator, we can without
    loss of generality assume that $a\equiv c ~(\mymod 2)$ and
    $b\equiv d ~(\mymod 2)$ or that $a\equiv b ~(\mymod 2)$ and
    $c\equiv d ~(\mymod 2)$. It then follows that either $Hu$ or
    $\omega u$ has denominator exponent $(0,\ell)$ since
    \[
    Hu = \frac{1}{2\rf{\ell}}
    \begin{bmatrix}
      (a+c) + i (b+d) \\
      (a-c) + i(b-d)
    \end{bmatrix} \quad \mbox{ and } \quad \omega u =
    \frac{1}{2\rf{\ell}}
    \begin{bmatrix}
      (a-b) + i (a+b) \\
      (c-d) + i(c+d)
    \end{bmatrix}.
    \]
  \item $k=2$. In this case $a^2 + b^2 +c^2 +d^2 = 4\cdot 5^\ell
    \equiv 0 ~(\mymod 4)$. This implies that $a,b,c$ and $d$ must have
    the same parity and thus, by minimality of $k$, must all be
    odd. Using a Pauli+$S$ operator, we can without loss of generality
    assume that $a\equiv b \equiv c \equiv d \equiv 1 ~(\mymod 4)$. It
    then follows that $H\omega u$ has denominator exponent $(0,\ell)$
    since
    \[
    H\omega u = \frac{1}{4\rf{\ell}}
    \begin{bmatrix}
      (a-b+c-d) + i (a+b+c+d) \\
      (a-b-c+d) + i(a+b-c-d)
    \end{bmatrix}.
    \]\qedhere
  \end{itemize}
\end{proof}

\begin{remark}
  Let $V$ be one of the $V$\!-gates, $u$ be a vector of the form
  (\ref{eq:vector}), and $\ell$ and $\ell'$ be the least $\sqrt
  5$-denominator exponents of $u$ and $Vu$ respectively. Then
  $\ell'\leq \ell+1$. Moreover, if it were the case that
  $\ell'<\ell-1$, then the least $\sqrt 5$-denominator exponent of
  $V\da Vu=u$ would be strictly less $\ell$ which is absurd. Thus
  $\ell-1 \leq \ell' \leq \ell+1$.
\end{remark}

\begin{lemma}
  \label{lem-column-k}
  If $u$ is a unit vector of the form (\ref{eq:vector}) with least
  denominator exponent $(0,\ell)$, then there exists a Pauli+$V$
  circuit $W$ of $V$\!-count $\ell$ such that $Wu = e_1$, the first
  standard basis vector.
\end{lemma}

\begin{proof}
  Write $u$ as in (\ref{eq:vector}) with $k=0$, $\alpha=a+ib$, and
  $\beta=c+id$. Since $u$ has unit norm, we have $a^2+b^2+c^2+d^2 =
  2^0\cdot 5^\ell=5^\ell$. We prove the lemma by induction on $\ell$.
  \begin{itemize}
  \item $\ell=0$. In this case $a^2+b^2+c^2+d^2 = 1$. It follows that
    exactly one of $a,b,c,d$ is $\pm 1$ while all the others are
    0. Then $u$ can be reduced to $e_1$ by acting on it using a Pauli
    operator.
  \item $\ell>0$. In this case $a^2+b^2+c^2+d^2\equiv 0~(\mymod 5)$.
    We will show that there exists a Pauli+$V$ operator $U$ of
    $V$\!-count 1 such that the least denominator exponent of $Uu$ is
    $\ell-1$. It then follows by the induction hypothesis that there
    exists $U'$ of $V$\!-count $\ell-1$ such that $U'Uu=e_1$, which
    then completes the proof.
    
    Consider the residues modulo 5 of $a,b,c,$ and $d$. Since $0,1,$
    and $4$ are the only squares modulo 5, then, up to a reordering of
    the tuple $(a,b,c,d)$, we must have:
    \[
    (a,b,c,d) \equiv \left\{ \begin{array}{l}
        (0,0,0,0) \\
        (\pm2,\pm1,0,0) \\
        (\pm2, \pm2, \pm1, \pm1).
      \end{array} \right. 
    \]
    However, by minimality of $\ell$, we know that $a\equiv b\equiv
    c\equiv d \equiv 0$ is impossible, so the other two cases are the
    only possible ones. We treat them in turn.
    
    First, assume that one of $a,b,c,d$ is congruent to $\pm 2$, one
    is congruent to $\pm 1$, and the remaining two are congruent to
    $0$. By acting on $u$ with a Pauli operator, we can moreover
    assume without loss of generality that $a\equiv 2$. Now if
    $b\equiv 1$, consider $V_Zu$:
    \[
    V_Zu = \frac{1}{\sqrt{5}^{k+1}}
    \begin{bmatrix}
      (a-2b) + i (2a+b) \\
      (c+2d) + i (d-2c)
    \end{bmatrix}.
    \]
    Since $a\equiv 2$, $b\equiv 1$, and $c\equiv d\equiv 0$, we get
    $(a-2b)\equiv (2a+b)\equiv (c+2d)\equiv (d-2c)\equiv 0~(\mymod
    5)$. The least denominator exponent of $V_Zu$ is therefore
    $\ell-1$. If on the other hand $b\equiv -1$ then
    \[
    {V_Z}\da u = \frac{1}{\sqrt{5}^{k+1}}
    \begin{bmatrix}
      (a+2b) + i (b-2a) \\
      (c-2d) + i (d+2c)
    \end{bmatrix}
    \]
    and reasoning analogously shows that the least denominator
    exponent of ${V_Z}\da u$ is $k-1$. A similar argument can be made
    in the remaining cases, i.e., when $c\equiv \pm 1$ or $d\equiv
    \pm1$. For brevity, we list the desired operators in the table
    below. The left column describes the residues of $a,b,c$, and $d$
    modulo 5 and the right column gives the operator $U$ such that
    $Uu$ has least denominator exponent $\ell-1$.
    \begin{center}
      \begin{tabular}{r|l}
        $(a,b,c,d)$   & $U$         \\
        \hline
        \hline \rule{0pt}{3ex} 
        $(2,1,0,0) $  & $V_Z$       \\
        $(2,0,1,0)$   & ${V_Y}\da$  \\
        $(2,0,0,1)$   & $V_X$       \\
        $(2,-1,0,0)$  & $V_Z\da$    \\  
        $(2,0,-1,0)$  & $V_Y$       \\
        $(2,0,0,-1)$  & ${V_X}\da$           
      \end{tabular} 
    \end{center}
    
    Now assume that two of $a,b,c,d$ are congruent to $\pm2$ while the
    remaining two are congruent to $\pm1$. We can use Pauli operators
    to guarantee that $a\equiv 2$ and $c\geq 0$. As above, we list the
    desired operators in a table for conciseness. It can be checked
    that in each case the given operator is such that the least
    denominator exponent of $Uu$ is $\ell-1$.
    \begin{center}
      \begin{tabular}{r|l}
        $(a,b,c,d)$   & $U$         \\
        \hline
        \hline \rule{0pt}{3ex}
        $(2,2,1,1)$   & ${V_Y}\da$  \\    
        $(2,1,2,1)$   & $V_X$       \\
        $(2,1,1,2)$   & $V_Z$       \\
        $(2,1,2,-1)$  & $V_Z$       \\
        $(2,-1,2,1)$  & ${V_Z}\da$  \\
        $(2,2,1,-1)$  & ${V_X}\da$  \\
        $(2,-2,1,1)$  & $V_X$       \\
        $(2,1,1,-2)$  & ${V_Y}\da$  \\
        $(2,-1,1,2)$  & ${V_Y}\da$  \\
        $(2,-1,1,-2)$ & ${V_Z}\da$  \\
        $(2,-1,2,-1)$ & ${V_X}\da$  \\
        $(2,-2,1,-1)$ & ${V_Y}\da$           
      \end{tabular} 
    \end{center}\qedhere
  \end{itemize}
\end{proof}

We can now solve Problem~\ref{pb:exact-synth-CliffordV}.

\begin{proposition}
  \label{prop-exact-cliffordV}
  A unitary operator $U\in U(2)$ is exactly representable by a
  Clifford+$V$ circuit if and only if $U\in\vset$. Moreover, there
  exists an efficient algorithm that computes a Clifford+$V$ circuit
  for $U$ with $V$\!-count equal to the least $\sqrt 5$-denominator
  exponent of $U$, which is minimal.
\end{proposition}

\begin{proof}
  The left-to-right implication is given by
  Corollary~\ref{cor:CliffordV-in-vset}. For the right-to-left
  implication, it suffices to show that there exists a Clifford+$V$
  circuit $W$ of $V$\!-count $\ell$ such that $WU = I$, since we then
  have $U=W\da$. To construct $W$, apply Lemma~\ref{lem-column-l} and
  Lemma~\ref{lem-column-k} to the first column $u_1$ of $U$. This
  yields a circuit $W'$ such that the first column of $W'U$ is
  $e_1$. Since $W'U$ is unitary, it follows that its second column
  $u_2$ is a unit vector orthogonal to $e_1$. Therefore $u_2 = \lambda
  e_2$ where $\lambda$ is a unit of the Gaussian integers. Since the
  determinant of $W'$ is $i^m$ for some integer $m$, the determinant
  of $W'U$ is $i^{n+m}$, so that $\lambda=i^{n+m}$. Thus one of the
  following equalities must hold
  \begin{center}
    $W'U =I$, $ZW'U=I$, $SW'U=I$ or $ZSW'U =I$.
  \end{center}
  To prove the second claim, suppose that the least $\sqrt
  5$-denominator exponent of $U$ is $\ell$. Then $W$ can be
  efficiently computed because the algorithm described in the proofs
  of Lemma~\ref{lem-column-l} and Lemma~\ref{lem-column-k} requires
  $O(\ell)$ arithmetic operations.  Moreover, $W$ has $V$\!-count
  $\ell$ by Lemma~\ref{lem-column-k}, which is minimal since any
  Clifford+$V$ circuit of $V$\!-count up to $\ell-1$ has least $\sqrt
  5$-denominator exponent at most $\ell-1$.
\end{proof}

\begin{remark}
  \label{rem:Pauli-V-exact}
  By restricting $k$ to be equal to 0 in (\ref{eq:matrix}) and the
  determinant of $U$ to be $\pm1$ we get a solution to the problem of
  exact synthesis for Pauli+$V$ operators.
\end{remark}

% ====================================================================
\section{Approximate synthesis of \texorpdfstring{$z$}{z}-rotations}
\label{sec:approx-synth-z}

We now turn to the approximate synthesis of \emph{$z$-rotations} over
the Clifford+$V$ gate set. A $z$-rotation is a unitary matrix of the
form
\begin{equation}
  \label{eq:z-rot-def}
  \Rz(\theta) = \begin{bmatrix}
    e^{-i\theta/2} & 0 \\
    0 & e^{i\theta/2}
  \end{bmatrix}
\end{equation}
for some real number $\theta$. Matrices of the form
(\ref{eq:z-rot-def}) are called $z$-rotations because they act as
rotations of the Bloch sphere along the $z$-axis.

\begin{problem}
  \label{pb:approx-synth-CliffordV}
  Given an angle $\theta$ and a precision $\epsilon >0$, construct a
  Clifford+$V$ circuit $U$ whose $V$\!-count is as small as possible
  and such that $\norm{U-\Rz(\theta)}\leq \epsilon$.
\end{problem}

% --------------------------------------------------------------------
\subsection{A reduction of the problem}
\label{ssect:red-V}

\begin{lemma}
  \label{lem:det-sol}
  Let $c_V=|1-e^{i\pi/4}|$. If $\epsilon<c_V$, then all solutions to
  Problem~\ref{pb:approx-synth-CliffordV} have determinant 1. If
  $\epsilon\geq c_V$, then there exists a solution of the form
  $\omega^n$ for some $n\in\N$.
\end{lemma}

\begin{proof}
  Every complex $2\by 2$ unitary operator $U$ can be written as
  \[
  U =
  \begin{bmatrix}
    a & -b\da e^{i\phi}\\
    b & a\da e^{i\phi}
  \end{bmatrix},
  \]
  for $a,b\in\Comp$ and $\phi\in[-\pi, \pi]$. This, together with the
  characterization of Clifford+$V$ operators given by
  Proposition~\ref{prop-exact-cliffordV}, implies that a complex $2\by
  2$ unitary operator $U$ can be exactly synthesized over the
  Clifford+$V$ gate set if and only if
  \[
  U = \frac{1}{\rt{k}\rf{\ell}}
  \begin{bmatrix}
    \alpha    & -\beta\da i^n \\
    \beta & \alpha\da i^n
  \end{bmatrix},
  \]
  with $k,\ell,n\in\N$, $\alpha,\beta\in\Z[i]$, and $0\leq\ell\leq
  2$. Now assume that $\epsilon<|1-e^{i\pi/4}|$ and
  $\norm{U-\Rz(\theta)}\leq \epsilon$. Let $e^{i\phi_1}$ and
  $e^{i\phi_2}$ be the eigenvalues of $U\Rz(\theta)\inv$, with
  $\phi_1,\phi_2\in[-\pi,\pi]$. Then $|1-e^{i\pi/4}|>\epsilon\geq
  \norm{U-\Rz(\theta)} =
  \norm{I-U\Rz(\theta)\inv}=\max\s{|1-e^{i\phi_1}|,|1-e^{i\phi_2}|}$,
  so that $|1-e^{i\phi_j}|<|1-e^{i\pi/4}|$. Therefore
  $-\pi/4<\phi_j<\pi/4$, for $j\in\s{1,2}$, which implies that
  $-\pi/2<\phi_1+\phi_2<\pi/2$. Hence
  $|1-e^{i(\phi_1+\phi_2)}|<|1-e^{i\pi/2}|=\sqrt{2}$. But
  $e^{i(\phi_1+\phi_2)}=\det(U\Rz(\theta)\inv)=i^n$. Thus
  $|1-i^n|<\sqrt{2}$ which proves that $i^n=1$.
  
  For the second statement, note that if $\theta/2\in[-\pi/4,\pi/4]$,
  then $\norm{I-\Rz(\theta)}= |1-e^{i\theta/2}|\leq
  |1-e^{i\pi/4}|$. Similarly, if $\theta/2$ belongs to one of
  $[\pi/4,3\pi/4]$, $[3\pi/4,5\pi/4]$, or $[5\pi/4,7\pi/4]$, then one
  of $\norm{\omega^2-\Rz(\theta)}$, $\norm{-I-\Rz(\theta)}$, or
  $\norm{-\omega^2-\Rz(\theta)}$ is less than $|1-e^{i\pi/4}|$. In
  each case, $\Rz(\theta)$ is approximated to within $\epsilon$ by a
  Clifford operator.
\end{proof}

\begin{definition}
  \label{def:repsilon}
  Let $\theta$ be an angle and $\epsilon>0$ a precision. The
  \emph{$\epsilon$-region for $\theta$} is the subset of the plane
  defined by
  \[
  \Repsilon = \s{u \in \Disk \such u \cdot z \geq
    1-\frac{\epsilon^2}{2}}
  \]
  where $z=e^{-i\theta/2}$ and $\Disk$ is the unit disk.
\end{definition}

The $\epsilon$-region is illustrated in Figure~\ref{fig:Repsilon}. We
now show how to reduce Problem~\ref{pb:approx-synth-CliffordV} to
three distinct problems.

% ......................................................................
\begin{figure}
  \[
  \mp{0.8}{\begin{tikzpicture}[scale=1.8, baseline=0]

  \fill[fill=blue!10, yscale=-1, rotate=35] (cos 50, sin 50) -- (cos 50, -sin 50) arc (-50:50:1) -- cycle;
  \path[color=gray] (-.3,.4) node {$\Disk$};
  \draw[color=gray,->] (-1.2,0) -- (1.4,0);
  \draw[color=gray,->] (0,-1.1) -- (0,1.2);
  \path[color=gray] (0,1) node[above left] {$i$};
  \path[color=gray] (1,0) node[above=6pt,right=-2pt] {$1$};
  \draw[color=gray] (0,0) circle (1);
  \draw[yscale=-1, rotate=35] (cos 50, -sin 50) arc (-50:50:1);
  \draw[->, yscale=-1, rotate=35] (0,0) -- (1,0) node[right] {$\vec z$};
  \draw[yscale=-1, rotate=35] (cos 50, sin 50) -- (cos 50,-1);
  \draw[yscale=-1, rotate=35] (1,sin 50) -- (1,-1);
  \draw[yscale=-1, rotate=35] (cos 50, sin 50) -- (1, sin 50);
  \draw[yscale=-1, rotate=35] (cos 50, -sin 50) -- (1, -sin 50);
  \draw[<->, yscale=-1, rotate=35] (cos 50,-0.9) -- (1,-0.9);
  \draw[yscale=-1, rotate=35] (0.9, -0.9) node[above] {$\frac{\epsilon^2}{2}$};
  \path[yscale=-1, rotate=35] (0.8,-0.2) node {$\Repsilon$};
  \draw (0.5,0) arc (0:-35:0.5);
  \path (-17.5:0.4) node {\small $\frac{\theta}{2}$};

\end{tikzpicture}}
  \]
  \caption{The $\epsilon$-region..}
  \label{fig:Repsilon}
  \rule{\textwidth}{0.1mm}
\end{figure}
% ......................................................................

\begin{proposition}
  \label{prop:red-synth-zrot-CliffordV}
  Problem~\ref{pb:approx-synth-CliffordV} reduces to a grid problem, a
  Diophantine equation, and an exact synthesis problem, namely:
  \begin{enumerate}
  \item[1.] find $k,\ell\in\N$ with $0\leq k\leq 2$ and $\alpha\in\Zi$
    such that $\alpha \in \rt{k}\rf{\ell}\Repsilon$,
  \item[2.] find $\beta\in\Zi$ such that $\beta\da\beta = 2^k 5^\ell -
    \alpha\da\alpha$, and
  \item[3.] find a Clifford+$V$ circuit for the unitary matrix
    \[
    U= \frac{1}{\rt{k}\rf{\ell}}\begin{bmatrix}
      \alpha & -\beta\da \\
      \beta & \alpha\da
    \end{bmatrix}.
    \]
  \end{enumerate}
  Moreover, the least $\rt{k}\rf{\ell}$ for which the above three
  problems can be solved yields an optimal solution to
  Problem~\ref{pb:approx-synth-CliffordV}.
\end{proposition}

\begin{proof}
  Let $U$ be the matrix
  \begin{equation}
    \label{eqn:form-sols}
    U = \frac{1}{\rt{k}\rf{\ell}}
    \begin{bmatrix}
      \alpha    & -\beta\da \\
      \beta & \alpha\da
    \end{bmatrix}
  \end{equation}
  with $\alpha, \beta \in \Zi$ and $k,\ell\in\N$ satisfying $0\leq k
  \leq 2$ and $\alpha\da\alpha = 2^k5^\ell$. Given $\epsilon$ and
  $\theta$, we can express the requirement $\norm{U-\Rz(\theta)}\leq
  \epsilon$ as a constraint on the top left entry
  $\alpha/(\rf{k}\rt{\ell})$ of $U$. Indeed, let $z=e^{-i\theta/2}$,
  $\alpha'=\alpha/(\rf{k}\rt{\ell})$, and
  $\beta'=\beta/(\rf{k}\rt{\ell})$. Since ${\alpha'}\da
  \alpha'+{\beta'}\da \beta'=1$ and $z\da z=1$, we have
  \begin{align*}
    \norm{U - \Rz(\theta) }^2
    & = |\alpha'-z|^2 + |\beta'|^2 \\
    & = (\alpha'-z)\da(\alpha'-z) + {\beta'}\da\beta' \\
    & = {\alpha'}\da\alpha'+{\beta'}\da\beta' -
    z\da\alpha'-{\alpha'}\da z + z\da z  \\
    & = 2 - 2\Realpart(z\da \alpha').
  \end{align*}
  Thus $\norm{\Rz(\theta)-U}\leq \epsilon$ if and only if
  $2-2\Realpart(z\da \alpha')\leq\epsilon^2$, or equivalently,
  $\Realpart(z\da \alpha') \geq 1-\frac{\epsilon^2}{2}$. If we
  identify the complex numbers $z=x+yi$ and $\alpha'=a+bi$ with
  2-dimensional real vectors $\vec z = (x,y)^T$ and $\vec \alpha' =
  (a,b)^T$, then $\Realpart(z\da \alpha')$ is just their inner product
  $\vec z\cdot \vec \alpha'$, and therefore $\norm{U -
    \Rz(\theta)}\leq \epsilon$ is equivalent to $\vec z\cdot \vec
  \alpha' \geq 1 - \epsilon^2/2$. Hence
  \[
  \norm{U-\Rz(\theta)}\leq \epsilon \iff
  \alpha\in\rt{k}\rf{\ell}\Repsilon.
  \]
  The fact that, by Lemma~\ref{lem:det-sol}, all the solutions to
  Problem~\ref{pb:approx-synth-CliffordV} are of the form
  (\ref{eqn:form-sols}) completes the reduction. Since by
  Proposition~\ref{prop-exact-cliffordV} the minimal $V$\!-count of an
  element $U\in\V$ is its least $\sqrt{5}$-denominator exponent, then
  the least $\rt{k}\rf{\ell}$ for which problems 1-3 can be solved is
  an optimal solution.
\end{proof}

% --------------------------------------------------------------------
\subsection{The algorithm}
\label{ssect:main-algo-V}

The reduction of Proposition~\ref{prop:red-synth-zrot-CliffordV}
describes an algorithm which we explicitly state below.

\begin{algorithm}
  \label{algo:CliffordV}
  Let $\theta$ and $\epsilon>0$ be given.
  \begin{enumerate}
  \item[1.] Use the algorithm from
    Proposition~\ref{prop:algo-grid-pb-scaled-Zi} of
    Chapter~\ref{chap:grid-pb} to enumerate the infinite sequence of
    solutions to the scaled grid problem over $\Zi$ for $\Repsilon$
    and $\rt{k}\rf{\ell}$ in order of increasing $\ell$.
  \item[2.] For each solution $\alpha$:
    \begin{enumerate}
    \item[(a)] Let $n = 2^k5^\ell - \alpha\da\alpha$.
    \item[(b)] Attempt to find a prime factorization of $n$. If $n\neq
      0$ but no prime factorization is found, skip step 2(c) and
      continue with the next $\alpha$.
    \item[(c)] Use the algorithm from
      Proposition~\ref{prop:diophantineZi} of Chapter~\ref{chap:nb-th}
      to solve the equation $\beta\da \beta = n$. If a solution
      $\beta$ exists, go to step 3; otherwise, continue with the next
      $\alpha$.
    \end{enumerate}
  \item[3.] Define $U$ as
    \[
    U= \frac{1}{\rt{k}\rf{\ell}}\begin{bmatrix}
      \alpha & -\beta\da \\
      \beta & \alpha\da
    \end{bmatrix}
    \]
    and use the exact synthesis algorithm of
    Proposition~\ref{prop-exact-cliffordV} to find a Clifford+$V$
    circuit for $U$. Output this circuit and stop.
  \end{enumerate}
\end{algorithm}

\begin{remark}
  \label{rem:Pauli-V-approx}
  In analogy with Remark~\ref{rem:Pauli-V-exact}, we note that
  restricting $k$ to be equal to 0 throughout the algorithm yields a
  method for the approximate synthesis of $z$-rotations in the
  Pauli+$V$ basis.
\end{remark}

% --------------------------------------------------------------------
\subsection{Analysis of the algorithm}
\label{ssect:analysis-algo-V}

We now discuss the properties of Algorithm~\ref{algo:CliffordV}. We
are interested in three aspects of the algorithm: its correctness, its
circuit complexity, and its time complexity. We treat each of these
aspects in turn. We note that the results established here also hold
for the restricted algorithm of Remark~\ref{rem:Pauli-V-approx}.

% --------------------------------------------------------------------
\subsubsection{Correctness}
\label{sssect:correctness}

\begin{proposition}[Correctness]
  \label{prop-correctness}
  If Algorithm~\ref{algo:CliffordV} terminates, then it yields a valid
  solution to the approximate synthesis problem, i.e., it yields a
  Clifford+$V$ circuit approximating $\Rz(\theta)$ up to $\epsilon$.
\end{proposition}

\begin{proof}
  By construction, following the reduction described by
  Proposition~\ref{prop:red-synth-zrot-CliffordV}.
\end{proof}

% --------------------------------------------------------------------
\subsubsection{Circuit complexity}
\label{sssect:circ-comp}

In the presence of a factoring oracle, the algorithm has optimal
circuit complexity.

\begin{proposition}[Optimality in the presence of a factoring oracle]
  \label{prop-optimality}
  In the presence of an oracle for integer factoring, the circuit
  returned by Algorithm~\ref{algo:CliffordV} has the smallest
  $V$\!-count of any single-qubit Clifford+$V$ circuit approximating
  $\Rz(\theta)$ up to $\epsilon$.
\end{proposition}

\begin{proof}
  By construction, step~(1) of the algorithm enumerates all solutions
  $\alpha$ to the scaled grid problem over $\Zi$ for $\Repsilon$ and
  $\rt{k}\rf{\ell}$ in order of increasing $\ell$. Step~2(a) always
  succeeds and, in the presence of the factoring oracle, so does
  step~2(b). When step~2(c) succeeds, the algorithm has found a
  solution of Problem~\ref{pb:approx-synth-CliffordV} of minimal
  $\sqrt{5}$-denominator exponent, which therefore has minimal
  $V$\!-count.
\end{proof}

In the absence of a factoring oracle, the algorithm is still nearly
optimal. Our proof of this near-optimality relies on the following
number-theoretic hypothesis. We do not have a proof of this
hypothesis, but it appears to be valid in practice.

\begin{hypothesis}
  \label{hyp-numbers-V}
  For each number $n$ produced in step~2(a) of
  Algorithm~\ref{algo:CliffordV}, write $n=2^jm$, where $m$ is
  odd. Then $m$ is asymptotically as likely to be a prime congruent to
  1 modulo 4 as a randomly chosen odd number of comparable
  size. Moreover, each $m$ can be modelled as an independent random
  event.
\end{hypothesis}

\begin{definition}
  Let $U'$ and $U''$ be the following two solutions of the approximate
  synthesis problem
  \begin{equation}\label{eqn-u1-u2}
    U' = \begin{bmatrix}\alpha' & -\beta'^{\dagger} \\ 
      \beta' & \alpha'^{\dagger} \end{bmatrix}
    \quad 
    \mbox{and}
    \quad
    U'' = \begin{bmatrix}\alpha'' & -\beta''^{\dagger} \\ 
      \beta'' & \alpha''^{\dagger}\end{bmatrix}.
  \end{equation}
  $U'$ and $U''$ are said to be {\em equivalent solutions} if
  $\alpha'=\alpha''$.
\end{definition}

\begin{proposition}[Near-optimality in the absence of a factoring
  oracle]
  \label{prop-near-optimality}
  Let $\ell$ be the $V$\!-count of the solution of the approximate
  synthesis problem found by Algorithm~\ref{algo:CliffordV} in the
  absence of a factoring oracle. Then
  \begin{enumerate}
  \item The approximate synthesis problem has at most
    $O(\log(1/\epsilon))$ non-equivalent solutions with $V$\!-count
    less than $\ell$.
  \item The expected value of $\ell$ is
    $\ell'''+O(\log(\log(1/\epsilon)))$, where $\ell', \ell'',$ and
    $\ell'''$ are the $V$\!-counts of the optimal, second-to-optimal,
    and third-to-optimal solutions of the approximate synthesis
    problem (up to equivalence).
  \end{enumerate}
\end{proposition}

\begin{proof}
  If $\epsilon\geq |1-e^{i\pi/4}|$, then by Lemma~\ref{lem:det-sol}
  there is a solution of $V$\!-count 0 and the algorithm easily finds
  it. In this case there is nothing to show, so assume without loss of
  generality that $\epsilon<|1-e^{i\pi/4}|$. Then by
  Lemma~\ref{lem:det-sol}, all solutions are of the form
  (\ref{eqn:form-sols}).
  \begin{enumerate}
  \item Consider the list $\alpha_1,\alpha_2,\ldots$ of candidates
    generated in step~(i) of the algorithm. Let $\ell_1,\ell_2,\ldots$
    be their least $\sqrt 5$-denominator exponent and let
    $n_1,n_2,\ldots$ be the corresponding integers calculated in
    step~(ii.a). Note that $n_j\leq 4\cdot 5^{\ell_j}$ for all
    $j$. Write $n_j=2^{z_j}m_j$ where $m_j$ is odd. By
    Hypothesis~\ref{hyp-numbers-V}, the probability that $m_j$ is a
    prime congruent to 1 modulo 4 is asymptotically no smaller than
    that of a randomly chosen odd integer less than $4\cdot
    5^{\ell_j}$, which, by the well-known prime number theorem, is
    \begin{equation}\label{eqn-pj}
      p_j:=\frac{1}{\ln (4\cdot 5^{\ell_j})} 
      = \frac{1}{\ell_j\ln 5 + \ln 4}.
    \end{equation}
    By the pigeon-hole principle, two of $\ell_1,\ell_2,$ and $\ell_3$
    must be congruent modulo 2. Assume without loss of generality that
    $\ell_2\equiv \ell_3 ~(\mymod 2)$. Then $\alpha_2$ and $\alpha_3$
    are two distinct solutions to the scaled grid problem over $\Zi$
    for $\Repsilon$ and $\rt{k}\rf{\ell}$ with (not necessarily least)
    $\sqrt 5$-denominator exponent $\ell_3$. It follows by
    Proposition~\ref{prop:evolution-grid-Zi} from
    Chapter~\ref{chap:grid-pb} that there are at least $5^r+1$
    distinct candidates of denominator exponent $\ell_3+2r$, for all
    $r\geq 0$. In other words, for all $j$, if $j\leq 5^r+1$, we have
    $\ell_j\leq \ell_3+2r$.  In particular, this holds for
    $r=\floor{1+\log_5 j}$, and therefore,
    \begin{equation}\label{eqn-kj}
      \ell_j\leq \ell_3+2(1+\log_5 j).
    \end{equation}
    Combining {\eqref{eqn-kj}} with {\eqref{eqn-pj}}, we have
    \begin{equation}\label{eqn-pj2}
      p_j \geq \frac{1}{(\ell_3+2(1+\log_5 j))\ln 5 + \ln 4}
      = \frac{1}{(\ell_3+2)\ln 5 + 2\ln j + \ln 4}
    \end{equation}
    Let $j_0$ be the smallest index such that $m_{j_0}$ is a prime
    congruent to 1 modulo 4. By Hypothesis~\ref{hyp-numbers-V}, we can
    treat each $m_j$ as an independent random variable. Therefore,
    \begin{eqnarray}\label{def:proba}
      P(j_0 > j) &=& P(\mbox{$n_1,\ldots,n_j$ are not prime})
      \nonumber\\
      &\leq & (1-p_1)(1-p_2)\cdots(1-p_j) \nonumber\\
      &\leq & (1-p_j)^j \nonumber\\
      &\leq & 
      \bigparen{1-\frac{1}{(\ell_3+2)\ln 5 + 2\ln j + \ln 4}}^j.         
      \nonumber
    \end{eqnarray}
    The expected value of $j_0$ is given by the sum $E(j_0)=
    \sum_{j=0}^{\infty} P(j_0 > j)$\label{def:expected}. It was proved
    in \cite{gridsynth} that this sum can be estimated as follows
    \begin{eqnarray}
      \label{eqn-ej0}
      E(j_0) &=& \sum_{j=0}^{\infty} P(j_0 > j) \nonumber\\
      &\leq & 1+\sum_{j=1}^{\infty} 
      \bigparen{1-\frac{1}{(\ell_3+2)\ln 5 + 2\ln j + \ln 4}}^j
      ~=~ O(\ell_3).
    \end{eqnarray}
          
    Next, we will estimate $\ell_3$. First note that if the $\epsilon$
    region contains a circle of radius greater than $1/\rf{\ell}$,
    then it contains at least 3 solutions to the scaled grid problem
    for $\Repsilon$ with $\sqrt{5}$-denominator exponent $\ell$. The
    width of the $\epsilon$-region $\Repsilon$ is $\epsilon^2/2$ at
    the widest point, and we can inscribe a disk of radius
    $r={\epsilon^2}/{4}$ in it. Hence the scaled grid problem over
    $\Z[i]$ for $\Repsilon$, as in step~1 of the algorithm, has at
    least three solutions with denominator exponent $\ell$, provided
    that
    \[
    r = \frac{\epsilon^2}{4} \geq \frac{1}{\rf \ell},
    \]
    or equivalently, provided that
    \[
    \ell\geq 2 \log_5(2) + 2\log_5 ({1}/{\epsilon}).
    \]
    It follows that
    \begin{equation}\label{eqn-k2}
      \ell_3=O(\log({1}/{\epsilon})),
    \end{equation}
    and therefore, using {\eqref{eqn-ej0}}, also
    \begin{equation}\label{eqn-ej0-eps}
      E(j_0)=O(\log({1}/{\epsilon})).
    \end{equation}
          
    To finish the proof of part~(i), recall that $j_0$ was defined to
    be the smallest index such that $m_{j_0}$ is a prime congruent to
    1 modulo 4. The primality of $m_{j_0}$ ensures that step~(ii.b) of
    the algorithm succeeds for the candidate
    $\alpha_{j_0}$. Furthermore, because $m_{j_0}\equiv 1\mmod{4}$,
    the equation $\beta\da \beta=n$ has a solution by
    Proposition~\ref{prop:prime-1mod4-Zi}. Hence the remaining steps
    of the algorithm also succeed for $\alpha_{j_0}$.

    Now let $s$ be the number of non-equivalent solutions of the
    approximate synthesis problem of $V$\!-count strictly less than
    $\ell$. As noted above, any such solution $U$ is of the form
    {\eqref{eqn:form-sols}}. Then the least denominator exponent of
    $\alpha$ is strictly smaller than $\ell_{j_0}$, so that
    $\alpha=\alpha_j$ for some $j<j_0$. In this way, each of the $s$
    non-equivalent solutions is mapped to a different index
    $j<j_0$. It follows that $s<j_0$, and hence that $E(s)\leq
    E(j_0)=O(\log({1}/{\epsilon}))$, as was to be shown.
  \item Let $U'$ be an optimal solution of the approximate synthesis
    problem, let $U''$ be optimal among the solutions that are not
    equivalent to $U'$ and let $U'''$ be optimal among the solutions
    that are not equivalent to either $U'$ or $U''$. Assume that $U',
    U'',$ and $U'''$ are written as in (\ref{eqn-u1-u2}) with top-left
    entry $\alpha',\alpha'',$ and $\alpha'''$ respectively. Now let
    $\ell'$, $\ell''$, and $\ell'''$ be the least denominator
    exponents of $\alpha'$, $\alpha''$, and $\alpha'''$,
    respectively. Let $\ell_3$ and $j_0$ be as in the proof of
    part~(i). Note that, by definition, $\ell_3\leq \ell'''$.  Let
    $\ell$ be the least denominator exponent of the solution of the
    approximate synthesis problem found by the algorithm. Then
    $\ell\leq \ell_{j_0}$. Using {\eqref{eqn-kj}}, we have
    \[
    \ell \leq \ell_{j_0} \leq \ell_3 + 2(1+\log_5 j_0) \leq \ell''' +
    2(1+\log_5 j_0).
    \]
    This calculation applies to any one run of the algorithm.  Taking
    expected values over many randomized runs, we therefore have
    \begin{equation}\label{eqn-em}
      E(\ell) \leq \ell''' + 2 + 2 E(\log_5 j_0) 
      \leq \ell''' + 2 + 2\log_5 E(j_0).
    \end{equation}
    Note that we have used the law $E(\log j_0)\leq \log(E(j_0))$,
    which holds because $\log$ is a concave function. Combining
    {\eqref{eqn-em}} with {\eqref{eqn-ej0-eps}}, we therefore have the
    desired result:
    \begin{equation}
      E(\ell) = \ell''' + O(\log(\log(1/\epsilon))). \nonumber
    \end{equation}\qedhere
  \end{enumerate}
\end{proof}

% --------------------------------------------------------------------
\subsubsection{Time complexity}
\label{sssect:time-comp}

Finally, we turn to the time complexity of the algorithm. For this
again, we rely on our number theoretic conjecture on the distribution
of primes to estimate how many candidates must be tried before one
that is prime is reached.

\begin{proposition}
  Algorithm~\ref{algo:CliffordV} runs in expected time
  $O(\polylog(1/\epsilon))$. This is true whether or not a
  factorization oracle is used.
\end{proposition}

\begin{proof}
  Let $M$ be the uprightness of the $\epsilon$-region. Let $j_0$ be
  the average number of candidates tried in steps 2(a)--(c) of the
  algorithm, and let $\ell_{j_0}$ be the least denominator exponent of
  the final candidate. Let $n$ be the largest integer that appears in
  step 2(a) of the algorithm.

  By Proposition~\ref{prop:algo-grid-pb-scaled-Zi}, step~1 of the
  algorithm requires $O(\log(1/M))$ arithmetic operations, plus a
  constant number per candidate. For each of the $j_0$ candidates,
  step~2(a) requires $O(1)$ arithmetic operations. Step 2(b) also
  requires $O(1)$ arithmetic operations, either due to the use of a
  factoring oracle, or else, because we can put an arbitrary fixed
  bound on the amount of effort invested in factoring any given
  integer. At minimum, this will succeed when the integer in question
  is prime, which is sufficient for the estimates of
  Proposition~\ref{prop-near-optimality}. Step 2(c) requires
  $O(\polylog(n))$ operations by
  Proposition~\ref{prop:diophantineZi}. Finally, step 3 requires
  $O(\ell_{j_0})$ arithmetic operations by
  Proposition~\ref{prop-exact-cliffordV}. So the total expected number
  of arithmetic operations is
  \begin{equation} \label{eqn-expected} O(\log(1/M)) + j_0\cdot
    O(\polylog(n)) + O(\ell_{j_0}).
  \end{equation}
  Recall that the $\epsilon$-region $\Repsilon$, shown in
  Figure~\ref{fig:Repsilon}, contains a disk of radius $\epsilon^2/4$;
  therefore, $\area(\Repsilon)\geq \frac{\pi}{16}\epsilon^4$. On the
  other hand, the square $[-1,1]\times[-1,1]$ is a (not very tight)
  bounding box for $\Repsilon$. It follows that
  \[
  M = \up(\Repsilon) =
  \frac{\area(\Repsilon)}{\area(\BBox(\Repsilon))} =
  \Omega(\epsilon^4),
  \]
  hence $\log(1/M) = O(\log(1/\epsilon))$. From {\eqref{eqn-ej0-eps}},
  the expected value of $j_0$ is $O(\log({1}/{\epsilon}))$.  Combining
  {\eqref{eqn-kj}} with {\eqref{eqn-k2}}, we therefore have
  \[ \ell_{j_0} \leq \ell_3+2(1+\log_2 j_0) = O(\log({1}/{\epsilon}))
  + O(\log(\log(1/\epsilon))) = O(\log(1/\epsilon)).
  \]
  Now note that for any $i$, $n_i \leq 5^{\ell_i+1}$. This, together
  with the fact that candidates are enumerated in order of increasing
  denominator exponent, we have $n\leq 4^{\ell_{j_0}}$, hence
  \[ \polylog(n) = O(\poly(\ell_{j_0})) = O(\polylog(1/{\epsilon})).
  \]
  Combining all of these estimates with {\eqref{eqn-expected}}, the
  expected number of arithmetic operations for the algorithm is
  $O(\polylog(1/{\epsilon}))$. Moreover, each individual arithmetic
  operation can be performed with precision $O(\log(1/\epsilon))$,
  taking time $O(\polylog(1/{\epsilon}))$. Therefore the total
  expected time complexity of the algorithm is
  $O(\polylog(1/{\epsilon}))$, as desired.
\end{proof}

% ====================================================================
\section{Approximate synthesis of special unitaries}
\label{sec:approx-synth}

The algorithm of the previous section allows us to approximate
$z$-rotation up to arbitrarily small accuracy. This method can be used
to solve the problem of approximate synthesis of arbitrary special
unitaries over the Clifford+$V$ gate set.

\begin{problem}
  \label{pb:approx-synth-CliffordV-gen}
  Given a special unitary $U$ and a precision $\epsilon >0$, construct
  a Clifford+$V$ circuit $U$ whose $V$\!-count is as small as possible
  and such that $\norm{U-\Rz(\theta)}\leq \epsilon$.
\end{problem}

Indeed, an element $U\in \suset(2)$ can always be decomposed as a
product of three rotations using \emph{Euler angles}. Hence
\[
U = \Rz(\theta_1) \Rx(\theta_2) \Rz(\theta_3).
\]
Using the Hadamard gate, the central $x$-rotation can be expressed as
a $z$-rotation. Thus
\[
U = \Rz(\theta_1) H \Rz(\theta_2) H \Rz(\theta_3).
\]
We can therefore use Algorithm~\ref{algo:CliffordV} to find a
Clifford+$V$ circuit approximating each of the $\Rz(\theta_i)$ up to
$\epsilon/3$. Since the Hadamard gate is a Clifford operator, this
yields a Clifford+$V$ approximation of $U$ up to $\epsilon$.

We note that optimality is lost in the process of decomposing an
operator as a product of $z$-rotations. Indeed, if we write a special
unitary $U$ as 
\[
U=\Rz(\theta_1) H \Rz(\theta_2) H \Rz(\theta_3)
\]
and approximate each $z$-rotation using
Algorithm~\ref{algo:CliffordV}, we obtain a circuit whose length
exceeds the optimal one by a factor of 3 in the typical case.

% ---------------------------------------------------------------------
\chapter{Clifford+\texorpdfstring{$T$}{T} approximate synthesis}
\label{chap:synth-T}

In this chapter, we introduce an efficient algorithm to solve the
problem of approximate synthesis of special unitaries over the
Clifford+$T$ gate set. Recall from Chapter~\ref{chap:intro} that the
$T$\label{defTgate} gate is the following matrix
\[
T = \begin{bmatrix}
  1 & 0 \\
  0 & \omega
\end{bmatrix}.
\]
The Clifford+$T$ gate set is obtained by adding the $T$ gate to the
generators $\omega$, $H$, and $S$ of the Clifford group.

The results presented in this chapter are obtained by adapting the
methods of Chapter~\ref{chap:synth-V} to the Clifford+$T$
setting. Like in Chapter~\ref{chap:synth-V}, we first consider the
exact synthesis of Clifford+$T$ operators. This provides a
characterization of Clifford+$T$ circuits which we then use to define
an algorithm for the approximate synthesis of $z$-rotations. Because
of the similarities between this chapter and the previous one, we omit
most proofs in order to avoid redundancy. However, we explain the
differences when they occur.

The approximate synthesis algorithms introduced in this chapter
(Algorithm~\ref{algo:CliffordT} and Algorithm~\ref{alg-phase}) have
been implemented in Haskell. The implementations are freely available
\cite{implementation}.

% ====================================================================
\section{Exact synthesis of Clifford+\texorpdfstring{$T$}{T}
  operators}
\label{sec:optimal-CliffordT-synth}

\begin{problem}[Exact synthesis of Clifford+$T$ operators]
  \label{pb:exact-synth-CliffordT}
  Given a unitary $U\in\uset(2)$, determine whether there exists a
  Clifford+$T$ circuit $W$ such that $U=W$ and, in case such a circuit
  exists, construct one whose $T$\!-count is minimal.
\end{problem}

Problem~\ref{pb:exact-synth-CliffordT} was first solved in
\cite{KMM-approx}. A version of Problem~\ref{pb:exact-synth-CliffordT}
generalized to multi-qubit circuits was solved in
\cite{Giles-Selinger}.

To characterize Clifford+$T$ operators, we consider the following set
of unitaries.

\begin{definition}
  \label{def:subgroup-RT}
  The set $\tset$ consists of unitary matrices of the form
  \begin{equation}
    \label{eq:matrixT}
    U =  \frac{1}{\rt{k}} \begin{bmatrix} 
      \alpha & \gamma \\
      \beta & \delta
    \end{bmatrix}
  \end{equation}
  where $k\in\N$ and $\alpha,\beta,\gamma,\delta\in\Zomega$.
\end{definition}

We note that $\tset$ is the subgroup of $U(2)$ consisting of matrices
with entries in $\Z[1/\sqrt{2}, i]$. This is slightly different than
the situation in the previous chapter, where we had defined $\vset$ to
be a strict subset of the group of unitary matrices over
$\Z[1/\sqrt{5},i]$. We will also use a notion of denominator exponent
for the elements of $\tset$.

\begin{definition}
  \label{def:2-denomexp}
  Let $U\in\tset$ be as in (\ref{eq:matrixT}). The integer $k$ is
  called a \emph{denominator exponent} of $U$. The least $k$ such that
  $U$ can be written as in (\ref{eq:matrixT}) is the \emph{least
    denominator exponent} of $U$. These notions extend naturally to
  vectors and scalars of the form
  \begin{equation}
    \label{eq:vectorT}
    \frac{1}{\rt{k}} \begin{bmatrix} 
      \alpha \\
      \beta 
    \end{bmatrix}
    \quad                           
    \mbox{ and }
    \quad
    \frac{1}{\rt{k}}~ \alpha,                            
  \end{equation}
  where $k\in\N$ and $\alpha, \beta\in\Zomega$.
\end{definition}

In the previous chapter, we used the set $\vset$ to characterize
Clifford+$V$ operators. Similarly, one can prove that Clifford+$T$
operators are exactly the elements of $\tset$.

\begin{proposition}[Kliuchnikov, Maslov, Mosca
  \cite{Kliuchnikov-etal}]
  \label{prop-exact-cliffordT}
  A unitary operator $U\in U(2)$ is exactly representable by a
  Clifford+$T$ circuit if and only if $U\in\tset$. Moreover, there
  exists an efficient algorithm that computes a Clifford+$T$ circuit
  for $U$ with minimal $T$-count.
\end{proposition}

The above proposition can be proved by a technique similar to the one
used to establish Proposition~\ref{prop-exact-cliffordV}. To this end
one first shows that every vector of the form (\ref{eq:vectorT}) can
be reduced to $e_1=\left[ \begin{smallmatrix} 1 \\ 0 \end{smallmatrix}
\right]$ by applying well-chosen Clifford+$T$ operators. Applying this
method to the first column of an element $U$ of $\tset$ then yields a
circuit for $U$.

Recall that in Proposition~\ref{prop-exact-cliffordV}, the minimal
$V$-count of the operator $U$ was equal to its least
$\sqrt{5}$-denominator exponent. The relation between denominator
exponent and minimal $T$-count is slightly more complicated.

\begin{proposition}
  \label{prop-exact-cliffordT-Tcount}
  Let $U\in\tset$ with least denominator exponent $k$ and minimal
  $T$-count $t$. Then $2k-3 \leq t \leq 2k+1$.
\end{proposition}

\begin{proof}
  See, e.g., \cite{ma-remarks}.
\end{proof}

% ====================================================================
\section{Approximate synthesis of \texorpdfstring{$z$}{z}-rotations}
\label{sec:approx-synth-z-T}

As in Chapter~\ref{chap:synth-V}, we consider the problem of
approximate synthesis of $z$-rotations.

\begin{problem}
  \label{pb:approx-synth-CliffordT}
  Given an angle $\theta$ and a precision $\epsilon >0$, construct a
  Clifford+$T$ circuit $U$ whose $T$\!-count is as small as possible
  and such that
  \begin{equation}
    \label{eqn-norm}
    \norm{U-\Rz(\theta)}\leq \epsilon.
  \end{equation}
\end{problem}

An algorithm to solve Problem~\ref{pb:approx-synth-CliffordT} can be
used to solve the problem of approximate synthesis of arbitrary
special unitaries using Euler angles, as in
Section~\ref{sec:approx-synth}.

Our algorithm solving Problem~\ref{pb:approx-synth-CliffordT} relies
on a reduction of the problem to a grid problem, a Diophantine
equation and an exact synthesis problem. This is analogous to the
reduction described in the Clifford+$V$ case by
Proposition~\ref{prop:red-synth-zrot-CliffordV}. In the Clifford+$T$
context, we must first show that enumerating candidate solutions in
order of increasing denominator exponents allows us to also enumerate
candidate solutions in order of minimal $T$-count. This is not
immediate, due to Proposition~\ref{prop-exact-cliffordT-Tcount}.

\begin{lemma}
  \label{lem:det-solT}
  If $\epsilon<|1-e^{i\pi/8}|$, then all solutions to
  Problem~\ref{pb:approx-synth-CliffordT} have the form
  \begin{equation}\label{eqn-u}
    U = \frac{1}{\rt{k}}
    \begin{bmatrix} 
      u & -t\da \\
      t & u\da
    \end{bmatrix}.
  \end{equation}
  If $\epsilon\geq|1-e^{i\pi/8}|$, then there exists a solution of
  $T$-count 0 (i.e., a Clifford operator), and it is also of the form
  {\eqref{eqn-u}}.
\end{lemma}

\begin{proof}
  Analogous to the proof of Lemma~\ref{lem:det-sol}.
\end{proof}

\begin{lemma}
  \label{lem-2k-2}
  Let $U$ be a unitary operator as in {\eqref{eqn-u}} with least
  denominator exponent $k$. Then the $T$-count of $U$ is either $2k-2$
  or $2k$. Moreover, if $k>0$ and $U$ has $T$-count $2k$, then
  $U'=TUT\da$ has $T$-count $2k-2$.  We further note that
  $\norm{\Rz(\theta) - U'} = \norm{\Rz(\theta) - U}$, so for the
  purpose of solving {\eqref{eqn-norm}}, it does not matter whether
  $U$ or $U'$ is used. Hence, without loss of generality, we may
  assume that $U$ as in {\eqref{eqn-u}} always has $T$-count exactly
  $2k-2$ when $k>0$, and $0$ when $k=0$.
\end{lemma}

\begin{proof}
  The claims about the $T$-counts of $U$ and $U'$ follow by inspection
  of Figure~2 of {\cite{ma-remarks}}. Using the terminology of
  Definitions~7.4 and 7.6 of {\cite{ma-remarks}}, this figure shows
  every possible $k$-residue of a Clifford+$T$ operator, modulo a
  right action of the group $\groupspan{S,X,\omega}$. Because $U$ is
  of the form {\eqref{eqn-u}}, only a subset of the $k$-residues is
  actually possible, and the figure shows that for this subset, the
  $T$-count is $2k$ or $2k-2$. Moreover, in each of the possible cases
  where $k>0$ and $U$ has $T$-count $2k$, the figure also shows that
  $U'=TUT\da$ has $T$-count $2k-2$.

  For the final claim, we have $\norm{\Rz(\theta) -
    U}=\norm{T\Rz(\theta)T\da - TUT\da} = \norm{\Rz(\theta) - U'}$
  because $\Rz(\theta)$ and $T$ commute.
\end{proof}

We can now state a reduction for
Problem~\ref{pb:approx-synth-CliffordT} as we did in
Proposition~\ref{prop:red-synth-zrot-CliffordT} for
Problem~\ref{pb:approx-synth-CliffordV}.

\begin{proposition}
  \label{prop:red-synth-zrot-CliffordT}
  Problem~\ref{pb:approx-synth-CliffordT} reduces to a grid problem, a
  Diophantine equation, and an exact synthesis problem, namely:
  \begin{enumerate}
  \item[1.] find $k\in\N$ and $\alpha\in\Zomega$ such that $\alpha \in
    \rt{k}\Repsilon$ and $\alpha\bul\in (-\sqrt{2})^k\Disk$,
  \item[2.] find $\beta\in\Zrt$ such that $ \beta\da\beta = 2^k -
    \alpha\da\alpha$, and
  \item[3.] define the unitary matrix $U$ as
    \[
    U= \frac{1}{\rt{k}}\begin{bmatrix}
      \alpha & -\beta\da \\
      \beta & \alpha\da
    \end{bmatrix}
    \]
    and find a Clifford+$T$ circuit for $U$ or $TUT\da$, whichever has
    the smaller $T$-count.
  \end{enumerate}
  Moreover, the least $k$ for which the above three problems can be
  solved yields an optimal solution to
  Problem~\ref{pb:approx-synth-CliffordT}.
\end{proposition}

Note that item~1 of Proposition~\ref{prop:red-synth-zrot-CliffordT} is
a grid problem over $\Zomega$. This is in contrast with the
corresponding item of Proposition~\ref{prop:red-synth-zrot-CliffordT},
which was a grid problem over $\Zi$. In both cases, we look for points
in a scaled $\epsilon$-region in order of increasing denominator
exponent. However, in the Clifford+$T$ case, the desired points must
be elements of $\Zomega$. Since $\Zomega$ is dense in $\R^2$, there
are infinitely many elements in $\Repsilon\cap\Zomega$ for any fixed
denominator exponent. To circumvent this issue, we only consider those
elements of $\Repsilon\cap\Zomega$ for which the Diophantine equation
of item~2 can potentially be solved. Since we have
\[
\alpha\da\alpha + \beta\da\beta = 2^k \Longrightarrow \alpha \in
\rt{k}\Disk \mbox{ and } \alpha\bul \in (-\sqrt{2})^k\Disk,
\]
where $\Disk$ is the closed unit disk, the points of interest are
precisely the solutions to the scaled grid problem over $\Zomega$ for
$\Repsilon$ and $\Disk$.

\begin{algorithm}
  \label{algo:CliffordT}
  Let $\theta$ and $\epsilon>0$ be given.
  \begin{enumerate}
  \item[1.] Use the algorithm from
    Proposition~\ref{prop:algo-grid-pb-scaled-Zomega} of
    Chapter~\ref{chap:grid-pb} to enumerate the infinite sequence of
    solutions to the scaled grid problem over $\Zomega$ for
    $\Repsilon$ and $\Disk$ and $k$ in order of increasing $k$.
  \item[2.] For each solution $\alpha$:
    \begin{enumerate}
    \item[(a)] Let $\xi = 2^k - \alpha\da\alpha$ and $n=\xi\bul\xi$.
    \item[(b)] Attempt to find a prime factorization of $n$. If $n\neq
      0$ but no prime factorization is found, skip step 2(c) and
      continue with the next $\alpha$.
    \item[(c)] Use the algorithm from
      Proposition~\ref{prop:diophantineZomega} of
      Chapter~\ref{chap:nb-th} to solve the equation $\beta\da \beta =
      n$. If a solution $\beta$ exists, go to step 3; otherwise,
      continue with the next $\alpha$.
    \end{enumerate}
  \item[3.] Define $U$ as
    \[
    U= \frac{1}{\rt{k}}\begin{bmatrix}
      \alpha & -\beta\da \\
      \beta & \alpha\da
    \end{bmatrix}
    \]
    and use the exact synthesis algorithm of
    Proposition~\ref{prop-exact-cliffordT} to find a Clifford+$V$
    circuit for $U$ or $TUT\da$, whichever has the smallest
    $T$-count. Output this circuit and stop.
  \end{enumerate}
\end{algorithm}

We now state the properties of Algorithm~\ref{algo:CliffordT}. In most
cases it enjoys the same properties as the Clifford+$V$ algorithm.

\begin{proposition}[Correctness]
  \label{prop-correctnessT}
  If Algorithm~\ref{algo:CliffordT} terminates, then it yields a valid
  solution to the approximate synthesis problem, i.e., it yields a
  Clifford+$T$ circuit approximating $\Rz(\theta)$ up to $\epsilon$.
\end{proposition}

\begin{proposition}[Optimality in the presence of a factoring oracle]
  \label{prop-optimalityT}
  In the presence of an oracle for integer factoring, the circuit
  returned by Algorithm~\ref{algo:CliffordT} has the smallest
  $T$\!-count of any single-qubit Clifford+$T$ circuit approximating
  $\Rz(\theta)$ up to $\epsilon$.
\end{proposition}

Correctness and optimality are proved like the corresponding
propositions in Chapter~\ref{chap:synth-V}.

Here also, we rely on a number-theoretic assumption on the
distribution of primes to establish the remaining properties of the
algorithm.

\begin{hypothesis}
  \label{hyp-numbers-T}
  For each $n$ produced in step 2(a) of
  Algorithm~\ref{algo:CliffordT}, write $n=2^jm$, where $m$ is
  odd. Then $m$ is asymptotically as likely to be prime as a randomly
  chosen odd number of comparable size. Moreover, the primality of
  each $m$ can be modelled as an independent random event.
\end{hypothesis}

Note that Hypothesis~\ref{hyp-numbers-T} is slightly different than
Hypothesis~\ref{hyp-numbers-V}. Indeed, Hypothesis~\ref{hyp-numbers-V}
makes an additional assumption on the residue class of the integer
$m$. Here it is not necessary to make such an assumption, since we can
prove that the number $n$ produced in step 2(a) of the algorithm
satisfies $n \geq 0$ and moreover is such that either $n=0$ or
$n\equiv 1 ~(\mymod 8)$. A proof of this fact can be found in Appendix
D of \cite{gridsynth}.

\begin{proposition}[Near-optimality in the absence of a factoring
  oracle]
  \label{prop-near-optimalityT}
  Let $m$ be the $T$\!-count of the solution of the approximate
  synthesis problem found by Algorithm~\ref{algo:CliffordT} in the
  absence of a factoring oracle. Then
  \begin{enumerate}
  \item The approximate synthesis problem has at most
    $O(\log(1/\epsilon))$ non-equivalent solutions with $T$-count less
    than $m$.
  \item The expected value of $m$ is $m''+O(\log(\log(1/\epsilon)))$,
    where $m'$ and $m''$ are the $T$-counts of the optimal and
    second-to-optimal solutions of the approximate synthesis problem
    (up to equivalence).
  \end{enumerate}
\end{proposition}

Note that in Proposition~\ref{prop-near-optimalityT}, we use the
second-to-optimal solution, rather than the third-to-optimal solution
as in Proposition~\ref{prop-near-optimality}. This is due to the fact
that $\sqrt{2}\in\Zomega$ whereas $\sqrt{5}\notin\Zi$. Indeed, if
$\alpha$ and $\alpha'$ are two solutions of least denominator exponent
$k$ and $k'$ with $k\leq k'$, then they are both solutions of
denominator exponent $k'$. But in the case of the Clifford+$V$ gates,
we need to have three solutions to guarantee that two will have the
same denominator exponent.

The last property of the algorithm can be proved just like the
corresponding one from Chapter~\ref{chap:synth-V}.

\begin{proposition}
  Algorithm~\ref{algo:CliffordT} runs in expected time
  $O(\polylog(1/\epsilon))$. This is true whether or not a
  factorization oracle is used.
\end{proposition}

% ====================================================================
\section{Approximation up to a phase}
\label{sect:up-to-phase}

So far, we have considered the problem of approximate synthesis ``on
the nose'', i.e., the operator $U$ in
Problem~\ref{pb:approx-synth-CliffordT} was literally required to
approximate $\Rz(\theta)$ in the operator norm.  However, it is
well-known that global phases have no observable effect in quantum
mechanics, so in quantum computing, it is also common to consider the
problem of approximate synthesis ``up to a phase''.  This is made
precise in the following definition.

\begin{problem}
  \label{def-approx-up-to-phase}
  Given $\theta$ and some $\epsilon>0$, the {\em approximate synthesis
    problem for $z$-rotations up to a phase} is to find an operator
  $U$ expressible in the single-qubit Clifford+$T$ gate set, and a
  unit scalar $\lambda$, such that
  \begin{equation}\label{eqn-approx-up-to-phase}
    \norm{\Rz(\theta)-\lambda U} \leq \epsilon.
  \end{equation}
  Moreover, it is desirable to find $U$ of smallest possible
  $T$-count. As before, the norm in {\eqref{eqn-approx-up-to-phase}}
  is the operator norm.
\end{problem}

In this section, we will give a version of
Algorithm~\ref{algo:CliffordT} that optimally solves the approximate
synthesis problem up to a phase.  The central insight is that it is in
fact sufficient to restrict $\lambda$ to only two possible phases,
namely $\lambda=1$ and $\lambda=\sqrt{\omega}=e^{i\pi/8}$.

First, note that if $W$ is a unitary $2\times 2$-matrix and $\det
W=1$, then $\tr W$ is real. This is obvious, because $\det W=1$
ensures that the two eigenvalues of $W$ are each other's complex
conjugates.

\begin{lemma}\label{lem-IW}
  Let $W$ be a unitary $2\times 2$-matrix, and assume that $\det W =
  1$ and $\tr W \geq 0$.  Then for all unit scalars $\lambda$, we have
  \[ \norm{I-W} \leq \norm{I-\lambda W}.
  \]
\end{lemma}

\begin{proof}
  We may assume without loss of generality that $W$ is diagonal. Since
  $\det W = 1$, we can write
  \[ W = \begin{bmatrix} e^{i\phi} & 0 \\ 0 & e^{-i\phi}
  \end{bmatrix}
  \]
  for some $\phi$. By symmetry, we can assume without loss of
  generality that $0\leq\phi\leq\pi$. Since $\tr W\geq 0$, we have
  $\phi\leq\pi/2$. Now consider a unit scalar $\lambda=e^{i\psi}$,
  where $-\pi\leq\psi\leq\pi$.  Then $\norm{I-\lambda W} =
  \max\s{|1-e^{i(\psi+\phi)}|, |1-e^{i(\psi-\phi)}|}$ and $\norm{I-W}
  = |1-e^{i\phi}|$. If $\psi\geq 0$, then
  $|1-e^{i\phi}|\leq|1-e^{i(\psi+\phi)}|$. Similarly, if $\psi\leq 0$,
  then $|1-e^{i\phi}|\leq|1-e^{i(\psi-\phi)}|$. In either case, we
  have $\norm{I-W} \leq \norm{I-\lambda W}$, as claimed.
\end{proof}

\begin{lemma}\label{lem-discrete-phase}
  Fix $\epsilon$, a unitary operator $R$ with $\det R=1$, and a
  Clifford+$T$ operator $U$. The following are equivalent:
  \begin{enumerate}
  \item There exists a unit scalar $\lambda$ such that
    \[ \norm{R-\lambda U}\leq\epsilon;
    \]
  \item There exists $n\in\Z$ such that
    \[ \norm{R-e^{in\pi/8} U}\leq\epsilon.
    \]
  \end{enumerate}
\end{lemma}

\begin{proof}
  It is obvious that (2) implies (1). For the opposite implication,
  first note that, because $U$ is a Clifford+$T$ operator, we have
  $\det U=\omega^k$ for some $k\in\Z$, and therefore $\det (R\inv U) =
  \omega^{k}$. Let $V=e^{-ik\pi/8}R\inv U$, so that $\det V=1$. If
  $\tr V\geq 0$, let $W=V$; otherwise, let $W=-V$. Either way, we have
  $W=e^{in\pi/8}R\inv U$, where $n\in\Z$, and $\det W=1$, $\tr W\geq
  0$. Let $\lambda' = e^{-in\pi/8}\lambda$. By Lemma~\ref{lem-IW}, we
  have
  \[ \begin{array}{crcl}
    & \norm{I-W} &\leq& \norm{I-\lambda' W} \\
    \imp &
    \norm{I-e^{in\pi/8}R\inv U}  &\leq& 
    \norm{I-\lambda' e^{in\pi/8}R\inv U} \\
    \imp &
    \norm{R-e^{in\pi/8} U}  &\leq& 
    \norm{R-\lambda' e^{in\pi/8} U},\\
    \imp &
    \norm{R-e^{in\pi/8} U}  &\leq& 
    \norm{R-\lambda U}, \\
  \end{array}
  \]
  which implies the desired claim.
\end{proof}

\begin{remark}
  A version of Lemma~\ref{lem-discrete-phase} applies to gate
  sets other than Clifford+$T$, as long as the gate set has discrete
  determinants.
\end{remark}

\begin{corollary}\label{cor-only-two-phases}
  In Definition~\ref{def-approx-up-to-phase}, it suffices without loss
  of generality to consider only the two scalars $\lambda=1$ and
  $\lambda=e^{i\pi/8}$.
\end{corollary}

\begin{proof}
  Suppose $U$ is a Clifford+$T$ operator satisfying
  {\eqref{eqn-approx-up-to-phase}} for some unit scalar $\lambda$. By
  Lemma~\ref{lem-discrete-phase}, there exists a $\lambda$ of the form
  $e^{in\pi/8}$ also satisfying {\eqref{eqn-approx-up-to-phase}}.
  Then we can write $\lambda=\omega^k\lambda'$, where $k\in\Z$ and
  $\lambda'\in \s{1,e^{i\pi/8}}$. Letting $U'=\omega^k U$, we have
  $\lambda'U' = \lambda U$, and therefore
  \[ \norm{\Rz(\theta)-\lambda' U'} \leq \epsilon,
  \]
  as claimed. Moreover, since $\omega=e^{i\pi/4}$ is a Clifford
  operator, $U$ and $U'$ have the same $T$-count.
\end{proof}

To solve the approximate synthesis problem up to a phase, we therefore
need an algorithm for finding optimal solutions of
{\eqref{eqn-approx-up-to-phase}} in case $\lambda=1$ and
$\lambda=e^{i\pi/8}$. For $\lambda=1$, this is of course just
Algorithm~\ref{algo:CliffordT}. So all that remains to do is to find
an algorithm for solving
\begin{equation}\label{eqn-norm8}
  \norm{\Rz(\theta)-e^{i\pi/8}U}\leq \epsilon.
\end{equation}
We use a sequence of steps very similar to those of
Proposition~\ref{prop:red-synth-zrot-CliffordT} to reduce this to a
grid problem and a Diophantine equation.  We first consider the form
of $U$.

\begin{lemma}\label{lem-ell8}
  If $\epsilon<|1-e^{i\pi/8}|$, then all solutions of
  {\eqref{eqn-norm8}} have the form
  \begin{equation}\label{eqn-u8}
    U = \begin{bmatrix} 
      u & -t\da\omega^{-1} \\ 
      t & u\da\omega^{-1} 
    \end{bmatrix}.
  \end{equation}
\end{lemma}

\begin{proof}
  This is completely analogous to the proof of
  Lemma~\ref{lem:det-solT}, using $e^{i\pi/8}U$ in place of $U$.
\end{proof}

Recall that $\delta = 1+\omega$, and note that
$\frac{\delta}{|\delta|} = e^{i\pi/8}$. Also note that
$\delta\omega\inv = \delta\da$, and that $\delta\inv =
(\omega-i)/\sqrt 2$. Suppose that $U$ is of the form
{\eqref{eqn-u8}}. Let $u'=\delta u$ and $t'=\delta t$. We have:
\begin{eqnarray}
  \norm{\Rz(\theta) - e^{i\pi/8}U}
  &=& 
  \bignorm{\Rz(\theta) - \frac{\delta}{|\delta|}\begin{bmatrix} u &
      -t\da\omega^{-1} \\ t & u\da\omega^{-1} \end{bmatrix}}
  \nonumber\\ &=&
  \bignorm{\Rz(\theta) - \frac{1}{|\delta|}\begin{bmatrix} \delta u &
      -\delta\da t\da \\ \delta t & \delta\da u\da \end{bmatrix}}
  \nonumber\\ &=&
  \bignorm{\Rz(\theta) - \frac{1}{|\delta|}\begin{bmatrix} u' &
      -{t'}\da \\ t' & {u'}\da \end{bmatrix}}.
  \nonumber
\end{eqnarray}
Using exactly the same argument as in
Proposition~\ref{prop:red-synth-zrot-CliffordT}, it follows that
{\eqref{eqn-norm8}} holds if and only if
$\frac{u'}{|\delta|}\in\Repsilon$, i.e., $u'\in |\delta|\Repsilon$.

As before, in order for $U$ to be unitary, of course it must satisfy
$u\da u+t\da t=1$, and a necessary condition for this is
$u,u\bul\in\Disk$. The latter condition can be equivalently
re-expressed in terms of $u'$ by requiring $u'\in|\delta|\,\Disk$ and
${u'}\bul\in|\delta\bul\!|\,\Disk$. Therefore, finding solutions to
{\eqref{eqn-norm8}} of the form {\eqref{eqn-u8}} reduces to the
two-dimensional grid problem $u'\in |\delta|\Repsilon$ and
${u'}\bul\in|\delta\bul\!|\,\Disk$, together with the usual
Diophantine equation $u\da u+t\da t=1$. The last remaining piece of
the puzzle is to compute the $T$-count of $U$, and in particular, to
ensure that potential solutions are found in order of increasing
$T$-count.

\begin{lemma}\label{lem-2k-1}
  Let $U$ be a Clifford+$T$ operator of the form {\eqref{eqn-u8}}, and
  let $k$ be the least denominator exponent of $u'= \delta u$. Then
  the $T$-count of $U$ is either $2k-1$ or $2k+1$. Moreover, if $k>0$
  and $U$ has $T$-count $2k+1$, then $U'=TUT\da$ has $T$-count
  $2k-1$. 
\end{lemma}

\begin{proof}
  This can be proved by a tedious but easy induction, analogous to
  Lemma~\ref{lem-2k-2}.
\end{proof}

We therefore arrive at the following algorithm for solving
{\eqref{eqn-norm8}}. Here we assume $\epsilon < |1-e^{i\pi/8}|$, so
that Lemma~\ref{lem-ell8} applies.

\begin{algorithm}\label{alg-phase8}
  Given $\theta$ and $\epsilon$, let $A=|\delta|\Repsilon$, and let
  $B=|\delta\bul\!|\,\Disk$.
  \begin{enumerate}
  \item[1.] Use Proposition~\ref{prop:algo-grid-pb-scaled-Zomega} to
    enumerate the infinite sequence of solutions to the scaled grid
    problem over $\Zomega$ for $A$, $B$, and $k$ in order of
    increasing $k$.
  \item[2.] For each such solution $u'$:
    \begin{enumerate}
    \item[(a)] Let $\xi=2^k - {u'}\da u'$, and $n = \xi\bul\xi$.
    \item[(b)] Attempt to find a prime factorization of $n$. If $n\neq
      0$ but no prime factorization is found, skip step 2(c) and
      continue with the next $u'$.
    \item[(c)] Use the algorithm of
      Proposition~\ref{prop:diophantineZomega} to solve the equation
      $t\da t = \xi$. If a solution $t$ exists, go to step 3;
      otherwise, continue with the next $u'$.
    \end{enumerate}
  \item[3.] Define $U$ as in equation {\eqref{eqn-u8}}, let
    $U'=TUT\da$, and use the exact synthesis algorithm of
    {\cite{Kliuchnikov-etal}} to find a Clifford+$T$ circuit
    implementing either $U$ or $U'$, whichever has smaller
    $T$-count. Output this circuit and stop.
  \end{enumerate}
\end{algorithm}

Algorithm~\ref{alg-phase8} is optimal in the presence of a factoring
oracle, and near-optimal in the absence of a factoring oracle, in the
same sense as Algorithm~\ref{algo:CliffordT}. Its expected time
complexity is $O(\polylog(1/\epsilon))$. The proofs are completely
analogous to those of the previous section. We then arrive at the
following composite algorithm for the approximate synthesis problem
for $z$-rotations up to a phase:

\begin{algorithm}[Approximate synthesis up to a phase]
  \label{alg-phase}
  Given $\theta$ and $\epsilon$, run both
  Algorithms~\ref{algo:CliffordT} and {\ref{alg-phase8}}, and return
  whichever circuit has the smaller $T$-count.
\end{algorithm}

\begin{proposition}[Correctness, time complexity, and optimality]
  Algorithm~\ref{alg-phase} yields a valid solution to the approximate
  synthesis problem up to a phase. It runs in expected time
  $O(\polylog(1/\epsilon)$. In the presence of a factoring oracle, the
  algorithm is optimal, i.e., the returned circuit has the smallest
  $T$-count of any single-qubit Clifford+$T$ circuit approximating
  $\Rz(\theta)$ up to $\epsilon$ and up to a phase.  Moreover, in the
  absence of a factoring oracle, the algorithm is near-optimal in the
  following sense. Let $m$ be the $T$-count of the solution
  found. Then:
  \begin{enumerate}
  \item The approximate synthesis problem up to a phase has an
    expected number of at most $O(\log(1/\epsilon))$ non-equivalent
    solutions with $T$-count less than $m$.
  \item The expected value of $m$ is $m''' +
    O(\log(\log(1/\epsilon)))$, where $m'''$ is the $T$-count of the
    third-to-optimal solution (up to equivalence) of the approximate
    synthesis problem up to a phase.
  \end{enumerate}
\end{proposition}

\begin{proof}
  The correctness and time complexity of Algorithm~\ref{alg-phase}
  follows from that of Algorithms~\ref{algo:CliffordT} and
  {\ref{alg-phase8}}. The optimality results follow from those of
  Algorithms~\ref{algo:CliffordT} and {\ref{alg-phase8}}, keeping in
  mind that Algorithm~\ref{algo:CliffordT} finds an optimal (or
  near-optimal) solution for the phase $\lambda=1$,
  Algorithm~\ref{alg-phase8} finds an optimal (or near-optimal)
  solution for the phase $\lambda=e^{i\pi/8}$, and by
  Corollary~\ref{cor-only-two-phases}, these are the only two phases
  that need to be considered.

  The only subtlety that must be pointed out is that in part (b) of
  the near-optimality, we use the $T$-count of the {\em
    third}-to-optimal solution, rather than the second-to-optimal one
  as in Proposition~\ref{prop-near-optimalityT}. This is because the
  optimal and second-to-optimal solutions may belong to
  Algorithms~\ref{algo:CliffordT} and {\ref{alg-phase8}},
  respectively, so that it may not be until the third-to-optimal
  solution that the near-optimality result of either
  Algorithm~\ref{algo:CliffordT} or Algorithm~\ref{alg-phase8} can be
  invoked.
\end{proof}

\begin{remark}
  Algorithms~\ref{algo:CliffordT} and {\ref{alg-phase8}} share the
  same $\epsilon$-region up to scaling, and therefore the uprightness
  computation only needs to be done once.
\end{remark}

\begin{remark}
  By Lemmas~\ref{lem-2k-2} and {\ref{lem-2k-1}},
  Algorithm~\ref{algo:CliffordT} always produces circuits with even
  $T$-count, and Algorithm~\ref{alg-phase8} always produces circuits
  with odd $T$-count. Instead of running both algorithms to
  completion, it is possible to interleave the two algorithms, so that
  all potential solutions are considered in order of increasing
  $T$-count. This is a slight optimization which does not, however,
  affect the asymptotic time complexity.
\end{remark}

% =====================================================================
% Logical Methods in Quantum Computation

% ---------------------------------------------------------------------
\chapter{The Proto-Quipper language}
\label{chap:pq}

In this chapter, we introduce the syntax and operational semantics of
the Proto-Quipper language.

% ====================================================================
\section{From the quantum lambda calculus to Proto-Quipper}
\label{sect:qlc-to-pq}

Proto-Quipper is based on the quantum lambda calculus. As was
discussed in Section~\ref{sect:qlc}, the execution of programs is
modelled in the quantum lambda calculus by a reduction relation
defined on closures, which are triples $[Q,L,a]$ consisting of a
quantum state $Q$, a list of term variables $L$, and a term $a$. The
quantum state is held in a quantum device capable of performing
certain operations (applying unitaries, measuring qubits,\ldots). The
reduction relation in the quantum lambda calculus is then defined as a
probabilistic rewrite procedure on these closures. Typically, the
reduction will be classical until a redex involving a quantum constant
is reached. At this point, the quantum device will be instructed to
perform the appropriate quantum operation. For example: ``Apply a
Hadamard gate to qubit number 3''.

Our approach in designing the Proto-Quipper language was to start with
a limited (but still expressive) fragment of the Quipper language and
make it completely type-safe. The central aspect of Quipper that we
chose to focus on is Quipper's circuit description abilities: to
generate and act on quantum circuits. Indeed, Quipper provides the
ability to treat circuits as data, and to manipulate them as a
whole. For example, Quipper has operators for reversing circuits,
decomposing them into gate sets, etc. This is in contrast with the
quantum lambda calculus, where one only manipulates qubits and all
quantum operations are immediately carried out on a quantum device,
not stored for symbolic manipulation.

We therefore extend the quantum lambda calculus with the minimal set
of features that makes it Quipper-like. The current version of
Proto-Quipper is designed to:
\begin{itemize}
\item incorporate Quipper's ability to generate and act on quantum
  circuits, and to
\item provide a linear type system to guarantee that the produced
  circuits are physically meaningful (in particular, properties like
  no-cloning are respected).
\end{itemize}

To achieve these goals, we define Proto-Quipper as a typed lambda
calculus, whose type system is similar to that of the quantum lambda
calculus. The main difference between Proto-Quipper and the quantum
lambda calculus is that the reduction relation of Proto-Quipper is
defined on closures $[C,a]$ that consist of a term $a$ and a
\emph{circuit state} $C$. Here, the state $C$ represents the circuit
currently being built. Instead of having a quantum device capable of
performing quantum operations, we assume that we have a \emph{circuit
  constructor} capable or performing certain circuit building
operations (such as appending gates, reversing, etc.). The reduction
is then defined as a rewrite procedure on closures. As in the quantum
lambda calculus, some redexes will affect the state by sending
instructions to the circuit constructor. For example: ``Append a
Hadamard gate to wire number 3''. In the current version of
Proto-Quipper, we make the simplifying assumption that no measurements
are available, so that the reduction relation is non-probabilistic.

% ====================================================================
\section{The syntax of Proto-Quipper}
\label{sec:pq:syntax}

In this section, we present in detail the syntax and type system of
Proto-Quipper.

\begin{definition}
  \label{def:pq-types}
  The \emph{types} of Proto-Quipper are defined by
  \[
  A,B \quad \bnf \quad \qubit \bor 1 \bor \bool \bor A\X B \bor A\loli
  B \bor \bang A \bor \Circ (T,U).
  \]
  Among the types, we single out the subset of \emph{quantum data
    types}
  \[
  T,U \quad \bnf \quad \qubit \bor 1 \bor T \X U.
  \]
\end{definition}

The types 1, $\bool$, $A\X B$, $A\loli B$, and $\bang A$ are inherited
from the quantum lambda calculus and should be interpreted as they
were in Section~\ref{sect:qlc}. The elements of $\qubit$ are
references to a logical qubit within a computation. They can be
thought of as references to quantum bits on some physical device, or
simply as references to quantum wires within the circuit currently
being constructed. Elements of quantum data types describe sets of
circuit endpoints, and consist of tuples of wire identifiers. We can
think of these as describing circuit interfaces. Finally, the type
$\Circ(T,U)$ is the set of all circuits having an input interface of
type $T$ and an output interface of type $U$.

\begin{definition}
  \label{def:pq-terms}
  The \emph{terms} of Proto-Quipper are defined by
  \begin{center}
    \begin{tabular}{rl}
      $a,b,c \quad \bnf$ & $x \bor q \bor (t,C,a) \bor \true 
      \bor \false \bor \p{a,b} \bor * \bor ab \bor \lambda x.a 
      \bor$ \\[0.05in]
      & $\rev \bor \unbox \bor \boxx^T \bor \Ifthenelse{a}{b}{c} \bor 
      \letin{*}{a}{b} \bor $\\[0.05in]
      & $\letin{\p{x,y}}{a}{b}.$
    \end{tabular}
  \end{center}
  where $x$ and $y$ come from a countable set $\vset$ of \emph{term
    variables}, $q$ comes from a countable set $\qset$ of
  \emph{quantum variables}, and $C$ comes from a countable set $\cset$
  of \emph{circuit constants}. Among the terms, we single out the
  subset of \emph{quantum data terms}
  \[
  t,u \quad \bnf \quad q \bor * \bor \p{t,u}.
  \]
  Moreover, we assume that $\cset$ is equipped with two functions
  $\In,\Out\from \cset\to\pset_f(\qset)$ and that $\qset$ is
  well-ordered. Here, $\pset_f(\qset)$ denotes the set of finite
  subsets of $\qset$.
\end{definition}

The meaning of most terms is intended to be the standard one. For
example $\p{a,b}$ is the pair of $a$ and $b$, $\true$ and $\false$ are
the booleans and $\lambda x.a$ is the function which maps $x$ to
$a$. We briefly discuss the meaning of the more unusual terms.
\begin{itemize}
\item A circuit constant $C$ represents a low-level quantum circuit.
  Because it would be complicated, and somewhat besides the point, to
  define a formal language for describing low-level quantum circuits,
  Proto-Quipper assumes that there exists a constant symbol for {\em
    every} possible quantum circuit. Each circuit $C$ is equipped with
  a finite set of {\em inputs} and a finite set of {\em outputs},
  which are subsets of the set of quantum variables
  $\qset$. Proto-Quipper's abstract treatment of quantum circuits is
  further explained in Section~\ref{sect:pq-op-sem}.

\item The term $(t,C,a)$ represents a quantum circuit, regarded as
  Proto-Quipper data. The purpose of the terms $t$ and $a$ is to
  provide structure on the (otherwise unordered) sets of inputs and
  outputs of $C$, so that these inputs and outputs can take the shape
  of Proto-Quipper quantum data.  For example, suppose that $C$ is a
  circuit with inputs $\s{q_1,q_2,q_3}$ and outputs
  $\s{q_4,q_5,q_6}$. Then the term
  \[
  (\p{q_2,\p{q_3,q_1}}, C, \p{\p{q_4,q_6},q_5})
  \]
  represents the circuit $C$, but also specifies what it means to
  apply this circuit to a quantum data term $\p{p,\p{r,s}}$. Namely,
  in this case, the circuit inputs $q_2$, $q_3$, and $q_1$ will be
  applied to qubits $p$, $r$, and $s$, respectively.  Moreover, if the
  output of this circuit is to be matched against the pattern
  $\p{\p{x,y},z}$, then the variables $x$, $y$, and $z$ will be bound,
  respectively, to the quantum bits at endpoints $q_4$, $q_6$, and
  $q_5$.

  Terms of the form $(t,C,a)$ are not intended to be written by the
  user of the programming language; in fact, a Proto-Quipper
  implementation would not provide a concrete syntax for such
  terms. Rather, these terms are internally generated during the {\em
    evaluation} of Proto-Quipper programs. However, the circuits for
  certain basic gates may be made available to the user as pre-defined
  symbols.
\item $\boxx^T$ is a built-in function to turn a circuit-producing
  function (for example, a function of type $T\loli U$) into a circuit
  regarded as data (for example, of type $\Circ(T,U)$).
\item $\unbox$ is a built-in function for turning a circuit regarded
  as data into a circuit-producing function. It is an inverse of
  $\boxx^T$.
\item $\rev$ is a built-in function for reversing a low-level circuit.
\end{itemize}
Note that the term $\boxx^T$ is parameterized by a type $T$. This
\emph{Church-style} typing of the language is the reason why types
were introduced before terms. Also note that in a term like $(t,C,a)$,
$t$ is assumed to be a quantum data term, but $a$ is not. The type
system to be introduced below will guarantee that even though $a$ is
not yet a quantum data term it will eventually reduce to one.

\begin{examples}
  \label{ex:pq-terms}
  Suppose that $H$ is the circuit constant for the Hadamard gate. The
  term $\p{q_1, H, q_2}$ then represents the circuit consisting only
  of the Hadamard gate, regarded as Proto-Quipper data. We can then
  define the circuit producing function $\Hgate=\unbox \p{q_1, H,
    q_2}$. Similarly, if $CNOT$ is the circuit constant for the
  controlled-not gate, then the term
  $\CNOTgate=\unbox\p{\p{q_1,q_2},CNOT, \p{q_3,q_4}}$ is the
  corresponding circuit producing function. In an implementation of
  the Proto-Quipper language, a finite set of such circuit producing
  functions would be provided as basic operations. For example, a
  candidate such gate set would consist of the Clifford+$T$ gate set
  extended with the controlled-not gate: $\Hgate, \Sgate, \Tgate,
  \CNOTgate$. Basic gates can then be combined. For example, the term
  \[
  \lambda x. \Tgate(\Sgate(\Hgate x))
  \]
  is the circuit producing function which applies in sequence the
  $\Hgate$, $\Sgate$, and $\Tgate$ gates. Using the $\boxx^T$
  operator, we can turn this circuit producing function into a circuit
  \[
  \boxx^\qubit (\lambda x. \Tgate(\Sgate(\Hgate x))).
  \]
\end{examples}

As in the quantum lambda calculus, the operational semantics of
Proto-Quipper will be defined according to a call-by-value reduction
strategy. We therefore define what it means, for a term of
Proto-Quipper, to be a value.

\begin{definition}
  \label{def:pq-values}
  The \emph{values} of Proto-Quipper are defined by
  \begin{center}
    \begin{tabular}{rl}
      $v,w \quad \bnf$ & $x \bor q \bor (t,C,u) \bor \true \bor 
      \false \bor \p{v,w} \bor$ \\
      & $* \bor \lambda x.a  \bor \boxx^T \bor \rev \bor \unbox 
      \bor \unbox ~v.$
    \end{tabular}
  \end{center}
\end{definition}

Note that according to Definition~\ref{def:pq-values}, some
applications are values, namely terms of the form $\unbox \, v$. This
is consistent with the meaning of the $\unbox$ constant discussed
above. Indeed, if $\unbox$ turns a circuit into a circuit-generating
function, then a term of the form $\unbox \, v$ should be seen as a
function awaiting an argument, much like a term of the form $\lambda
x.a$, and therefore considered a value.

We now introduce some useful syntactic operations on types and
terms. We start by defining the notion of free variable for
Proto-Quipper terms.

\begin{definition}
  \label{def:pq-fv}
  The set of \emph{free (term) variables} of a term $a$, written
  $\FV(a)$, is defined as
  \begin{itemize}
  \item $\FV(x)=\s{x}$,
  \item $\FV(\p{a,b})=\FV(a)\cup \FV(b)$,
  \item $\FV(ab)=\FV(a)\cup \FV(b)$,
  \item $\FV(\lambda x.a)=\FV(a)\setminus\s{x}$,
  \item $\FV(\Ifthenelse{a}{b}{c}) = \FV(a) \cup \FV(b) \cup \FV(c)$,
  \item $\FV(\letin{*}{a}{b}) = \FV(a)\cup \FV(b)$,
  \item $\FV(\letin{\p{x,y}}{a}{b}) = \FV(a)\cup (\FV(b)\setminus
    \s{x,y})$,
  \item $\FV((t,C,a))= \FV(a)$, and
  \item $\FV(a)=\emptyset$ in all remaining cases.
  \end{itemize}
\end{definition}

The above definition of free variables extends the standard one. Note
that the free variables of a term of the form $(t,C,a)$ are the free
variables of $a$. This is justified since no variables ever appear in
the quantum data term $t$.

The notions of $\alpha$-equivalence, capture-avoiding substitution,
etc., are defined in a straightforward manner.

By analogy with the free term variables of a term, we introduce a
notion of \emph{quantum variable of a term}.

\begin{definition}
  \label{def:pq-fqv}
  The set of \emph{free quantum variables} of a term $a$, written
  $\FQ(a)$, is defined as
  \begin{itemize}
  \item $\FQ(q)=\s{q}$,
  \item $\FQ(\p{a,b})=\FQ(a)\cup \FQ(b)$,
  \item $\FQ(ab)=\FQ(a)\cup \FQ(b)$,
  \item $\FQ(\lambda x.a)=\FQ(a)$,
  \item $\FQ(\Ifthenelse{a}{b}{c}) = \FQ(a) \cup \FQ(b) \cup \FQ(c)$,
  \item $\FQ(\letin{*}{a}{b})= \FQ(a)\cup \FQ(b)$,
  \item $\FQ(\letin{\p{x,y}}{a}{b})= \FQ(a)\cup \FQ(b)$, and
  \item $\FQ(a)=\emptyset$ in all remaining cases.
  \end{itemize}
\end{definition}

Note that $\FQ((t,C,a))=\emptyset$. This reflects the idea that the
quantum variables appearing in $t$ and $a$ are ``bound'' in $(t,C,a)$.

To append circuits, we will need to be able to express the way in
which wires should be connected. For this, we use the notion of a
\emph{binding}.

\begin{definition}
  \label{def:fin-bij}
  A \emph{finite bijection} on a set $X$ is a bijection between two
  finite subsets of $X$. We write $\finbij(X)$ for the set of finite
  bijections on $X$. The domain and codomain of a finite bijection
  $\binding$ are denoted $\dom(\binding)$ and $\cod(\binding)$,
  respectively.
\end{definition}

\begin{definition}
  \label{def:binding}
  A \emph{binding} is a finite bijection on $\qset$. We will usually
  denote bindings by $\binding$.
\end{definition}

\begin{definition}
  \label{def:binding-ap}
  If $a$ is a term, $\binding$ is a binding and
  $\FQ(a)=\s{q_1,\ldots,q_n}\seq\dom(\binding)$, then $\binding(a)$ is
  the following term
  \[
  \binding(a)= a[\binding(q_1)/q_1,\ldots ,\binding(q_n)/q_n].
  \]
\end{definition}

\begin{definition}
  \label{def:bind}
  The partial function $\bind: \qdataterm^2\to \finbij(\qset)$ is
  defined as
  \begin{itemize}
  \item $\bind (*,*)= \emptyset$;
  \item $\bind (q_1,q_2)= \s{(q_1,q_2)}$;
  \item $\bind (\p{t_1,t_2},\p{u_1,u_2})=\bind (t_1,u_1) \cupdot \bind
    (t_2,u_2)$, provided that $\bind (t_1,u_1) \cap \bind (t_2,u_2) =
    \emptyset$;
  \item $\bind (t,u)= \mbox{undefined}$, in all remaining cases.
  \end{itemize}
\end{definition}

\begin{definition}
  \label{def:spec}
  Let $T$ be a quantum data type and $X$ a finite subset of
  $\qset$. An \emph{$X$-specimen} for $T$ is quantum data term written
  $\spec_X(T)$ defined as
  \begin{itemize}
  \item $\spec_X(1)=*$,
  \item $\spec_X(\qubit)=q$ where $q$ is the smallest quantum index of
    $\qset\setminus X$,
  \item $\spec_X(T\X U)=\p{t,u}$ where $t=\spec_X(T)$ and
    $u=\spec_{X\cup \FQ(t)}(U)$.
  \end{itemize}
\end{definition}

Informally, an $X$-specimen for $T$ is a quantum data term $t$ that is
``fresh'' with respect to the quantum variables appearing in $X$. If
$X$ is clear from the context, we simply write $\spec (T)$. Note that
the definition of specimen uses the fact that $\qset$ is well-ordered.

As in the quantum lambda calculus, we use a subtyping relation to deal
with the $\bang$ modality.

\begin{definition}
  \label{def:subtyping}
  The \emph{subtyping relation} $<:$ is the smallest relation on types
  satisfying the rules given in Figure~\ref{subtyping_congruences}.
\end{definition}

% ....................................................................
\begin{figure}
  \[
  \infer[]{\qubit <: \qubit}{} \quad \infer[]{1 <: 1}{} \quad
  \infer[]{\bool <: \bool}{}
  \]
  \[
  \infer[]{ (A_1\X A_2) <: (B_1 \X B_2)}{A_1<:B_1 & A_2<:B_2} \quad
  \infer[]{ (A_1\loli B_1) <: (A_2\loli B_2)}{A_2<:A_1 & B_1<:B_2}
  \]
  \[
  \infer[]{ \Circ(A_1, B_1) <: \Circ(A_2, B_2)}{A_2<:A_1 & B_1<:B_2}
  \]
  \[
  \infer[]{ \bang^nA <: \bang^mB}{ A<:B & (n=0 \imp m=0) }
  \]
  \label{subtyping_congruences}
  \caption{Subtyping rules for Proto-Quipper.}
  \rule{\textwidth}{0.1mm}
\end{figure}
% ....................................................................

Note that the subtyping of $A\loli B$ and $\Circ(A,B)$ is {\em
  contravariant} in the left argument, i.e., $A<:A'$ implies $A'\loli
B<:A\loli B$.

\begin{remark}
  \label{subtyping_shape}
  If $A<:B$ then:
  \begin{enumerate}
  \item if $A\in\s{\qubit, 1, \bool}$, then $A=B$;
  \item if $A=A_1\X A_2$, then $B=B_1\X B_2$, $A_1<:B_1$ and
    $A_2<:B_2$;
  \item if $A=A_1\loli A_2$, then $B=B_1\loli B_2$, $B_1<:A_1$ and
    $A_2<:B_2$;
  \item if $A=\Circ(A_1, A_2)$, then $B=\Circ(B_1,B_2)$, $B_1<:A_1$
    and $A_2<:B_2$;
  \item if $B= \bang B'$, then $A= \bang A'$;\label{subtype_bang}
  \item if $A$ is not of the form $\bang A'$, then $B$ is not of the
    form $\bang B'$.
  \end{enumerate}
\end{remark}

\begin{proposition}
  The subtyping relation is reflexive and transitive.
\end{proposition}

As in the quantum lambda calculus, the following subtyping rule is
derivable
\[
\infer[]{\bang A <: A}{}.
\]

\begin{definition}
  A \emph{typing context} is a finite set $\s{x_1:A_1,\ldots,x_n:A_n}$
  of pairs of a variable and a type, such that no variable occurs more
  than once. A \emph{quantum context} is a finite set of quantum
  variables. The expressions of the form $x:A$ in a typing context are
  called \emph{type declarations}.
\end{definition}

We write $\Gamma$ or $\Delta$ for a typing context and $Q$ for a
quantum context. We also adopt the previous notational conventions
when dealing with typing contexts: $|\Gamma|$, $\Gamma(x_i)$,
$\bang\Gamma$, and $\Gamma <: \Gamma'$. Moreover, we still write
$\Gamma,\Gamma'$ to denote the union of two contexts, which is defined
when $|\Gamma|\cap|\Gamma'|=\emptyset$.

\begin{definition}
  \label{def:pq-constant-types}
  Let $T,U$ be quantum data types. For each of the constants
  $\boxx^T$, $\unbox$, and $\rev$, we introduce a type as follows
  \begin{itemize}
  \item $A_{\boxx^T}(T,U)=\bang (T\loli U)\loli \bang \Circ(T,U)$,
  \item $A_{\unbox}(T,U)=\Circ(T,U)\loli \bang (T\loli U)$, and
  \item $A_{\rev}(T,U)=\Circ(T,U) \loli \bang \Circ(U,T)$.
  \end{itemize}
\end{definition}

\begin{definition}
  A \emph{typing judgment} is an expression of the form:
  \[
  \Gamma ; Q \entails a:A
  \]
  where $\Gamma$ is a typing context, $Q$ is a quantum context, $a$ is
  a term and $A$ is a type. A typing judgment is \emph{valid} if it
  can be inferred from the rules given in Figure~\ref{trules}. In the
  rule $\rul{cst}$, $c$ ranges over the set $\s{\boxx^T, \unbox,
    \rev}$. Each typing rule carries an implicit side condition that
  the judgements appearing in it are well-formed. In particular, a
  rule containing a context of the form $\Gamma_1,\Gamma_2$ may not be
  applied unless $|\Gamma_1|\cap|\Gamma_2|=\emptyset$.
\end{definition}

Note that in the typing judgements of Proto-Quipper, quantum variables
and variables are kept separate. As a result, we do not have to
specify that $q:\qubit$ for every quantum variable $q$ since the
typing rules implicitly enforce this. However, when a future version
of Proto-Quipper will be equipped with the ability to manipulate
quantum \emph{and} classical wires, the type of a wire might have to
be explicitly stated.

% Another reason for this is because in the reduction rule for (t,c,u)
% we need to be able to construct, e.g., t in a context that is empty
% but for the quantum variables appearing in t.

% ....................................................................
\begin{figure}[!ht]
  \[
  \infer[\rul{ax_c}]{\bang \Delta, x:A;\emptyset\entails x:B}{ A<:B }
  \quad \infer[\rul{ax_q}]{\bang \Delta;\s{q}\entails q:\qubit}{ }
  \]
  \[
  \infer[\rul{cst}]{\bang \Delta;\emptyset \entails c:B}{ \bang
    A_{c}(T,U)<:B } \quad \infer[\rul{*_i}]{\bang
    \Delta;\emptyset\entails *:\bang^n 1}{ }
  \]
  \[
  \infer[\rul{\lambda_1}]{\Gamma;Q\entails \lambda x.b:A\loli B}{
    \Gamma,x:A;Q \entails b:B } \quad \infer[\rul{\lambda_2}]{\bang
    \Delta;\emptyset \entails \lambda x.b:~\bang^{n+1}(A\loli B)}{
    \bang \Delta, x:A;\emptyset \entails b:B }
  \]
  \[
  \infer[\rul{app}]{\Gamma_1,\Gamma_2, \bang \Delta;Q_1,Q_2\entails
    ca:B}{ \Gamma_1, \bang \Delta;Q_1\entails c:A\loli B & \Gamma_2,
    \bang \Delta ;Q_2\entails a:A }
  \]
  \[
  \infer[\rul{\X_i}]{\Gamma_1,\Gamma_2, \bang \Delta;Q_1,Q_2\entails
    \p{a,b}:\bang^n(A\X B)}{ \Gamma_1, \bang \Delta;Q_1\entails
    a:\bang^nA & \Gamma_2, \bang \Delta ;Q_2\entails b:\bang^nB }
  \]
  % Note that in fact the \X intro rule could have equivalently been
  % written with a:\bang^nA b:\bang^mB and (a,b):\bang^o(A\X B)
  % where o=min {n,m} Indeed the rules are equivalent via the type
  % isomorphism {!!}A~\bang A and the subtyping relation \bang
  % A<:A. The moral here is that a pair is as reusable as its least
  % reusable component.
  \[
  \infer[\rul{\X_e}]{\Gamma_1,\Gamma_2, \bang \Delta;Q_1,Q_2\entails
    \letin{\p{x,y}}{b}{a}:A}{ \Gamma_1, \bang \Delta;Q_1\entails
    b:\bang^n(B_1\X B_2) & \Gamma_2, \bang \Delta, x:\bang^nB_1,
    y:\bang^nB_2 ;Q_2\entails a:A }
  \]
  \[
  \infer[\rul{*_e}]{\Gamma_1,\Gamma_2, \bang \Delta;Q_1,Q_2\entails
    \letin{*}{b}{a}:A}{ \Gamma_1, \bang \Delta;Q_1\entails b:\bang^n1
    & \Gamma_2, \bang \Delta ;Q_2\entails a:A }
  \]
  \[
  \infer[\rul{\top}]{\bang \Delta;\emptyset\entails \true:\bang^n
    \bool}{ } \quad \infer[\rul{\bot}]{\bang \Delta;\emptyset\entails
    \false:\bang^n \bool}{ }
  \]
  \[
  \infer[\rul{if}]{\Gamma_1,\Gamma_2, \bang \Delta;Q_1,Q_2\entails
    \Ifthenelse{b}{a_1}{a_2}:A}{ \Gamma_1, \bang \Delta;Q_1\entails
    b:\bool & \Gamma_2, \bang \Delta;Q_2 \entails a_1:A & \Gamma_2,
    \bang \Delta;Q_2 \entails a_2:A }
  \]
  \[
  \infer[\rul{circ}]{\bang \Delta;\emptyset \entails
    (t,C,a):\bang^n\Circ(T,U)}{ Q_1\entails t:T & \bang \Delta ;
    Q_2\entails a:U & \In(C)=Q_1 & \Out(C)=Q_2 }
  \]
  \label{trules}
  \caption{Typing rules for Proto-Quipper.}
  \rule{\textwidth}{0.1mm}
\end{figure}
% ....................................................................

As a first illustration of the safety properties of the type system,
note that the $\rul{\X_i}$ rule ensures that $\lambda x. \p{x,x}$
cannot be given the type $\qubit\loli \qubit \X \qubit$.

% ====================================================================
\section{The operational semantics of Proto-Quipper}
\label{sect:pq-op-sem}

As mentioned Section~\ref{sect:qlc-to-pq}, the reduction relation for
Proto-Quipper is defined in the presence of a \emph{circuit
  constructor}. This is a device capable of performing certain basic
circuit building operations. It is not necessary to have a detailed
description of the inner workings of this device. In fact, all that is
required for the definition of Proto-Quipper's operational semantics
is the existence of some primitive operations. We now axiomatize these
operations. Their intuitive meaning will be explained following
Definition~\ref{circuit_constructor}.

\begin{definition}
  \label{circuit_constructor}
  A \emph{circuit constructor} consists of a pair of countable sets
  $\p{Q,S}$ together with the following maps
  \begin{itemize}
  \item $\New\from \pow_f(Q) \to S$,
  \item $\Inn\from S\to \pow_f(Q)$,
  \item $\Outt\from S\to \pow_f(Q)$,
  \item $\Rev\from S \to S$,
  \item $\Append\from S\times S\times \finbij(Q) \to S \times
    \finbij(Q)$
  \end{itemize}
  satisfying the following conditions
  \begin{enumerate}
  \item $\Rev\circ\Rev=1_S$,
  \item $\Inn\circ\Rev= \Outt$ and $\Outt\circ\Rev=
    \Inn$\label{in_out_rev},
  \item $\Inn\circ\New = \Outt\circ\New = 1_{\pset_f(\qset)}$, and
  \item if $\Append (C,D,b)=(C',b')$ and $\dom(b)\seq\Outt(C)$ and
    $\cod(b)=\Inn(D)$, then \label{Append_cond_x}
    \begin{enumerate}
    \item $\Inn(C') = \Inn(C)$,\label{Append_cond_2b}
    \item $\dom(b')=\Outt(D)$ and $\cod(b')\seq\Outt(C')$
      and\label{Append_cond_2}
    \item $\Outt(C')= (\Outt(C)\setminus
      \dom(b))\cupdot\cod(b')$.\label{Append_cond_3}
    \end{enumerate}
  \end{enumerate}
\end{definition}

If $\p{Q,S}$ is a circuit constructor, we call the elements of $S$
\emph{circuit states} and the elements of $Q$ \emph{wire identifiers}.
We now explain the intended meaning of a circuit constructor and its
constituents. An element $C\in S$ is a quantum circuit, such as
\[
C \quad=\quad \mp{.8}{\Qcircuit @C=1em @R=.7em {
    & \lstick{q_1} & \gate{H} & \targ & \rstick{q_3} \qw \\
    & \lstick{q_2} & \qw & \ctrl{-1} & \rstick{q_4.}\qw }}
\]
Each circuit has a finite set of inputs and a finite set of outputs,
given by the functions $\Inn$ and $\Outt$. For example, $\Inn(C) =
\s{q_1,q_2}$ and $\Outt(C)=\s{q_3,q_4}$. For $X\seq Q$, the circuit
$\New(X)$ is the identity circuit with inputs and outputs $X$; for
example,
\[
\New(q_1,q_2,q_3)\quad=\quad~~ \mp{0.5}{\Qcircuit @C=1em @R=1.3em {
    &\lstick{q_1}& \qw & \qw & \qw & \rstick{q_1}\qw \\
    &\lstick{q_2}& \qw & \qw & \qw & \rstick{q_2}\qw \\
    &\lstick{q_3}& \qw & \qw & \qw & \rstick{q_3.}\qw }}
\]
The operator $\Rev$ reverses a circuit, swapping its inputs and
outputs in the process. When $(C',b') = \Append(C,D,b)$, the circuit
$C'$ is obtained by appending the circuit $D$ to the end of the
circuit $C$. The function $b$ is used to specify along which wires to
compose $C$ and $D$ while the function $b'$ updates the wire names
post composition. An illustration of this is given in
Figure~\ref{rep_unencap}.

% ....................................................................
\begin{figure}
  \[
  \mbox{ \Qcircuit @C=.5em @R=1.7em { & & & & & & & & & & & \ustick{b}
      & & & & & &
      & & & & & & \ustick{b'} & & & \\
      &\qw &\ustick{q_1}\qw &\qw &\qw &\multigate{4}{~~~C~~~} &\qw
      &\qw &\ustick{q_1}\qw &\qw &\qw &\qw &\qw &\qw &\ustick{q_1'}\qw
      &\qw &\qw &\multigate{2}{~~~D~~~} &\qw &\qw &\ustick{p_1}\qw
      &\qw &\qw &\qw &\qw &\qw &
      \ustick{p_1'}\qw \\
      &\qw &\ustick{q_2}\qw &\qw &\qw &\ghost{~~~C~~~} &\qw &\qw
      &\ustick{q_2}\qw &\qw &\qw &\qw &\qw &\qw &\ustick{q_2'}\qw &\qw
      &\qw &\ghost{~~~D~~~} &\qw &\qw &\ustick{p_2}\qw &\qw &\qw &\qw
      &\qw &\qw &
      \ustick{p_2'}\qw \\
      &\qw &\ustick{q_3}\qw &\qw &\qw &\ghost{~~~C~~~} &\qw &\qw
      &\ustick{q_3}\qw &\qw &\qw &\qw &\qw &\qw &\ustick{q_3'}\qw &\qw
      &\qw &\ghost{~~~D~~~} &\qw &\qw &\ustick{p_3}\qw &\qw &\qw &\qw
      &\qw &\qw &
      \ustick{p_3'}\qw \\
      &\qw &\ustick{q_4}\qw &\qw &\qw &\ghost{~~~C~~~} &\qw &\qw
      &\ustick{q_4}\qw &\qw &\qw &\qw &\qw &\qw &\qw &\qw &\qw &\qw
      &\qw &\qw &\qw &\qw &\qw &\qw &\qw &\qw &\ustick{q_4}\qw \\
      &\qw &\ustick{q_5}\qw &\qw &\qw &\ghost{~~~C~~~} &\qw &\qw
      &\ustick{q_5}\qw &\qw &\qw &\qw &\qw &\qw &\qw &\qw &\qw &\qw
      &\qw &\qw &\qw &\qw &\qw &\qw &\qw &\qw &\ustick{q_5}\qw
      \gategroup{2}{3}{6}{8}{3.3em}{--}
      \gategroup{2}{16}{4}{20}{3.3em}{--}
      \gategroup{2}{23}{5}{24}{4.5em}{^\}}
      \gategroup{2}{11}{5}{13}{4.5em}{^\}} }}
  \]
  \caption{A representation of $\Append(C,D,b)$.}
  \label{rep_unencap}
  \rule{\textwidth}{0.1mm}
\end{figure}
% ....................................................................

We note that the axiomatization of
Definition~\ref{circuit_constructor} does not mention the concept of a
{\em gate}. Indeed, any gate is a circuit, and thus a member of the
set $S$; conversely, any circuit can be used as a gate. In
Proto-Quipper, we simply assume that certain members of $S$ are
available as pre-defined constants, serving as ``elementary'' gates.
The operation of appending a gate to a circuit is subsumed by the more
general operation of composing circuits.

Proto-Quipper's quantum variables and circuit constants are supposed
to be the syntactic representatives of a circuit constructor's wire
identifiers. This idea is formalized in the following definition.

\begin{definition}
  \label{def:pq-adequate-cc}
  A circuit constructor $\p{Q,S}$ is \emph{adequate} if it can be
  equipped with bijections $\Wire\from \qset \to Q$ and $\Name\from
  \cset \to S$ such that:
  \[\In=\Wire' \circ \Inn \circ \Name
  \quad \mbox{and} \quad \Out=\Wire'\circ\Outt\circ \Name
  \]
  where $\Wire'$ denotes the lifting of $\Wire^{-1}$ from $\qset$ to
  $\pset(\qset)$.
\end{definition}

\begin{remark}
  \label{structure-transfer}
  The existence of the bijections $\Wire$ and $\Name$ has the
  following consequences:
  \begin{itemize}
  \item $\cset$ can be equipped with an involution:
    \[
    (.)^{-1} = \Name\circ \Rev \circ \Name^{-1} \from \cset\to\cset
    \]
    such that $\In(C^{-1})=\Out(C)$ and $\Out(C^{-1})=\In(C)$.
  \item If $t$ and $u$ are quantum data terms such that
    $\bind(t,u)=\binding$, then we can define $b = \Wire\circ
    \binding\circ \Wire^{-1} \in \finbij(Q)$.
  \end{itemize}
  From now on, we always assume an adequate circuit
  constructor. Moreover, we work under the simplifying assumptions
  that $\qset=Q$, $\cset=S$, $\Wire=1_Q$, and $\Name=1_S$. This
  notably implies that $\In = \Inn$ and $\Out = \Outt$.
\end{remark}

We are now in a position to define Proto-Quipper's operational
semantics.

\begin{definition}
  \label{def:pq-closure}
  Let $\p{Q,S}$ be an adequate circuit constructor. A \emph{closure}
  is a pair $[C,a]$ where $C\in S$, $a$ is a term and
  $\FQ(a)\seq\Outt(C)$.
\end{definition}

\begin{definition}
  \label{def:pq-beta-red}
  The \emph{one-step reduction relation}, written $\to$, is defined on
  closures by the rules given in Figure~\ref{fig:pq-red-rules}. The
  \emph{reduction relation}, written $\to^*$, is defined to be the
  reflexive and transitive closure of $\to$.
\end{definition}

% ....................................................................
\begin{figure}
  \[
  \infer[\rul{fun}]{[C,ab]\to[C',a'b]}{ [C,a]\to[C',a'] } \quad
  \infer[\rul{arg}]{[C,vb]\to [C',vb']}{ [C,b]\to [C',b'] }
  \]
  \[
  \infer[\rul{right}]{[C,\p{a,b}]\to [C',\p{a,b'}]}{ [C,b]\to [C',b']
  } \quad \infer[\rul{left}]{[C,\p{a,v}]\to [C',\p{a',v}]}{[C,a]\to
    [C',a']}
  \]
  \[
  \infer[\rul{let*}]{[C,\letin{*}{a}{b}]\to [C', \letin{*}{a'}{b}]}{
    [C,a]\to [C',a'] }
  \]
  \[
  \infer[\rul{let}]{[C,\letin{\p{x,y}}{a}{b}]\to [C',
    \letin{\p{x,y}}{a'}{b}]}{ [C,a]\to [C',a'] }
  \]
  \[
  \infer[\rul{cond}]{[C, \Ifthenelse{a}{b}{c}]\to [C',
    \Ifthenelse{a'}{b}{c}]}{ [C,a]\to [C',a'] }
  \]
  \[
  \infer[\rul{circ}]{[C, (t,D,a)]\to [C, (t,D',a')]}{ [D,a]\to [D',a']
  }
  \]
  ~
  \[
  \infer[\rul{\beta}]{[C,(\lambda x.a)v]\to [C, a[v/x]]}{}
  \]
  \[
  \infer[\rul{unit}]{[C,\letin{*}{*}{a}]\to[C,a]}{}
  \]
  \[
  \infer[\rul{pair}]{[C,\letin{\p{x,y}}{\p{v,w}}{a}]
    \to[C,a[v/x,w/y]]}{}
  \]
  \[
  \infer[\rul{ifF}]{[C,\Ifthenelse{\false}{a}{b}] \to [C,b]}{}
  \]
  \[
  \infer[\rul{ifT}]{[C,\Ifthenelse{\true}{a}{b}] \to [C,a]}{}
  \]
  ~
  \[
  \infer[(\boxx)]{[C,\boxx^T(v)]\to [C,(t,D,vt)]}{ \spec_{\FQ(v)}(T)=t
    & \mathtt{new}(\FQ(t))=D }
  \]
  \[
  \infer[(\unbox)]{[C,(\unbox\,(u,D,u'))v]\to [C',\binding'(u')]}{
    \bind(v,u)=\binding & \mathtt{Append}(C,D,\binding) =
    (C',\binding') & \FQ(u')\seq\dom(\binding') }
  \]
  \[
  \infer[(\rev)]{[C,\rev\,(t,C,t')]\to [C,(t',C^{-1},t)]}{}
  \]
  \label{fig:pq-red-rules}
  \caption{Reduction rules for Proto-Quipper.}
  % \rule{\textwidth}{0.1mm}
\end{figure}
% ....................................................................

The rules are separated in three groups. The first group and second
group contain the \emph{congruence rules} and \emph{classical rules}
respectively.  Except for the rule for $(t,D,a)$, these rules are
standard. They describe a call-by-value reduction strategy. The
circuit generating rule for $\rev(t,C,t')$ rule is straightforward. We
briefly discuss the remaining rules.

The rule for $\boxx^T(v)$ is to be understood as follows. To reduce a
closure of the form $[C, \boxx^T(v)]$, start by generating a specimen
of type $T$. Then apply the function $v$ on the input $t$ in the
context of an empty circuit of the appropriate arity. By the
congruence rule for $(t,D,a)$, this computation will continue until a
value is reached, i.e., a term of the form $(t,D,t')$. Note that while
this computation is taking place, the state $C$ is not
accessible. When a value of the form $(t,D,t')$ is reached, the
construction of $C$ can resume. Note that it was necessary to know the
type $T$ in order to generate the appropriate specimen. This explains
the choice of a Church-style typing of the $\boxx$ operator.

The rule for $(\unbox\,(u,D,u'))v$ will first generate a binding from
$v$ and the input $u$ of $D$. Then, it will compose $C$ and $D$ along
that binding and update the names of the wire identifiers appearing in
$u'$ according to $\binding'$.

The recursive nature of the reduction rules explains why closures are
not required to satisfy $\FQ(a)=\Outt(C)$. The requirement that
$\FQ(a)\seq \Outt(C)$ is justified by the idea that a term should not
affect a wire outside of $C$. But if we also asked for the opposite
inclusion, it would not be possible to define a recursive reduction in
a straightforward way. For example, the reduction of a pair is done
component-wise: to reduce $\p{a,b}$ one first reduces $b$. The
simplest way to express this in terms of closures is to carry the
whole circuit state along. This implies that if both $a$ and $b$
contain wire identifiers, then the equality $\FQ(a)=\Outt(C)$ cannot
be satisfied.

Unlike in the quantum lambda calculus, Proto-Quipper's reduction is
not probabilistic, in the sense that the right member of any reduction
rule is a unique closure. The following proposition establishes that
Proto-Quipper's reduction is moreover \emph{deterministic}.

\begin{proposition}
  \label{determinicity}
  If $[C,a]$ is a closure, then at most one reduction rule applies to
  it.
\end{proposition}

\begin{proof}
  By case distinction on $a$.
\end{proof}

To close this chapter, we illustrate the reduction of Proto-Quipper
with an example. Assume that we are given the basic circuit generating
functions $\Hgate$, $\Sgate$, and $\CNOTgate$ of
Examples~\ref{ex:pq-terms} and let $F$ be the following term
\[
F=\lambda z.(\letin{\p{x,y}}{z}{\CNOTgate\p{\Hgate x, \Sgate y}})
\]
Since $F$ can be given the type $\qubit \x \qubit \loli \qubit \x
\qubit$, we can use $\boxx^{\qubit \x \qubit}$ to turn $F$ into a
Proto-Quipper circuit. Now consider the closure
\begin{equation} 
  \label{eq:ex-red1}
  [-, \boxx^{\qubit \x \qubit}F]
\end{equation}
where $-$ is any circuit state. The $\rul{\boxx}$ rule applies, so
that a specimen of type $\qubit \x \qubit$ is created, say
$\p{q_1,q_2}$, and (\ref{eq:ex-red1}) reduces to
\begin{equation*}
  [-,(\p{q_1,q_2},C,F\p{q_1,q_2})]
\end{equation*} 
where $C=\mathtt{new}(\s{q_1,q_2})$ is the empty circuit on
$\s{q_1,q_2}$. Since $F\p{q_1,q_2}$ is not a value, the $\rul{circ}$
rule applies. This means that we consider the closure
\begin{equation*}
  [C,F\p{q_1,q_2}].
\end{equation*} 
We now repeatedly consider reducts of this closure until a value is
reached. For clarity, we represent circuit states as circuits. In two
classical reductions we reach the closure
\begin{equation*}
  \mp{.8}{\Qcircuit @C=1em @R=1.7em { & \lstick{q_1} & \qw &\qw & 
      \qw & \qw & \rstick{q_1} \qw \\
      & \lstick{q_2} & \qw & \qw & \qw & \qw & \rstick{q_2} \qw }} 
  \quad ~ \quad ~ \quad ~ \quad \CNOTgate\p{\Hgate q_1, \Sgate q_2}.
\end{equation*}
Following Proto-Quipper's reduction strategy, the right argument is
reduced first, yielding
\begin{equation*}
  \mp{.8}{\Qcircuit @C=1em @R=.7em { & \lstick{q_1} & \qw &
      \qw & \qw & \rstick{q_1}\qw \\
      & \lstick{q_2} & \qw & \gate{S} & \qw &\rstick{q_2}\qw }}
  \quad ~ \quad ~ \quad ~ \quad \CNOTgate\p{\Hgate q_1, q_2}.
\end{equation*}
Here we assumed for simplicity that the output wire of the $\Sgate$
was not renamed. Since $q_1$ is a value, we now reduce $\Sgate
q_2$. This yields
\begin{equation*}
  \mp{.8}{\Qcircuit @C=1em @R=.7em { & \lstick{q_1} & \qw &
      \gate{H} & \qw & \rstick{q_1}\qw \\
      & \lstick{q_2} & \qw & \gate{S} & \qw &\rstick{q_2}\qw }}
  \quad ~ \quad ~ \quad ~ \quad \CNOTgate\p{q_1, q_2}.
\end{equation*}
Finally, the $CNOT$ gate is applied
\begin{equation}
  \label{eq:ex-red7}
  \mp{.8}{\Qcircuit @C=1em @R=.7em { & \lstick{q_1} & \gate{H} &
      \targ &  \rstick{q_1} \qw \\
      & \lstick{q_2} & \gate{S} & \ctrl{-1} & \rstick{q_2}\qw }}
  \quad ~ \quad ~ \quad ~ \quad \p{q_1, q_2}.
\end{equation}
Since $\p{q_1,q_2}$ is a value, the execution is finished. The final
circuit is now returned in the form of a term of the language, e.g.,
as
\begin{equation*}
  [-,(\p{q_1,q_2}, D, \p{q_1,q_2})]
\end{equation*}
where $D$ is the constant representing the circuit on the left hand
side of (\ref{eq:ex-red7}).

% ---------------------------------------------------------------------
\chapter{Type-safety of Proto-Quipper}
\label{chap:pq-safe}

In this chapter, we establish that Proto-Quipper is a \emph{type safe}
language. As discussed in Section~\ref{ssect:st-strong-norm}, type
safety is established by proving that the language enjoys the subject
reduction and progress properties.

% ====================================================================
\section{Properties of the type system}
\label{sect:pq-ppties-type-syst}

Before proving the subject reduction and progress, we record some
properties of the type system, including the technical but important
\emph{Substitution Lemma}. Note that the typing rules enforce a
\emph{strict} linearity on variables and quantum variables. In
particular, if a quantum variable appears in the quantum context of a
valid typing judgement for a term $a$, then it must belong to the free
quantum variables of $a$.

\begin{lemma}~
  \label{prop:type_syst}
  \begin{enumerate}
  \item If $\Gamma; Q \entails a:A$ is valid, then
    $Q=\FQ(a)$.\label{q_context}
  \item If $\Gamma,x:B;Q \entails a:A$ is valid, and $x\notin \FV(a)$,
    then $B=\bang B'$ and $\Gamma;Q \entails a:A$ is
    valid.\label{unused_var}
  \item If $\Gamma; Q \entails a:A$ is valid, then $\Gamma, \bang
    \Delta ; Q \entails a:A$ is valid.\label{weakening}
  \item If $\Gamma ; Q \entails a:A$ is valid, $\Delta <: \Gamma$ and
    $A<:B$, then $\Delta ; Q \entails a:B$ is valid.\label{subtype}
  \end{enumerate}
\end{lemma}

\begin{proof}
  By induction on the corresponding typing derivation.
\end{proof}

\begin{lemma}
  \label{specimen}
  If $T$ is a quantum data type and $X$ is a finite subset of $\qset$,
  then $\FQ(\spec_X(T))\entails \spec_X(T):T$ is valid.
\end{lemma}

\begin{proof}
  We prove the Lemma by induction on $T$.
  \begin{itemize}
  \item If $T=1$, then $\spec_X(T)=*$ and we can use the $\rul{*_i}$
    rule.
  \item If $T=\qubit$, then $\spec_X(T)=q$ for some quantum variable
    $q$ and we can use the $\rul{ax_q}$ rule.
  \item If $T=T_1\X T_2$, then $\spec_X(T)=\p{t_1,t_2}$ where
    $t_1=\spec_X(T_1)$ and $u=\spec_{X\cup \FQ(t_1)}(T_2)$. By the
    induction hypothesis, both $\FQ(t_1)\entails t_1:T_1$ and
    $\FQ(t_2)\entails t_2:T_2$ are valid typing judgements. We can
    therefore conclude by applying the $\rul{\X_i}$ rule.\qedhere
  \end{itemize}
\end{proof}

\begin{lemma}
  \label{binding_judgement}
  If $\Gamma;Q \entails a:A$ is valid and $\binding$ is a binding such
  that $FQ(a)\seq\dom(\binding)$ then $\Gamma;\binding (Q) \entails
  \binding(a):A$ is valid.
\end{lemma}

\begin{proof}
  By induction on the typing derivation of $\Gamma;Q \entails a:A$.
\end{proof}

\begin{lemma}
  \label{context_value}
  If $v\in\Val$ and $\Gamma;Q \entails v:\bang A$ is valid, then
  $Q=\emptyset$ and $\Gamma=\bang \Delta$ for some $\Delta$.
\end{lemma}

\begin{proof}
  By induction on the typing derivation of $\Gamma;Q \entails v:\bang
  A$. In the case of $\rul{ax_c}$, use
  Lemma~\rref{subtyping_shape}{subtype_bang}.
\end{proof}

\begin{lemma}
  \label{non_values}
  If a term $a$ is not a value then it is of one of the following
  forms
  \begin{itemize}
  \item $(t,C,a')$ with $a'\notin \Val$,
  \item $\p{a_1,a_2}$ with $a_1\notin \Val$ or $a_2\notin \Val$,
  \item $\Ifthenelse{a_1}{a_2}{a_3}$,
  \item $\letin{*}{a_1}{a_2}$,
  \item $\letin{\p{x,y}}{a_1}{a_2}$, or
  \item $a_1a_2$ with $a_1 \neq \unbox$ or $a_2\notin \Val$.
  \end{itemize}
\end{lemma}

\begin{proof}
  By definition of terms and values.
\end{proof}

\begin{lemma}
  \label{form_values}
  A well-typed value $v$ is either a variable, a quantum variable, a
  constant or one of the following case occurs
  \begin{itemize}
  \item if it is of type $\bang^n\Circ(T,U)$, it is of the form
    $(t,C,u)$ with $t$ and $u$ values,
  \item if it is of type $\bang^n\bool$ it is either $\true$ or
    $\false$,
  \item if it is of type $\bang^n (A\X B)$, it is of the form
    $\p{w,w'}$, with $w$ and $w'$ values and
    $\FQ(w)\cap\FQ(w')=\emptyset$,
  \item if it is of type $\bang^n1$, it is precisely the term $*$, or
  \item if it is of type $\bang^n (A\loli B)$, it is a lambda
    abstraction, a constant, or of the form $\unbox\,(t,C,u)$.
  \end{itemize}
\end{lemma}

\begin{proof}
  By induction on the typing derivation of $v$.
\end{proof}

\begin{corollary}
  \label{typed_qd_term}
  If $T$ is a quantum data type and $v$ is a well-typed value of type
  $T$ then $v$ is a quantum data term.
\end{corollary}

\begin{lemma}
  \label{bind_qd_terms}
  If $T$ is a quantum data type and $v_1$, $v_2$ are well-typed values
  of type $T$, then $\binding=\bind(v_1,v_2)$ is a well-defined
  binding, $\dom(\binding)=\FQ(v_1)$, and $\cod(\binding)=\FQ(v_2)$.
\end{lemma}

\begin{proof}
  By Corollary~\ref{typed_qd_term}, we know that $v_1$ and $v_2$ are
  quantum data terms so that the statement of the lemma makes
  sense. The proof then proceeds by induction on $T$, using
  Lemma~\ref{form_values}.
\end{proof}

\begin{lemma}[Substitution]
  \label{substitution}
  If $v\in\Val$ and both $\Gamma',\bang \Delta;Q' \entails v:B$ and
  $\Gamma,\bang \Delta,x:B;Q \entails a:A$ are valid typing
  judgements, then $\Gamma,\Gamma',\bang \Delta;Q,Q' \entails
  a[v/x]:A$ is also valid.
\end{lemma}

\begin{proof}
  Let $\pi_1$ and $\pi_2$ be the typing derivations of $\Gamma,\bang
  \Delta,x:B;Q \entails a:A$ and $\Gamma',\bang \Delta;Q' \entails
  v:B$ respectively. We prove the Lemma by induction on $\pi_1$.
  \begin{itemize}
  \item If the last rule of $\pi_1$ is $\rul{ax_c}$ and $a=x$, then
    $\pi_1$ is
    \[
    \infer[\rul{ax_c}]{\bang \Delta, x:B;\emptyset\entails x:A}{ B<:A
    }
    \]
    with $\Gamma=Q=\emptyset$. Then $a[v/x]=v$ and can conclude by
    applying Lemma~\rref{prop:type_syst}{subtype} to $\pi_2$.
  \item If the last rule of $\pi_1$ is $\rul{ax_c}$ and $a=y\neq x$,
    then $\pi_1$ is
    \[
    \infer[\rul{ax_c}]{\bang \Delta, x:\bang B',y:A';\emptyset\entails
      y:A}{ A'<:A }
    \]
    with $B=\bang B'$, $Q=\emptyset$ and $\Gamma =\s{y:A'}$ or
    $\Gamma=\emptyset$ depending on whether or not $A'$ is
    duplicable. Therefore $v$ is a value of type $\bang B'$ and by
    Lemma~\ref{context_value}, we know that
    $\Gamma'=Q'=\emptyset$. Since $a[v/x]=y$ and $x\notin \FV(y)$ we
    can conclude by applying Lemma~\rref{prop:type_syst}{unused_var}
    to $\pi_1$.
  \item If the last rule of $\pi_1$ is one of $\rul{ax_q}$,
    $\rul{cst}$, $\rul{*_i}$, $\rul{\top}$ and $\rul{\bot}$, and $a$
    is the corresponding constant, then $x\notin \FV(a)$ and $x$ must
    be declared of some type $\bang B'$.  We can therefore reason as
    in the previous case.
  \item If the last rule of $\pi_1$ is $\rul{\lambda_1}$ and
    $a=\lambda y.b$, then $\pi_1$ is
    \[
    \infer[\rul{\lambda_1}]{\Gamma,\bang \Delta,x:B;Q\entails \lambda
      y.b:A_1\loli A_2}{ \deduce[]{\Gamma, \bang \Delta,x:B,y:A_1;Q
        \entails b:A_2}{ \vdots } }
    \]
    with $A=A_1\loli A_2$. By the induction hypothesis, $\Gamma,
    \Gamma',\bang \Delta,y:A_1;Q,Q' \entails b[v/x]:A_2$ is valid and
    we can conclude by applying $\rul{\lambda_1}$.
  \item If the last rule of $\pi_1$ is $\rul{\lambda_2}$ and
    $a=\lambda y.b$, then $\pi_1$ is
    \[
    \infer[\rul{\lambda_2}]{\bang \Delta, x:\bang B';\emptyset
      \entails \lambda y.b:~\bang^{n+1}(A_1\loli A_2)}{
      \deduce[]{\bang \Delta, x:\bang B',y:A_1;\emptyset \entails
        b:A_2}{ \vdots } }
    \]
    with $A=\bang^{n+1}(A_1\loli A_2)$ and $B=\bang B'$. Hence $v$ is
    a value of type $\bang B'$ and by Lemma~\ref{context_value}, we
    know that $\Gamma'=Q'=\emptyset$. The induction hypothesis
    therefore implies that $\bang \Delta,y:A_1;\emptyset \entails
    b[v/x]:A_2$ is valid and we can conclude by applying
    $\rul{\lambda_2}$.
  \item If the last rule of $\pi_1$ is $\rul{app}$, and $a=ca'$, then
    $\pi_1$ can be of one of three forms depending on $B$. If $B$ is
    duplicable, then $\pi_1$ is
    \[
    \infer[\rul{app}]{\Gamma_1,\Gamma_2,\bang \Delta,x:\bang B';
      Q_1,Q_2\entails ca':A}{ \deduce[]{\Gamma_1, x:\bang B',\bang
        \Delta;Q_1\entails c:A'\loli A}{ \vdots } &
      \deduce[]{\Gamma_2, x:\bang B',\bang \Delta ;Q_2\entails a':A'
      }{ \vdots } }
    \]
    with $B=\bang B'$. Using Lemma~\ref{context_value} again, we know
    that $\Gamma'=Q'=\emptyset$. The induction hypothesis therefore
    implies that $\Gamma_1,\bang \Delta;Q_1\entails c[v/x]:A'\loli A$
    and $\Gamma_2,\bang \Delta ;Q_2\entails a'[v/x]:A'$ are valid and
    we can conclude by applying $\rul{app}$. If, instead, $B$ is
    non-duplicable, then the declaration $x:B$ can only appear in one
    branch of the derivation. This means that $\pi_1$ is either
    \[
    \infer[\rul{app}]{\Gamma_1,\Gamma_2,\bang
      \Delta,x:B;Q_1,Q_2\entails ca':A}{ \deduce[]{\Gamma_1, x:B,\bang
        \Delta;Q_1\entails c:A'\loli A}{ \vdots } &
      \deduce[]{\Gamma_2, \bang \Delta ;Q_2\entails a':A'}{ \vdots } }
    \]
    or
    \[
    \infer[\rul{app}.]{\Gamma_1,\Gamma_2,\bang
      \Delta,x:B;Q_1,Q_2\entails ca':A}{ \deduce[]{\Gamma_1, \bang
        \Delta;Q_1\entails c:A'\loli A}{ \vdots } &
      \deduce[]{\Gamma_2, x:B,\bang \Delta ;Q_2\entails a':A'}{ \vdots
      } }
    \]
    In the first case, the induction hypothesis implies that
    $\Gamma_1,\Gamma' \bang \Delta;Q_1,Q'\entails c[v/x]:A'\loli A$ is
    valid and we can conclude by $\rul{app}$. The second case is
    treated analogously.
  \item If the last rule of $\pi_1$ is one of $\rul{\X_i}$,
    $\rul{\X_e}$, $\rul{*_e}$ and $\rul{if}$, and $a$ is the
    corresponding term, then we can reason as above by considering in
    turn the case where $B$ is duplicable and the case where $B$ is
    non-duplicable.
  \item If the last rule of $\pi_1$ is $\rul{circ}$, and $a=(t,C,a')$,
    then $\pi_1$ is
    \[
    \infer[\rul{circ}]{\bang \Delta,x:\bang B';\emptyset \entails
      (t,C,a'):\bang^n\Circ(T,U)}{ \deduce[]{Q_1\entails t:T}{ \vdots
      } & \deduce[]{\bang \Delta,x:\bang B' ; Q_2\entails a':U}{
        \vdots } & \Inn(C)=Q_1 & \Outt(C)=Q_2 }
    \]
    with $A=\bang^n\Circ(T,U)$ and $B=\bang B'$ for some types $T$,
    $U$ and $B'$. Using Lemma~\ref{context_value} again, we know that
    $\Gamma'=Q'=\emptyset$. The induction hypothesis therefore implies
    that $\bang \Gamma; Q_2\entails a'[v/x]:U$ is valid and we can
    conclude by applying $\rul{circ}$. \qedhere
  \end{itemize}
\end{proof}

% ====================================================================
\section{Subject reduction}
\label{sect:pq-subj-red}

We now prove that Proto-Quipper enjoys the subject reduction
property. Since the reduction relation is defined on closures but the
typing rules apply to terms, we start by extending the notions of
typing judgement and validity to closures.

\begin{definition}
  \label{def:pq-typed-closure}
  A \emph{typed closure} is an expression of the form:
  \[
  \Gamma;Q\entails [C,a]:A,(Q'|Q'').
  \]
  It is \emph{valid} if $\Inn(C)=Q'$ and $\Outt(C)=Q,Q''$, and
  $\Gamma;Q\entails a:A$ is a valid typing judgement.
\end{definition}

\begin{lemma}
  \label{Inwires}
  If $[C,a]\to[C',a']$ then $\Inn(C)=\Inn(C')$.
\end{lemma}

\begin{proof}
  By induction on the derivation of $[C,a]\to[C',a']$. In all but the
  $(\unbox)$ case, the result follows either from the induction
  hypothesis or from the fact that $C=C'$. In the $(\unbox)$ case, use
  Definition~\rref{circuit_constructor}{Append_cond_2b}.
\end{proof}

\begin{theorem}[Subject reduction]
  \label{thm-subject-red}
  If $\Gamma;\FQ(a)\entails [C,a]:A,(Q'|Q'')$ is a valid typed closure
  and $[C,a]\to [C',a']$, then $\Gamma;\FQ(a')\entails
  [C',a']:A,(Q'|Q'')$ is a valid typed closure.
\end{theorem}

\begin{proof}
  We prove the theorem by induction on the derivation of the reduction
  $[C,a]\to[C',a']$. In each case, we start by reconstructing the
  unique typing derivation $\pi$ of $\Gamma;\FQ(a)\entails a:A$ and we
  use it to prove that $\Gamma;\FQ(a')\entails [C',a']:A,(Q'|Q'')$ is
  valid. By Lemma~\ref{Inwires} we never need to verify that
  $\Inn(C')=Q'$ so that we only need to show:
  \begin{itemize}
  \item $\Outt(C')=\FQ(a'),Q''$ and
  \item $\Gamma;\FQ(a')\entails a':A$ is valid.
  \end{itemize}
  Throughout the proof, we write $\IH(\pi)$ to denote the proof
  obtained by applying the induction hypothesis to $\pi$.

  \begin{description}
  \item[Congruence rules:] These rules are treated uniformly. We
    illustrate the $\rul{fun}$ and $\rul{circ}$ cases.
    \begin{itemize}
    \item $\rul{fun}$: the reduction rule is
      \[
      \infer[]{[C,cb]\to[C',c'b]}{ [C,c]\to[C',c'] }
      \]
      with $a=cb$ and $a'=c'b$. The typing derivation $\pi$ is
      therefore
      \[
      \infer[]{\Gamma_1,\Gamma_2, \bang \Delta;\FQ(c),\FQ(b)\entails
        cb:A}{ \deduce[]{\Gamma_1, \bang \Delta;\FQ(c)\entails
          c:B\loli A}{ \vdots~\pi_1 } & \deduce[]{\Gamma_2, \bang
          \Delta ;\FQ(b)\entails b:B}{ \vdots~\pi_2 } }
      \]
      and $\Gamma_1,\Gamma_2;\FQ(c),\FQ(b)\entails [C, cb],(Q'|Q'')$
      is valid. It follows that
      \[
      \Gamma_1, \bang \Delta; \FQ(c) \entails [C,c]:B\loli
      A,(Q'|\FQ(b),Q'')
      \]
      is valid and, by the induction hypothesis, this implies that
      $\Gamma_1, \bang \Delta;\FQ(c')\entails [C',c']:B\loli
      A,(Q'|\FQ(b),Q'')$ is also valid.  In particular, it follows
      that $\Outt(C')=\FQ(c'),\FQ(b),Q''$. This, together with the
      following typing derivation,
      \[
      \infer[]{\Gamma_1,\Gamma_2, \bang \Delta;\FQ(c'),\FQ(b)\entails
        c'b:A}{ \deduce[]{\Gamma_1, \bang \Delta;\FQ(c')\entails
          c':B\loli A}{ \vdots~\IH(\pi_1) } & \deduce[]{\Gamma_2,
          \bang \Delta ;\FQ(b)\entails b:B}{ \vdots~\pi_2 } }
      \]
      shows that $\Gamma_1,\Gamma_2, \bang
      \Delta;\FQ(c'),\FQ(b)\entails [C',c'b]:A,(Q',Q'')$ is valid.
    \item $\rul{circ}$: the reduction rule is
      \[
      \infer[\rul{circ}]{[C, (t,D,b)]\to [C, (t,D',b')]}{ [D,b]\to
        [D',b'] }
      \]
      with $a=(t,D,b)$ and $a'=(t,D',b')$. The typing derivation $\pi$
      is therefore
      \[
      \infer[]{\bang \Delta;\emptyset\entails
        (t,D,b):\bang^n\Circ(T,U)}{ \deduce[]{\FQ(t)\entails t:T}{
          \vdots~\pi_1 } & \deduce[]{\bang \Delta ; \FQ(b)\entails
          b:U}{ \vdots~\pi_2 } & \deduce[]{\Inn(D)=\FQ(t)}{
          \Outt(D)=\FQ(b) } }
      \]
      and $\bang \Delta ; \emptyset \entails
      [C,(t,D,b)]:\bang^n\Circ(T,U),(Q'|Q'')$ is valid.  Disregarding
      $\pi_1$, it follows from the assumptions in the above rule that
      $\bang \Delta ; \FQ(b)\entails [D,b]:U, (\FQ(t)|\emptyset)$ is
      valid and, by the induction hypothesis, this implies that $\bang
      \Delta,\FQ(b')\entails [D',b']:U,(\FQ(t)|\emptyset)$ is also
      valid.  This, together with the following typing derivation,
      \[
      \infer[.]{\bang \Delta;\emptyset\entails
        (t,D,b'):\bang^n\Circ(T,U)}{ \deduce[]{\FQ(t)\entails t:T}{
          \vdots~\pi_1 } & \deduce[]{\bang \Delta ; \FQ(b')\entails
          b':U}{ \vdots~\IH(\pi_2) } & \deduce[]{\Inn(D')=\FQ(t)}{
          \Outt(D')=\FQ(b') } }
      \]
      shows that $\bang \Delta;\emptyset\entails [C,(t,D',b')]
      :\bang^n\Circ(T,U),(Q'|Q'')$ is valid.
    \end{itemize}
  \item[Classical rules:] These rules are also treated uniformly, we
    illustrate the $\rul{\beta}$ case.
    \begin{itemize}
    \item $\rul{\beta}$: the reduction rule is
      \[
      \infer[]{[C,(\lambda x.b)v]\to [C, b[v/x]]}{}
      \]
      with $a=(\lambda x.b)v$ and $a'=b[v/x]$. The typing derivation
      $\pi$ is therefore
      \[
      \infer[]{\Gamma_1,\Gamma_2,\bang \Delta;\FQ(b),\FQ(v)\entails
        (\lambda x.b)v:A}{ \infer[]{\Gamma_1,\bang
          \Delta;\FQ(b)\entails \lambda x.b:B\loli A}{
          \deduce[]{\Gamma_1,\bang \Delta,x:B;\FQ(b) \entails b:A}{
            \vdots~\pi_1 } } & \deduce[]{\Gamma_2,\bang
          \Delta;\FQ(v)\entails v:B}{ \vdots~\pi_2 } }
      \]
      and $\Gamma_1,\Gamma_2,\bang \Delta;\FQ(b),\FQ(v)\entails
      [C,(\lambda x.b)v]:A,(Q'|Q'')$ is valid. We then know, by
      Lemma~\ref{substitution}, that $\Gamma_1,\Gamma_2,\bang
      \Delta;\FQ(b),\FQ(v)\entails b[v/x]:A$ is a valid typing
      judgement which implies that
      \[
      \Gamma_1,\Gamma_2,\bang \Delta;\FQ(b),\FQ(v)\entails
      [C,b[v/x]]:A,(Q'|Q'')
      \]
      is a valid typed closure.
    \end{itemize}
  \item[Circuit generating rules:] These rules represent the most
    interesting cases. We treat them individually.
    \begin{itemize}
    \item $\rul{\boxx}$: the reduction rule is
      \[
      \infer[]{[C,\boxx^T(v)]\to [C,(t,D,vt)]}{ \spec(T)=t &
        \New(\FQ(t))=D }
      \]
      with $a=\boxx^T(v)$ and $a'=(t,D,vt)$. Since $v$ is a value, we
      know by Lemma~\ref{context_value} that the typing derivation
      $\pi$ is
      \[
      \infer[]{\bang \Delta;\emptyset\entails
        \boxx^T(v):\bang^n\Circ(T,U)}{ \infer[]{\bang \Delta;\emptyset
          \entails \boxx^T:\bang (T\loli U)\loli \bang^n\Circ(T,U)}{ }
        & \deduce[]{\bang \Delta ;\emptyset\entails v:\bang (T\loli
          U)}{ \vdots ~\pi_1 } }
      \]
      and $\bang \Delta;\emptyset\entails
      [C,\boxx^T(v)]:\bang^n\Circ(T,U),(Q'|Q'')$ is valid. There
      exists a typing derivation $\pi_2$ of $\FQ(t)\entails t:T$, by
      Lemma~\ref{specimen}. Applying
      Lemma~\rref{prop:type_syst}{subtype} to $\pi_1$ we get a
      derivation $\pi_1'$ of $\bang \Delta ;\emptyset\entails v:T\loli
      U$. We can therefore construct the following derivation $\tau$:
      \[
      \infer[]{\bang \Delta;\FQ(t)\entails vt : U}{ \deduce[]{\bang
          \Delta ;\emptyset\entails v:T\loli U}{ \vdots ~\pi_1' } &
        \deduce[]{\bang \Delta;\FQ(t)\entails t:T}{ \vdots ~\pi_2' } }
      \]
      where $\pi_2'$ is obtained from $\pi_2$ by
      Lemma~\rref{prop:type_syst}{weakening}. Moreover, since
      $\FQ(vt)=\FQ(t)=\Outt(D)=\Inn(D)$, we have:
      \[
      \infer[.]{\bang \Delta;\emptyset \entails
        (t,D,vt):\bang^n\Circ(T,U)}{ \deduce[]{\FQ(t)\entails t:T}{
          \vdots~\pi_2 } & \deduce[]{\bang \Delta ; \FQ(vt)\entails
          vt:U}{ \vdots~\tau } & \deduce[]{\Inn(D)=\FQ(t) }{
          \Outt(D)=\FQ(vt) } }
      \]
      Hence $\bang \Delta ;\emptyset \entails
      [C,(t,D,vt)]:\bang^n\Circ(T,U) (Q'|Q'')$ is a valid typed
      closure.
    \item $\rul{\unbox}$: the reduction rule is
      \[
      \infer[]{[C,(\unbox\,(u,D,u'))v]\to [C',\binding'(u')]}{
        \bind(v,u)=\binding & \Append(C,D,\binding) = (C',\binding') &
        \FQ(u')\seq\dom(\binding') }
      \]
      with $a=(\unbox\,(u,D,u'))v$ and $a'=\binding'(u')$.  To
      reconstruct the typing derivation $\pi$, first note that we have
      the following derivation $\pi_1$ of $\bang
      \Delta;\emptyset\entails \unbox\,(u,D,u'):T\loli U$
      \begin{footnotesize}
        \[
        \infer[]{\bang \Delta;\emptyset\entails
          \unbox\,(u,D,u'):T\loli U}{ \infer[]{\bang \Delta;\emptyset
            \entails \unbox:\Circ(T,U)\loli (T\loli U)}{ } &
          \infer[]{\bang \Delta ;\emptyset\entails
            (u,D,u'):\Circ(T,U)}{ \deduce[]{\FQ(u)\entails u:T}{
              \vdots ~\pi_1^1 } & \deduce[]{\bang
              \Delta;\FQ(u')\entails u':U}{ \vdots ~\pi_1^2 } } }
        \]
      \end{footnotesize}
      with $\Inn(D)=\FQ(u)$, $\Outt(D)=\FQ(u')$. We can then use
      $\pi_1$ to rebuild $\pi$ as follows:
      \begin{footnotesize}
        \[
        \infer[]{\bang \Delta; \FQ(v)\entails (\unbox\,(u,D,u'))v :U}{
          \infer[]{\bang \Delta;\emptyset\entails
            \unbox\,(u,D,u'):T\loli U}{ \vdots~\pi_1 } &
          \infer[]{\bang \Delta ; \FQ(v)\entails v:T}{ \vdots~\pi_2 }
        }
        \]
      \end{footnotesize}
      and the typed closure
      \[
      \bang \Delta; \FQ(v) \entails [C,(\unbox\,(u,D,u'))v]
      :U,(Q'|Q'')
      \]
      is valid. In the conclusion of $\pi_2$, all the term variables
      are declared of a duplicable type. This follows from Corollary
      \ref{typed_qd_term} and
      Lemma~\rref{prop:type_syst}{unused_var}. By assumption, we know
      that $\FQ(u')\seq\dom(\binding')$. We can therefore apply
      Lemma~\ref{binding_judgement} to $\pi_1^2$ to get a typing
      derivation $\tau$ of
      \[
      \bang \Delta;\FQ(\binding'(u'))\entails \binding(u'):U.
      \]
      Now by Definition~\rref{circuit_constructor}{Append_cond_3} we
      have:
      \[
      \begin{array}{rcl}
        \Outt(C') & = & \binding'(\Outt(D))\cupdot 
        (\Outt(C)\setminus\binding^{-1}(\Inn(D))) \\
        & = & \binding'(\FQ(u')) \cupdot 
        ((Q''\cupdot\FQ(v))\setminus \binding^{-1}(\FQ(u))) \\
        & = & \FQ(\binding'(u')) \cupdot 
        ((Q''\cupdot\FQ(v))\setminus \FQ(v)) \\
        & = & \FQ(\binding'(u'))\cupdot Q''.                   
      \end{array}
      \]
      Hence $\bang \Delta; \FQ(\binding(u'))\entails
      [C',\binding'(u')] :U,(Q'|Q'')$ is valid.
    \item $\rul{\rev}$: the reduction rule is
      \[
      \infer[\rul{\rev}]{[C,\rev\,(t,D,t')]\to [C,(t',D^{-1},t)]}{}
      \]
      with $a=\rev\,(t,D,t')$ and $a'=(t',D^{-1},t)$. The typing
      derivation $\pi$ is therefore
      \begin{footnotesize}
        \[
        \infer[]{\bang \Delta;\emptyset\entails
          \rev\,(t,D,t'):\bang^n\Circ(U,T)}{ \infer[]{\bang
            \Delta;\emptyset \entails \rev:\Circ(T,U)\loli
            \bang^n\Circ(U,T)}{ } & \infer[]{\bang \Delta
            ;\emptyset\entails (t,D,t'):\Circ(T,U)}{
            \deduce[]{\FQ(t)\entails t:T}{ \vdots ~\pi_1 } &
            \deduce[]{\bang \Delta;\FQ(t')\entails t':U}{ \vdots
              ~\pi_2 } } }
        \]
      \end{footnotesize}
      with $\Inn(D)=\FQ(t)$, $\Outt(D)=\FQ(t')$ and
      \[
      \bang \Delta;\emptyset\entails
      \rev\,(t,D,t'):\bang^n\Circ(T,U),(Q'|Q'')
      \]
      is valid. Now note that since $t'$ is a quantum data term, it
      contains no variables. Applying
      Lemma~\rref{prop:type_syst}{unused_var} to $\pi_2$ repeatedly we
      therefore get a derivation $\pi_2'$ of $\FQ(t')\entails
      t':U$. Moreover, by applying
      Lemma~\rref{prop:type_syst}{weakening} to $\pi_1$ we get a
      typing derivation $\pi_1'$ of $\bang \Delta ,\FQ(t)\entails
      t:T$.  Since, by Remark~\ref{structure-transfer}, we have
      $\Outt(D^{-1})=\Inn(D)=t$ and $\Inn(D^{-1})=\Outt(D)=t'$, we can
      construct the following typing derivation:
      \[
      \infer[]{\bang \Delta ;\emptyset\entails
        (t',D^{-1},t):\Circ(U,T)}{ \deduce[]{\FQ(t')\entails t':U}{
          \vdots ~\pi_2' } & \deduce[]{\bang \Delta;\FQ(t)\entails
          t:T}{ \vdots ~\pi_1' } &
        \deduce[]{\Inn(D^{-1})=\FQ(t')}{\Outt(D^{-1})=\FQ(t') } }
      \]
      Hence $\bang \Delta;\emptyset\entails
      (t',D,t):\bang^n\Circ(T,U),(Q'|Q'')$ is valid. \qedhere
    \end{itemize}
  \end{description}
\end{proof}

\begin{corollary}
  If $\Gamma;\FQ(a)\entails [C,a]:A,(Q'|Q'')$ is a valid typed closure
  and $[C,a]\to^* [C',a']$, then $\Gamma;\FQ(a')\entails
  [C',a']:A,(Q'|Q'')$ is also a valid typed closure.
\end{corollary}

\begin{proof}
  By induction on the length of the reduction sequence. The base case
  is provided by Theorem~\ref{thm-subject-red}.
\end{proof}

The above formulation of Subject Reduction explains why a typed
closure contains information about the input and output wires of the
circuit state. Indeed, Subject Reduction now guarantees (1) that the
input wires of a circuit remain unchanged through reduction and (2)
that a term can only affect wires whose identifiers are among its
quantum variables.

% ====================================================================
\section{Progress}
\label{sect:pq-progress}

We now prove that Proto-Quipper enjoys the progress property.

\begin{theorem}[Progress]
  If $\FQ(a)\entails [C,a]:A,(Q'|Q'')$ is a valid typed closure then
  either $a\in\Val$ or there exists a closure $[C',a']$ such that
  $[C,a]\to [C',a']$.
\end{theorem}

First, note that the Progress property is stated for a typed closure
whose typing context is empty. This is because the property is not
expected to hold if we allow for a non-empty typing context. Indeed,
it is easy to see that there are well-typed, non-closed closures such
as $[C,xy]$, which are neither values nor reduce. We now prove the
theorem.

\begin{proof}
  We prove the theorem by induction on the typing derivation $\pi$ of
  $\FQ(a)\entails a:A$. If $a$ is a value then there is nothing to
  prove. If $a$ is not a value, then by Lemma~\ref{non_values} there
  are 6 cases to consider.  In each case we show that $[C,a]$ is
  reducible in the sense that there exists a closure $[C,b]$ such that
  $[C,a]\to[C,b]$
  \begin{enumerate}
  \item If $a=(t,D,a')$ with $a'\notin \Val$, then the typing
    derivation $\pi$ is:
    \[
    \infer[.]{\emptyset\entails (t,D,a'):\Circ(T,U)}{
      \deduce[]{\FQ(t)\entails t:T}{ \vdots ~\pi_1 } &
      \deduce[]{\FQ(a')\entails a':U}{ \vdots ~\pi_2 } &
      \deduce[]{\Inn(D)=\FQ(t)}{ \Outt(D)=\FQ(a') } }
    \]
    The typed closure
    \[
    \begin{array}{rcl}
      \FQ(a') & \entails & [D,a']:U,(\FQ(t)|\emptyset)
    \end{array}
    \]
    is therefore valid. Since $a'$ is not a value, the induction
    hypothesis implies that there exists $a''$ such that $[D,a']\to
    [D',a'']$ and $[C,(t,D,a')]$ therefore reduces to $[C,(t,D',a'')]$
    by the $\rul{circ}$ reduction rule.
  \item If $a=\p{a_1,a_2}$ with $a_1\notin \Val$ or $a_2\notin \Val$,
    then the typing derivation $\pi$ is:
    \[
    \infer[.]{\FQ(a_1),\FQ(a_2)\entails \p{a_1,a_2}:\bang^n(A_1\X
      A_2)}{ \deduce[]{\FQ(a_1)\entails a_1:\bang^nA_1}{ \vdots~\pi_1
      } & \deduce[]{\FQ(a_2)\entails a_2:\bang^nA_2}{ \vdots~\pi_2 } }
    \]
    The typed closures
    \[
    \begin{array}{rcl}
      \FQ(a_1) & \entails & [C,a_1]:\bang^nA_1,(Q'|\FQ(a_2),Q'')\\
      \FQ(a_2) & \entails & [C,a_2]:\bang^nA_1, (Q'|\FQ(a_1),Q'')
    \end{array}
    \]
    are therefore both valid. Now if $a_2\notin\Val$, then by the
    induction hypothesis $[C,a_2]\to[C',a_2']$. Hence
    $[C,\p{a_1,a_2}]$ reduces to $[C',\p{a_1,a_2'}]$ by the
    $\rul{right}$ reduction rule. If on the other hand $a_2\in\Val$,
    then it must be the case that $a_1\notin\Val$ and we can conclude
    by reasoning analogously that $[C,\p{a_1,a_2}]$ reduces to some
    $[C',\p{a_1',a_2}]$ by the $\rul{left}$ reduction rule.
  \item If $a=\Ifthenelse{a_1}{a_2}{a_3}$, then the typing derivation
    $\pi$ is:
    \[
    \infer[.]{\FQ(a_1),Q\entails \Ifthenelse{a_1}{a_2}{a_3}:A}{
      \deduce[]{\FQ(a_1)\entails a_1:\bool }{ \vdots~\pi_1 } &
      \deduce[]{Q \entails a_2:A}{ \vdots~\pi_2 } & \deduce[]{Q
        \entails a_3:A}{ \vdots~\pi_3 } }
    \]
    The typed closure
    \[
    \begin{array}{rcl}
      \FQ(a_1) & \entails & [C,a_1]:\bool,(Q'|\FQ(a_2),\FQ(a_3),Q'')
    \end{array}
    \]
    is therefore valid. Now if $a_1\notin\Val$, then by the induction
    hypothesis $[C,a_1]\to[C',a_1']$ and thus
    $[C,\Ifthenelse{a_1}{a_2}{a_3}]$ reduces to
    $[C',\Ifthenelse{a_1'}{a_2}{a_3}]$ by the $\rul{cond}$ reduction
    rule.  If on the other hand $a_1\in\Val$, then either $a_1=\true$
    or $a_1=\false$, by Lemma~\ref{form_values}. Thus
    $[C,\Ifthenelse{a_1}{a_2}{a_3}]$ reduces either to $[C,a_2]$ by
    the $\rul{ifT}$ reduction rule, or to $[C,a_3]$ by the $\rul{ifF}$
    reduction rule.
  \item If $a=(\letin{*}{a_1}{a_2})$, then we can reason as above to
    show that if $a_1$ is not a value, then the $\rul{let*}$
    congruence rule applies, and that if $a_1$ is a value then
    Lemma~\ref{form_values} guarantees that the $\rul{*}$ rule
    applies.
  \item If $a=(\letin{\p{x,y}}{a_1}{a_2})$, then we can reason as
    above to show that if $a_1$ is not a value, then the $\rul{let}$
    congruence rule applies and that if $a_1$ is a value then
    Lemma~\ref{form_values} guarantees that the $\rul{pair}$ rule
    applies.
  \item If $a=a_1a_2$ then the typing derivation $\pi$ is:
    \[
    \infer[.]{\FQ(a_1),\FQ(a_2)\entails a_1a_2:A}{
      \deduce[]{\FQ(a_1)\entails a_1:B\loli A}{ \vdots~\pi_1 } &
      \deduce[]{\FQ(a_2)\entails a_2:B}{ \vdots~\pi_2 } }
    \]
    The typed closures
    \begin{eqnarray}
      \FQ(a_1) & \entails & [C,a_1]:B\loli A,(Q'|\FQ(a_2),Q'')\\
      \FQ(a_2) & \entails & [C,a_2]:B, (Q'|\FQ(a_1),Q'')\label{eqn-fqa2}
    \end{eqnarray}
    are therefore valid. There are three cases to treat.
    \begin{itemize}
    \item If $a_1\notin\val$, then $[C,a_1a_2]\to[C',a_1'a_2]$ by the
      induction hypothesis and the $\rul{fun}$ rule.
    \item If $a_1\in\val$ and $a_2\notin\Val$, then
      $[C,a_1a_2]\to[C',a_1a_2']$ by the induction hypothesis and the
      $\rul{arg}$ rule.
    \item If $a_1,a_2\in\val$ then by Lemma~\ref{form_values}, $a_1$
      is either an abstraction, a constant, or of the form $\unbox\,
      (t,C,u)$. If $a_1$ is a lambda abstraction, then $[C, a_1a_2]$
      reduces by the $\rul{\beta}$ rule. If $a_1$ is a constant, then
      it cannot be $\unbox$, since $a_1a_2$ would then be a value. If
      $a_1=\rev$, then $a_2$ is a value of type $\Circ(T,U)$. Hence
      $a_2$ is of the form $(t,C,u)$ by Lemma~\ref{form_values}, so
      that $[C, a_1a_2]$ reduces by the $\rul{\rev}$ rule. If
      $a_1=\boxx^T$, then $[C, a_1a_2]$ reduces by the $\rul{\boxx}$
      rule.

      Remains to treat the case of $a_1=\unbox\,(u,D,t)$. For the
      $\rul{\unbox}$ rule to apply, we need to show that
      $\bind(a_2,u)$ is well-defined, and that
      $\FQ(t)\seq\dom(\binding')$, where
      $\Append(C,D,\binding)=(C',\binding')$. The typing derivation
      $\pi_1$ is the following:
      \begin{footnotesize}
        \[
        \infer[]{\bang \Delta;\emptyset\entails \unbox\,(u,D,t):T\loli
          U}{ \infer[]{\bang \Delta;\emptyset \entails
            \unbox:\Circ(T,U)\loli (T\loli U)}{ } & \infer[]{\bang
            \Delta ;\emptyset\entails (u,D,t):\Circ(T,U)}{
            \deduce[]{\FQ(u)\entails u:T}{ \vdots ~\pi_1^1 } &
            \deduce[]{\bang \Delta;\FQ(t)\entails t:U}{ \vdots
              ~\pi_1^2 } } }
        \]
      \end{footnotesize}
      with $\Inn(D)=\FQ(u)$, $\Outt(D)=\FQ(t)$, $A=U$ and $B=T$ for
      some quantum data types $T,U$. It follows that $a_2$ and $u$ are
      two values of type $T$, so that, by Lemma~\ref{bind_qd_terms},
      $\binding=\bind(a_2,u)$ is defined with
      $\dom(\binding)=\FQ(a_2)$ and $\cod(\binding)=\FQ(u)$. Moreover,
      $\FQ(a_2)\seq\Outt(C)$ by the validity of the typed closure
      {\eqref{eqn-fqa2}}, and $\FQ(u)=\Inn(D)$ as noted
      above. Therefore, $\dom(\binding)\seq\Outt(C)$ and
      $\cod(\binding)=\Inn(D)$ hold, as required by
      Definition~\rref{circuit_constructor}{Append_cond_2}. By
      Definition~\rref{circuit_constructor}{Append_cond_2}, we
      conclude that $\dom(\binding')=\Outt(D) =\FQ(t)$, so that the
      $(\unbox)$ rule in fact applies. \qedhere
    \end{itemize}
  \end{enumerate}
\end{proof}

% =====================================================================
\chapter{Conclusion}
\label{chap:conc}

In this thesis, we applied tools from algebraic number theory and
mathematical logic to problems in the theory of quantum
computation. We described algorithms to solve the problem of
approximate synthesis of special unitaries over the Clifford+$V$ and
Clifford+$T$ gate sets. We also defined a typed lambda calculus for
quantum computation called Proto-Quipper which serves as a
mathematical foundation for the Quipper quantum programming
language. In conclusion, we briefly describe some avenues for future
research.

% ====================================================================
\section{Approximate synthesis}
\label{sect:conc-synth}

The synthesis methods described in chapters~\ref{chap:synth-V} and
\ref{chap:synth-T} belong to a very recent family of number-theoretic
algorithms. Many generalization of these methods can be considered.
\begin{itemize}
\item The algorithms described in chapters~\ref{chap:synth-V} and
  \ref{chap:synth-T} are only optimal for $z$-rotations. While Euler
  angle decompositions can be used to extend these methods to
  arbitrary special unitaries, optimality is lost in the process. A
  first potential generalization of the methods of
  chapters~\ref{chap:synth-V} and \ref{chap:synth-T} is to define
  optimal number-theoretic synthesis algorithms for arbitrary special
  unitaries.
\item A second restriction of the decomposition methods presented in
  chapters~\ref{chap:synth-V} and \ref{chap:synth-T} is that they are
  only defined for specific gate sets, namely the Clifford+$V$ and
  Clifford+$T$ gate sets. Another future generalization of this work
  is to extend the number-theoretic synthesis methods to different
  gate sets. This line of enquiry has already been pursued in the
  recent literature with encouraging results. In particular, it is
  known that asymptotically optimal number-theoretic decomposition
  methods can be defined for certain gate sets based on anyonic
  braidings (see \cite{KBS2013} and \cite{BCKW2015}). Further, exact
  synthesis methods have recently been devised for a relatively
  general family of gate sets (\cite{FGKM15} and \cite{KY15}). While
  these exact synthesis methods have not yet been extended to
  approximate synthesis algorithms, we expect that, at least in some
  cases, this extension should carry through with relative ease.
\item A further possible generalization is to consider unitary groups
  in higher dimensions. Higher-dimensional versions of the Clifford
  gates have previously been studied in the literature
  \cite{Gottesman98b}. Moreover, higher-dimensional analogues of the
  $T$ gate were recently introduced \cite{HowardVala}. Together, these
  define a higher-dimensional Clifford+$T$ gate set which stands as a
  natural candidate for an adaptation of the methods of
  chapters~\ref{chap:synth-V} and \ref{chap:synth-T}.
\end{itemize}
We note that the decomposition algorithm of Fowler \cite{Fowler04} as
well as the Solovay-Kitaev algorithm \cite{Dawson-Nielsen} are both
very general. Indeed, both algorithms allow the synthesis of arbitrary
special unitaries over any gate set and in any dimension. Since the
algorithms presented in chapters~\ref{chap:synth-V} and
\ref{chap:synth-T} rely on specific properties of the rings of
algebraic integers associated with the Clifford+$V$ and Clifford+$T$
gate set, there is no reason for these methods to generalize to
arbitrary gate sets. However, it might be possible to identify a
general class of gate sets to which these methods apply.

Another interesting avenue of future research lies in a modification
of the statement of the synthesis problem itself, by allowing a
broader notion of circuit. As a first such generalization, one can
introduce \emph{ancillary} qubits in the approximating
circuit. Suppose a special unitary $U$ and a precision $\epsilon>0$
are given. Instead of searching for $W\in\uset(2)$ such that
$\norm{U-W}<\epsilon$ one can look for $W\in\uset(2^{n+1})$ such that
for any state $\ket{\phi}$ we have
\[
\norm{U(\ket{\phi})\ket{0\ldots 0} - W(\ket{\phi}\ket{0\ldots
    0})}<\epsilon.
\]
In other words, unitaries acting on more than one qubit can be
considered, provided that they return the additional qubits nearly
unchanged. The advantage of such a generalized notion of circuit is
that it opens the door to a certain form of parallelization. In
particular, even if the number of non-Clifford gates in the
approximating circuit remains unchanged, applying them in parallel,
rather than sequentially, may represent a gain.

Another generalization of the synthesis problem is to allow the
approximating circuits to use measurements or other adaptive
methods. It is known that using such methods can decrease the gate
count below the information-theoretic lower bound. These methods are
relatively new, but very promising results have already been achieved
(\cite{WiebeKli2013}, \cite{PaetSvor2013}, \cite{BRS2014}, and
\cite{WiebeRoet2014}).

% ====================================================================
\section{Proto-Quipper}
\label{sect:conc-pq}

As already mentioned in Chapter~\ref{chap:intro}, the rationale behind
the design of Proto-Quipper was to start with the simplest language
possible, establish type-safety, and then extend the language in small
steps with the goal of eventually adding most of Quipper's features to
Proto-Quipper in a type-safe way. This defines a natural set of
problems for future work. Many such extensions are conceivable, but we
only describe a few here.

\begin{itemize}
\item In the current version of Proto-Quipper, all circuits are
  reversible. This follows from the definition of the $(\rev)$ and
  $\rul{circ}$ typing rules and will have to be modified to
  accommodate non-reversible gates such as measurements. In such a
  setting, the type system should ensure that circuits are reversed
  only if it is meaningful to do so. In particular, if a circuit
  contains a measurement, then it should not be possible to reverse
  it.
\item A circuit generating function that inputs a list of qubits does
  not define just one circuit, but rather a family of circuits
  parameterized by the length $n$ of the list.  To box such a
  function, a particular value of $n$ has to be given. In Quipper, we
  refer to $n$ as the ``shape'' of the argument of the
  function. Operations such as boxing and reversing often require
  shape information. An alternative solution would be to equip
  Proto-Quipper with a {\em dependent type system}. This would allow
  shape information to be stored at the type level.
\item In contrast to the quantum lambda calculus, the reduction in
  Proto-Quipper is non-probabilistic. Of course, the hypothetical
  quantum device running the circuit produced by Proto-Quipper would
  have to perform probabilistic operations, but the circuit generation
  itself does not have to. Even if the language is extended with
  measurement gates, it is still possible to generate the circuits
  deterministically. This is justified by the ``principle of deferred
  measurement" which states that any quantum circuit is equivalent to
  one where all measurements are performed as the very last operations
  (see, e.g., \cite{Nielsen-Chuang} p.186). We therefore do not need
  to rely on the result of a measurement to construct circuits and, in
  theory, no computational power is lost by making this assumption. In
  practice, however, this delaying of measurement may significantly
  increase the size of the circuit. Thus in terms of computational
  resources it is sometimes advantageous to permit circuit generating
  functions that access previous measurement results. Several existing
  quantum algorithms rely on such interactive circuit building. In
  Quipper, this capability is captured by the notion of \emph{dynamic
    lifting}. Adding such a feature to Proto-Quipper would make the
  reduction relation probabilistic. It is an interesting research
  problem how such an extension can be carried out in a type-safe way.
\end{itemize}

% =====================================================================
\bibliographystyle{plain} \bibliography{thesis}

\end{document}